\documentclass[12pt]{amsart}
\usepackage{amsfonts}
\usepackage{changepage}
\usepackage{natbib}
\usepackage{eurosym}
\usepackage{amssymb, amsmath, url,color,  amsthm, amssymb,graphicx,multirow, appendix}
\usepackage[margin=0.8in]{geometry}
\usepackage{bigints}
\usepackage[scr=boondoxo,scrscaled=1.05]{mathalfa}
\usepackage{lscape,amssymb, amsmath,verbatim,graphicx}
\usepackage{epsfig,float}
\usepackage{graphicx}
\usepackage{epsfig}
\usepackage{threeparttable}
\usepackage{ragged2e}
\usepackage{epstopdf}
\usepackage{caption}
\usepackage{booktabs}
\usepackage{float}
\usepackage{multirow}
\usepackage{adjustbox}
\usepackage{lscape,lipsum}
\usepackage{siunitx}
\usepackage{url}
\usepackage[T1]{fontenc}
\usepackage{a4wide, url}
\usepackage[english]{babel}
\usepackage{graphicx}
\usepackage{array}
\usepackage{multirow}
\usepackage{amsthm}
\usepackage{amsmath}
\usepackage[foot]{amsaddr}
\usepackage{amsfonts}
\usepackage{amssymb}
\usepackage{pgf}
\usepackage{paralist}
\usepackage{pdflscape}
\usepackage{lscape,amssymb, amsmath,verbatim,graphicx}
\usepackage{graphicx}
\usepackage{epsfig}
\usepackage{amssymb}
\usepackage{amsmath}
\usepackage{amsfonts}
\usepackage{geometry}
\usepackage{scalefnt}
\usepackage{graphicx}
\usepackage{caption}
\usepackage{setspace}
\usepackage{array}
\usepackage{multirow}
\usepackage{amsthm}
\usepackage{amsmath}
\usepackage{amsfonts}
\usepackage{amssymb}
\usepackage{pgf}
\usepackage{paralist}
\usepackage{pdflscape}
\usepackage{rotating}
\usepackage{subcaption}
\usepackage{scalefnt}
\usepackage{xr}
\usepackage{hyperref}
\usepackage{algpseudocode}
\usepackage{algorithm}
\usepackage{ragged2e}

\DeclareUnicodeCharacter{2060}{\nolinebreak}
\allowdisplaybreaks
\DeclareFontFamily{OT1}{pzc}{}
\DeclareFontShape{OT1}{pzc}{m}{it}{<-> s * [1.10] pzcmi7t}{}
\DeclareMathAlphabet{\mathpzc}{OT1}{pzc}{m}{it}

\newtheoremstyle{exampstyle}
{4pt} 
{4pt} 
{\itshape} 
{} 
{\bfseries} 
{.} 
{.75em} 
{} 
\theoremstyle{exampstyle}
\newtheorem{remark}{Remark}[section]

\numberwithin{table}{section}
\numberwithin{figure}{section}
\newtheorem{theorem}{Theorem}[section]
\newtheorem{lemma}{Lemma}[section]

\newtheorem{assumption}{Assumption}[section]
\newtheorem{definition}{Definition}[section]

\def\beq{\begin{equation}}
\def\eeq{\end{equation}}
\def\bals{\begin{align*}}
\def\eals{\end{align*}}
\def\bal{\begin{align}}
\def\eal{\end{align}}

\numberwithin{equation}{section}
\numberwithin{theorem}{section}
\numberwithin{corollary}{section}
\allowdisplaybreaks
\input{tcilatex}

\begin{document}

\title[Testing for alpha]{A general randomized test for Alpha}
\author{ Daniele Massacci\textsuperscript{1}, Lucio Sarno\textsuperscript{2}%
, Lorenzo Trapani\textsuperscript{3}$^{,}$\textsuperscript{4}, Pierluigi Vallarino\textsuperscript{5} 
}
\address{\textsuperscript{1}King's Business School, King's College London, UK%
}
\address{\textsuperscript{2}Cambridge Judge Business School and Girton
College, University of Cambridge, UK; and Centre for Economic Policy
Research (CEPR)}
\address{\textsuperscript{3}Universita' di Pavia, Italy}
\address{\textsuperscript{4}University of
Leicester, UK}
\address{\textsuperscript{5}Institute of Finance, Universita' della
Svizzera Italiana (USI Lugano), Switzerland}
\thanks{\textbf{Acknowledgments.} We thank Marine Carrasco, Marcelo Fernandes, Christian Francq,
Yingying Li, Oliver Linton, Roberto Ren\`{o}, Fabio Trojani, Dick van Dijk, Kevin Schneider,
Dacheng Xiu, Chen Zhou and the participants to: the 2nd
FinML conference; the 2025 Annual SoFiE conference; Erasmus University
Rotterdam; ESSEC Business School; LUISS Guido Carli; the 2025 EC$^2$
conference;  CREST-ENSAE; Toulouse Business School; University of Rome Tor Vergata. The usual disclaimer applies.}
\maketitle


\begin{adjustwidth}{-10pt}{-10pt}
    
\begin{abstract}
We propose a methodology to test for the null hypothesis that the alphas of a panel of asset returns are jointly equal to zero in a linear factor pricing model with observable and tradable factors --- that is, the
null of \textquotedblleft zero alpha{\textquotedblright}. The test is based on equation-by-equation estimation, using a \textit{randomized} version of the estimated alphas, which only requires rates of convergence. The distinct features of the proposed methodology are that it does not require the estimation of any covariance matrix, and that it allows for both $N$ and $T$ to pass to infinity, with the former possibly faster than the latter. Further, unlike extant approaches, the procedure can accommodate conditional heteroskedasticity, non-Gaussianity, and strong cross-sectional dependence in the error terms. We also propose a derandomized decision rule to choose in favor or against the correct specification of a linear factor pricing model. Monte Carlo simulations show that the test has satisfactory properties and it compares favorably
to several existing tests. The usefulness of the testing procedure is illustrated through an application of linear factor pricing models to the constituents of the S{\&}P 500.
\end{abstract}

\end{adjustwidth}

\doublespacing

\section{Introduction\label{intro}}

Testing for \textquotedblleft alpha\textquotedblright\ -
i.e. the component of expected returns that cannot be explained by a linear factor model - is central in asset pricing, and yet
poses a number of issues. Available testing procedures are often marred by
low power or poor size control, hinge on strong assumptions on the data
generating process (DGP), and generally involve the estimation or inversion
of (often large) covariance matrices \citep{giglio2022factor}. \newline
In this paper, we propose a general methodology to test for
\textquotedblleft zero alpha\textquotedblright\ - that is the null that for each test asset, on average, realized excess returns  equal those implied by a linear factor pricing model. In order to make the presentation easier to follow, we mainly
focus on the following linear factor pricing model with \textit{tradable}
and \textit{observable} factors 
\begin{equation}
y_{i,t}=\alpha _{i}+\beta _{i}^{\prime }f_{t}+u_{i,t},\text{ \ \ }1\leq
i\leq N,\text{ \ }1\leq t\leq T,  \label{eq:fac_model}
\end{equation}%
where: $y_{i,t}$ is the excess return on the $i$-th security at time $t$; $%
f_{t}$ is a $K$-dimensional vector of pricing factors; $%
\beta _{i}$ is a $K$-dimensional vector of factor loadings; and $u_{i,t}$ is
a zero mean error term. Equation (\ref{eq:fac_model}) is the workhorse model
employed in asset pricing, and it encompasses several popular specifications
such as the Capital Asset Pricing Model (CAPM; see \citealp{sharpe1964capital}), and the three- and five-factors models (\citealp{fama1993common}; and \citealp{fama2015five}).\footnote{In Sections \ref{nontradable} and \ref{latent} in the Supplement, we study the extension to the cases of non-tradable and of latent factors respectively.} In the context of (\ref%
{eq:fac_model}), we propose a novel approach to test for the null hypothesis
that all the $\alpha _{i}$s are jointly equal to zero versus the alternative
that at least one $\alpha _{i}$ is nonzero, viz. 
\begin{eqnarray}
&\mathbb{H}_{0}\, :\,\alpha _{i}=0 & \text{   for all \ \ }1\leq i\leq N,
\label{h0} \\
&\mathbb{H}_{A}\, :\,\alpha _{i}\neq 0 & \text{   for at least one }1\leq
i\leq N.  \label{hA}
\end{eqnarray}%
We base our analysis on the equivalent version of (\ref{h0})%
\begin{equation}
\mathbb{H}_{0}\,:\,\max_{1\leq i\leq N}\text{ }\left\vert \alpha
_{i}\right\vert =0.  \label{null}
\end{equation}%
Testing for (\ref{null}) is interesting \textit{per se}. In the
context of asset pricing, (\ref{eq:fac_model}) implies $\mathbb{E}%
\left( y_{i,t}\right) =\alpha _{i}+\beta _{i}^{\prime }\mathbb{E}\left(
f_{t}\right) $; hence, $\alpha _{i}$ represents the excess return on the $i$-th test asset not explained by the $K$ factors. When no relevant pricing factor is omitted from \eqref{eq:fac_model}, a non-zero $\alpha_i$ can still appear due, for example, to market rigidities that prevent security prices from reaching their theoretical fair values and that are not allowed for in the factor pricing  model. As these securities are misspriced, this $\alpha_{i}$ is often called \emph{pricing error}. However, a non-zero $\alpha _{i}$ can originate from multiple other sources of model misspecification, such as the omission of one or more pricing factors and/or non-linear pricing relations. Hence, depending on the nature of model misspecification, $\alpha _{i}\neq0$ is not necessarily a pure pricing error but rather a signal that the model in (1.1) is misspecified from an asset pricing perspective.\footnote{In the case of testing for alpha in fund returns, rather than security returns, alpha can also capture skill of the fund manager  \citep{giglio2021thousands}.}

Arguably, the first contribution to propose a test for \textquotedblleft
zero alpha\textquotedblright\ in linear pricing models is the
paper by \citet[][GRS henceforth]{gibbons1989test}, where an $F$-test is
proposed for the joint null hypothesis that $\alpha _{1}=\alpha
_{2}=...=\alpha _{N}=0$. Such an approach is entirely natural, but - as well as needing several restrictions on the error terms $u_{i,t}$ - it requires the restriction that $N$ is fixed with $%
N<T$, because the test uses a (consistent) estimate of the $N\times N$
covariance matrix of the regression residuals, which subsequently needs to
be inverted. Hence, the case $N>T$, often encountered in
applied work, cannot be addressed with the GRS test. \citet{fan2015power}, %
\citet{gagliardini2016time} and \citet{pesaran2023testing} propose several
solutions towards this technical difficulty by developing average-type alpha
tests - i.e.\ tests where individual statistics on each of the $\alpha _{i}$%
s are averaged - for the joint
asymptotics case $\min \left\{ N,T\right\}
\rightarrow \infty $. All these tests can deal with cross-sectional
dependence among the errors (in essence, allowing for
cross-sectional dependence as long as a Central Limit Theory holds for
cross-sectional averages). However, their implementation still hinges on
estimating a large dimensional, $N\times N$ covariance matrix, which
requires several restrictions on the covariance structure of the errors, and
which can become computationally intensive when certain estimators are
considered (this is e.g.\ the case for the threshold estimator of %
\citealp{bickel2008regularized}, used by \citealp{gagliardini2016time}). Further, all the tests mentioned above are based on the
assumption of serially independent errors, which is bound to cause problems
in the presence of conditional heteroskedasticity, and the approaches by %
\citet{fan2015power} and \citet{gagliardini2016time} further require \textit{%
Gaussian} errors. In a recent contribution, \citet{feng2022high} propose a
max-type test for the null that $\max_{1\leq i\leq N}\left\vert \alpha
_{i}\right\vert =0$; their test is based on using the maximal \textit{t}-statistic, and therefore its asymptotics hinges on an
Extreme-Value-type argument, rather than a \textquotedblleft
central\textquotedblright\ one like average-type tests. As we also discuss in greater detail in Section \ref{model}, their
asymptotics requires \textit{weak} cross-sectional dependence (and serial independence). \citet{ardia2024robust} test the null hypothesis in (\ref{h0}) whether all alphas are zero by combining $p$-values from asset-specific
tests using the Cauchy combination approach proposed 
by \citet{liu2020cauchy}. While their test can be applied under quite
general assumptions, its implementation hinges on several tuning parameters
which significantly impact the outcome of the procedure. Finally, %
\citet{chernov2025test} directly generalize the GRS test to the
high-dimensional case by considering a ridge-regularized estimator of the $%
N\times N$ covariance matrix of the residuals, still resting on similar assumptions to those of the average-type tests discussed above.\footnote{%
Other relevant references include \citet{gungor2016multivariate}, %
\citet{ma2020testing}, and \citet{raponi2020testing}. Further, the correct
specification of linear factor pricing models can also be assessed by
testing whether they imply a pricing kernel with zero Hansen-Jagannathan
distance (see e.g., \citealp{hansen1997assessing}, %
\citealp{hodrick2001evaluating}, and \citealp{carrasco2022hansen}).}

\smallskip

\textit{Testing methodology and the contribution of this paper}

\smallskip

In this paper, we fill the gaps mentioned above, by proposing a test for (%
\ref{null}) which allows for the errors to have: 
\textit{(i)} (weak) serial dependence, including conditional
heteroskedasticity such as volatility clustering, or the
\textquotedblleft leverage effect\textquotedblright\ (%
\citealp{black1976studies}); \textit{(ii)} non-Gaussianity, such as
skewness and excess kurtosis, indeed relaxing some of the moment assumptions in the aforementioned papers; \textit{(iii)} strong cross-sectional dependence;
and also \textit{(iv)} allowing for $N>T$. Importantly, as we further explain in Section \ref{comparison}, our test statistic
only requires a consistent estimator (and its rate of convergence) for the individual $\alpha _{i}$s, whereas no second order property (such as asymptotic Gaussianity, or asymptotic efficiency) is needed. This allows to ignore the cross-sectional structure of the data, and therefore the test
does not require the estimation of any $N\times N$ covariance matrices, and
needs virtually no tuning. \\
The technical details are described in the next sections; here we
offer a preview of the main arguments, which are based on the construction
of a \textit{randomized} test statistic. In order to construct the test
statistic, we estimate the $\alpha _{i}$s from $N$ separate time series
regressions, unit by unit: at no stage do we require joint estimation.
Although any consistent estimator can be employed, here we use OLS,
obtaining, say, $\widehat{\alpha }_{i}$. Under our assumptions, it can be
expected that $\widehat{\alpha }_{i}$ will converge to $\alpha _{i}$ at a
rate $T^{-1/2}$. We then pre-multiply each $\left\vert \widehat{\alpha }%
_{i}\right\vert $ by a function of $T$ which diverges as $T\rightarrow
\infty $, but at a rate slower than $O\left( T^{1/2}\right)$. Hence, we
obtain $N$ statistics which drift to zero when $\alpha
_{i}=0$, and diverge to positive infinity whenever $\alpha _{i}\neq 0$. We
then perturb the resulting $N$ statistics by adding to each of them a $%
\mathcal{N}\left( 0,1\right) $ shock, with the $N$ shocks forming an \textit{%
i.i.d.} sequence. As a consequence, we obtain an $N$-dimensional sequence
which, under $\mathbb{H}_{0}$, is (roughly) \textit{i.i.d.}$\mathcal{N}%
\left( 0,1\right) $ conditionally on the sample. Finally, we take the
largest of these perturbed statistics as a test statistic for $\mathbb{H}_{0}
$ in (\ref{null}): conditionally on the sample, this
is distributed as a Gumbel under $\mathbb{H}_{0}$, whereas it diverges under the alternative that
at least one asset is mis-priced. This methodology has, at least
conceptually, some similarities with the one proposed in \citet{fan2015power}%
, where a sequence is added to a test statistic, constructed so as to drift
to zero under the null (thus introducing no distortion in the asymptotics
under the null), and diverging under the alternative (so as to boost the
power). However, unlike \citet{fan2015power}, we do not need to estimate at
any stage the asymptotic covariance matrix of the estimated $\mathbf{\alpha }%
=\left( \alpha _{1},...,\alpha _{N}\right) ^{\prime }$, or define a
high-dimensional weight matrix: only the individual, unit-by-unit, estimates
of the $\alpha _{i}$s are required. Moreover, we only require a rate of
convergence for the estimate $\widehat{\alpha }_{i}$. Hence, the assumptions
on serial dependence and moment existence can be relatively mild.

\smallskip

To the best of our knowledge, our contribution  is the first application to asset pricing of
tests based on randomizing a function of the data, i.e.\ a statistic. While novel to this setting, this type of randomized tests have been
used in econometrics and statistics, particularly where a limiting
distribution is unavailable or non-pivotal, or where its derivation requires
excessively restrictive assumptions; although a comprehensive literature
review goes beyond the scope of this paper, we refer to \citet{corradi2006}
for a first  application of these randomized tests in econometrics, and
to the paper by \citet{he2023one} for references.\footnote{We emphasize that this paper is not, in general, the first alpha testing procedure based on a randomization. The literature has considered bootstrap-based tests for alpha; see, for example, \citet{sullivan1999data}, \citet{white2000reality}, \citet{kosowski2006can} and \citet{fama2010luck}. While  bootstrap and randomization are both based on an added source of randomness, their core statistical mechanisms are very different, and we extensively discuss similarities and differences between bootstrap and randomized tests in Section \ref{comparison}.
}

We make at least four contributions to the current literature. First, whilst
we focus on the specific case of \textquotedblleft testing for
alpha\textquotedblright , we propose a novel methodology to construct tests
involving a growing number of parameters with no need for joint estimation
and, in essence, no need to take the dimensionality of the problem into
account. In a similar spirit, although in (\ref{eq:fac_model}) we assume
that the common factors $f_{t}$ are the same across all units, our approach
can be readily extended e.g.\ to the case where pricing factors are
heterogeneous across assets. Indeed, the factor structure could be so
heterogeneous that no factor influences all the assets, which makes our
testing procedure robust to the presence of weak and semi-strong pricing
factors, and 
we do not need any factor to be \textquotedblleft strong\textquotedblright\
in our model.\footnote{%
Following \cite{chudik2011weak}, the $k$-th pricing factor is \textit{strong }if $\beta
_{i,k}\neq 0$ for all $i=1,\dots ,N$; and it is \textit{weak} (resp. \textit{semi-strong}) when $\beta
_{i,k}\neq 0$ for $\lfloor N^{\gamma }\rfloor $ of the assets with $0<\gamma
<1/2$ (resp. $1/2\leq \gamma <1)$. %
} Second, we relax several
technical conditions required in the extant literature, allowing for
serial and cross-sectional dependence, conditional
heteroskedasticity, thicker tails and a larger $N/T$ ratio than allowed for
in other contributions. Third, we enhance the randomized test by
proposing a decision rule which shows excellent control for Type I and Type
II errors in simulations. Fourth and last, our methodology is flexible and can be readily extended to other frameworks. 

\smallskip

The remainder of the paper is organized as follows. We discuss our set-up
and assumptions in Section \ref{model}. The hypotheses of interest and the
randomized testing approach are discussed in Section \ref{tests}; we offer a heuristic description of the core statistical mechanism of our methodology, and its relationship with other randomization methods, in Section \ref{comparison}; we report
the derandomized decision rule in Section \ref{derandom}; and, in Section %
\ref{practice}, we offer guidelines for the practical implementation of our
procedure. Simulations are in Section \ref{simulations}, while Section %
\ref{empirical} contains an empirical illustration. Conclusions and further
lines of research are in Section \ref{conclusion}. Further Monte Carlo and empirical
evidence, extensions, technical lemmas and proofs are in the Supplement.

NOTATION. We define the probability space $\left( \Omega ,\boldsymbol{B},%
\mathbb{P}\right) $, where $\Omega $ is the sample space with elements $%
\omega \in \Omega $, $\boldsymbol{B}$ the space of events, and $\mathbb{P}
$ the probability function. Given a random
variable $X$, $\mathbb{E}\left( X\right) $ is the mean and $\mathcal{V}%
\left( X\right) $ is the covariance viz. $\mathcal{V}\left( X\right) =%
\mathbb{E}\left[ \left( X-\mathbb{E}\left( X\right) \right) \left( X-\mathbb{%
E}\left( X\right) \right) ^{\prime }\right] $; further, given $r>0$, we
denote the $\mathcal{L}_{r}$-norm of $X$ as $\left\vert X\right\vert
_{r}=\left( \mathbb{E}\left\vert X\right\vert ^{r}\right) ^{1/r}$. The
indicator function of a set $\mathcal{A}$ is denoted as $\mathbb{I}\left( 
\mathcal{A}\right) $. We use: \textquotedblleft a.s\textquotedblright\ for
\textquotedblleft almost sure(ly)\textquotedblright ; \textquotedblleft $%
\rightarrow $\textquotedblright\ to denote the ordinary limit; and
\textquotedblleft $\overset{a.s.}{\rightarrow }$\textquotedblright\ to
denote almost sure convergence. Orders of magnitude for an almost surely
convergent sequence with multiparameter index $\Pi $, say $s_{\Pi }$ are
denoted as $O_{a.s.}\left( b_{\Pi }\right) $\ and $o_{a.s.}\left( b_{\Pi
}\right) $\ when, respectively, $\mathbb{P}\left( \limsup_{\Pi \rightarrow
\infty }\left\vert b_{\Pi }^{-1}s_{\Pi }\right\vert <\infty \right) $ $=$ $1$%
, and $b_{\Pi }^{-1}s_{\Pi }\overset{a.s.}{\rightarrow }0$, as $\Pi
\rightarrow \infty $. Positive, finite constants are denoted as $c_{0}$, $%
c_{1}$, ... and their value may change from line to line. Other, relevant
notation is introduced later on in the paper.

\section{Model and assumptions\label{model}}

Recall our workhorse model (\ref{eq:fac_model}) 
\begin{equation*}
y_{i,t}=\alpha _{i}+\beta _{i}^{\prime }f_{i,t}+u_{i,t}.
\end{equation*}%
We begin with a definition of weak dependence which we use throughout the
paper.

\begin{definition}
\label{bernoulli}The sequence $\left\{ m_{t},-\infty <t<\infty \right\} $
forms an $\mathcal{L}_{\nu }$-decomposable Bernoulli shift if and only if it
holds that $m_{t}=g\left( \eta _{t},\eta _{t-1},...\right) $, where: (i) $%
g:S^{\infty }\rightarrow \mathbb{R}^{k}$ is a non random measurable
function; (ii) $\left\{ \eta _{t},-\infty <t<\infty \right\} $ is an \textit{%
i.i.d.} sequence with values in a measurable space $S$; (iii) $\left\vert
m_{t}\right\vert _{\nu }<\infty $; and (iv) $\left\vert m_{t}-m_{t,\ell
}^{\ast }\right\vert _{\nu }\leq c_{0}\ell ^{-a}$, for some $c_{0}>0$ and $%
a>0$, where $m_{t,\ell }^{\ast }=g\left( \eta _{t},...,\eta _{t-\ell
+1},\eta _{t-\ell ,t,\ell }^{\ast },\eta _{t-\ell -1,t,\ell }^{\ast
}...\right) $, with $\left\{ \eta _{s,t,\ell }^{\ast },-\infty <s,\ell
,t<\infty \right\} $ \textit{i.i.d.} copies of $\eta _{0}$ independent of $%
\left\{ \eta _{t},-\infty <t<\infty \right\} $.
\end{definition}

Decomposable Bernoulli shifts (\citealp{ibragimov1962some}) are a convenient way to model stationary,
dependent time series. Virtually
all the most common DGPs in econometrics and statistics satisfy Definition \ref{bernoulli}. \citet{liu2009strong} provide various theoretical results, and numerous examples including ARMA-GARCH sequences, and other nonlinear time series models (e.g. Random Coefficient AutoRegressive and threshold models). 

\smallskip

We are now ready to present our assumptions.

\begin{assumption}
\label{error} For all $1\leq i\leq N$ and some $\nu \geq 4$, $\left\{
u_{i,t},-\infty <t<\infty \right\} $ is an $\mathcal{L}_{\nu }$-decomposable
Bernoulli shift, with $a>\left( \nu -1\right) /\left( \nu -2\right) $, $%
\mathbb{E}u_{i,t}=0$, and $\min_{1\leq i\leq N}\mathbb{E}u_{i,t}^{2}>0$.
\end{assumption}

\begin{assumption}
\label{regressor} For some $\nu \geq 4$, $\left\{ f_{t},-\infty <t<\infty
\right\} $ is an $\mathcal{L}_{\nu }$-decomposable, $K$-dimensional
Bernoulli shift, with $a>\left( \nu -1\right) /\left( \nu -2\right) $ and
positive definite covariance matrix $\mathcal{V}\left( f_{t}\right) $.
\end{assumption}

\begin{assumption}
\label{exogeneity} It holds that $\mathbb{E}\left( f_{t}u_{i,t}\right) =0$,
for all $1\leq i\leq N$.
\end{assumption}

Assumptions \ref{error} and \ref{regressor} allow for (weak) serial dependence in $u_{i,t}$ and $f_{t}$, e.g. due to nonlinear phenomena such as conditional heteroskedasticity. In
contrast, the tests by \citet{fan2015power}%
, \citet{gagliardini2016time}, \citet{feng2022high} and \citet{pesaran2023testing} all assume
independence over time of $u_{i,t}$, thus being unable to accommodate idiosyncratic
conditional heteroskedasticity in asset returns. The assumptions imply that the unconditional variances of errors and
pricing factors are constant over time, similarly to Assumptions A1 and A2
in \cite{feng2022high}. However, as we also discuss in Section \ref{comparison}, we can extend our set-up to the case of \textit{unconditional} heteroskedasticity, allowing for different regimes along
similar lines as Assumption 2.2 in \citet{horvath2025detecting}.\footnote{That is, assuming $%
\left\{ f_{t},1\leq t\leq T\right\} =\bigcup_{\ell =1}^{q}\left\{
f_{\ell ,t},t_{\ell -1}<t\leq t_{\ell }\right\} $ with $t_{0}=1$\ and $%
t_{q}=T$\, and $\left\{ f_{\ell ,t},-\infty <t<\infty \right\} $\ satisfying Assumption \ref{regressor} for each segment $1\leq \ell \leq q$\
(and similarly for $u_{i,t}$).} Our assumptions also require error terms to have four finite moments, as opposed to
the Gaussianity assumption in \citet{fan2015power}. Similarly,
Assumption \ref{regressor} is substantially milder than the sub-exponential tails
constraint on $f_t$ in \citet{fan2015power} and \citet{feng2022high}.\\
Assumption \ref{error} also allows for more general forms  of
cross-sectional dependence in the errors, compared with the extant literature. For instance, existing max-type tests such as the one by \citet{feng2022high}  require some restrictions on the covariance matrix of the error term.\footnote{In particular, Assumption 3 in \citet{feng2022high} implies that the covariance matrix of the error term be invertible and that cross-covariances are absolutely summable. In turn, this is required in order for the individual \textit{t}-statistics associated with the $\alpha_i$s to be weakly correlated, thus being able to apply standard Extreme Value Theory when deriving the asymptotic distribution of their maximum.} We do not impose any such requirement and, in principle, our methodology can also deal  with  strong cross-sectional dependence in $u_{i,t}$ as long as the maintained zero-mean assumption  is satisfied. However, in the context of asset pricing such a case requires attention. Indeed, a possible cause of cross-sectional dependence could be the omission of strong or semi-strong factors, which then show up in the error term $u_{i,t}$. If these omitted tradable factors are priced - thus having a non-zero risk-premium (mean) - then $\alpha_i$ captures both the pure pricing error \textit{and} the risk premium of the omitted factors, regardless of their correlation with the factors included in the model \citep{giglio2021thousands}. However, $\alpha_i$ remains a pure pricing error when there are  {\textquotedblleft time series factors\textquotedblright} in $u_{i,t}$, that is \textit{tradable} factors $g_t$ which have zero mean $\mathbb{E}\left(
g_{t}\right) =0$ -- and hence do not contribute to the pricing of the assets -- and are orthogonal to $f_t$, $\mathbb{E}\left(
f_{t}g_{t}^{\prime }\right) =0$. Omitting $g_t$ results in $u_{i,t}=\gamma _{i}^{\prime
}g_{t}+\xi _{i,t}$, where $\xi _{i,t}$ is a purely idiosyncratic component. This, in turn, renders all the existing procedures invalid, but it is allowed under our Assumptions   \ref{error} - \ref{exogeneity}.

Assumption \ref{exogeneity} is a weak exogeneity requirement, less restrictive than the independence assumption in %
\citet{fan2015power}, \citet{feng2022high}, and \citet{pesaran2023testing}.

\section{The test\label{tests}}

Recall the null hypothesis of (\ref{null}), i.e. $ 
\mathbb{H}_{0}:\max_{1\leq i\leq N}\left\vert \alpha _{i}\right\vert =0$. As mentioned in the introduction, \textit{any} consistent estimator of $\alpha_i$ could be employed to construct our test statistics. Here, we focus on the unit-by-unit OLS estimator,\footnote{We note that the $1 \leq i \leq N$ equations in (\ref{eq:fac_model}) share the same set of regressors; hence, the OLS estimator coincides with a system-based (Feasible) GLS estimation.} viz.
\begin{equation}
\,\widehat{\alpha }_{i,T}=\overline{y}_{i}-\widehat{\beta }_{i,T}^{\prime }%
\overline{f},  \label{alpha-hat}
\end{equation}%
for  $\overline{y}_{i}=T^{-1}\sum_{t=1}^{T}y_{i,t}$, $\overline{f}%
=T^{-1}\sum_{t=1}^{T}f_{t}       $, and 
$
\widehat{\beta }_{i,T}=\left[ \sum_{t=1}^{T}\left( f_{t}-\overline{f}\right)
\left( f_{t}-\overline{f}\right) ^{\prime }\right] ^{-1}\left[
\sum_{t=1}^{T}\left( f_{t}-\overline{f}\right) y_{i,t}\right]
$.
Recall that the data admit at least $\nu \geq 4$ moments. We use the
transformation%
\begin{equation}
\psi _{i,NT}=\left\vert \frac{T^{1/\nu }\widehat{\alpha }_{i,T}}{\widehat{s}%
_{NT}}\right\vert ^{\nu /2},  \label{psi_iT}
\end{equation}
where the rescaling sequence is based on the squared OLS residuals $\widehat{%
u}_{i,t}^{2}$: 
\begin{equation}
\widehat{s}_{NT}=\sqrt{\frac{1}{NT}\sum_{i=1}^{N}\sum_{t=1}^{T}\widehat{u}%
_{i,t}^{2}}.  \label{s_NT}
\end{equation}%
As we only work with rates of convergence -- requiring $\psi _{i,NT}$
to drift to zero under the null and diverge under the alternative --   in principle any rescaling sequence which removes the measurement unit
from $\widehat{\alpha }_{i,T}$ can be used in (\ref{psi_iT}).
For example, one
may use an estimate of the variance of the individual $\widehat{\alpha }%
_{i,T}$, and indeed such estimate does not even need to be consistent. We
propose $\widehat{s}_{NT}$ because the cross-sectional averaging smooths away any potentially problematic behavior such as the presence of spikes in the unit specific variances; this choice also turns out to deliver the best performance in simulations, and to help in the empirics (see the discussion in Section \ref{empirical}). 

\smallskip

We now explain the construction of the test statistic. Heuristically, under our assumptions it should hold that $\widehat{%
\alpha }_{i,T}-\alpha _{i}=O_{a.s.}\left( T^{-1/2}\right) $. Hence, under
the null of (\ref{null}), $\psi _{i,NT}%
\overset{a.s.}{\rightarrow }0$ by construction; conversely, under the alternative, $\psi _{i,NT}\overset{a.s.}{\rightarrow }\infty $. In order to
have a statistic to test for $\mathbb{H}_{0}$, we now perturb the $\psi
_{i,NT}$s by adding a sequence of \textit{i.i.d.} Gaussian variables%
\begin{equation}
z_{i,NT}=\psi _{i,NT}+\omega _{i}, \label{z}
\end{equation}%
where $\omega _{i}\overset{i.i.d.}{\sim }\mathcal{N}\left( 0,1\right) $,
generated independently of the sample $\left\{ \left( u_{i,t},f_{t}^{\prime
}\right) ^{\prime },1\leq i\leq N,1\leq t\leq T\right\} $. Thus,
under the null $z_{i,NT}$ should be an \textit{i.i.d.} sequence of standard
normals; under the alternative, there should be (at least) one spike due to
the fact that $\alpha _{i}\neq 0$ for some $i$. Hence, we base our test on
the maximally selected $z_{i,NT}$:%
\begin{equation*}
Z_{N,T}=\max_{1\leq i\leq N}z_{i,NT}.
\end{equation*}%
In order to study the asymptotics of $Z_{N,T}$, let%
\begin{equation*}
b_{N}=\sqrt{2\log N}-\frac{\log \log N+\ln \left( 4\pi \right) }{2\sqrt{%
2\log N}}\text{, \ \ and \ \ }a_{N}=\frac{b_{N}}{1+b_{N}^{2}},
\end{equation*}%
and consider the following restriction:

\begin{assumption}
\label{asymptotics}It holds that $N=O\left( T^{\frac{1}{2}\left(\frac{\nu}{2}%
-1\right)-\varepsilon }\right) $ for some $\varepsilon >0$.
\end{assumption}

Assumption \ref{asymptotics} poses a constraint on the relative rate of
divergence of $N$ and $T$ as they pass to infinity: the more
moments the data admit, the larger $N$ can be relative to $T$. A comparison
with the similar Assumption A1(iii) in \citet{feng2022high} may
shed further light: if, similarly to \citet{feng2022high}, we assumed independence
between $f_{t}$ and $u_{i,t}$, then Assumption \ref{asymptotics} would
become $N=O\left( T^{\nu /2-1-\varepsilon }\right) $ for some (arbitrarily
small) $\varepsilon >0$, which coincides with Assumption A1(iii) in %
\citet{feng2022high}. Similarly, the asymptotics in %
\citet{pesaran2023testing} requires $N=o\left( T^{2}\right) $, under the
assumptions of deterministic regressors and at least eight finite moments
for the errors. In our case, as long as $\nu \geq 6$, the condition that $%
N=o\left( T^{2}\right) $ is satisfied, and therefore we have either the same
asymptotic regime with a milder moment condition, or, with the same moment
condition, a larger $N$ relative to $T$. In Section \ref{complements} in
the Supplement, we discuss the possibility of relaxing -- given $\nu $ --
Assumption \ref{asymptotics} to allow for a broader set of combinations of $%
N $ and $T$.

\smallskip

Let $\mathbb{P}^{\ast }$ denote the probability conditional on the sample $\left\{ \left( u_{i,t},f_{t}^{\prime }\right) ^{\prime },1\leq i\leq
N,1\leq t\leq T\right\} $.

\begin{theorem}
\label{asy-max} Let Assumptions \ref{error}-\ref{asymptotics} hold. Then, under $\mathbb{H}_{0}$ of (\ref{h0}), it holds that%
\begin{equation}
\lim_{\min \left\{ N,T\right\} \rightarrow \infty }\mathbb{P}^{\ast }\left[ a_{N}^{-1} \left( Z_{N,T}-b_{N} \right) \leq x \right] =\exp \left( -\exp \left( -x\right)
\right) ,  \label{th-null}
\end{equation}%
for almost all realizations of $\left\{ \left( u_{i,t},f_{t}^{\prime
}\right) ^{\prime },1\leq i\leq N,1\leq t\leq T\right\} $, and all $-\infty
<x<\infty $. Under $\mathbb{H}_{A}$ of (\ref{hA}), it holds that%
\begin{equation}
\lim_{\min \left\{ N,T\right\} \rightarrow \infty }\mathbb{P}^{\ast }\left[ a_{N}^{-1} \left( Z_{N,T}-b_{N} \right) \leq x \right] =0,  \label{th-alt}
\end{equation}%
for almost all realizations of $\left\{ \left( u_{i,t},f_{t}^{\prime
}\right) ^{\prime },1\leq i\leq N,1\leq t\leq T\right\} $, and all $-\infty
<x<\infty $.
\end{theorem}

Theorem \ref{asy-max} describes the limiting behavior of the test statistic $%
Z_{N,T}$ both under the null and under the alternative hypotheses. By (\ref%
{th-null}) the suitably normed version of $Z_{N,T}$ converges (in
distribution, a.s.\ conditionally on the sample) to a Gumbel distribution.\footnote{In Section \ref{app:MC_newCV} in the Supplement, we also propose an alternative, fixed $N$ version of the critical values. As discussed in Section \ref{simulations}, these alternative critical values return better finite sample results  when $T$ grows larger for a fixed $N$.}
Equation (\ref{th-null}) implies that asymptotic critical values at nominal
level $\tau $ are given by 
\begin{equation}
c_{\tau }=b_{N}-a_{N}\log \left( -\log \left( 1-\tau \right) \right) .
\label{evt-asy}
\end{equation}%
Similarly, (\ref{th-alt}) roughly states that under the alternative (the
suitably normed version of) $Z_{N,T}$ diverges to positive infinity in
probability, a.s.\ conditional on the sample.\footnote{Following the proof of the theorem, it can be readily shown that a sufficient condition to have asymptotic unit power is that, as $\min \{N,T\} \rightarrow \infty$, it holds that $\sqrt{T/\log N} \max_{1 \leq i \leq N} |\alpha_i| \rightarrow \infty$. This entails that our test has power as long as one alpha is (mildly) larger than zero. The same result can be shown to hold, \textit{a fortiori}, in the case of a \textquotedblleft small\textquotedblright , pervasive alternative whereby $|\alpha_i|=|\alpha|$ for all $i$, with $\sqrt{T/\log N}|\alpha| \rightarrow \infty$ - a case known in the literature as the \textquotedblleft diffuse alpha\textquotedblright\ case (\citealp{chernov2025test}).}

\subsection{Discussion\label{comparison}}

We now discuss how and why
randomization works in our case, and what relationship it has with other
approaches based on \textquotedblleft adding randomness\textquotedblright\
such as the bootstrap. As \citet{zhang2023randomization} put it, \textquotedblleft [t]he meaning of randomization tests has become obscure in statistics education and practice over the last century \textquotedblright\ (p.\ 2928). Hence, some clarifications on the core statistical mechanism underpinning our randomized test are in order. As we expound hereafter, the main feature of our approach is that our randomization is based on adding randomness \textit{to the test statistic}, rather than \textit{to the data}.

Our approach works as follows. To start, we construct the statistic $\psi _{i,NT}$ defined in (\ref{psi_iT}),
based on the data, as in any \textquotedblleft
traditional\textquotedblright\ testing approach. However, we do not require its second order properties (that is, its limiting distribution and/or asymptotic efficiency),
and use only  its first order properties (that is, its rate of
convergence). Being able to focus only on rates requires simpler arguments and milder assumptions; furthermore, at no stage do we require the estimation of asymptotic variances.\footnote{Indeed, this entails that our approach, by its very nature, places more emphasis on the \textit{robustness} of the estimator employed.} In the construction of $\psi_{i,NT}$, we pre-multiply $\widehat{\alpha }_{i}$\ by the scaling factor $T^{1/\nu }$, which - heuristically - is designed to \textquotedblleft wash out\textquotedblright\ the estimation error $\widehat{\alpha}_i-\alpha_i$, and therefore the randomness
coming from the data. Our theory uses \textit{almost sure} rates for the statistic $\psi _{i,NT}$; the results in Lemma \ref{psi} in the Supplement yield that, for each $i$%
\begin{equation}
\mathbb{P}\left( \omega :\lim_{\min \left\{ N,T\right\} \rightarrow \infty }\psi
_{i,NT}=0\right) =1\text{, under }\mathbb{H}_{0},  \label{as-0}
\end{equation}%
so that, in our proofs, we can work with the premise that (under the
null) $\lim_{\min \left\{ N,T\right\} \rightarrow \infty}\psi _{i,NT}=0$.\footnote{Similarly, the
theory also entails that, for each $i$, $\mathbb{P}\left( \omega :\lim_{\min \left\{ N,T\right\} \rightarrow \infty }\psi _{i,NT}=\infty
\right) =1\text{, under }\mathbb{H}_{A}$; hence, in our proofs, we can work with the premise that (under the
alternative) $\lim_{\min \left\{ N,T\right\} \rightarrow \infty }\psi _{i,NT}=\infty$. }
Seeing as the randomness of the data has been washed out by pre-multiplying $\widehat{\alpha }_{i}$\ by $T^{1/\nu }$, in order to construct a test, we add randomness in the form of (\textit{i.i.d.} and
Gaussian) noise to each $\psi _{i,NT}$, thus constructing the sequence $%
\left\{ z_{i,NT},1\leq i\leq N\right\} $\ defined in (\ref{z}). By (%
\ref{as-0}), we can derive the limiting law of $Z_{N,T} = \max_{1\leq i\leq N}z_{i,NT}$%
, and - crucially - show that this depends solely on the probability measure of the added randomness. Given that this is the only driver of the asymptotic behavior of $Z_{N,T}$, and that it is fully under the researcher's control, it can be studied
straightforwardly, with no need for further assumptions on the data apart
from the ones required to derive the rates of $\psi _{i,NT}$. This is the key
argument whereby our test can be employed, among other things, under arbitrary levels of
cross-sectional dependence: the only cross-sectional dependence that matters for the asymptotics of $Z_{N,T}$ is the one in the
added random noise - which can be forced to equal zero by the researcher. Considering as a leading term of comparison the max-type test by \citet{feng2022high}, in that case the asymptotic behavior of the test statistic is determined by the probability measure of the data, and therefore it is affected by the presence and extent of cross-sectional dependence across $i$. By the same token, in (\ref{psi_iT}) we do not require a consistent estimator of the variance of the estimated $\alpha_i$, but merely a rescaling factor such as $\widehat{s}_{NT}$, designed to make $\psi_{i,NT}$ adimensional. An immediate consequence is that - as mentioned in Section \ref{model} - in our set-up we can readily allow for unconditional heteroskedasticity in both factors and errors. Seeing as we rely only on rates of convergence for $\widehat{\alpha}_i$, and no estimation of the asymptotic variance of the estimators is required, our results would hold even in this case, without requiring any modifications or even prior knowledge as to the presence of heteroskedasticity. Conversely, \citet{feng2022high} use a sequence of \textit{t}-statistics across $i$, and thus a consistent estimator of the variance of $\widehat{\alpha}_i$ is required - which, especially in the presence of heteroskedasticity, is well known to be fraught with difficulties (see e.g. the review, and the proposed solution, in \citealp{casini2024fixed}).

\smallskip 

Further light on our approach can be shed by comparing it with approaches where the randomness is added to the data, a prime example being the bootstrap. The
mode of convergence in Theorem \ref{asy-max} - where, under the null, $%
Z_{N,T}$ converges \textquotedblleft in distribution,
almost surely conditional on the sample\textquotedblright\ - is the same as
one would find in the case of the bootstrap (\citealp{bickelfreedman}).
Notwithstanding this analogy, the way in which randomness is added, the
way in which the theory works, and the assumptions required on the data are
profoundly different to the randomization method proposed herein. In the bootstrap, randomness is added by resampling the data
multiple times, and constructing a (pseudo) version of the test statistic at
each resampling: the randomness of the data is not washed away. Hence, the asymptotic behavior of the
resampled statistic is still affected by the features of the data -
e.g., by serial and/or cross-sectional dependence.\footnote{E.g. in a \textquotedblleft traditional\textquotedblright\ resampling scheme, given data $\left\{ y_{i},1\leq i\leq n\right\} $, at each iteration $%
1\leq b\leq B$ the pseudosample $\left\{ y_{i,b}^{\ast },1\leq i\leq
n\right\} $ is constructed such that $\mathbb{P}\left( y_{i,b}^{\ast
}=y_{j}|\left\{ y_{i},1\leq i\leq n\right\} \right) =1/n$. Hence, the
(conditional) law of $y_{i,b}^{\ast }$ is 
$\mathbb{P}_{n}\left( y\right) =n^{-1}\sum_{i=1}^{n}\mathbb{I}\left(
y_{i}\leq y\right)$,
where $\mathbb{I}\left( \cdot \right) $ denotes the indicator function, which clearly depends on the features of the data $\left\{
y_{i},1\leq i\leq n\right\} $.} Furthermore, a typical way of proving the
validity of the bootstrap is to show that the distribution of the resampled
test statistic, conditional on the sample, converges in some sense (e.g., in
distribution \textit{a.s.} conditional on the sample) to the asymptotic distribution of
the original test statistic.
This, however, requires deriving such asymptotic distribution, which is likely to be more complicated and to require stronger assumptions than simply deriving its convergence rate.\footnote{Similar considerations also hold for other approaches based on adding randomness to the data. For example, the randomized tests studied in \citet{canay2017randomization} require that the distribution of the data be \textquotedblleft approximately symmetric\textquotedblright\ - that is, invariant under certain transformations. Such shape restrictions are not required by our approach. } Moreover, applying the bootstrap in our context is fraught with
difficulties: \citet{Huang2023} show that the approaches proposed
by \citet{kosowski2006can} and \citet{fama2010luck} may suffer from (even
severe) undersizing and low power, especially when the data exhibit features that are typical of financial returns (e.g., skewed unconditional distributions and large cross-sectional sizes). The bootstrap corrections suggested in \citet{Huang2023} ameliorate
these issues, but still require \textit{weak} cross-sectional dependence.
\\
Finally, a crucial difference between our approach and approaches based on adding randomness to the data is that, in the latter case, the added randomness vanishes in the
limit, and thus it does not affect the limiting behavior of the resulting test statistic. Conversely, the randomness added in our method does not
vanish asymptotically, which is a
well-known feature of this type of randomized tests (see e.g. %
\citealp{corradi2006}). We propose a solution in
the next section.

\subsection{Derandomized inference\label{derandom}}

The discussion above indicates that the results in Theorem \ref{asy-max} are different to \textquotedblleft
standard\textquotedblright\ inferential theory. In particular, (\ref{th-alt}%
) entails $\lim_{\min \left\{ N,T\right\} \rightarrow \infty }$ $\mathbb{P}^{\ast }\left(
Z_{N,T}\geq c_{\tau }|\mathbb{H}_{A}\right) =1$, 
which corresponds to the notion of power. The
result under the null is more delicate: whilst it holds that $\lim_{\min \left\{ N,T\right\} \rightarrow \infty }\mathbb{P}^{\ast }\left(
Z_{N,T}\geq c_{\tau }|\mathbb{H}_{0}\right) =\tau$, 
this result is not the standard notion of size. The fact that the added randomness in the construction of $Z_{N,T}$ does not vanish asymptotically entails that, under
the null, different researchers using the same data will obtain different
values of $Z_{N,T}$, and thus different p-values. \newline
We propose a decision rule to discern between $\mathbb{H}_{0}$ and $%
\mathbb{H}_{A}$ which is not driven by the added randomness, and is
therefore the same across all researchers using the same dataset. Following %
\citet{HT2019}, each researcher will compute $Z_{N,T}$ over $B$
replications, at each replication $1\leq b\leq B$ constructing a statistic $%
Z_{N,T}^{\left( b\right) }$ using a random sequence $\omega _{i}^{\left(
b\right) }\overset{i.i.d.}{\sim }\mathcal{N}\left( 0,1\right) $ for $1\leq
i\leq N$, independent across $1\leq b\leq B$ and of the sample. Let
\begin{equation}
Q_{N,T,B}\left( \tau \right) =B^{-1}\sum_{b=1}^{B}\mathbb{I}\left(
Z_{N,T}^{\left( b\right) }\leq c_{\tau }\right) ,  \label{q}
\end{equation}%
be the percentage of times that the researcher does not reject the null
at nominal significance level $\tau $. An immediate consequence of Theorem \ref%
{asy-max} is that, as $\min \left\{ N,T,B\right\} \rightarrow \infty$
\begin{equation}
\mathbb{P}^{\ast
}\left( Q_{N,T,B}\left( \tau \right) =1-\tau |\mathbb{H}_{0}  \right) =1\text{\qquad and\qquad }%
\mathbb{P}^{\ast
}\left( Q_{N,T,B}\left( \tau \right) =0 |\mathbb{H}_{A} \right) =1,  \label{q-null}
\end{equation}%
for almost all realizations of $\left\{ \left( u_{i,t},f_{t}^{\prime
}\right) ^{\prime },1\leq i\leq N,1\leq t\leq T\right\} $. This result holds
for all different researchers, and therefore averaging across the
replications $1\leq b\leq B$ removes the added randomness in $%
Q_{N,T,B}\left( \tau \right) $: hence, all researchers will
obtain the same value of $Q_{N,T,B}\left( \tau \right) $. As noted in %
\citet{he2023one}, $Q_{N,T,B}\left( \tau \right) $ corresponds to (the
complement to one of) the \textquotedblleft fuzzy decision\textquotedblright\ in equation (1.1a) in \citet{geyer}. This notion can be illustrated
by considering a random variable, say $\mathcal{D}$, which takes two values:
\textquotedblleft do not reject $\mathbb{H}_{0}$\textquotedblright
\thinspace\ with probability $Q_{N,T,B}(\tau )$, and \textquotedblleft
reject $\mathbb{H}_{0}$\textquotedblright . According to (\ref{q-null}),
asymptotically it holds that, a.s. conditionally on the sample $\mathbb{P}%
^{\ast }\left( \omega :\mathcal{D}=\text{\textquotedblleft reject }\mathbb{H}%
_{0}\text{\textquotedblright }|\mathbb{H}_{0}\right) =\tau $, across all
researchers, which reconciles the procedure with the notion of \textit{size}
of a test. Similarly, (\ref{q-null}) states that, asymptotically, $\mathbb{P}%
^{\ast }\left( \omega :\mathcal{D}=\text{\textquotedblleft reject }\mathbb{H}%
_{0}\text{\textquotedblright }|\mathbb{H}_{A}\right) =1$ a.s. conditionally
on the sample, which corresponds to the notion of \textit{power} of a test.

\begin{theorem}
\label{strong-rule} Let Assumptions {\ref{error}-\ref{asymptotics}%
} hold, and $B=O\left( \left(\log N\right)^2 \right) $. Then it holds that%
\begin{equation}
\limsup_{\min \left\{ N,T,B\right\} \rightarrow \infty }\sqrt{\frac{B}{2\log
\log B}} \left\vert \frac{Q_{\tau }-\left( 1-\tau \right) }{\sqrt{\tau
\left( 1-\tau \right) }}\right\vert =1\text{ a.s.},  \label{st1}
\end{equation}%
under $\mathbb{H}_{0}$, for almost all realizations of $\left\{ \left(
u_{i,t},f_{t}^{\prime }\right) ^{\prime },1\leq i\leq N,1\leq t\leq
T\right\} $. Under $\mathbb{H}_{A}$, it holds that $Q_{N,T,B}\left( \tau \right)
=o_{a.s.}\left( 1\right) $, for almost all realizations of $\left\{ \left(
u_{i,t},f_{t}^{\prime }\right) ^{\prime },1\leq i\leq N,1\leq t\leq
T\right\} $.
\end{theorem}


Building on Theorem \ref{strong-rule}, a \textquotedblleft
derandomized\textquotedblright\ decision rule can be proposed. By (\ref{st1}%
), under $\mathbb{H}_{0}$ there exists a triplet of random variables $\left(
N_{0},T_{0},B_{0}\right) $ such that 
\begin{equation}
Q_{N,T,B}\left( \tau \right) \geq \ell_{\tau}\equiv 1-\tau -\sqrt{\tau \left( 1-\tau \right) }%
\sqrt{\frac{2\log \log B}{B}},  \label{eq:derand_LIL}
\end{equation}%
for all $\left( N,T,B\right) $ with $N\geq N_{0}$, $T\geq T_{0}$, and $B\geq
B_{0}$. Similarly, under $\mathbb{H}_{A}$ there exists a triplet of random
variables $\left( N_{0},T_{0},B_{0}\right) $ such that $Q_{N,T,B}\left( \tau
\right) \leq \epsilon $, for all $\epsilon >0$ and $\left( N,T,B\right) $
with $N\geq N_{0}$, $T\geq T_{0}$, and $B\geq B_{0}$. This dichotomous
behavior can be exploited to construct a decision rule based on $%
Q_{N,T,B}\left( \tau \right) $, in a way that is more akin to information criteria than tests: $\mathbb{H}_{0}$ is not rejected when $%
Q_{N,T,B}\left( \tau \right) $ exceeds a threshold, whereas it is rejected
otherwise. In theory, one could use the threshold based on the Law of the Iterated Logarithm (LIL) in (\ref{eq:derand_LIL}); under $\mathbb{H}_0$, it follows that $\mathbb{P}^{*}\left[ Q_{N,T,B} < \ell_{\tau} \right]=0$ as $\min\{N,T,B\} \rightarrow \infty$. Albeit valid asymptotically, this criterion turns out to be biased against $\mathbb{H}_0$ in simulations, especially in small samples; this is not entirely surprising, since the bound induced by the LIL is not likely to \textquotedblleft bite\textquotedblright\ unless $B$ is large (which, in light of the restriction $B=O((\log N)^2)$, requires $N$ to be \textquotedblleft very large\textquotedblright). A decision rule that is more favorable towards the null
could be%
\begin{equation}
Q_{N,T,B}\left( \tau \right) \geq \left( 1-\tau \right) -f\left( B\right) ,
\label{eq:derand_fB}
\end{equation}%
with $f\left( B\right) $ a user-specified, non-increasing function of $B$
such that $\lim_{B\rightarrow \infty }f\left( B\right) =0$, and $\limsup_{B\rightarrow \infty }\left( f\left( B\right) \right) ^{-1}\sqrt{%
 2 \log \log B /B}=0$. 

\subsection{From theory to practice: guidelines and recommendations\label%
{practice}}

The procedure proposed in Section \ref{tests} depends on  the nuisance
parameter $\nu$ and on the 
tuning quantities $B$
and $f\left( B\right) $, i.e.\ the number of trials and the threshold in the derandomized approach. We offer
a set of guidelines/suggestions which could inform the practical application
of these procedures.\\


There are at least two ways in
which $\nu $ in (\ref{psi_iT}) can be determined:

\begin{enumerate}
\item A \textit{direct} approach, based on using a tail index estimator for the largest moment $\nu$ admitted by the data.
This approach would offer a consistent estimator, but its properties may be rather poor in finite samples (%
\citealp{embrechts}).

\item An \textit{indirect} approach, based upon noting that a lower bound (as opposed to an exact value) for 
$\nu $ would suffice. In order to find such a bound, e.g. the tests by %
\citet{trapani16} and \citet{degiannakis2023superkurtosis} could be employed
to test for the null hypothesis $\mathbb{H}_{0}:\mathbb{E}\left\vert
y_{i,t}\right\vert ^{\nu _{0}}=\infty $. Upon rejecting, it follows
that $\nu \geq \nu _{0}$, and therefore $\nu _{0}$ can be used in (\ref%
{psi_iT}).
\end{enumerate}
In addition --- as we do in our empirical illustration --- one can use $\nu_0
=4$ when constructing $\psi_{i,NT}$, i.e. the smallest finite moment prescribed by our theory, which
we would recommend when the sample size $T$ does not afford reliable
inference. 

\smallskip

Turning to the specifics of the derandomization, we note that:

\begin{enumerate}
\item The choice of $B$ is constrained by the condition $B=O\left( \left( \log
N\right)^2\right) $; in
our simulations, we employ $B=\lfloor\left(\mathrm{\log }\,N\right)^{2}\rfloor$, which we
recommend as a guideline.\footnote{In the proof of Theorem \ref{strong-rule}, we show that the rate of approximation of (\ref{eq:derand_LIL}) is very fast in $B$ (see equations (\ref{revesz}), (\ref{lilB}) and (\ref{revesz2}) in the Supplement), which guarantees that the lower bound in (\ref{eq:derand_LIL}) is accurate even when using a \textquotedblleft small\textquotedblright\ value of $B$.}

\item The choice of $f\left( B\right) $ is based on (\ref{eq:derand_fB}); \citet{he2023one} show that the derandomized decision
rule is relatively robust to the specification of $f\left( B\right) $. We recommend $f\left( B\right) =B^{-1/4}$, which is also found
to deliver the best results in \citet{he2023one}.
\end{enumerate}

\section{Simulations\label{simulations}}

We use a similar DGP to \citet{feng2022high}: 
\begin{equation}
y_{i,t} =\alpha _{i}+\sum_{p=1}^{3}\beta _{i,p}f_{p,t}+u_{i,t}, \text{   with   } f_{t} =\overline{f}+\Phi f_{t-1}+\zeta _{t},
\label{eq:measurement} 
\end{equation}%
where $f_{t}=(f_{1,t},f_{2,t},f_{3,t})^{\prime }$, $\overline{f}%
=(0.53,0.19,0.19)^{\prime }$, $\Phi =\mathrm{diag}\left\{
-0.1,0.2,-0.2\right\} $ and $\zeta _{t}\overset{i.i.d.}{\sim }\mathcal{N}%
_{3}(0,I_{3})$. Loadings are generated as $\beta _{i,1}\overset{i.i.d.}{\sim 
}\mathcal{U}(0.3,1.8)$, $\beta _{i,2}\overset{i.i.d.}{\sim }\mathcal{U}%
(-1,1) $, and $\beta _{i,3}\overset{i.i.d.}{\sim }\mathcal{U}(-0.6,0.9)$ for
all $i$. We allow for strong cross-sectional dependence in the innovations $%
\mathbf{u}_{t}=(u_{1,t},\dots ,u_{N,t})^{\prime }$ via a factor model: 
\begin{equation}
\mathbf{u}_{t}=\boldsymbol{\gamma }g_{t}+\boldsymbol{\xi }_{t},\text{ \ \
with \ \ }g_{t}=\phi _{g}g_{t-1}+\chi _{t},  \label{eq:innov2}
\end{equation}%
where $\boldsymbol{\gamma }=(\gamma _{1},\dots ,\gamma _{N})^{\prime }$ for $%
\gamma _{i}\overset{i.i.d.}{\sim }\mathcal{U}(0.7,0.9)$, $\phi _{g}=0.4$,
and $\chi _{t}\overset{i.i.d.}{\sim }\mathcal{N}(0,1)$, with $\left\{ \chi
_{t},1\leq t\leq T\right\} $ generated independently of $\left\{ \zeta
_{t},1\leq t\leq T\right\} $. In (\ref%
{eq:innov2}), the $N$-dimensional random vectors $\left\{ \boldsymbol{\xi }%
_{t},1\leq t\leq T\right\} $ are generated independently of $\left\{ \left(\zeta
_{t}, \chi_t\right),1\leq t\leq T\right\} $  
under the following three set-ups (all with mean zero and covariance matrix $\Sigma
_{\xi }$):

\begin{enumerate}
\item \textbf{The Gaussian case}: $\boldsymbol{\xi }_{t}\overset{i.i.d.}{%
\sim }\mathcal{N}_{N}(0,I_{N})$.

\item \textbf{The Student's $t$ case}: where $\xi _{i,t}$ follows a
Students's $t$ distribution with $d=5.5$ degrees of freedom, zero mean and
unit scale, independent across $i$. In this case, $\xi _{i,t}$ and $y_{i,t}$
have regularly varying tails;

\item \textbf{The GARCH case}: we generate $\boldsymbol{\xi }_{t}=\mathbf{H}%
_{t}\boldsymbol{z}_{t}$, with: $\boldsymbol{z}_{t}=\left(
z_{1,t},...,z_{N,t}\right) ^{\prime }$ and $z_{i,t}\overset{i.i.d.}{\sim }%
\mathcal{N}(0,1)$; and $\mathbf{H}_{t}=\mathrm{diag}\left\{ h_{1,t},\dots
,h_{N,t}\right\} $ with $h_{i,t}^{2}=\omega _{i}+\pi_{i}\xi
_{i,t-1}^{2}+\beta _{i}h_{i,t-1}^{2}$, with $\omega _{i}\overset{i.i.d.}{\sim }%
\mathcal{U}(0.01,0.05)$, $\pi_{i}\overset{i.i.d.}{\sim }\mathcal{U}%
(0.01,0.04)$ and $\beta _{i}\overset{i.i.d.}{\sim }\mathcal{U}(0.85,0.95)$.%
\footnote{%
These parameter values imply that $\xi _{i,t}$ has finite sixth moment for
any $i$.}
\end{enumerate}

In all scenarios, we report  rejection frequencies under the null
and under the alternative, for nominal level $\tau =5\%$, using $N\in
\left\{ 100,200,500\right\} $ and $T\in \left\{ 100,200,300,500, 1000, 2000\right\} $. As far as the alternative hypothesis is concerned, we consider a rather
sparse alternative where $\alpha _{i}\overset{i.i.d.}{\sim }\mathcal{N}(0,1)$
for $5\%$ percent of the cross-sectional units $1\leq i\leq N$. We construct
the test statistic using a notional value of $\nu =5$ in (\ref{psi_iT}), and
consider both the \textquotedblleft one-shot\textquotedblright\ test in
Section \ref{tests}, and the derandomized version in Section \ref{derandom}%
. For the latter, we examine results using both \eqref{eq:derand_LIL} and \eqref{eq:derand_fB} with $f(B)=B^{-1/4}$. We
compare our test with the tests by \citet[][FLLM]{feng2022high}, %
\citet[][PY]{pesaran2023testing}, and \citet[][AS]{ardia2024robust}.\footnote{We set tuning parameters of the AS test as
suggested in their Monte Carlo exercise (in particular, using their
notation, we set $L=0$ and $\psi =1/3$). When applicable, i.e.\ when $N<T$, we also check performances of the test by \citet[][GRS]{gibbons1989test}
Moreover, Monte Carlo results in \citet{feng2022high} and %
\citet{pesaran2023testing} show that their tests consistently outperform
that of \citet{gungor2016multivariate}. Hence, we omit comparisons with this
last approach. The comparison with the approaches of \cite{fan2015power} and \cite{gagliardini2016time} is reported in Section \ref{MC_moreBenchmarks} of the Supplement.}

\smallskip

We start from the Gaussian case; results using the one-shot test are in
Table \ref{tab:test_Gaussian}, whereas in Table \ref{tab:derand_Gaussian} we
report rejection frequencies from the derandomization approach. Our test is
the best one at controlling the size 
for all combinations of $N$ and $T$, whereas the other tests are
consistently oversized. Our test becomes slightly undersized, and subsequently plateaus, as $T$ increases. This is not accompanied by any loss of power; further, the alternative critical values presented in Section \ref{app:MC_newCV} of the Supplement yield empirical rejection frequencies that are extremely close to the nominal size when $T\geq 500$. Table \ref{tab:derand_Gaussian} suggest that the
derandomization procedure based on $f(B)=B^{-1/4}$ also works very well.
Indeed, as predicted by the theory, the empirical rejection frequencies are extremely close to zero under
the null and quickly converge to one under the alternative; the latter is
particularly true when $N$ gets large, thus showing that our approach is
particularly suitable when $N>T$. As expected, the
threshold based on the LIL
leads to higher empirical rejection frequencies under both the null and the
alternative. As far as power is concerned, the right panels of the table
show that our tests performs satisfactorily, whilst at the same time
guaranteeing size control. We note that the test by FLLM outperforms ours in
terms of power in most cases, but it is also oversized. Turning to the case
of data with heavier tails, Tables \ref{tab:test_T} and \ref{tab:derand_T}
report empirical rejection frequencies for the Student's $t$ case. Results
for the randomized test are in line with those of Table \ref%
{tab:test_Gaussian}. As in the Gaussian case, the other tests are oversized, while ours becomes slighlty  undersized when $T$ grows large. Again, this can be solved by using the alternative critical values  from Section \ref{app:MC_newCV}.
The derandomized procedure works as expected also in the case of heavier
tails. Similar considerations hold for power as for the Gaussian case.
Finally, results under the GARCH case are in Tables \ref{tab:test_GARCH} and %
\ref{tab:derand_GARCH}; size and power of all tests behave as in the other
cases, and so does the decision rule based on Theorem \ref{strong-rule}. Our test is slighlty oversized when $N=500$ and $T=100$, but still outperforms all the others when sizes and powers are considered. The use of the alternative critical values still improves the (slight) under-rejection for large values of $T$. Notably, GRS performs very well whenever applicable, exhibiting good finite sample properties even when the DGP violates its assumption of \textit{i.i.d.} Gaussian errors. However, the fact that the test requires $T>N$ makes it inapplicable whenever one deals with a rather large number of test assets.\footnote{For instance, using the GRS  with monthly data on $N=200$ test assets implicitly  assumes that  assets' $\beta$s  are constant  over more than 16 years. This  is at odds with  all the available empirical findings.}

\begin{table}[h]
\caption{\scriptsize{Empirical rejection frequencies for the test in Theorem \protect\ref%
{asy-max}, Gaussian case.}}
\label{tab:test_Gaussian}\captionsetup{font=small}
\par
\begin{center}
{\scriptsize \begin{tabular}{ll|
                S S S S S S
                c
                S S S S S S}
\toprule
&&\multicolumn{6}{c}{$\phi_{g} = 0.4$; $\alpha_{i} = 0$ for all $i$}&&
  \multicolumn{6}{c}{$\phi_{g} = 0.4$; $\alpha_{i} \sim N(0,1)$ for 5\% of units} \\
\midrule
$N$&$\text{Test }\backslash T$
&{100}&{200}&{300}&{500}&{1000}&{2000}&&
 {100}&{200}&{300}&{500}&{1000}&{2000}\\
\midrule
100&Thm. 1 & 5.8 & 4.2 & 4.0 & 3.9 & 3.3 & 3.4 &&
             93.5 & 96.0 & 96.9 & 97.5 & 98.7 & 99. \\ 
&FLLM&12.4&11.2&10.1&9.5&9.7&10.2&&
       98.6&99.5&99.6&99.8&100.0&100.0 \\
&GRS& \multicolumn{1}{c}{--} &5.3&4.5&4.7&3.5&4.1&&
      \multicolumn{1}{c}{--}&99.4&99.5&99.9&100.0&100.0\\
&PY&21.5&19.5&17.8&19.1&18.7&18.8&&
    78.5&94.1&97.0&98.8&99.8&100.0\\
&AS&10.4&7.5&6.6&8.9&7.9&6.6&&
    22.2&54.7&85.2&97.6&100.0&100.0\\
\midrule
200&Thm. 1&7.1&3.9&3.6&3.5&4.1&4.1&&
          99.7&99.9&100.0&100.0&100.0&100.0\\ 
&FLLM&15.1&8.7&9.9&10.4&10.5&10.1&&
       100.0&100.0&100.0&100.0&100.0&100.0 \\
&GRS& \multicolumn{1}{c}{--} & \multicolumn{1}{c}{--}
     &3.8&4.1&4.7&3.9&&
       \multicolumn{1}{c}{--}& \multicolumn{1}{c}{--}
     &100.0&100.0&100.0&100.0\\
&PY&21.9&17.4&18.2&21.7&20.2&19.0&&
    86.5&98.6&99.5&100.0&100.0&100.0\\
&AS&11.3&8.3&8.5&7.8&7.7&8.6&&
    22.5&65.7&94.4&99.9&100.0&100.0 \\
\midrule
500&Thm. 1&7.0&4.6&4.0&3.8&4.0&4.1&&
          100.0&100.0&100.0&100.0&100.0&100.0 \\
&FLLM&12.1&13.4&12.0&8.9&9.7&10.5&&
       100.0&100.0&100.0&100.0&100.0&100.0 \\
&GRS& \multicolumn{1}{c}{--}& \multicolumn{1}{c}{--}
     & \multicolumn{1}{c}{--}& \multicolumn{1}{c}{--}
     &6.3&4.6&&
       \multicolumn{1}{c}{--}& \multicolumn{1}{c}{--}
     & \multicolumn{1}{c}{--}& \multicolumn{1}{c}{--}
     &100.0&100.0\\
&PY&19.2&22.0&22.1&19.1&19.3&20.4&&
    94.6&100.0&100.0&100.0&100.0&100.0\\
&AS&10.0&10.6&11.2&10.2&9.2&10.0&&
    24.8&76.8&97.9&100.0&100.0&100.0 \\
\bottomrule
\end{tabular} }
\end{center}
\par
{\scriptsize 
\textbf{Note: }{The nominal size is 5\% and powers are assessed at 5\% level
of significance; frequencies are computed across $M=1000$ Monte Carlo
samples and we set $\nu=5$ when computing $\psi_{i,NT}$.} }
\end{table}
\begin{table}[h]
\caption{\scriptsize{Empirical rejection frequencies for the decision rule of Theorem 
\protect\ref{strong-rule}, Gaussian case.}}
\label{tab:derand_Gaussian}\captionsetup{font=small}
\par
\begin{center}
{\scriptsize \centering
\begin{tabular}{ll|
                S S S S S S
                c
                S S S S S S}
\toprule
&&\multicolumn{6}{c}{$\phi_{g} = 0.4$; $\alpha_{i} = 0$ for all $i$}&&
  \multicolumn{6}{c}{$\phi_{g} = 0.4$; $\alpha_{i} \sim N(0,1)$ for 5\% of units} \\
\midrule
$N$&$\text{C.V. }\backslash T$
&{100}&{200}&{300}&{500}&{1000}&{2000}&&
 {100}&{200}&{300}&{500}&{1000}&{2000}\\
\midrule
100&LIL
&6.7 & 1.1 & 0.1 & 0.0 & 0.0 & 0.0 &&
 96.2 & 97.6 & 98.2 & 98.9 & 99.8 & 99.8\\
&$f(B)=B^{-1/4}$
&0.2 & 0.0 & 0.0 & 0.0 & 0.0 & 0.0 &&
 93.0 & 95.8 & 97.1 & 97.5 & 99.8 & 99.8\\
\midrule
200&LIL
&5.9 & 0.6 & 0.0 & 0.0 & 0.0 & 0.0 &&
 99.9 & 99.9 & 100.0 & 100.0 & 100.0 & 100.0\\
&$f(B)=B^{-1/4}$
&0.9 & 0.0 & 0.0 & 0.0 & 0.0 & 0.0 &&
 99.6 & 99.7 & 99.9 & 100.0 & 100.0 & 100.0\\
\midrule
500&LIL
&38.1 & 22.4 & 16.2 & 7.7 & 3.6 & 0.5 &&
 100.0 & 100.0 & 100.0 & 100.0 & 100.0 & 100.0\\
&$f(B)=B^{-1/4}$
&1.2 & 0.0 & 0.0 & 0.0 & 0.0 & 0.0 &&
 100.0 & 100.0 & 100.0 & 100.0 & 100.0 & 100.0\\
\bottomrule
\end{tabular}
 }
\end{center}
\par
{\scriptsize 
\textbf{Note: }{Results using either LIL-based critical values (LIL) or
critical values based on $f(B)=B^{-1/4}$. The derandomized statistic is
based on nominal level $\tau=5\%$. We set $B=\left(\mathrm{\log }%
\,N\right)^{2}$ for $Q_{N,T,B}(\tau)$ and $\nu=5$ for $\psi_{i,NT}$.} }
\end{table}
\begin{table}[h]
\caption{\scriptsize{Empirical rejection frequencies for the test in Theorem \protect
\ref{asy-max}, Student's $t$ case.}}
\label{tab:test_T}\captionsetup{font=small}
\par
\begin{center}
{\scriptsize \centering
\begin{tabular}{ll|
                S S S S S S
                c
                S S S S S S}
\toprule
&&\multicolumn{6}{c}{$\phi_{g} = 0.4$; $\alpha_{i} = 0$ for all $i$}&&
  \multicolumn{6}{c}{$\phi_{g} = 0.4$; $\alpha_{i} \sim N(0,1)$ for 5\% of units} \\
\midrule
$N$&$\text{Test }\backslash T$
&{100}&{200}&{300}&{500}&{1000}&{2000}&&
 {100}&{200}&{300}&{500}&{1000}&{2000}\\
\midrule
100&Thm. 1
&5.8&4.6&3.9&3.8&3.5&3.3&&
 88.8&93.0&94.7&96.3&98.2&98.8\\
&FLLM
&12.5&9.4&10.2&10.0&10.6&11.0&&
 97.7&99.2&99.9&100.0&100.0&100.0\\
&GRS
&\multicolumn{1}{c}{--}&4.3&4.4&5.1&4.5&3.8&&
 \multicolumn{1}{c}{--}&99.1&99.6&99.9&\multicolumn{1}{c}{--}&\multicolumn{1}{c}{--}\\
&PY
&24.2&19.6&20.2&20.8&20.6&20.8&&
 82.6&93.4&97.3&99.1&99.9&99.9\\
&AS
&9.3&7.9&8.8&8.8&8.6&9.0&&
 17.0&47.4&79.8&96.5&99.9&100.0\\
\midrule
200&Thm. 1
&5.6&4.0&3.7&3.5&4.2&4.0&&
 98.2&99.3&99.7&99.8&99.9&100.0\\
&FLLM
&10.4&11.4&10.8&11.0&10.2&9.3&&
 99.9&100.0&100.0&100.0&100.0&100.0\\
&GRS
&\multicolumn{1}{c}{--}&\multicolumn{1}{c}{--}&6.2&4.4&3.6&3.6&&
 \multicolumn{1}{c}{--}&\multicolumn{1}{c}{--}&100.0&100.0&\multicolumn{1}{c}{--}&\multicolumn{1}{c}{--}\\
&PY
&24.0&21.7&19.7&21.7&18.6&18.6&&
 89.1&97.9&99.3&100.0&100.0&100.0\\
&AS
&10.2&7.7&8.2&9.6&8.2&7.5&&
 19.2&55.3&90.4&99.3&100.0&100.0\\
\midrule
500&Thm. 1
&6.8&3.9&3.8&3.9&4.1&4.0&&
 100.0&100.0&100.0&100.0&100.0&100.0\\
&FLLM
&13.5&10.3&12.3&11.1&12.4&11.2&&
 100.0&100.0&100.0&100.0&100.0&100.0\\
&GRS
&\multicolumn{1}{c}{--}&\multicolumn{1}{c}{--}&\multicolumn{1}{c}{--}&\multicolumn{1}{c}{--}&5.0&5.0&&
 \multicolumn{1}{c}{--}&\multicolumn{1}{c}{--}&\multicolumn{1}{c}{--}&\multicolumn{1}{c}{--}&\multicolumn{1}{c}{100.0}&\multicolumn{1}{c}{100.0}\\
&PY
&27.0&20.9&20.1&20.7&21.5&20.1&&
 96.9&99.9&100.0&100.0&100.0&100.0\\
&AS
&9.7&8.5&9.9&10.8&11.5&9.1&&
 22.3&63.7&96.6&99.9&100.0&100.0\\
\bottomrule
\end{tabular}
 }
\end{center}
\par
{\scriptsize 
\textbf{Note: }{The nominal size is 5\% and powers are assessed at 5\% level
of significance; frequencies are computed across $M=1000$ Monte Carlo
samples and we set $\nu=5$ when computing $\psi_{i,NT}$. } }
\end{table}
\begin{table}[h]
\caption{\scriptsize{Empirical rejection frequencies for the decision rule of Theorem 
\protect\ref{strong-rule}, Student's $t$ case.}}
\label{tab:derand_T}\captionsetup{font=small}
\par
\begin{center}
{\scriptsize \centering
\begin{tabular}{ll|
                S S S S S S
                c
                S S S S S S}
\toprule
&&\multicolumn{6}{c}{$\phi_{g} = 0.4$; $\alpha_{i} = 0$ for all $i$}&&
  \multicolumn{6}{c}{$\phi_{g} = 0.4$; $\alpha_{i} \sim N(0,1)$ for 5\% of units} \\
\midrule
$N$&$\text{C.V. }\backslash T$
&{100}&{200}&{300}&{500}&{1000}&{2000}&&
 {100}&{200}&{300}&{500}&{1000}&{2000}\\
\midrule
100&LIL
&5.6&0.7&0.1&0.0&0.0&0.0&&
 94.5&96.6&97.4&98.3&99.4&99.8\\
&$f(B)=B^{-1/4}$
&0.3&0.1&0.0&0.0&0.0&0.0&&
 88.5&92.7&94.6&97.1&98.7&99.2\\
\midrule
200&LIL
&4.0&0.1&0.0&0.0&0.0&0.0&&
 99.6&99.7&99.9&100.0&100.0&100.0\\
&$f(B)=B^{-1/4}$
&0.4&0.0&0.0&0.0&0.0&0.0&&
 98.5&99.4&99.5&99.6&99.8&100.0\\
\midrule
500&LIL
&43.3&23.0&14.0&7.8&2.6&0.4&&
 100.0&100.0&100.0&100.0&100.0&100.0\\
&$f(B)=B^{-1/4}$
&1.3&0.1&0.0&0.0&0.0&0.0&&
 100.0&100.0&100.0&100.0&100.0&100.0\\
\bottomrule
\end{tabular}
 }
\end{center}
\par
{\scriptsize 
\textbf{Note: }{Results using either LIL-based critical values (LIL) or
based on $f(B)=B^{-1/4}$. The derandomized statistic is
based on nominal level $\tau=5\%$. We set $B=\left(\mathrm{\log }%
\,N\right)^{2}$ for $Q_{N,T,B}(\tau)$ and $\nu=5$ for $\psi_{i,NT}$. } }
\end{table}
\begin{table}[h]
\caption{\scriptsize{Empirical rejection frequencies for the test in Theorem \protect\ref%
{asy-max}, GARCH case.}}
\label{tab:test_GARCH}\captionsetup{font=small}
\par
\begin{center}
{\scriptsize \centering
\begin{tabular}{ll|
                S S S S S S
                c
                S S S S S S}
\toprule
&&\multicolumn{6}{c}{$\phi_{g} = 0.4$; $\alpha_{i} = 0$ for all $i$}&&
  \multicolumn{6}{c}{$\phi_{g} = 0.4$; $\alpha_{i} \sim N(0,1)$ for 5\% of units} \\
\midrule
$N$&$\text{Test }\backslash T$
&{100}&{200}&{300}&{500}&{1000}&{2000}&&
 {100}&{200}&{300}&{500}&{1000}&{2000}\\
\midrule
100&Thm. 1
&7.9&5.0&4.4&4.0&3.3&3.5&&
 96.5&97.5&98.0&98.7&99.6&99.9\\
&FLLM
&9.8&8.9&9.5&8.1&9.5&8.6&&
 98.6&99.5&99.6&99.8&100.0&100.0\\
&GRS
&\multicolumn{1}{c}{--}&4.6&4.6&3.5&4.3&4.2&&
 \multicolumn{1}{c}{--}&99.9&99.9&99.9&100.0&100.0\\
&PY
&21.5&20&20.2&20.2&21.0&21.7&&
 78.5&94.1&97.0&98.8&99.9&100.0\\
&AS
&10.4&8.9&8.2&8.1&7.9&6.9&&
 22.3&61.0&87.2&98.5&100.0&100.0\\
\midrule
200&Thm. 1
&7.0&4.2&3.6&3.4&4.0&4.1&&
 100.0&100.0&100.0&100.0&100.0&100.0\\
&FLLM
&9.3&9.7&9.1&9.2&8.4&6.4&&
 100.0&100.0&100.0&100.0&100.0&100.0\\
&GRS
&\multicolumn{1}{c}{--}&\multicolumn{1}{c}{--}&5.7&5.1&4.6&3.4&&
 \multicolumn{1}{c}{--}&\multicolumn{1}{c}{--}&100.0&100.0&100.0&100.0\\
&PY
&18.7&19.7&20.7&22.4&20.7&19.3&&
 86.5&98.6&99.5&100.0&100.0&100.0\\
&AS
&9.6&10.1&9.4&9.1&6.6&7.2&&
 24.9&64.1&95.6&99.4&100.0&100.0\\
\midrule
500&Thm. 1
&10.9&5.3&4.3&3.9&4.0&3.8&&
 100.0&100.0&100.0&100.0&100.0&100.0\\
&FLLM
&9.2&8.3&7.3&6.9&7.4&6.3&&
 100.0&100.0&100.0&100.0&100.0&100.0\\
&GRS
&\multicolumn{1}{c}{--}&\multicolumn{1}{c}{--}&\multicolumn{1}{c}{--}&\multicolumn{1}{c}{--}&4.0&4.3&&
 \multicolumn{1}{c}{--}&\multicolumn{1}{c}{--}&\multicolumn{1}{c}{--}&\multicolumn{1}{c}{--}&100.0&100.0\\
&PY
&18.3&21.1&19.4&18.5&19.8&19.0&&
 94.6&100.0&100.0&100.0&100.0&100.0\\
&AS
&11.4&10.3&9.1&7.5&8.1&6.3&&
 22.0&75.0&97.9&100.0&100.0&100.0\\
\bottomrule
\end{tabular}
 }
\end{center}
\par
{\scriptsize 
\textbf{Note: }{The nominal size is 5\% and powers are assessed at 5\% level
of significance; frequencies are computed across $M=1000$ Monte Carlo
samples and we set $\nu=5$ when computing $\psi_{i,NT}$.} }
\end{table}
\begin{table}[h]
\caption{\scriptsize{Empirical rejection frequencies for the decision rule of Theorem 
\protect\ref{strong-rule}, GARCH case.}}
\label{tab:derand_GARCH}\captionsetup{font=small}
\par
\begin{center}
{\scriptsize \centering
\begin{tabular}{ll|
                S S S S S S
                c
                S S S S S S}
\toprule
&&\multicolumn{6}{c}{$\phi_{g} = 0.4$; $\alpha_{i} = 0$ for all $i$}&&
  \multicolumn{6}{c}{$\phi_{g} = 0.4$; $\alpha_{i} \sim N(0,1)$ for 5\% of units} \\
\midrule
$N$&$\text{C.V. }\backslash T$
&{100}&{200}&{300}&{500}&{1000}&{2000}&&
 {100}&{200}&{300}&{500}&{1000}&{2000}\\
\midrule
100&LIL
&8.5&1.3&0.3&0.1&0.0&0.0&&
 98.3&98.6&99.4&99.5&99.4&99.8\\
&$f(B)=B^{-1/4}$
&1.0&0.2&0.0&0.0&0.0&0.0&&
 95.6&97.5&97.9&98.8&98.7&99.2\\
\midrule
200&LIL
&9.3&1.6&0.5&0.1&0.0&0.0&&
 100.0&100.0&100.0&100.0&100.0&100.0\\
&$f(B)=B^{-1/4}$
&1.9&0.0&0.0&0.0&0.0&0.0&&
 99.9&100.0&100.0&100.0&100.0&100.0\\
\midrule
500&LIL
&47.4&27.6&17.8&9.4&3.7&0.7&&
 100.0&100.0&100.0&100.0&100.0&100.0\\
&$f(B)=B^{-1/4}$
&5.9&0.5&0.1&0.0&0.0&0.0&&
 100.0&100.0&100.0&100.0&100.0&100.0\\
\bottomrule
\end{tabular}
}
\end{center}
\par
{\scriptsize 
\textbf{Note: }{Results using either LIL-based critical values (LIL) or
critical values based on $f(B)=B^{-1/4}$. The derandomized statistic is
based on nominal level $\tau=5\%$. We set $B=\left(\mathrm{\log }%
\,N\right)^{2}$ for $Q_{N,T,B}(\tau)$ and $\nu=5$ for $\psi_{i,NT}$.} }
\end{table}

We conclude this section with a further experiment, where we allow for cross-sectional heteroskedasticity. We simulate assets with higher variance, the higher $\alpha_i$. The tendency of severely misspriced assets to exhibit large variances is a clear feature of the data employed in our empirical analysis (see Section \ref{empirical}), especially for the most severely misspriced assets. Hence, we consider the case where $\sigma^{2}_{\xi_i} = Var\left(\xi_{i,t}\right)$. In the Gaussian and Student's $t$ DGP, we simulate $\bf{\xi}_t$ so that  $\sigma_{\xi_i} = 1+\vartheta|\alpha_{i}|$, while for the GARCH case set $\omega_{i} = (1+\vartheta|\alpha_{i}|)(1-\beta_i-\pi_i)$ so that $Var\left(\xi_{i,t}\right) = \left(1+\vartheta|\alpha_{i}|\right)^2$ under all DGPs, and we consider ten equally spaced values of $\vartheta$ between one and five percent. Results are in Figure \ref{fig:asset_specific}, showing that our test has unit power for all levels of asset specific variance. Conversely, the power of all the other tests decreases, with the test by \cite{feng2022high} exhibiting the largest loss of power under all DGPs. This makes sense since the test by \cite{feng2022high} is a max-$t^{2}$ test, hence being driven by both the alpha and the asset-specific variance.\footnote{Further results are in Section \ref{simulations_bis} of the Supplement, where we report results for: the
case $\phi _{g}=0$ (Section \ref{MC_phizero}); the cases where $g_{t}$ is a weak or semi-strong
omitted factor (Section \ref{MC_omitted_factors_strength}); additional power analyses for different levels of sparsity under the alternative (Section \ref{MC_sparsity}); the case when only one factor is strong,
while the others are at most semi-strong (Section \ref{MC_pricing_factors_streng}); the cases of non-tradable
factors (studied in Section \ref{nontradable}) and latent factors (studied
in Section \ref{latent}) (Sections \ref{MC-FF} and \ref{MC-latent} respectively); using alternative critical values in the rejection rule based on Theorem \ref{asy-max} (Section \ref{app:MC_newCV}).}
\begin{figure}[t]
\captionsetup{font=scriptsize} \centering
\begin{subfigure}[t]{0.325\textwidth}
                 \centering
                 \includegraphics[ width=0.98\textwidth, trim = 0cm 1.5cm 0cm 1.5cm, clip]{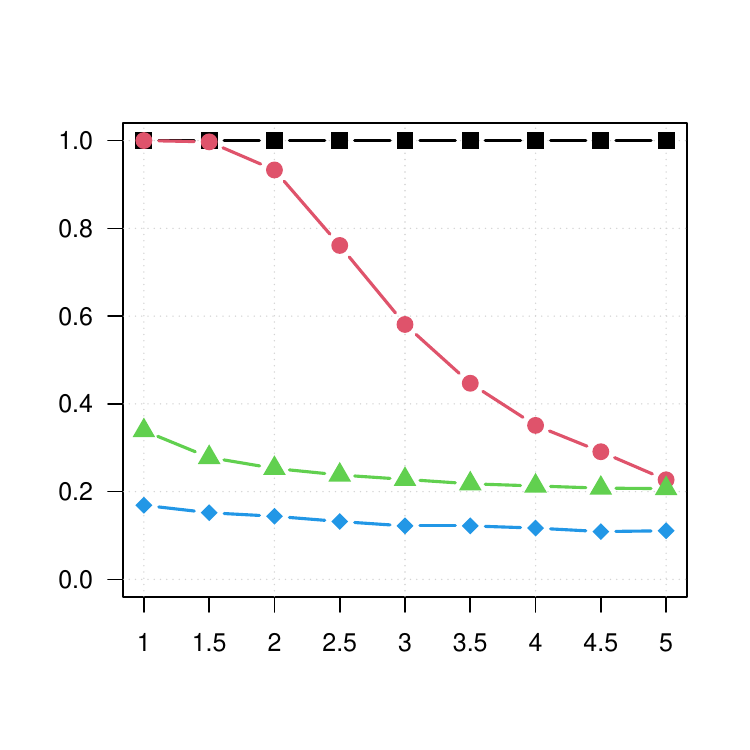}
                               \caption{Gaussian DGP, $\phi_{g} = 0.40$.}
         \end{subfigure}
\begin{subfigure}[t]{0.325\textwidth}
                 \centering
                 \includegraphics[width=0.98\textwidth, trim = 0cm 1.5cm 0cm 1.5cm, clip]{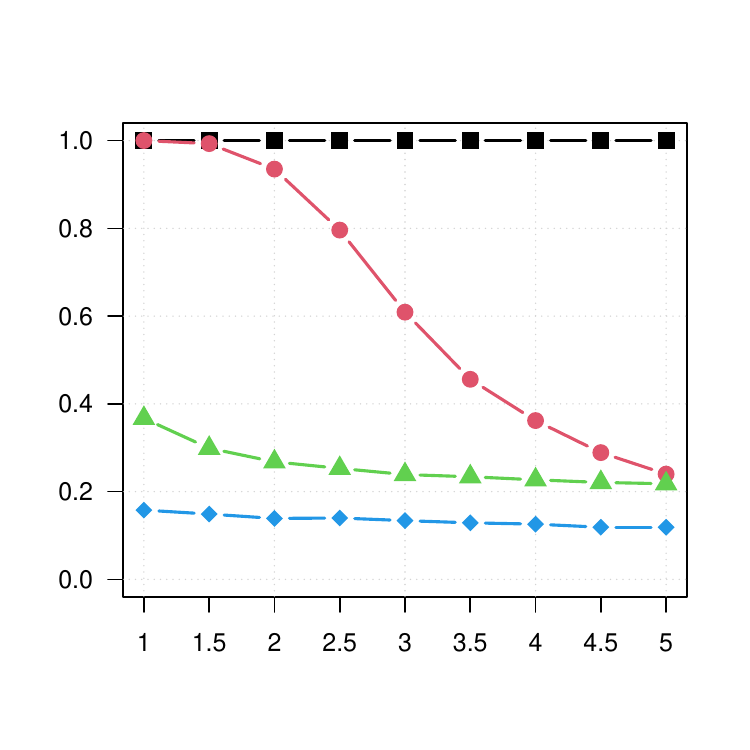}
                 \caption{Student's $t$ DGP, $\phi_{g} = 0.40$.}
         \end{subfigure}
\begin{subfigure}[t]{0.325\textwidth}
                 \centering
                 \includegraphics[width=0.98\textwidth, trim = 0cm 1.5cm 0cm 1.5cm, clip]{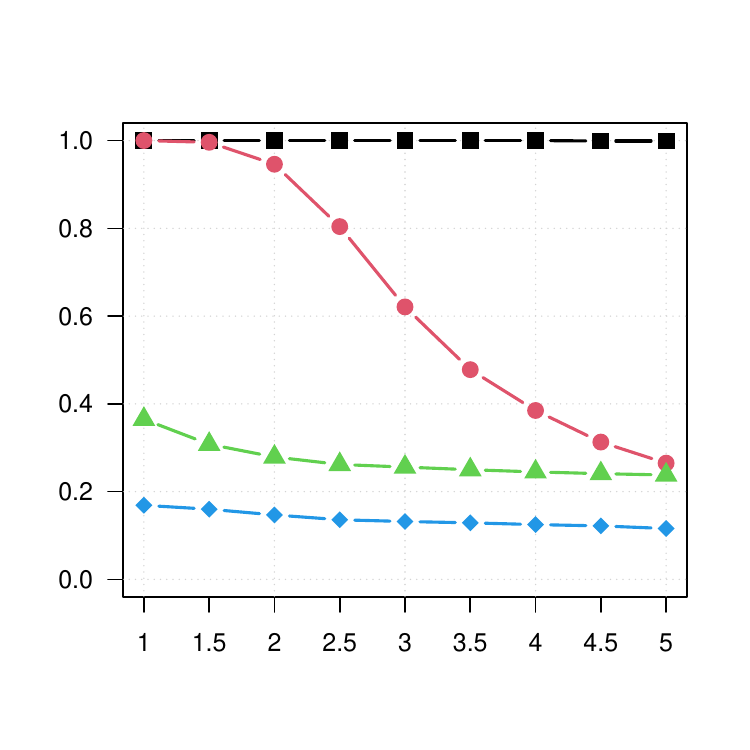}
\caption{GARCH(1,1) DGP, $\phi_{g} = 0.40$.}
         \end{subfigure}
\caption{Power curves. The horizontal axis reports $\vartheta$ in the  magnifying factor $(1+\vartheta|\alpha_i|)$ for asset-specific variances. The legend is as follows: our test, black squares; \cite{feng2022high}, red dots; \cite{pesaran2023testing}, green triangles; \cite{ardia2024robust}, light blue diamonds.}
\label{fig:asset_specific}
\end{figure}

\section{Empirical illustration\label{empirical}}

We illustrate our procedure by testing whether several linear factor pricing
models correctly price the constituents of the S{\&}P 500 index. We collect monthly data on all stocks that were part of the S\&P 500 for at
least five years between January 1985 and December 2024, using
simple percentage returns on the  $i$-th stock gross of  dividend yield
\begin{equation}
r_{i,t}=100\frac{P_{i,t}-P_{i,t-1}}{P_{i,t-1}}+\frac{DY_{i,t}}{12},
\label{pmodel}
\end{equation}%
where $P_{i,t}$ is the end of the month stock price and $DY_{i,t}$ is the percent
per annum dividend yield.\footnote{%
To ensure that the index accurately represents the US stock market, S{\&}P
regularly updates its composition. We account for these changes by revising
the set of included assets every month. Security data are sourced from \emph{%
Datastream}, while we obtain those on the pricing factors (and on the
risk-free rate) from the website of Professor French.} We define the excess returns as $y_{i,t}$, and test six linear factor pricing models
that are all encompassed by the following regression 
\begin{equation}  \label{eq:general_pricingmodel}
y_{i,t}=\alpha _{i}+\beta _{i,1}\text{MKT}_{t}+\beta _{i,2}\text{SMB}%
_{t}+\beta _{i,3}\text{HML}_{t}+\beta _{i,4}\text{RMW}_{t}+\beta _{i,5}\text{%
CMA}_{t} + \beta_{i,6}\text{MOM}_{t}+u_{i,t},
\end{equation}%
where: $\text{MKT}_{t}$ is the excess return on the market, $\text{SMB}_{t}$
the \emph{size} factor, $\text{HML}_{t}$ the \emph{book-to-market} factor, $%
\text{RMW}_{t}$ the \emph{profitability} factor, $\text{CMA}_{t}$ the \emph{%
investment strategy} factor, and $\text{MOM}_{t}$ the \emph{momentum} factor.%
\footnote{%
All factors are observable and tradable; see the website
of Professor French for their description.} Both the excess return $y_{i,t}$
and the market factor use the one-month Treasury bill rate as risk-free
rate. The first model that we test is the CAPM, which is obtained when only
the market factor is considered. We then consider a two-factor (FF2
henceforth) model based on the market and momentum factors, as well as the
usual three- and five-factors models of \citet[][FF3]{fama1993common} and %
\citet[][FF5]{fama2015five}, which  we also augment with
momentum  to obtain a four- (FF4) and a six-factors (FF6) model. In line with \citet{fan2015power} and \citet{pesaran2023testing}, we mitigate  the impact of possible time variation in the factor loadings by running our inferential procedure on  5-years rolling windows.\footnote{We use $T=60$ (i.e. $%
5 $ years) in estimation, while $N$ ranges between $437$ and $568$. Within each window, we discard securities with at least one missing
observation. In Section \ref{mht} in the Supplement, we show that using the derandomized confidence function as suggested in Section \ref{derandom} allows to control for the family-wise rejection frequency even when the number of windows passes to infinity, with no changes or corrections required to the nominal level of the individual tests.}

\subsection{Results\label{results}}

Table \ref{tab:derand_RWS&P} reports the percentage of windows for which we
can reject the null that a pricing model is correctly specified. These
empirical rejection frequencies are computed using our derandomized
procedure based on $B=(\log N)^2$, at $5\%$ nominal  significance level. We calculate $\psi _{i,NT}$ using both $\nu=5$ (in the light of the very good finite sample results found in Section \ref{simulations}), and $\nu=4$ (by way of robustness check, as suggested in our guidelines in Section \ref{practice}). We also report results obtained with the test by \cite{feng2022high}, as our closest benchmark.  

\begin{table}[h]
\caption{\scriptsize{{Results of the derandomized procedure for constituents of the S{\&}P 500 index.}}}
\label{tab:derand_RWS&P}\captionsetup{font=small}
{\scriptsize 
\begin{tabular}{ll|cccccc}
\toprule
&& CAPM & FF2 & FF3 & FF4 & FF5 & FF6 \\ 
\midrule 
Full Sample &  $\nu  = 4$ & 0.607 & 0.624 & 0.710 & 0.724 & 0.840 & 0.840 \\ 
&   $\nu  = 5$ &0.448 & 0.479 & 0.595 & 0.631 & 0.793 & 0.799 \\ 
& \cite{feng2022high} & 0.205& 0.117 & 0.348 & 0.245 & 0.438 & 0.419 \\ 
\midrule
Asian financial crisis &   $\nu  = 4$&  1.000 & 1.000 & 1.000 & 1.000 & 1.000 & 1.000 \\ 
(Jul-1997 to Dec-1998) &  $\nu = 5$&{0.860} & {0.930} & {1.000} & {1.000} & {1.000} & {1.000} \\ 
& \cite{feng2022high} & {0.722} & {0.167} & {0.722} & {0.444} & {0.333} & {0.389} \\
\midrule
Dot-com Bubble &  $\nu  = 4$ & 0.875 & 0.844 & 1.000 & 1.000 & 1.000 & 1.000 \\ 
 (Mar-2000 to Oct-2002)&  $\nu = 5$&{0.546} & {0.535} & {1.000} & {1.000} & {1.000} & {1.000} \\ 
& \cite{feng2022high} & {0.000} & {0.000} & {0.438} & {0.219} & {0.938} & {0.750} \\
\midrule
Great financial crisis &   $\nu=4$ & 1.000 & 1.000 & 1.000 & 1.000 & 1.000 & 1.000 \\ 
(Dec-2007 to Jun-2009) &  $\nu = 5$&{0.936} & {0.936} & {0.955} & {0.951} & {0.961} & {0.970} \\
& \cite{feng2022high} & 0.895 & 0.474 & 0.842 & 0.632 & 0.842 & 0.368 \\
\midrule
COVID &   $\nu=4$ & 0.765 & 0.588 & 0.529 & 0.647 & 0.647 & 0.647 \\ 
(Jan-2020 to May-2021) &  $\nu = 5$&0.448 & 0.326 & 0.304 & 0.336 & 0.392 & 0.412 \\ 
& \cite{feng2022high} & 0.118 & 0.000 & 0.000 & 0.000 & 0.000 & {0.000} \\
\hline
Market distress &  $\nu=4$ & 0.907 & 0.860 & 0.907 & 0.930 & 0.930 & 0.930 \\ 
&  $\nu = 5$&0.679 & 0.665 & 0.852 & 0.858 & 0.871 & 0.877 \\ 
& \cite{feng2022high} & 0.372 & 0.140 & 0.500 & 0.314 & 0.605 & 0.442 \\
\midrule
Market stability &  $\nu=4$ & 0.530 & 0.563 & 0.659 & 0.671 & 0.817 & 0.817 \\ 
& $\nu = 5$&0.389 & 0.432 & 0.528 & 0.573 & 0.773 & 0.778 \\ 
& \cite{feng2022high} & 0.162 & 0.111 & 0.308 & 0.228 & 0.395 & 0.413 \\
\bottomrule 
\end{tabular}
}
\par
\par
{\tiny 
\begin{tablenotes}
\item Percentage of rolling windows for which the
null is rejected. We consider both overall percentages, i.e.\ for the whole
sample between January 1985 and December 2024, and those specific to relevant subperiods. The significance level is always $\tau=5\%$
and all the windows contain $T=60$ observations. Calculation of $\psi_{i,NT}$
is based on $\nu=4$ and $\nu = 5$. The derandomized procedure uses $B=(\mathrm{log}\, N)^2$ repetitions of the test and thresholds based on $%
f(B)=B^{-1/4}$.
\end{tablenotes}
}
\par
\end{table}

The first sub-panel of Table \ref{tab:derand_RWS&P} shows rejection
frequencies for the whole sample. Results suggest that the most parsimonious models - the CAPM and FF2 - are
the least rejected models (i.e.\ the ones that produce the lowest percentage
of rejections of the null of zero  alphas).\footnote{This result is somewhat surprising in that, if the true model was a larger model that nests the CAPM or FF2, then underconditioning would entail a bias in the alpha estimates of the more parsimonious models. However, \citet{moskowitz2025} show that, in the absence of mispricing, the CAPM performs better in terms of alpha tests than several prominent multifactor models. This is because, while multifactor betas can help capture mispricing, persistence in those betas leads the multifactor models to distort expected returns well after that information gets priced correctly. Also, it is well known that  individual stocks have poorly estimated betas, with the estimation error acting both as a source of downward bias due to an errors-in-variables problem \citep{blume1975betas} and as a driver of the aforementioned persistence. However,
\citet{ang2020using} argue that, although the measurement errors in the
betas are larger in individual stocks relative to e.g. portfolio returns,
the larger cross-sectional spread in the betas of individual stocks more
than offsets this error, leading to a more accurate estimate of the market
risk premium.} 
Results using the approach of \cite{feng2022high} also generally favor the use of more parsimonious models, i.e.\ CAPM and FF2. However, the rejection frequencies appear unrealistically small, as they suggest that a linear model based on  market and momentum  correctly prices large cap US stocks for almost 90\% of the sample, which comprises several periods of market turmoil.

The next four sub-panels of Table \ref{tab:derand_RWS&P} show rejection
frequencies over four periods of market turmoil:  the Asian financial crisis,
the Burst of the Dot-com Bubble, the Global Financial Crisis, and the
COVID-19 pandemic.\footnote{%
For the COVID-19 pandemic, we consider rolling windows between January 2020
and May 2021, where the end date corresponds to the termination of lockdown
policies in most of the world. The dates for the other periods of market
turmoil are set as in \cite{pesaran2023testing}.} Our procedure suggests  that  these models almost never   price large cap US stocks during the Asian and the Global Financial Crisis. The picture is only slightly  more positive during the early 2000s crisis.  Finally, we see that the size and book-to-market factors play an important
role during the COVID-19 pandemic, as the three-factor model is the best
over this period.\footnote{%
To interpret this result note that, while the COVID-19 period shares with
other crises the fact that the stock market yielded low returns and was
characterized by high volatility and low liquidity, it also has features
that make it very distinct. Indeed, during the COVID-19 pandemic households
were required to stay home to slow the spread of the virus and a variety of
firms were severely restricted in producing their goods and services,
essentially constraining output production and limiting consumption
decisions. This distinctive feature of the COVID-19 period implies that,
while uncertainty about the end of the pandemic was very high, in the short
term recession fears and low growth expectations strongly characterize that
period \citep{gormsen2023financial}. It is not surprising, therefore, that
value stocks --- generally considered long-horizon investments --- came
under huge pressure as economic uncertainty prompted investors to shorten
their time horizons, and indeed we find that the value factor contributes to the slightly better performance of larger models during this
period (further results are available upon request). This is consistent with
the evidence in \cite{campbell2025drives}, who document that value
experienced a historically striking drop in performance during the pandemic.%
}  

This sub-periods analysis also indicates that some of the results using the approach   of   \citet{feng2022high} may be overly optimistic. Indeed, their test suggests that the FF2 model perfectly explains the cross-section of expected excess returns on large cap US stocks during the Dot-com Bubble and the COVID pandemic. More generally, and also looking at the last two sub-panels, their procedure tells us that FF2 almost always prices all large cap US stocks, no matter whether the market is experiencing a period of distress or stability. While possible, this result does not seem plausible, and we argue that it is likely due to the loss of power suffered by the test of \citet{feng2022high} when the asset-specific variance of the misspriced assets increases, as highlighted in the earlier simulation section. We study this presumption starting from  Figure \ref{fig:alpha_var}. Its upper panels present the time-series of the three largest estimated alphas for model FF2 (each point of the series represents the absolute value of the largest estimate of $\alpha_i$ for a given window). We see that these are maximal between January 1999 and June 2009, which is the period corresponding to the red shaded areas. The lower panels report the time series of asset-specific standard deviation for the assets whose estimated alpha is reported in panels (a) - (c). These standard deviations are strongly associated with the estimated alphas during the highlighted period, suggesting that, over this period, we are in a setting which is very similar to that of our Monte Carlo results in Figure \ref{fig:asset_specific}, where we would expect the test by \cite{feng2022high} to display low power. Indeed, their test rejects the FF2 model for  14\% of the windows between January 1999 and June 2006, while our procedure based on $\nu=4$ ($\nu=5$, respectively) rejects 92\% (74\%, respectively) of the times. The results in Table \ref{tab:derand_RWS&P} and Figure \ref{fig:alpha_var} suggest a similar conclusion when looking at the Asian financial crisis and the outbreak of the COVID pandemic. These empirical findings highlight the usefulness of our testing procedure particularly when  asset specific variances are highly associated with the size of absolute pricing errors, as seems to be the case for US stocks during periods of market turmoil.
\begin{figure}[!t]
\captionsetup{font=scriptsize} \centering
\begin{subfigure}[t]{0.325\textwidth}
                 \centering
                 \includegraphics[ width=0.98\textwidth, trim = 0cm 1.5cm 0cm 1.5cm, clip]{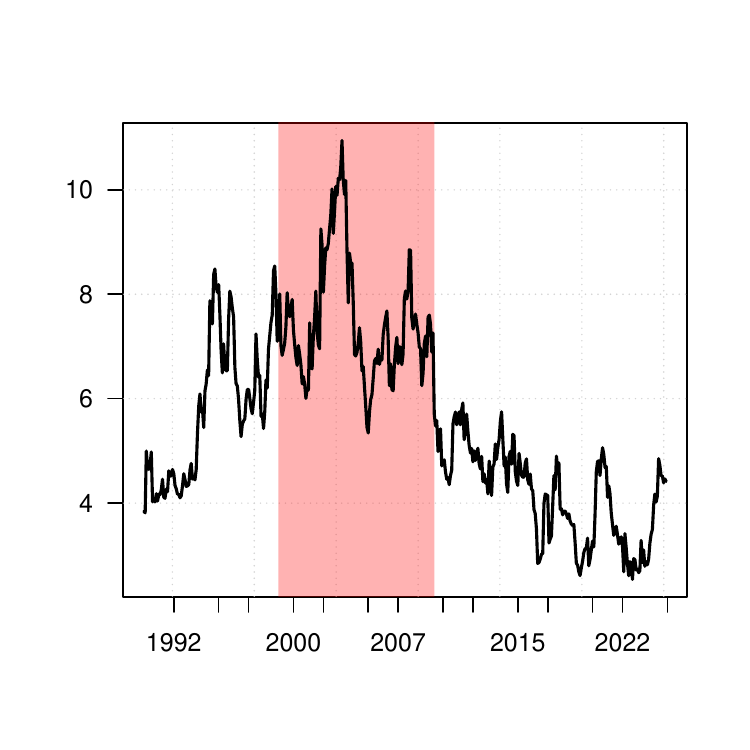}
\caption{Largest  $|\hat{\alpha}_i|$.}
         \end{subfigure}
\begin{subfigure}[t]{0.325\textwidth}
                 \centering
\includegraphics[width=0.98\textwidth, trim = 0cm 1.5cm 0cm 1.5cm, clip]
           {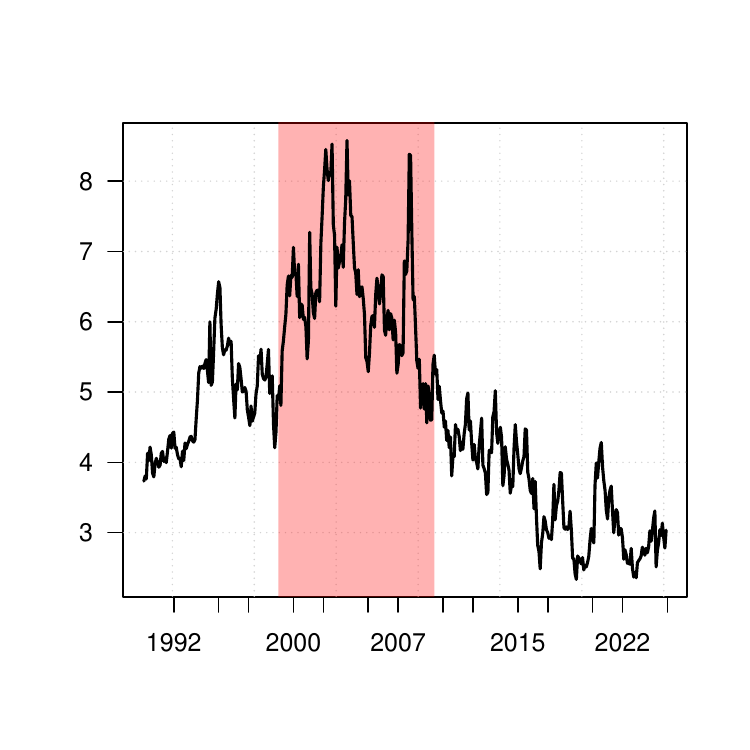}
\caption{Second largest  $|\hat{\alpha}_i|$.}
         \end{subfigure}
\begin{subfigure}[t]{0.325\textwidth}
                 \centering
\includegraphics[width=0.98\textwidth, trim = 0cm 1.5cm 0cm 1.5cm, clip]{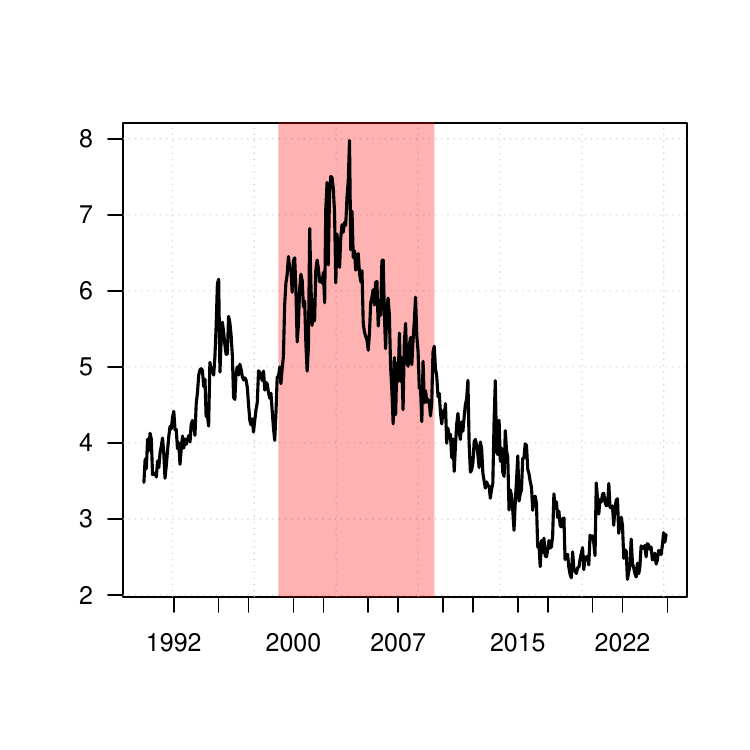}
\caption{Third largest $|\hat{\alpha}_i|$.}
         \end{subfigure}
         
\begin{subfigure}[t]{0.325\textwidth}
                 \centering
                \includegraphics[width=0.98\textwidth, trim = 0cm 1.5cm 0cm 1.5cm, clip]{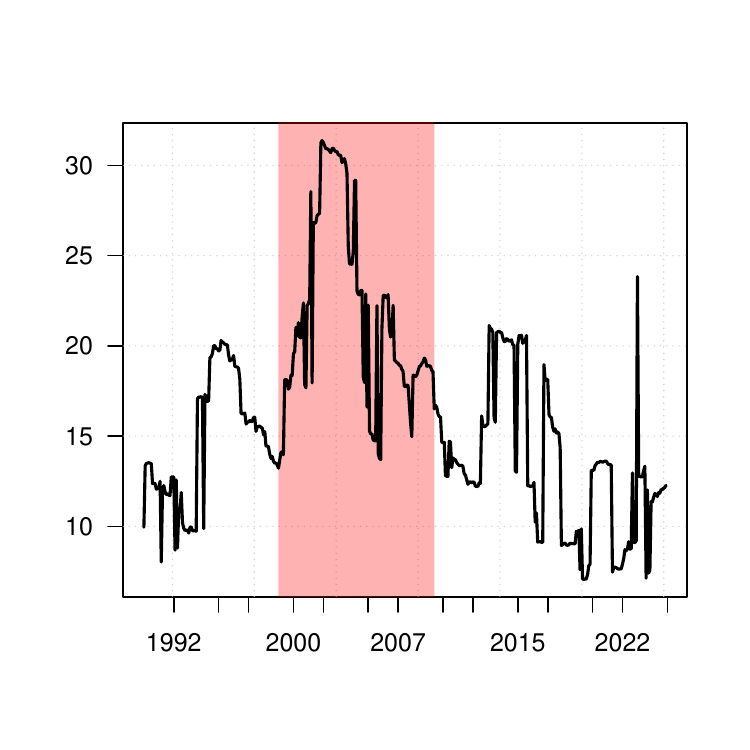}
\caption{Associated standard deviation.}
         \end{subfigure}
\begin{subfigure}[t]{0.325\textwidth}
                 \centering
\includegraphics[width=0.98\textwidth, trim = 0cm 1.5cm 0cm 1.5cm, clip]{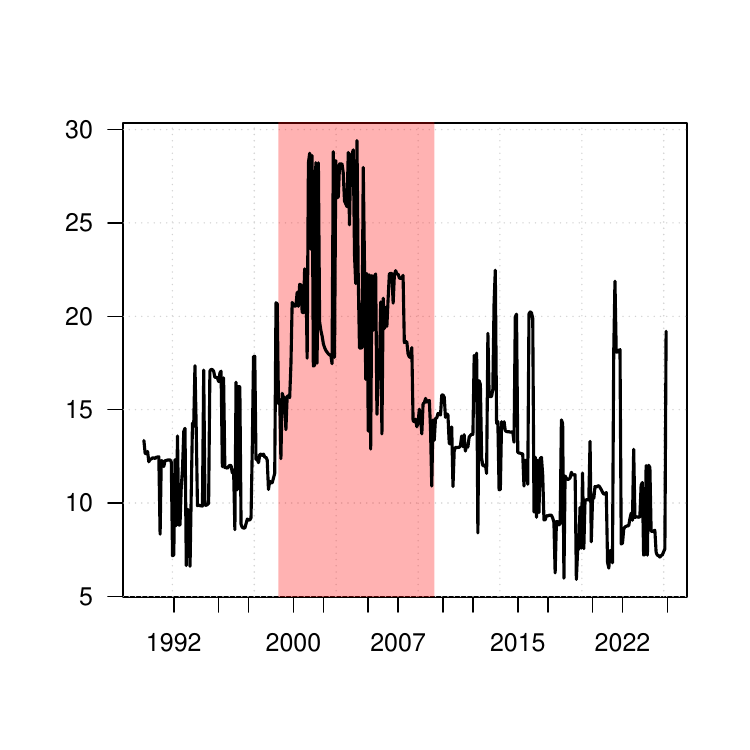}
\caption{Associated standard deviation.}
          \end{subfigure}
\begin{subfigure}[t]{0.325\textwidth}
                 \centering
\includegraphics[width=0.98\textwidth, trim = 0cm 1.5cm 0cm 1.5cm, clip]{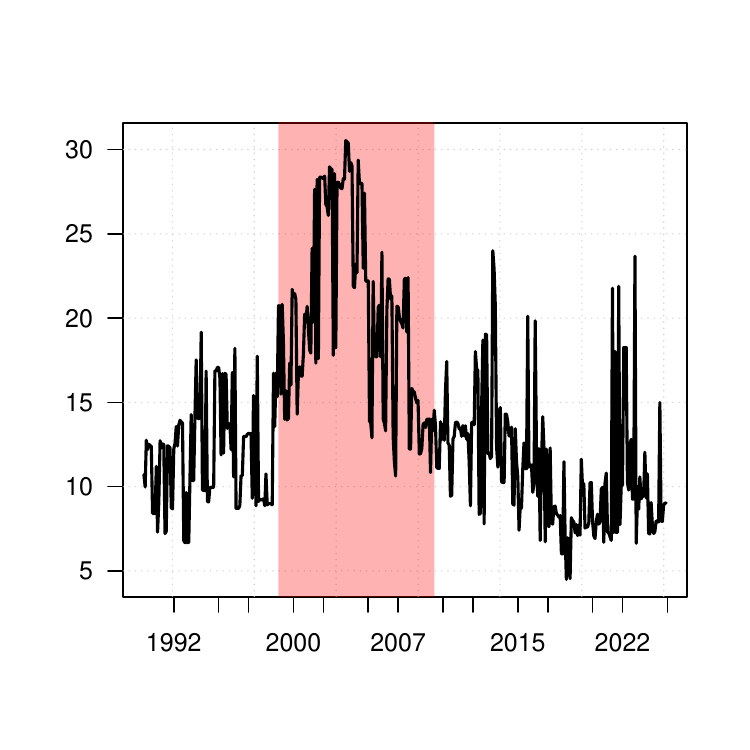}
\caption{Associated standard deviation.}
         \end{subfigure}
\caption{Upper panels: three largest estimated $\alpha$s under model FF2 for each window. Lower panels: idiosyncratic standard deviations  of the  associated stocks under model FF2. The red shaded area denotes the period January 1999 to June 2009.}
\label{fig:alpha_var}
\end{figure}

\section{Conclusions\label{conclusion}}

We propose a novel, general methodology to test for
\textquotedblleft zero alpha\textquotedblright\ in linear factor pricing
models, with observable or latent, tradable or non-tradable pricing factors. Our proposed test relies on a randomization scheme and can be used even when residuals exhibit 
 conditional heteroskedasticity, strong cross-sectional
dependence and non-Gaussianity, and allows for the number
of assets, $N$, to pass to infinity, also at a faster rate than the sample
size $T$. Extensive Monte Carlo analysis shows that the test has very good
power properties, and that it is - compared to other extant approaches -
the one that better controls the Type I Error probability in all scenarios.

\smallskip

We discuss two possible extensions which are under investigation by the authors. First, building on the theory developed herein, a randomized
test could be developed to test for the null hypotheses $\mathbb{H}%
_{0,i}\,:\,\alpha _{i}=0$ for $i=1,\dots ,N$, while controlling for
multiple testing.\footnote{By the same token, the test statistic could be constructed also when the equality is replaced by an inequality, such as $\mathbb{H}_{0,i}\,:\,\alpha _{i}\leq 0$, as in %
\citet{giglio2021thousands}.} Indeed, the individual test statistics would be
perturbed by adding randomness \textit{independently} across $i$, thus making the
randomized statistics (conditionally) independent across $i$, which would
facilitate the application of customarily employed procedures for size control under multiple testing. This would find a natural application in testing
mutual/hedge-funds performances, as well as assessing trading
strategies. \\
A second extension/application of randomized tests based on rates of convergence is model comparison. Consider two models for $y_{i,t}$, which will be denoted
henceforth using the superscripts $^{\left( 1\right) }$ and $^{\left(
2\right) }$, with $\alpha _{i}^{\left( 1\right)
} $ and $\alpha _{i}^{\left( 2\right) } $ respectively. The two models can be
compared to check whether $\alpha _{i}^{\left( 1\right)
} $ and $\alpha _{i}^{\left( 2\right) } $ differ in some
sense; in the spirit of the max-type test proposed in this paper, one can test  whether the maximally
selected difference of absolute alphas is zero, viz.
\begin{equation}
\mathbb{H}_{0}:\max_{1\leq i\leq N}\left\vert \left\vert \alpha _{i}^{\left(
1\right) }\right\vert -\left\vert \alpha _{i}^{\left( 2\right) }\right\vert
\right\vert =0.  \label{model_selection}
\end{equation}%
Under (\ref{model_selection}), the worst-case scenario discrepancy between
the alphas of the two models is zero, indicating a comparable
performance.    A test can be constructed - adapting our methodology - by defining $
\psi _{i,NT}^{\left( 1,2\right) }=\left\vert T^{1/\nu }\left[ \left\vert 
\widehat{\alpha }_{i}^{\left( 1\right) }\right\vert -\left\vert \widehat{%
\alpha }_{i}^{\left( 2\right) }\right\vert \right] /\widehat{s}_{NT}^{\left(
1,2\right) }\right\vert ^{\nu /2}, $
where $\widehat{s}_{NT}^{\left( 1,2\right) }$ denotes a scaling factor.%
\footnote{%
Even in this case, any scaling factor which is bounded away from zero and
infinity, and which renders the estimators $\widehat{\alpha }_{i}^{\left(
1\right) }$ and $\widehat{\alpha }_{i}^{\left( 2\right) }$ scale free, can
be used e.g., $
\widehat{s}_{NT}^{\left( 1,2\right) }=\sqrt{\sum_{i=1}^{N}\left( \left( 
\widehat{u}_{i,t}^{\left( 1\right) }\right) ^{2}+\left( \widehat{u}%
_{i,t}^{\left( 2\right) }\right) ^{2}\right) /NT,}$
where $\widehat{u}_{i,t}^{\left( 1\right) }$ and $\widehat{u}_{i,t}^{\left(
2\right) }$\ are the residuals from models $^{\left( 1\right) }$ and $%
^{\left( 2\right) }$ respectively, by adapting the definition of $\widehat{s}%
_{NT}$ in (\ref{s_NT}).} It is immediate to see that, under the null, $\psi
_{i,NT}^{\left( 1,2\right) }=o_{a.s.}\left( 1\right) $, whereas it would
diverge at a rate $T^{1/2}$ as long as $\left\vert \alpha _{i}^{\left(
1\right) }\right\vert \neq \left\vert \alpha _{i}^{\left( 2\right)
}\right\vert $\ for at least one $i$. Hence, our approach can be applied 
\textit{verbatim} to $\psi _{i,NT}^{\left( 1,2\right) }$, perturbing each $%
\psi _{i,NT}^{\left( 1,2\right) }$ by an \textit{i.i.d.} standard Gaussian
shock $\omega _{i}$ and using  $
Z_{NT}^{\left( 1,2\right) }=\max_{1\leq i\leq N}\left\vert \psi
_{i,NT}^{\left( 1,2\right) }+\omega _{i}\right\vert$. This test could be applied to compare two non-nested models, but also to compare a nested and a nesting model - in the latter case, upon not
rejecting $\mathbb{H}_{0}$ of (\ref{model_selection}), the applied user can
conclude that the nested, \textquotedblleft smaller\textquotedblright\ model
does not result in any worsening. Building on the previous paragraph, model comparison may also be approached as a multiple testing problem in the cross-section. In particular, we may consider the $N$-dimensional sequence of pair-wise null hypotheses $\left\{\mathbb{H}^{(m,j)}_{0,1},\dots,\mathbb{H}_{0,N}^{(m,j)}\right\}$ with generic entry $
\mathbb{H}^{(m,j)}_{0,i}\,:\, |\alpha_{i,m}| = |\alpha_{i,j}|
$, and test it against the one-sided alternative  $|\alpha_{i,j}|>|\alpha_{i,m}|$. \\
In conclusion, we would like to note that our approach - whilst focused here on the issue of testing for alpha in asset pricing models - can be exported to other set-ups. Using essentially the same arguments as in this paper, for example, one could test for the two sample problem in high dimension, or for the equality of intercept and slope in a high-dimensional regression model, or for \textquotedblleft pooling versus not pooling\textquotedblright\ in a large panel data model. In all cases, our approach would be applicable even in the presence of a whole world of misspecification for the errors - thus, being a viable alternative to other, more standard, methodologies when these fail due to their underlying assumptions not being satisfied or due to the need for a robust estimator of the long-run covariance, or for a large covariance matrix estimation.

\begin{adjustwidth}{-5pt}{-5pt}

{\footnotesize {\ 
\bibliographystyle{chicago}
\bibliography{LTbiblio}
} }

\end{adjustwidth}

\clearpage
\newpage


\doublespacing

\renewcommand*{\thesection}{\Alph{section}}

\setcounter{section}{0} 
\setcounter{equation}{0} \setcounter{lemma}{0} \setcounter{theorem}{0} %
\renewcommand{\theassumption}{A.\arabic{assumption}} 
\renewcommand{\thetheorem}{A.\arabic{theorem}} \renewcommand{\thelemma}{A.%
\arabic{lemma}} \renewcommand{\theproposition}{A.\arabic{proposition}} %
\renewcommand{\thecorollary}{A.\arabic{corollary}} \renewcommand{%
\theequation}{A.\arabic{equation}}

\section{Further simulations\label{simulations_bis}}

This section extends the Monte Carlo analysis of the main body by studying
the finite sample performances of the test based on Theorem \ref{asy-max}
under different data generating processes. Specifically, Section \ref%
{MC_moreBenchmarks} displays empirical rejection frequencies for other
competing approaches that were not included in the main body; Sections \ref%
{MC_phizero} and \ref{MC_omitted_factors_strength} modify the persistence
and factor structure for the error terms of the pricing models $u_{i,t}$,
while Section \ref{MC_pricing_factors_streng} discusses the implications of
changes in the strength of the pricing factors $%
f_{t}=(f_{1,t},f_{2,t},f_{3,t})^{\prime }$. Section \ref{MC_sparsity} studies how the power of the tests varies as a function of the percentage of assets that are misspriced under the alternative. Sections \ref{MC-FF} and \ref{MC-latent} present finite results when pricing factors are non-tradable and latent, respectively. Section \ref{app:MC_newCV} contains size and power results for the test in Theorem  \ref{asy-max} using alternative critical values.

\subsection{Further competing tests}
\label{MC_moreBenchmarks}
Tables \ref{tab:test_Gaussian_supp} to \ref{tab:test_GARCH_supp} report
empirical rejection frequencies for the tests of \citet[][FLY
henceforth]{fan2015power}, \citet[][GOS]{gagliardini2016time} under the DGPs
of the main body.\footnote{
The test of FLY is based on estimating the covariance matrix of $u_{t}$ with
the POET method of \cite{fan2013large}. As in \cite{fan2015power}, we
consider the \emph{soft} thresholding function; results based on hard and
SCAD thresholding \citep{fan2001variable} are numerically identical and
available upon request.} Here and in what follows, we do not report results for sample sizes $T=1000$ and $T=2000$ as the computational cost of both approaches was excessively demanding. Finally, for completeness, we also repeat the results of our
testing procedure based on Theorem \ref{asy-max}. As for the average-type
test of \cite{pesaran2023testing}, these approaches substantially
over-reject the null under $\mathbb{H}_0$. Estimation of a large dimensional
covariance matrix in the presence of strongly cross-sectionally correlated
and mildly persistent residuals  drives these over-rejections.

\begin{table}[!t]
\caption{Empirical rejection frequencies for the test in Theorem \protect\ref%
{asy-max}, Gaussian case.}
\label{tab:test_Gaussian_supp}\captionsetup{font=small}
\par
\begin{center}
{\scriptsize 
\begin{tabular}{ll|
                S S S S S S
                c
                S S S S S S}
\toprule
&&\multicolumn{6}{c}{$\phi_{g}=0.4$; $\alpha_i=0$ for all $i$}&&
  \multicolumn{6}{c}{$\phi_{g}=0.4$; $\alpha_i\sim N(0,1)$ for 5\% of units} \\
\midrule
$N$&$\text{Test }\backslash T$
&{100}&{200}&{300}&{500}&{1000}&{2000}&&
 {100}&{200}&{300}&{500}&{1000}&{2000}\\
\midrule
100&Thm.\ 1
&5.8&4.2&4.0&3.9&4.0&3.7&&
 93.5&96.0&96.9&97.5&99.3&99.6\\
&FLY
&35.3&21.8&15.0&13.8&
 \multicolumn{1}{c}{--}&\multicolumn{1}{c}{--}&&
 99.6&99.7&99.8&100.0&
 \multicolumn{1}{c}{--}&\multicolumn{1}{c}{--}\\
&GOS
&21.2&20.0&17.9&19.3&
 \multicolumn{1}{c}{--}&\multicolumn{1}{c}{--}&&
 79.2&94.5&97.0&98.8&
 \multicolumn{1}{c}{--}&\multicolumn{1}{c}{--}\\
\midrule
200&Thm.\ 1
&7.1&3.9&3.6&3.5&3.4&3.2&&
 99.7&99.9&100.0&100.0&99.9&100.0\\
&FLY
&38.6&18.1&15.3&11.7&
 \multicolumn{1}{c}{--}&\multicolumn{1}{c}{--}&&
 100.0&100.0&100.0&100.0&
 \multicolumn{1}{c}{--}&\multicolumn{1}{c}{--}\\
&GOS
&21.4&17.6&18.3&21.7&
 \multicolumn{1}{c}{--}&\multicolumn{1}{c}{--}&&
 86.3&98.7&99.5&100.0&
 \multicolumn{1}{c}{--}&\multicolumn{1}{c}{--}\\
\midrule
500&Thm.\ 1
&7.0&4.6&4.0&3.8&3.7&3.6&&
 100.0&100.0&100.0&100.0&100.0&100.0\\
&FLY
&45.8&26.8&19.0&12.5&
 \multicolumn{1}{c}{--}&\multicolumn{1}{c}{--}&&
 100.0&100.0&100.0&100.0&
 \multicolumn{1}{c}{--}&\multicolumn{1}{c}{--}\\
&GOS
&18.1&22.2&22.4&19.2&
 \multicolumn{1}{c}{--}&\multicolumn{1}{c}{--}&&
 93.6&100.0&100.0&100.0&
 \multicolumn{1}{c}{--}&\multicolumn{1}{c}{--}\\
\bottomrule
\end{tabular}
 }
\end{center}
\par
{\scriptsize 
\textbf{Note: }{The nominal size is 5\% and powers are assessed at 5\% level
of significance; frequencies are computed across $M=1000$ Monte Carlo
samples and we set $\nu=5$ when computing $\psi_{i,NT}$. Results are in percentage points.} }
\end{table}

\begin{table}[!t]
\caption{ Empirical rejection frequencies for the test in Theorem \protect
\ref{asy-max}, Student's $t$ case.}
\label{tab:test_T_supp}\captionsetup{font=small}
\par
\begin{center}
{\scriptsize \centering
\begin{tabular}{ll|
                S S S S S S
                c
                S S S S S S}
\toprule
&&\multicolumn{6}{c}{$\phi_{g}=0.4$; $\alpha_i=0$ for all $i$}&&
  \multicolumn{6}{c}{$\phi_{g}=0.4$; $\alpha_i\sim N(0,1)$ for 5\% of units} \\
\midrule
$N$&$\text{Test }\backslash T$
&{100}&{200}&{300}&{500}&{1000}&{2000}&&
 {100}&{200}&{300}&{500}&{1000}&{2000}\\
\midrule
100&Thm.\ 1
&5.8&4.6&3.9&3.8&3.8&3.7&&
 88.8&93.0&94.7&96.3&98.6&99.4\\
&FLY
&33.7&17.5&13.9&11.3&
 \multicolumn{1}{c}{--}&\multicolumn{1}{c}{--}&&
 99.6&99.9&100.0&100.0&
 \multicolumn{1}{c}{--}&\multicolumn{1}{c}{--}\\
&GOS
&0.0&0.0&0.0&0.0&
 \multicolumn{1}{c}{--}&\multicolumn{1}{c}{--}&&
 79.4&93.4&97.3&99.1&
 \multicolumn{1}{c}{--}&\multicolumn{1}{c}{--}\\
\midrule
200&Thm.\ 1
&5.6&4.0&3.7&3.5&3.3&3.3&&
 98.2&99.3&99.7&99.8&99.9&100.0\\
&FLY
&37.9&22.8&15.6&11.9&
 \multicolumn{1}{c}{--}&\multicolumn{1}{c}{--}&&
 100.0&100.0&100.0&100.0&
 \multicolumn{1}{c}{--}&\multicolumn{1}{c}{--}\\
&GOS
&0.0&0.0&0.0&21.9&
 \multicolumn{1}{c}{--}&\multicolumn{1}{c}{--}&&
 85.8&98.0&99.3&100.0&
 \multicolumn{1}{c}{--}&\multicolumn{1}{c}{--}\\
\midrule
500&Thm.\ 1
&6.8&3.9&3.8&3.9&3.6&3.6&&
 100.0&100.0&100.0&100.0&100.0&100.0\\
&FLY
&48.8&23.2&17.7&13.2&
 \multicolumn{1}{c}{--}&\multicolumn{1}{c}{--}&&
 100.0&100.0&100.0&100.0&
 \multicolumn{1}{c}{--}&\multicolumn{1}{c}{--}\\
&GOS
&27.1&21.1&20.3&21.1&
 \multicolumn{1}{c}{--}&\multicolumn{1}{c}{--}&&
 94.2&99.8&100.0&100.0&
 \multicolumn{1}{c}{--}&\multicolumn{1}{c}{--}\\
\bottomrule
\end{tabular}
 }
\end{center}
\par
{\scriptsize 
\textbf{Note: }{The nominal size is 5\% and powers are assessed at 5\% level
of significance; frequencies are computed across $M=1000$ Monte Carlo
samples and we set $\nu=5$ when computing $\psi_{i,NT}$. Results are in percentage points.} }
\end{table}

\begin{table}[!t]
\caption{Empirical rejection frequencies for the test in Theorem \protect\ref%
{asy-max}, GARCH case.}
\label{tab:test_GARCH_supp}\captionsetup{font=small}
\par
\begin{center}
{\scriptsize \centering
\begin{tabular}{ll|
                S S S S S S
                c
                S S S S S S}
\toprule
&&\multicolumn{6}{c}{$\phi_{g}=0.4$; $\alpha_i=0$ for all $i$}&&
  \multicolumn{6}{c}{$\phi_{g}=0.4$; $\alpha_i\sim N(0,1)$ for 5\% of units} \\
\midrule
$N$&$\text{Test }\backslash T$
&{100}&{200}&{300}&{500}&{1000}&{2000}&&
 {100}&{200}&{300}&{500}&{1000}&{2000}\\
\midrule
100&Thm.\ 1
&7.9&5.0&4.4&4.0&3.8&3.7&&
 96.5&97.5&98.0&98.7&99.6&99.7\\
&FLY
&31.3&17.4&15.4&12.3&
 \multicolumn{1}{c}{--}&\multicolumn{1}{c}{--}&&
 100.0&100.0&100.0&100.0&
 \multicolumn{1}{c}{--}&\multicolumn{1}{c}{--}\\
&GOS
&21.2&20.0&17.9&19.3&
 \multicolumn{1}{c}{--}&\multicolumn{1}{c}{--}&&
 79.2&94.5&97.0&98.8&
 \multicolumn{1}{c}{--}&\multicolumn{1}{c}{--}\\
\midrule
200&Thm.\ 1
&7.0&4.2&3.6&3.4&3.3&3.2&&
 100.0&100.0&100.0&100.0&100.0&100.0\\
&FLY
&31.5&21.2&14.4&12.0&
 \multicolumn{1}{c}{--}&\multicolumn{1}{c}{--}&&
 100.0&100.0&100.0&100.0&
 \multicolumn{1}{c}{--}&\multicolumn{1}{c}{--}\\
&GOS
&21.4&17.6&18.3&21.7&
 \multicolumn{1}{c}{--}&\multicolumn{1}{c}{--}&&
 86.3&98.7&99.5&100.0&
 \multicolumn{1}{c}{--}&\multicolumn{1}{c}{--}\\
\midrule
500&Thm.\ 1
&10.9&5.3&4.3&3.9&3.9&3.7&&
 100.0&100.0&100.0&100.0&100.0&100.0\\
&FLY
&47.1&24.7&16.5&10.3&
 \multicolumn{1}{c}{--}&\multicolumn{1}{c}{--}&&
 100.0&100.0&100.0&100.0&
 \multicolumn{1}{c}{--}&\multicolumn{1}{c}{--}\\
&GOS
&18.1&22.2&22.4&19.2&
 \multicolumn{1}{c}{--}&\multicolumn{1}{c}{--}&&
 93.6&100.0&100.0&100.0&
 \multicolumn{1}{c}{--}&\multicolumn{1}{c}{--}\\
\bottomrule
\end{tabular}
 }
\end{center}
\par
{\scriptsize 
\textbf{Note: }{The nominal size is 5\% and powers are assessed at 5\% level
of significance; frequencies are computed across $M=1000$ Monte Carlo
samples and we set $\nu=5$ when computing $\psi_{i,NT}$. Results are in percentage points.} }
\end{table}
\newpage

\subsection{No persistence in the omitted factor\label{MC_phizero}}

We complement the results of Section \ref{simulations} under $\phi _{g}=0$
in (\ref{eq:innov2}) - that is, the omitted (strong) common factor has no
persistence.

Tables \ref{tab:test_Gaussian_phi0}, \ref{tab:test_T_phi0} and \ref%
{tab:test_GARCH_phi0} contain the empirical rejection frequencies under the
null (left panels) and the alternative (right panels) for our one-shot test.
For average type tests (FLY, GOS and PY) results are in line with those of
the main body, though the degree of over-rejection under $\mathbb{H}_{0}$ is
smaller. The $p$-value combination of AS is still mildly oversized, while
the max-type approach of FLLM performs on par with ours.

Results for the \textquotedblleft strong\textquotedblright\ decision rule of
Theorem \ref{strong-rule} are in Tables \ref{tab:derand_Gaussian_phi0}, \ref%
{tab:derand_T_phi0} and \ref{tab:derand_GARCH_phi0}. As desired, empirical
rejection frequencies quickly converge to zero under the null and for all
cases considered. As in the main body, convergence is faster when using $%
f(B)=B^{-1/4}$. Similarly good results hold under the alternative, where the
rejection frequencies always converge to one.

\begin{table}[!t]
\caption{Empirical rejection frequencies for the test in Theorem \protect\ref%
{asy-max}, Gaussian case.}
\label{tab:test_Gaussian_phi0}\captionsetup{font=small}
\par
\begin{center}
{\scriptsize 
\begin{tabular}{ll|
                S S S S S S
                c
                S S S S S S}
\toprule
&&\multicolumn{6}{c}{$\phi_{\nu}=0$; $\alpha_i=0$ for all $i$}&&
  \multicolumn{6}{c}{$\phi_{\nu}=0$; $\alpha_i\sim N(0,1)$ for 5\% of units} \\
\midrule
$N$&$\text{Test }\backslash T$
&{100}&{200}&{300}&{500}&{1000}&{2000}&&
 {100}&{200}&{300}&{500}&{1000}&{2000}\\
\midrule
100&Thm.\ 1
&3.9&3.8&3.7&3.9&4.0&3.7&&
 93.5&96.3&97.3&97.9&99.4&99.7\\
&FLLM
&4.5&4.5&3.5&2.4&2.1&4.1&&
 98.8&99.5&99.7&99.8&100.0&100.0\\
&GRS
&\multicolumn{1}{c}{--}&5.1&3.7&4.0&3.0&3.4&&
 \multicolumn{1}{c}{--}&99.4&99.5&99.9&100.0&100.0\\
&FLY
&24.9&14.8&9.6&10.0&
 \multicolumn{1}{c}{--}&\multicolumn{1}{c}{--}&&
 99.5&99.6&99.8&100.0&
 \multicolumn{1}{c}{--}&\multicolumn{1}{c}{--}\\
&GOS
&8.1&7.1&6.3&6.5&
 \multicolumn{1}{c}{--}&\multicolumn{1}{c}{--}&&
 80.5&94.6&97.6&99.2&
 \multicolumn{1}{c}{--}&\multicolumn{1}{c}{--}\\
&PY
&8.7&6.8&6.0&6.5&6.2&6.5&&
 81.0&94.6&97.6&99.2&99.9&100.0\\
&AS
&6.8&6.9&6.4&7.2&7.3&7.8&&
 21.9&68.5&94.1&99.2&100.0&100.0\\
\midrule
200&Thm.\ 1
&4.9&3.4&3.3&3.3&3.3&3.2&&
 99.8&99.9&100.0&100.0&100.0&100.0\\
&FLLM
&5.5&2.7&2.7&2.8&2.1&2.7&&
 100.0&100.0&100.0&100.0&100.0&100.0\\
&GRS
&\multicolumn{1}{c}{--}&\multicolumn{1}{c}{--}&4.1&3.5&4.2&3.4&&
 \multicolumn{1}{c}{--}&\multicolumn{1}{c}{--}&100.0&100.0&100.0&100.0\\
&FLY
&29.6&13.1&10.8&9.1&
 \multicolumn{1}{c}{--}&\multicolumn{1}{c}{--}&&
 100.0&100.0&100.0&100.0&
 \multicolumn{1}{c}{--}&\multicolumn{1}{c}{--}\\
&GOS
&8.8&6.2&6.4&6.1&
 \multicolumn{1}{c}{--}&\multicolumn{1}{c}{--}&&
 88.8&98.6&99.7&100.0&
 \multicolumn{1}{c}{--}&\multicolumn{1}{c}{--}\\
&PY
&9.6&5.9&6.2&6.1&7.3&7.1&&
 90.2&98.6&99.7&100.0&100.0&100.0\\
&AS
&7.9&5.9&8.0&8.2&7.2&7.9&&
 25.3&83.3&98.4&100.0&100.0&100.0\\
\midrule
500&Thm.\ 1
&4.4&3.6&3.7&3.6&3.6&3.6&&
 100.0&100.0&100.0&100.0&100.0&100.0\\
&FLLM
&3.5&4.4&3.6&2.9&2.7&3.4&&
 100.0&100.0&100.0&100.0&100.0&100.0\\
&GRS
&\multicolumn{1}{c}{--}&\multicolumn{1}{c}{--}&\multicolumn{1}{c}{--}&\multicolumn{1}{c}{--}&5.8&4.2&&
 \multicolumn{1}{c}{--}&\multicolumn{1}{c}{--}&\multicolumn{1}{c}{--}&\multicolumn{1}{c}{--}&100.0&100.0\\
&FLY
&38.5&19.6&14.1&8.9&
 \multicolumn{1}{c}{--}&\multicolumn{1}{c}{--}&&
 100.0&100.0&100.0&100.0&
 \multicolumn{1}{c}{--}&\multicolumn{1}{c}{--}\\
&GOS
&7.1&8.4&8.7&6.8&
 \multicolumn{1}{c}{--}&\multicolumn{1}{c}{--}&&
 95.1&100.0&100.0&100.0&
 \multicolumn{1}{c}{--}&\multicolumn{1}{c}{--}\\
&PY
&8.2&8.3&8.5&6.7&6.3&7.5&&
 97.3&100.0&100.0&100.0&100.0&100.0\\
&AS
&6.9&8.9&9.9&9.0&7.7&10.7&&
 29.1&91.3&99.9&100.0&100.0&100.0\\
\bottomrule
\end{tabular}
 }
\end{center}
\par
{\scriptsize 
\textbf{Note: }{The nominal size is 5\% and powers are assessed at 5\% level
of significance; frequencies are computed across $M=1000$ Monte Carlo
samples and we set $\nu=5$ when computing $\psi_{i,NT}$. Results are in percentage points.} }
\end{table}
\begin{table}[!t]
\caption{Empirical rejection frequencies for the decision rule of Theorem 
\protect\ref{strong-rule}, Gaussian case.}
\label{tab:derand_Gaussian_phi0}\captionsetup{font=small}
\par
\begin{center}
{\scriptsize \centering
\begin{tabular}{ll|
                S S S S S S
                c
                S S S S S S}
\toprule
&&\multicolumn{6}{c}{$\phi_g=0$; $\alpha_i=0$ for all $i$}&&
  \multicolumn{6}{c}{$\phi_g=0$; $\alpha_i\sim N(0,1)$ for 5\% of units} \\
\midrule
$N$&$\text{C.V. }\backslash T$
&{100}&{200}&{300}&{500}&{1000}&{2000}&&
 {100}&{200}&{300}&{500}&{1000}&{2000}\\
\midrule
100&LIL
&1.6&0.0&0.0&0.0&0.0&0.0&&
 96.9&98.0&98.4&98.9&99.7&99.9\\
&$f(B)=B^{-1/4}$
&0.0&0.0&0.0&0.0&0.0&0.0&&
 96.9&98.0&98.4&98.9&99.6&99.8\\
\midrule
200&LIL
&0.9&0.0&0.0&0.0&0.0&0.0&&
 99.9&99.9&100.0&100.0&100.0&100.0\\
&$f(B)=B^{-1/4}$
&0.2&0.0&0.0&0.0&0.0&0.0&&
 99.9&99.9&100.0&100.0&100.0&100.0\\
\midrule
500&LIL
&27.6&13.7&8.0&2.8&1.4&0.3&&
 100.0&100.0&100.0&100.0&100.0&100.0\\
&$f(B)=B^{-1/4}$
&0.1&0.0&0.0&0.0&0.0&0.0&&
 100.0&100.0&100.0&100.0&100.0&100.0\\
\bottomrule
\end{tabular}
 }
\end{center}
\par
{\scriptsize 
\textbf{Note: }{Results using either LIL-based critical values (LIL) or
critical values based on $f(B)=B^{-1/4}$. The derandomized statistic is
based on nominal level $\tau=5\%$. We set $B=\log(N)^{2}$ for the
calculation of $Q_{N,T,B}(\tau)$ and $\nu=5$ for that of $\psi_{i,NT}$. Results are in percentage points.} }
\end{table}

\begin{table}[!t]
\caption{ Empirical rejection frequencies for the test in Theorem \protect
\ref{asy-max}, Student's $t$ case.}
\label{tab:test_T_phi0}\captionsetup{font=small}
\par
\begin{center}
{\scriptsize \centering
\begin{tabular}{ll|
                S S S S S S
                c
                S S S S S S}
\toprule
&&\multicolumn{6}{c}{$\phi_g=0$; $\alpha_i=0$ for all $i$}&&
  \multicolumn{6}{c}{$\phi_g=0$; $\alpha_i\sim N(0,1)$ for 5\% of units} \\
\midrule
$N$&$\text{Test }\backslash T$
&{100}&{200}&{300}&{500}&{1000}&{2000}&&
 {100}&{200}&{300}&{500}&{1000}&{2000}\\
\midrule
100&Thm.\ 1
&5.1&4.2&3.8&3.8&4.0&3.7&&
 89.3&93.1&95.3&96.6&98.8&99.5\\
&FLLM
&5.2&3.3&3.6&3.7&3.3&2.5&&
 97.9&99.2&99.9&100.0&100.0&100.0\\
&GRS
&\multicolumn{1}{c}{--}&4.3&4.4&5.1&3.6&3.2&&
 \multicolumn{1}{c}{--}&99.2&99.6&99.9&100.0&100.0\\
&FLY
&24.8&11.6&9.4&7.3&
 \multicolumn{1}{c}{--}&\multicolumn{1}{c}{--}&&
 99.5&99.9&100.0&100.0&
 \multicolumn{1}{c}{--}&\multicolumn{1}{c}{--}\\
&GOS
&9.9&6.7&7.1&8.2&
 \multicolumn{1}{c}{--}&\multicolumn{1}{c}{--}&&
 78.6&94.0&97.6&99.0&
 \multicolumn{1}{c}{--}&\multicolumn{1}{c}{--}\\
&PY
&11.6&7.7&7.3&8.2&7.7&8.0&&
 84.2&94.4&97.6&98.8&99.9&100.0\\
&AS
&5.4&5.6&5.8&6.0&7.1&7.8&&
 17.5&55.5&86.9&98.0&99.9&100.0\\
\midrule
200&Thm.\ 1
&4.5&3.9&3.6&3.5&3.4&3.2&&
 98.4&99.3&99.7&99.9&99.9&100.0\\
&FLLM
&4.3&4.8&3.2&4.0&3.9&4.5&&
 99.9&100.0&100.0&100.0&100.0&100.0\\
&GRS
&\multicolumn{1}{c}{--}&\multicolumn{1}{c}{--}&6.2&4.4&3.2&3.3&&
 \multicolumn{1}{c}{--}&\multicolumn{1}{c}{--}&100.0&100.0&100.0&100.0\\
&FLY
&28.2&16.2&11.9&8.0&
 \multicolumn{1}{c}{--}&\multicolumn{1}{c}{--}&&
 100.0&100.0&100.0&100.0&
 \multicolumn{1}{c}{--}&\multicolumn{1}{c}{--}\\
&GOS
&9.8&7.7&6.8&8.3&
 \multicolumn{1}{c}{--}&\multicolumn{1}{c}{--}&&
 88.1&98.4&99.6&100.0&
 \multicolumn{1}{c}{--}&\multicolumn{1}{c}{--}\\
&PY
&11.8&8.3&6.9&8.2&6.1&5.5&&
 92.3&98.5&99.6&100.0&100.0&100.0\\
&AS
&7.5&6.0&7.6&6.1&7.5&7.3&&
 18.4&64.7&95.4&99.6&100.0&100.0\\
\midrule
500&Thm.\ 1
&5.0&3.8&3.9&3.7&3.5&3.5&&
 100.0&100.0&100.0&100.0&100.0&100.0\\
&FLLM
&5.3&3.8&4.1&4.2&2.5&3.4&&
 100.0&100.0&100.0&100.0&100.0&100.0\\
&GRS
&\multicolumn{1}{c}{--}&\multicolumn{1}{c}{--}&\multicolumn{1}{c}{--}&\multicolumn{1}{c}{--}&4.8&4.8&&
 \multicolumn{1}{c}{--}&\multicolumn{1}{c}{--}&\multicolumn{1}{c}{--}&\multicolumn{1}{c}{--}&100.0&100.0\\
&FLY
&42.0&17.9&14.1&8.1&
 \multicolumn{1}{c}{--}&\multicolumn{1}{c}{--}&&
 100.0&100.0&100.0&100.0&
 \multicolumn{1}{c}{--}&\multicolumn{1}{c}{--}\\
&GOS
&16.2&9.5&7.4&6.6&
 \multicolumn{1}{c}{--}&\multicolumn{1}{c}{--}&&
 95.8&100.0&100.0&100.0&
 \multicolumn{1}{c}{--}&\multicolumn{1}{c}{--}\\
&PY
&12.5&9.5&7.7&6.5&8.1&7.5&&
 98.8&100.0&100.0&100.0&100.0&100.0\\
&AS
&7.7&7.0&8.4&8.3&9.6&8.8&&
 22.5&79.1&98.9&100.0&100.0&100.0\\
\bottomrule
\end{tabular}
 }
\end{center}
\par
{\scriptsize 
\textbf{Note: }{The nominal size is 5\% and powers are assessed at 5\% level
of significance; frequencies are computed across $M=1000$ Monte Carlo
samples and we set $\nu=5$ when computing $\psi_{i,NT}$. Results are in percentage points.} }
\end{table}
\begin{table}[!t]
\caption{Empirical rejection frequencies for the decision rule of Theorem 
\protect\ref{strong-rule}, Student's $t$ case.}
\label{tab:derand_T_phi0}\captionsetup{font=small}
\par
\begin{center}
{\scriptsize \centering
\begin{tabular}{ll|
                S S S S S S
                c
                S S S S S S}
\toprule
&&\multicolumn{6}{c}{$\phi_g=0$; $\alpha_i=0$ for all $i$}&&
  \multicolumn{6}{c}{$\phi_g=0$; $\alpha_i\sim N(0,1)$ for 5\% of units} \\
\midrule
$N$&$\text{C.V. }\backslash T$
&{100}&{200}&{300}&{500}&{1000}&{2000}&&
 {100}&{200}&{300}&{500}&{1000}&{2000}\\
\midrule
100&LIL
&1.6&0.2&0.1&0.0&0.0&0.0&&
 94.8&96.8&97.7&98.2&99.4&99.7\\
&$f(B)=B^{-1/4}$
&0.1&0.1&0.0&0.0&0.0&0.0&&
 89.5&93.5&95.2&96.4&98.9&99.2\\
\midrule
200&LIL
&1.6&0.0&0.0&0.0&0.0&0.0&&
 99.6&99.7&100.0&100.0&100.0&100.0\\
&$f(B)=B^{-1/4}$
&0.0&0.0&0.0&0.0&0.0&0.0&&
 98.8&99.4&99.6&99.7&99.8&100.0\\
\midrule
500&LIL
&35.6&13.7&9.3&3.9&0.7&0.1&&
 100.0&100.0&100.0&100.0&100.0&100.0\\
&$f(B)=B^{-1/4}$
&0.3&0.1&0.0&0.0&0.0&0.0&&
 100.0&100.0&100.0&100.0&100.0&100.0\\
\bottomrule
\end{tabular}
 }
\end{center}
\par
{\scriptsize 
\textbf{Note: }{Results using either LIL-based critical values (LIL) or
critical values based on $f(B)=B^{-1/4}$. The derandomized statistic is
based on nominal level $\tau=5\%$. We set $B=\log(N)^{2}$ for the
calculation of $Q_{N,T,B}(\tau)$ and $\nu=5$ for that of $\psi_{i,NT}$. Results are in percentage points.} }
\end{table}
\begin{table}[!t]
\caption{Empirical rejection frequencies for the test in Theorem \protect\ref%
{asy-max}, GARCH case.}
\label{tab:test_GARCH_phi0}\captionsetup{font=small}
\par
\begin{center}
{\scriptsize \centering
\begin{tabular}{ll|
                S S S S S S
                c
                S S S S S S}
\toprule
&&\multicolumn{6}{c}{$\phi_g=0$; $\alpha_i=0$ for all $i$}&&
  \multicolumn{6}{c}{$\phi_g=0$; $\alpha_i\sim N(0,1)$ for 5\% of units} \\
\midrule
$N$&$\text{Test }\backslash T$
&{100}&{200}&{300}&{500}&{1000}&{2000}&&
 {100}&{200}&{300}&{500}&{1000}&{2000}\\
\midrule
100&Thm.\ 1
&5.6&4.3&4.0&3.9&3.7&3.7&&
 96.7&98.3&98.9&99.2&99.7&99.7\\
&FLLM
&4.5&4.5&3.5&2.4&1.8&1.7&&
 98.8&99.5&99.7&99.8&100.0&100.0\\
&GRS
&\multicolumn{1}{c}{--}&4.4&3.7&3.4&4.0&3.6&&
 \multicolumn{1}{c}{--}&99.9&99.9&99.9&100.0&100.0\\
&FLY
&20.9&10.9&9.6&8.8&
 \multicolumn{1}{c}{--}&\multicolumn{1}{c}{--}&&
 100.0&100.0&100.0&100.0&
 \multicolumn{1}{c}{--}&\multicolumn{1}{c}{--}\\
&GOS
&8.1&7.1&6.3&6.5&
 \multicolumn{1}{c}{--}&\multicolumn{1}{c}{--}&&
 80.5&94.6&97.6&99.2&
 \multicolumn{1}{c}{--}&\multicolumn{1}{c}{--}\\
&PY
&8.7&6.8&6.0&6.5&8.6&7.2&&
 81.0&94.6&97.6&99.2&100.0&100.0\\
&AS
&6.8&6.9&6.4&7.2&7.2&7.8&&
 21.9&68.5&94.1&99.2&100.0&100.0\\
\midrule
200&Thm.\ 1
&4.8&3.5&3.3&3.2&3.3&3.2&&
 100.0&100.0&100.0&100.0&100.0&100.0\\
&FLLM
&5.5&2.7&2.7&2.8&2.0&1.9&&
 100.0&100.0&100.0&100.0&100.0&100.0\\
&GRS
&\multicolumn{1}{c}{--}&\multicolumn{1}{c}{--}&5.7&4.4&3.9&2.8&&
 \multicolumn{1}{c}{--}&\multicolumn{1}{c}{--}&100.0&100.0&100.0&100.0\\
&FLY
&23.4&13.9&10.4&8.5&
 \multicolumn{1}{c}{--}&\multicolumn{1}{c}{--}&&
 100.0&100.0&100.0&100.0&
 \multicolumn{1}{c}{--}&\multicolumn{1}{c}{--}\\
&GOS
&8.8&6.2&6.4&6.1&
 \multicolumn{1}{c}{--}&\multicolumn{1}{c}{--}&&
 88.8&98.6&99.7&100.0&
 \multicolumn{1}{c}{--}&\multicolumn{1}{c}{--}\\
&PY
&9.6&5.9&6.2&6.1&6.5&6.7&&
 90.2&98.6&99.7&100.0&100.0&100.0\\
&AS
&7.9&5.9&8.0&8.2&7.3&8.1&&
 25.3&83.3&98.4&100.0&100.0&100.0\\
\midrule
500&Thm.\ 1
&6.5&4.1&3.5&3.5&3.5&3.5&&
 100.0&100.0&100.0&100.0&100.0&100.0\\
&FLLM
&3.5&4.4&3.6&2.9&1.4&1.4&&
 100.0&100.0&100.0&100.0&100.0&100.0\\
&GRS
&\multicolumn{1}{c}{--}&\multicolumn{1}{c}{--}&\multicolumn{1}{c}{--}&\multicolumn{1}{c}{--}&4.0&4.1&&
 \multicolumn{1}{c}{--}&\multicolumn{1}{c}{--}&\multicolumn{1}{c}{--}&\multicolumn{1}{c}{--}&100.0&100.0\\
&FLY
&38.8&18.9&12.9&8.3&
 \multicolumn{1}{c}{--}&\multicolumn{1}{c}{--}&&
 100.0&100.0&100.0&100.0&
 \multicolumn{1}{c}{--}&\multicolumn{1}{c}{--}\\
&GOS
&7.1&8.4&8.7&6.8&
 \multicolumn{1}{c}{--}&\multicolumn{1}{c}{--}&&
 95.1&100.0&100.0&100.0&
 \multicolumn{1}{c}{--}&\multicolumn{1}{c}{--}\\
&PY
&8.2&8.3&8.5&6.7&6.6&5.7&&
 97.3&100.0&100.0&100.0&100.0&100.0\\
&AS
&6.9&8.9&9.9&9.0&9.0&7.0&&
 29.1&91.3&99.9&100.0&100.0&100.0\\
\bottomrule
\end{tabular}
 }
\end{center}
\par
{\scriptsize 
\textbf{Note: }{The nominal size is 5\% and powers are assessed at 5\% level
of significance; frequencies are computed across $M=1000$ Monte Carlo
samples and we set $\nu=5$ when computing $\psi_{i,NT}$. Results are in percentage points.} }
\end{table}
\begin{table}[!t]
\caption{Empirical rejection frequencies for the decision rule of Theorem 
\protect\ref{strong-rule}, GARCH case.}
\label{tab:derand_GARCH_phi0}\captionsetup{font=small}
\par
\begin{center}
{\scriptsize \centering
\begin{tabular}{ll|
                S S S S S S
                c
                S S S S S S}
\toprule
&&\multicolumn{6}{c}{$\phi_g=0$; $\alpha_i=0$ for all $i$}&&
  \multicolumn{6}{c}{$\phi_g=0$; $\alpha_i\sim N(0,1)$ for 5\% of units} \\
\midrule
$N$&$\text{C.V. }\backslash T$
&{100}&{200}&{300}&{500}&{1000}&{2000}&&
 {100}&{200}&{300}&{500}&{1000}&{2000}\\
\midrule
100&LIL
&2.6&0.3&0.0&0.0&0.0&0.0&&
 98.5&98.9&99.5&99.5&99.4&99.7\\
&$f(B)=B^{-1/4}$
&0.2&0.0&0.0&0.0&0.0&0.0&&
 96.5&97.8&98.8&99.3&98.9&99.2\\
\midrule
200&LIL
&3.6&0.2&0.0&0.0&0.0&0.0&&
 100.0&100.0&100.0&100.0&100.0&100.0\\
&$f(B)=B^{-1/4}$
&0.4&0.0&0.0&0.0&0.0&0.0&&
 100.0&100.0&100.0&100.0&100.0&100.0\\
\midrule
500&LIL
&37.0&14.2&8.8&3.9&1.2&0.2&&
 100.0&100.0&100.0&100.0&100.0&100.0\\
&$f(B)=B^{-1/4}$
&1.8&0.2&0.0&0.0&0.0&0.0&&
 100.0&100.0&100.0&100.0&100.0&100.0\\
\bottomrule
\end{tabular}
}
\end{center}
\par
{\scriptsize 
\textbf{Note: }{Results using either LIL-based critical values (LIL) or
critical values based on $f(B)=B^{-1/4}$. The derandomized statistic is
based on nominal level $\tau=5\%$. We set $B=\log(N)^{2}$ for the
calculation of $Q_{N,T,B}(\tau)$ and $\nu=5$ for that of $\psi_{i,NT}$. Results are in percentage points.} }
\end{table}

\newpage 

\subsection{Different strengths of the omitted factor\label%
{MC_omitted_factors_strength}}

We consider the same three-factor pricing model of Section \ref{simulations}%
. In the main body, $\gamma _{i}\overset{i.i.d.}{\sim }\mathcal{U}(0.7,0.9)$
which implies that $\gamma _{i}>0$ for any cross-sectional unit. We now
depart from that assumption in two ways: first, we consider the case where
only $\lfloor N^{0.4}\rfloor $ assets have a non-zero loading on $g_{t}$;
secondly, we look at a situation where $\lfloor N^{0.8}\rfloor $
cross-sectional units have a non-zero loading on $g_{t}$. These two
experiments correspond to the case of a weak and of a semi-strong omitted
pricing factors, respectively. Notably, the largest eigenvalue of the
covariance matrix of $\mathbf{u}_{t}$ is bounded in the weak factor case,
while it diverges to infinity in the semi-strong one. We consider the same
sample sizes and alternative hypothesis as in Section \ref{simulations}.

Empirical rejection frequencies for the weak factor case are in Table \ref%
{tab:test_T_weakOmitted}. Results for our test and for those of \cite%
{feng2022high} and \cite{fan2015power} are very similar to those in the main
body (Table \ref{tab:test_T}) both under both the null and the alternative
hypothesis. Actual sizes of the GOS and PY tests are much closer to the
nominal one, suggesting that strong-cross sectional dependence in the
residuals was the driver of their overrejections. Their powers are
substantially unaltered. The test of \cite{ardia2024robust} performs well in
terms of size, but still exhibits a lack of power when $T$ is small. All
previous considerations hold unchanged with respect to the value of $\phi_g$.

Table \ref{tab:test_T_semiStrongOmitted} shows results for the semi-strong
case. We would like to point out the importance of this data generating
process, as \cite{bailey2021measurement} described tens of semi-strong
pricing factor for the cross-section of US excess returns. Empirical
rejection frequencies are very close to those of the main body, as the FLY,
GOS, PY, and FLLM tests all become oversized when $\phi_g = 0.4$. This is
most likely an effect of strong cross-sectional dependence in the residuals,
as implied by a diverging eigenvalue in their covariance matrix. Results on
the test by of \cite{ardia2024robust} are equivalent to those of Tables \ref%
{tab:test_T} and \ref{tab:test_T_weakOmitted}.

\begin{table}[!t]
\caption{Empirical rejection frequencies for the test in Theorem \protect\ref%
{asy-max}, Student's $t$ innovations with weak omitted common factor.}
\label{tab:test_T_weakOmitted}
\begin{center}
{\scriptsize \centering
\begin{tabular}{ll|
                S S S S S S
                c
                S S S S S S}
\toprule
&&\multicolumn{6}{c}{$\phi_g=0$; $\alpha_i=0$ for all $i$}&&
  \multicolumn{6}{c}{$\phi_g=0$; $\alpha_i\sim N(0,1)$ for 5\% of units} \\
\midrule
$N$&$\text{Test }\backslash T$
&{100}&{200}&{300}&{500}&{1000}&{2000}&&
 {100}&{200}&{300}&{500}&{1000}&{2000}\\
\midrule
100&Thm.\ 1
&4.0&2.7&3.5&3.3&3.2&3.2&&
 94.1&95.9&96.8&98.0&99.4&99.7\\
&FLLM
&4.8&5.7&4.4&3.4&4.4&5.3&&
 98.4&99.8&99.9&100.0&100.0&100.0\\
&GRS
&\multicolumn{1}{c}{--}&4.5&4.7&3.4&3.6&2.7&&
 \multicolumn{1}{c}{--}&98.5&99.7&99.9&100.0&100.0\\
&FLY
&27.9&16.8&11.4&7.4&
 \multicolumn{1}{c}{--}&\multicolumn{1}{c}{--}&&
 99.0&99.8&99.9&100.0&
 \multicolumn{1}{c}{--}&\multicolumn{1}{c}{--}\\
&GOS
&8.4&6.9&6.3&5.7&
 \multicolumn{1}{c}{--}&\multicolumn{1}{c}{--}&&
 98.3&99.4&99.8&100.0&
 \multicolumn{1}{c}{--}&\multicolumn{1}{c}{--}\\
&PY
&6.3&5.3&5.8&5.3&5.5&5.9&&
 98.1&99.4&99.7&100.0&100.0&100.0\\
&AS
&5.0&5.3&6.4&5.9&5.2&5.5&&
 15.3&51.3&88.1&98.5&100.0&100.0\\
\midrule
200&Thm.\ 1
&5.1&3.5&4.0&3.0&4.2&4.0&&
 99.6&99.9&100.0&100.0&100.0&100.0\\
&FLLM
&6.1&4.6&3.2&5.4&4.6&4.4&&
 100.0&100.0&100.0&100.0&100.0&100.0\\
&GRS
&\multicolumn{1}{c}{--}&\multicolumn{1}{c}{--}&3.9&4.6&3.8&3.8&&
 \multicolumn{1}{c}{--}&\multicolumn{1}{c}{--}&100.0&100.0&100.0&100.0\\
&FLY
&32.5&14.6&10.7&8.5&
 \multicolumn{1}{c}{--}&\multicolumn{1}{c}{--}&&
 100.0&100.0&100.0&100.0&
 \multicolumn{1}{c}{--}&\multicolumn{1}{c}{--}\\
&GOS
&9.3&7.4&6.0&6.6&
 \multicolumn{1}{c}{--}&\multicolumn{1}{c}{--}&&
 99.6&100.0&100.0&100.0&
 \multicolumn{1}{c}{--}&\multicolumn{1}{c}{--}\\
&PY
&6.1&5.7&5.2&5.9&6.4&4.9&&
 100.0&100.0&100.0&100.0&100.0&100.0\\
&AS
&6.2&5.8&5.5&5.9&5.9&6.1&&
 16.3&64.3&96.5&99.9&100.0&100.0\\
\midrule
500&Thm.\ 1
&3.0&5.4&3.1&3.1&4.2&3.9&&
 100.0&100.0&100.0&100.0&100.0&100.0\\
&FLLM
&7.1&5.1&4.9&4.3&3.6&2.6&&
 100.0&100.0&100.0&100.0&100.0&100.0\\
&GRS
&\multicolumn{1}{c}{--}&\multicolumn{1}{c}{--}&\multicolumn{1}{c}{--}&\multicolumn{1}{c}{--}&4.5&5.1&&
 \multicolumn{1}{c}{--}&\multicolumn{1}{c}{--}&\multicolumn{1}{c}{--}&\multicolumn{1}{c}{--}&100.0&100.0\\
&FLY
&43.1&17.3&12.2&8.4&
 \multicolumn{1}{c}{--}&\multicolumn{1}{c}{--}&&
 100.0&100.0&100.0&100.0&
 \multicolumn{1}{c}{--}&\multicolumn{1}{c}{--}\\
&GOS
&9.9&6.4&6.2&5.3&
 \multicolumn{1}{c}{--}&\multicolumn{1}{c}{--}&&
 99.3&100.0&100.0&100.0&
 \multicolumn{1}{c}{--}&\multicolumn{1}{c}{--}\\
&PY
&4.4&4.6&5.0&4.8&4.8&4.2&&
 100.0&100.0&100.0&100.0&100.0&100.0\\
&AS
&4.6&4.3&5.0&5.4&4.8&5.4&&
 17.7&80.7&99.4&100.0&100.0&100.0\\
\midrule
&&\multicolumn{6}{c}{$\phi_g=0.4$; $\alpha_i=0$ for all $i$}&&
  \multicolumn{6}{c}{$\phi_g=0.4$; $\alpha_i\sim N(0,1)$ for 5\% of units} \\
\midrule
$N$&$\text{Test }\backslash T$
&{100}&{200}&{300}&{500}&{1000}&{2000}&&
 {100}&{200}&{300}&{500}&{1000}&{2000}\\
\midrule
100&Thm.\ 1
&4.8&3.6&2.8&3.5&3.2&3.2&&
 94.4&96.3&97.1&97.8&99.4&99.7\\
&FLLM
&5.5&6.8&5.6&4.6&6.1&7.0&&
 98.4&99.9&100.0&100.0&100.0&100.0\\
&GRS
&\multicolumn{1}{c}{--}&5.1&4.9&4.0&4.3&3.6&&
 \multicolumn{1}{c}{--}&98.4&99.7&99.9&100.0&100.0\\
&FLY
&31.3&19.0&13.3&8.5&
 \multicolumn{1}{c}{--}&\multicolumn{1}{c}{--}&&
 99.0&99.9&100.0&100.0&
 \multicolumn{1}{c}{--}&\multicolumn{1}{c}{--}\\
&GOS
&12.0&11.0&8.9&8.0&
 \multicolumn{1}{c}{--}&\multicolumn{1}{c}{--}&&
 98.4&99.5&99.7&100.0&
 \multicolumn{1}{c}{--}&\multicolumn{1}{c}{--}\\
&PY
&8.8&8.9&8.2&7.4&7.8&9.3&&
 98.3&99.4&99.7&100.0&100.0&100.0\\
&AS
&4.5&5.6&6.5&5.9&5.4&5.5&&
 15.3&51.2&87.4&98.4&100.0&100.0\\
\midrule
200&Thm.\ 1
&3.8&4.0&3.5&4.0&4.2&4.0&&
 99.7&99.9&100.0&100.0&100.0&100.0\\
&FLLM
&7.0&5.3&4.4&6.3&5.3&5.3&&
 100.0&100.0&100.0&100.0&100.0&100.0\\
&GRS
&\multicolumn{1}{c}{--}&\multicolumn{1}{c}{--}&4.7&4.7&4.5&4.4&&
 \multicolumn{1}{c}{--}&\multicolumn{1}{c}{--}&100.0&100.0&100.0&100.0\\
&FLY
&34.6&17.0&12.3&10.0&
 \multicolumn{1}{c}{--}&\multicolumn{1}{c}{--}&&
 100.0&100.0&100.0&100.0&
 \multicolumn{1}{c}{--}&\multicolumn{1}{c}{--}\\
&GOS
&12.0&9.7&8.8&9.6&
 \multicolumn{1}{c}{--}&\multicolumn{1}{c}{--}&&
 99.6&100.0&100.0&100.0&
 \multicolumn{1}{c}{--}&\multicolumn{1}{c}{--}\\
&PY
&8.2&8.0&7.9&9.1&8.2&7.0&&
 100.0&100.0&100.0&100.0&100.0&100.0\\
&AS
&6.4&5.8&5.5&5.8&5.9&6.0&&
 17.0&64.8&96.3&99.9&100.0&100.0\\
\midrule
500&Thm.\ 1
&5.1&5.0&3.6&3.0&4.2&3.9&&
 100.0&100.0&100.0&100.0&100.0&100.0\\
&FLLM
&7.3&5.9&5.0&4.5&4.4&3.2&&
 100.0&100.0&100.0&100.0&100.0&100.0\\
&GRS
&\multicolumn{1}{c}{--}&\multicolumn{1}{c}{--}&\multicolumn{1}{c}{--}&\multicolumn{1}{c}{--}&4.9&5.5&&
 \multicolumn{1}{c}{--}&\multicolumn{1}{c}{--}&\multicolumn{1}{c}{--}&\multicolumn{1}{c}{--}&100.0&100.0\\
&FLY
&46.9&19.2&14.1&9.6&
 \multicolumn{1}{c}{--}&\multicolumn{1}{c}{--}&&
 100.0&100.0&100.0&100.0&
 \multicolumn{1}{c}{--}&\multicolumn{1}{c}{--}\\
&GOS
&12.2&8.8&7.9&7.0&
 \multicolumn{1}{c}{--}&\multicolumn{1}{c}{--}&&
 99.3&100.0&100.0&100.0&
 \multicolumn{1}{c}{--}&\multicolumn{1}{c}{--}\\
&PY
&6.5&5.9&6.1&6.5&7.4&6.7&&
 100.0&100.0&100.0&100.0&100.0&100.0\\
&AS
&4.9&4.2&4.9&5.2&4.9&5.5&&
 17.8&80.9&99.4&100.0&100.0&100.0\\
\bottomrule
\end{tabular}
 }
\end{center}
\par
{\small \justifying
\textbf{Note: }{The nominal size is 5\% and powers are assessed at 5\% level
of significance; frequencies are computed across $M=1000$ Monte Carlo
samples and we set $\nu=5$ when computing $\psi_{i,NT}$. Only $\lfloor
N^{0.4}\rfloor$ cross-sectional entities have a non-zero loading on the
omitted common factor. Results are in percentage points.} }
\end{table}
\begin{table}[!t]
\caption{Empirical rejection frequencies for the test in Theorem \protect\ref%
{asy-max}; Student's $t$ innovations with semi-strong omitted common factor.}
\label{tab:test_T_semiStrongOmitted}
\begin{center}
{\scriptsize \centering
\begin{tabular}{ll|
                S S S S S S
                c
                S S S S S S}
\toprule
&&\multicolumn{6}{c}{$\phi_g=0$; $\alpha_i=0$ for all $i$}&&
  \multicolumn{6}{c}{$\phi_g=0$; $\alpha_i\sim N(0,1)$ for 5\% of units} \\
\midrule
$N$&$\text{Test }\backslash T$
&{100}&{200}&{300}&{500}&{1000}&{2000}&&
 {100}&{200}&{300}&{500}&{1000}&{2000}\\
\midrule
100&Thm.\ 1
&3.8&2.9&3.7&3.3&3.2&3.2&&
 94.0&95.7&96.8&98.3&99.0&99.5\\
&FLLM
&4.6&4.4&4.1&4.4&4.3&4.2&&
 98.6&99.6&99.9&100.0&100.0&100.0\\
&GRS
&\multicolumn{1}{c}{--}&4.9&3.7&3.2&2.9&3.0&&
 \multicolumn{1}{c}{--}&98.8&99.5&99.9&100.0&100.0\\
&FLY
&26.7&14.7&11.4&8.1&
 \multicolumn{1}{c}{--}&\multicolumn{1}{c}{--}&&
 99.6&99.8&99.9&100.0&
 \multicolumn{1}{c}{--}&\multicolumn{1}{c}{--}\\
&GOS
&8.7&6.3&6.0&6.4&
 \multicolumn{1}{c}{--}&\multicolumn{1}{c}{--}&&
 97.9&99.0&99.6&99.9&
 \multicolumn{1}{c}{--}&\multicolumn{1}{c}{--}\\
&PY
&7.3&5.8&5.8&6.1&7.9&6.4&&
 97.3&99.0&99.6&99.9&100.0&100.0\\
&AS
&5.2&6.4&6.6&4.9&5.2&5.5&&
 13.9&55.1&87.9&98.4&100.0&100.0\\
\midrule
200&Thm.\ 1
&5.1&3.7&4.1&3.0&4.2&4.0&&
 99.5&99.8&99.8&99.9&100.0&100.0\\
&FLLM
&5.3&4.8&4.2&4.6&5.1&4.8&&
 100.0&100.0&100.0&100.0&100.0&100.0\\
&GRS
&\multicolumn{1}{c}{--}&\multicolumn{1}{c}{--}&6.4&3.9&2.9&4.3&&
 \multicolumn{1}{c}{--}&\multicolumn{1}{c}{--}&100.0&100.0&100.0&100.0\\
&FLY
&32.6&16.1&12.3&7.5&
 \multicolumn{1}{c}{--}&\multicolumn{1}{c}{--}&&
 100.0&100.0&100.0&100.0&
 \multicolumn{1}{c}{--}&\multicolumn{1}{c}{--}\\
&GOS
&11.9&7.8&5.6&5.4&
 \multicolumn{1}{c}{--}&\multicolumn{1}{c}{--}&&
 99.5&100.0&100.0&100.0&
 \multicolumn{1}{c}{--}&\multicolumn{1}{c}{--}\\
&PY
&9.1&6.9&4.8&5.3&5.6&5.0&&
 99.8&100.0&100.0&100.0&100.0&100.0\\
&AS
&4.2&4.6&4.5&5.8&5.9&6.1&&
 15.4&65.3&96.5&99.9&100.0&100.0\\
\midrule
500&Thm.\ 1
&3.7&5.2&3.7&3.2&4.2&4.0&&
 100.0&100.0&100.0&100.0&100.0&100.0\\
&FLLM
&4.3&5.6&3.7&3.7&4.0&3.9&&
 100.0&100.0&100.0&100.0&100.0&100.0\\
&GRS
&\multicolumn{1}{c}{--}&\multicolumn{1}{c}{--}&\multicolumn{1}{c}{--}&\multicolumn{1}{c}{--}&4.2&4.8&&
 \multicolumn{1}{c}{--}&\multicolumn{1}{c}{--}&\multicolumn{1}{c}{--}&\multicolumn{1}{c}{--}&100.0&100.0\\
&FLY
&41.3&17.0&14.0&9.9&
 \multicolumn{1}{c}{--}&\multicolumn{1}{c}{--}&&
 100.0&100.0&100.0&100.0&
 \multicolumn{1}{c}{--}&\multicolumn{1}{c}{--}\\
&GOS
&13.6&8.6&6.7&7.1&
 \multicolumn{1}{c}{--}&\multicolumn{1}{c}{--}&&
 99.1&100.0&100.0&100.0&
 \multicolumn{1}{c}{--}&\multicolumn{1}{c}{--}\\
&PY
&9.4&6.5&5.8&6.7&6.2&6.7&&
 100.0&100.0&100.0&100.0&100.0&100.0\\
&AS
&4.6&5.5&5.0&6.5&4.8&5.4&&
 16.8&79.4&99.2&100.0&100.0&100.0\\
\midrule
&&\multicolumn{6}{c}{$\phi_g=0.4$; $\alpha_i=0$ for all $i$}&&
  \multicolumn{6}{c}{$\phi_g=0.4$; $\alpha_i\sim N(0,1)$ for 5\% of units} \\
\midrule
$N$&$\text{Test }\backslash T$
&{100}&{200}&{300}&{500}&{1000}&{2000}&&
 {100}&{200}&{300}&{500}&{1000}&{2000}\\
\midrule
100&Thm.\ 1
&5.0&3.8&3.2&3.4&3.3&3.3&&
 93.1&95.8&97.0&98.2&98.9&99.4\\
&FLLM
&8.1&6.9&6.9&7.0&9.2&9.2&&
 98.5&99.8&99.9&100.0&100.0&100.0\\
&GRS
&\multicolumn{1}{c}{--}&5.2&4.5&3.7&3.6&3.8&&
 \multicolumn{1}{c}{--}&98.8&99.5&99.9&100.0&100.0\\
&FLY
&33.2&18.3&14.0&10.9&
 \multicolumn{1}{c}{--}&\multicolumn{1}{c}{--}&&
 99.5&99.9&100.0&100.0&
 \multicolumn{1}{c}{--}&\multicolumn{1}{c}{--}\\
&GOS
&18.3&17.9&16.6&16.2&
 \multicolumn{1}{c}{--}&\multicolumn{1}{c}{--}&&
 97.8&99.0&99.6&99.9&
 \multicolumn{1}{c}{--}&\multicolumn{1}{c}{--}\\
&PY
&16.8&16.8&15.9&15.8&19.1&17.7&&
 97.5&99.0&99.6&99.9&100.0&100.0\\
&AS
&5.1&6.6&6.9&4.6&5.4&5.5&&
 13.3&52.1&87.0&98.2&100.0&100.0\\
\midrule
200&Thm.\ 1
&4.6&4.2&3.3&4.0&4.2&4.1&&
 99.5&99.7&100.0&100.0&100.0&100.0\\
&FLLM
&8.3&8.6&8.6&7.0&8.8&7.2&&
 100.0&100.0&100.0&100.0&100.0&100.0\\
&GRS
&\multicolumn{1}{c}{--}&\multicolumn{1}{c}{--}&6.7&4.4&3.0&4.3&&
 \multicolumn{1}{c}{--}&\multicolumn{1}{c}{--}&100.0&100.0&100.0&100.0\\
&FLY
&39.0&19.4&15.2&9.4&
 \multicolumn{1}{c}{--}&\multicolumn{1}{c}{--}&&
 100.0&100.0&100.0&100.0&
 \multicolumn{1}{c}{--}&\multicolumn{1}{c}{--}\\
&GOS
&21.4&18.8&16.9&16.2&
 \multicolumn{1}{c}{--}&\multicolumn{1}{c}{--}&&
 99.2&100.0&100.0&100.0&
 \multicolumn{1}{c}{--}&\multicolumn{1}{c}{--}\\
&PY
&19.7&18.2&16.3&15.8&17.7&17.8&&
 99.7&100.0&100.0&100.0&100.0&100.0\\
&AS
&4.1&6.0&5.6&6.0&5.9&6.0&&
 14.2&63.1&95.8&99.8&100.0&100.0\\
\midrule
500&Thm.\ 1
&6.2&5.3&3.7&3.1&4.0&3.9&&
 100.0&100.0&100.0&100.0&100.0&100.0\\
&FLLM
&8.3&8.7&7.8&7.2&8.5&9.2&&
 100.0&100.0&100.0&100.0&100.0&100.0\\
&GRS
&\multicolumn{1}{c}{--}&\multicolumn{1}{c}{--}&\multicolumn{1}{c}{--}&\multicolumn{1}{c}{--}&4.6&5.1&&
 \multicolumn{1}{c}{--}&\multicolumn{1}{c}{--}&\multicolumn{1}{c}{--}&\multicolumn{1}{c}{--}&100.0&100.0\\
&FLY
&48.0&21.3&16.3&11.6&
 \multicolumn{1}{c}{--}&\multicolumn{1}{c}{--}&&
 100.0&100.0&100.0&100.0&
 \multicolumn{1}{c}{--}&\multicolumn{1}{c}{--}\\
&GOS
&26.3&21.3&19.0&20.5&
 \multicolumn{1}{c}{--}&\multicolumn{1}{c}{--}&&
 98.8&100.0&100.0&100.0&
 \multicolumn{1}{c}{--}&\multicolumn{1}{c}{--}\\
&PY
&22.0&19.2&17.8&20.2&19.3&18.7&&
 100.0&100.0&100.0&100.0&100.0&100.0\\
&AS
&6.4&5.5&5.5&6.9&4.9&5.5&&
 16.8&76.0&99.1&100.0&100.0&100.0\\
\bottomrule
\end{tabular}
 }
\end{center}
\par
{\small \justifying
\textbf{Note: }{The nominal size is 5\% and powers are assessed at 5\% level
of significance; frequencies are computed across $M=1000$ Monte Carlo
samples and we set $\nu=5$ when computing $\psi_{i,NT}$. Only $\lfloor
N^{0.8}\rfloor$ assets have a non-zero loading on the omitted common factor. Results are in percentage points.} }
\end{table}

\newpage
\clearpage

\subsection{Different levels of sparsity under the alternative\label{MC_sparsity}}

In this section, we consider power under twelve different
alternatives, each based on a different percentage of mis-priced assets: $1\%,$ $2\%,$ $3\%,$ $\dots ,$ $9\%,$ $10\%,$ $15\%$ and $20\%$. For brevity, we only focus on sample sizes $T=100$ and $N=500$, which are the most relevant for the empirical analysis of Section \ref{empirical}. Figure \ref%
{fig:multi_alpha} reports the empirical rejection frequencies for the twelve
levels of sparsity (the horizontal axis reports the percentage of mis-priced
assets) across all DGPs. The upper panels consider $\phi_g = 0$, while the case $\phi_g = 0.40$ is in the lower ones. Tests by FLY and FLLM always achieve unit power, while our approach performs almost equally well (and better than the other
tests), as its power converges to one almost immediately.

\begin{figure}[!t]
\captionsetup{font=scriptsize} \centering
\begin{subfigure}[t]{0.325\textwidth}
                 \centering
                 \includegraphics[ width=0.98\textwidth, trim = 0cm 1.5cm 0cm 1.5cm, clip]{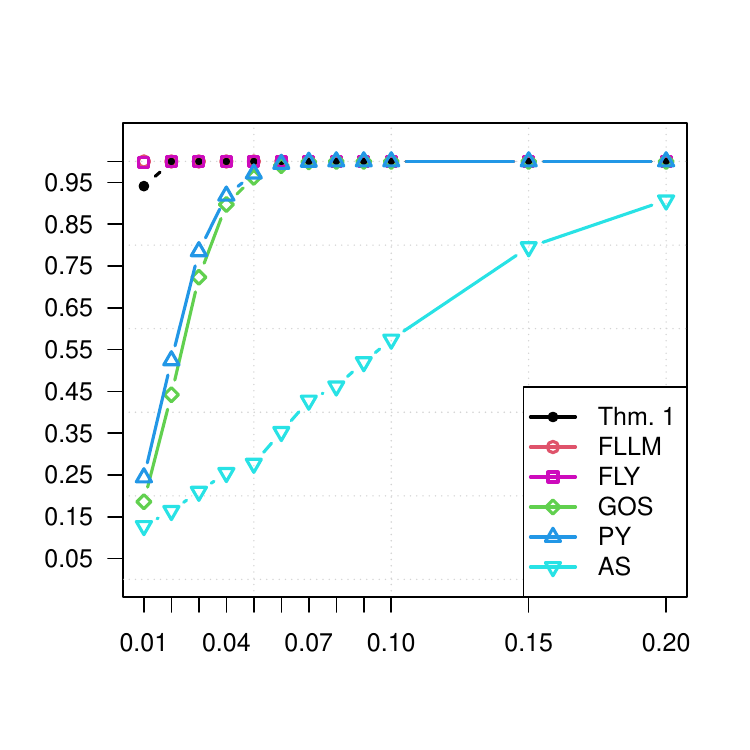}
                               \caption{Gaussian DGP, $\phi_{g} = 0$.}
                 \label{fig:phi0_gau}
         \end{subfigure}
\begin{subfigure}[t]{0.325\textwidth}
                 \centering
                 \includegraphics[width=0.98\textwidth, trim = 0cm 1.5cm 0cm 1.5cm, clip]{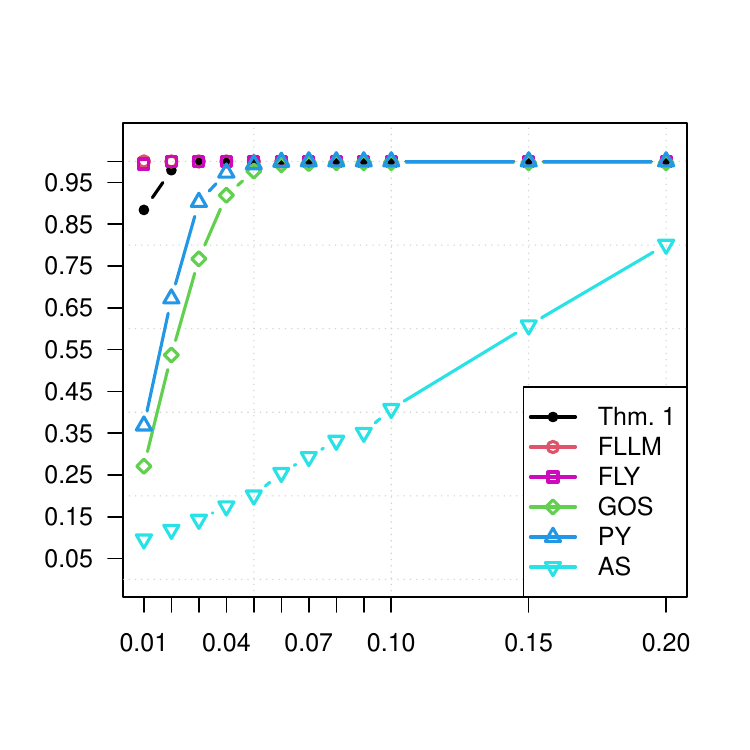}
                 \caption{Student's $t$ DGP, $\phi_{g} = 0$.}
                 \label{fig:phi_0_T}
         \end{subfigure}
\begin{subfigure}[t]{0.325\textwidth}
                 \centering
                 \includegraphics[width=0.98\textwidth, trim = 0cm 1.5cm 0cm 1.5cm, clip]{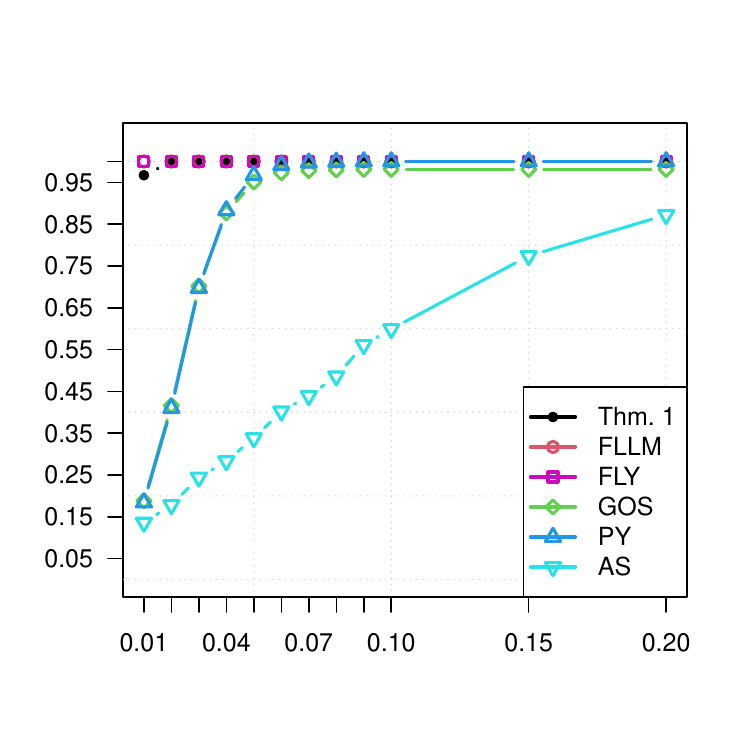}
\caption{GARCH(1,1) DGP, $\phi_{g} = 0$.}
                 \label{fig:phi_0_garch}
         \end{subfigure}

         \begin{subfigure}[t]{0.325\textwidth}
                 \centering
                 \includegraphics[ width=0.98\textwidth, trim = 0cm 1.5cm 0cm 1.5cm, clip]{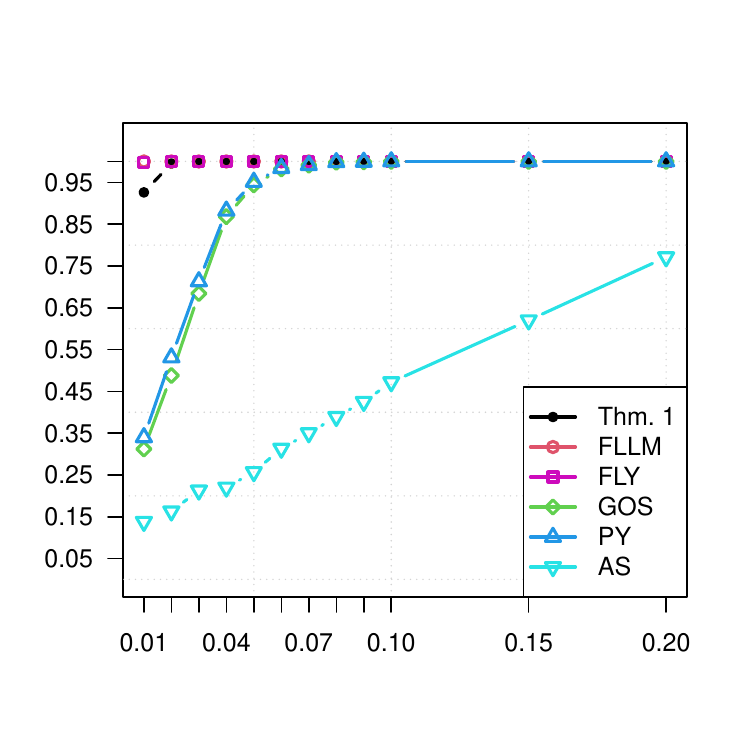}
                               \caption{Gaussian DGP, $\phi_{g} = 0.40$.}
                 \label{fig:phi040_gau}
         \end{subfigure}
\begin{subfigure}[t]{0.325\textwidth}
                 \centering
                 \includegraphics[width=0.98\textwidth, trim = 0cm 1.5cm 0cm 1.5cm, clip]{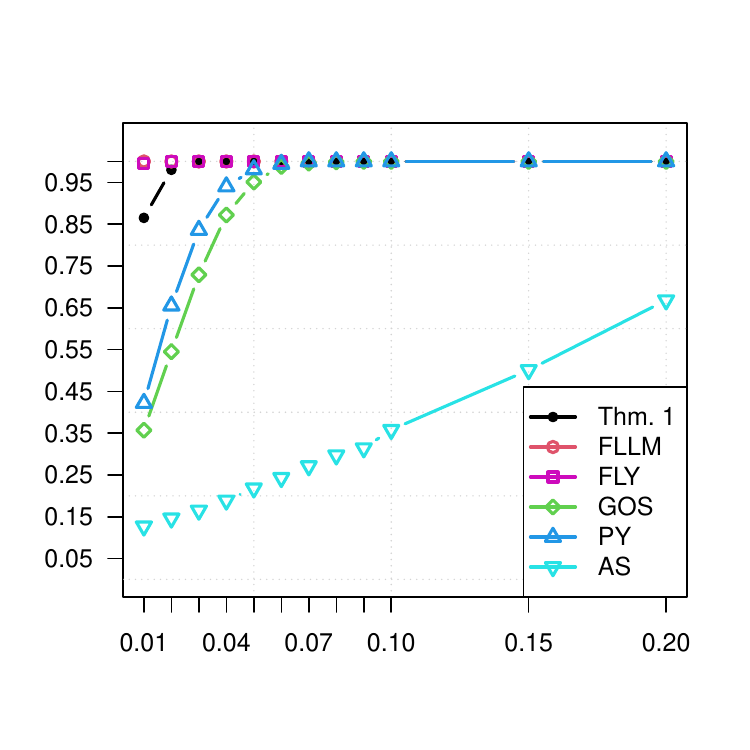}
                 \caption{Student's $t$ DGP, $\phi_{g} = 0.40$.}
                 \label{fig:phi_040_T}
         \end{subfigure}
\begin{subfigure}[t]{0.325\textwidth}
                 \centering
                 \includegraphics[width=0.98\textwidth, trim = 0cm 1.5cm 0cm 1.5cm, clip]{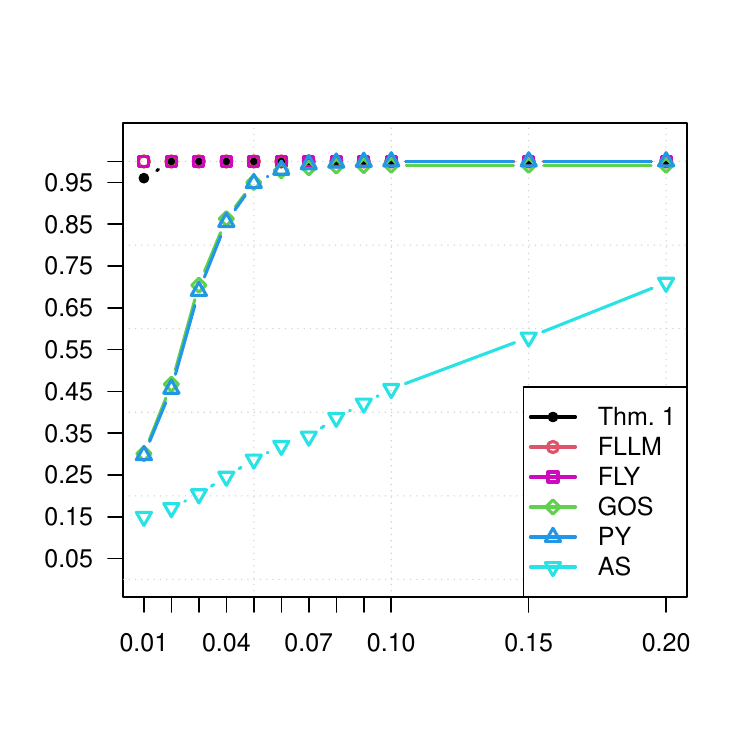}
\caption{GARCH(1,1) DGP, $\phi_{g} = 0.40$.}
                 \label{fig:phi_040_garch}
         \end{subfigure}
\caption{Power curves. The horizontal axis reports the percentages of
mis-priced assets. }
\label{fig:multi_alpha}
\end{figure}

\newpage

\subsection{Strong and semi-strong pricing factors\label%
{MC_pricing_factors_streng}}

We now consider the same data generating process of Section \ref%
{MC_omitted_factors_strength} but assuming that the omitted factor $g_{t}$
is strong. This time, however, we assume that $\beta _{i,2}=\beta _{i,3}=0$
for $N-\lfloor N^{0.8}\rfloor $ randomly chosen assets (cf. %
\citealp{bailey2021measurement}, who found that only the market factor is
strong, while other $140$ factors are at most semi-strong). Empirical
rejection frequencies for this data generating process are reported in Table %
\ref{tab:test_T_semiStrongPricing}. Results on all tests are substantially
equivalent to those of Table \ref{tab:test_T}, thus validating results of
the main body also in the case where only one pricing factor is strong. This
holds true irrespectively of the value of $\phi _{g }$.

\begin{table}[!t]
\caption{Empirical rejection frequencies for the test in Theorem \protect\ref%
{asy-max}; Student's $t$ innovations with one strong and two semi-strong
pricing factors.}
\label{tab:test_T_semiStrongPricing}
\begin{center}
{\scriptsize \centering
\begin{tabular}{ll|
                S S S S S S
                c
                S S S S S S}
\toprule
&&\multicolumn{6}{c}{$\phi_g=0$; $\alpha_i=0$}&&
  \multicolumn{6}{c}{$\phi_g=0$; $\alpha_i\sim N(0,1)$ (5\%)}\\
\midrule
$N$&Test$\backslash T$
&100&200&300&500&1000&2000&&
 100&200&300&500&1000&2000\\
\midrule
100&Thm.\ 1
&3.6&2.9&3.8&3.4&3.3&3.2&&
 89.3&93.2&95.1&96.9&98.5&99.2\\
&FLLM
&4.1&3.3&2.8&4.1&3.8&3.4&&
 98.0&99.5&100.0&100.0&100.0&100.0\\
&GRS
&\multicolumn{1}{c}{--}&5.3&5.5&3.6&2.9&3.0&&
 \multicolumn{1}{c}{--}&98.9&99.6&100.0&100.0&100.0\\
&FLY
&23.2&15.4&11.5&10.1&\multicolumn{1}{c}{--}&\multicolumn{1}{c}{--}&&
 99.4&99.9&100.0&100.0&\multicolumn{1}{c}{--}&\multicolumn{1}{c}{--}\\
&GOS
&6.9&7.1&7.4&7.4&\multicolumn{1}{c}{--}&\multicolumn{1}{c}{--}&&
 79.4&94.7&97.3&98.8&\multicolumn{1}{c}{--}&\multicolumn{1}{c}{--}\\
&PY
&8.7&7.2&7.2&7.2&7.3&8.0&&
 85.0&94.9&97.3&98.7&99.9&99.9\\
&AS
&6.8&5.7&6.6&6.5&5.2&5.6&&
 14.8&55.2&87.2&98.2&99.9&100.0\\
\midrule
200&Thm.\ 1
&4.6&3.9&3.9&3.1&4.1&4.0&&
 98.4&99.6&99.8&99.8&100.0&100.0\\
&FLLM
&4.7&4.2&3.2&3.5&3.3&3.3&&
 100.0&100.0&100.0&100.0&100.0&100.0\\
&GRS
&\multicolumn{1}{c}{--}&\multicolumn{1}{c}{--}&6.2&4.0&2.9&4.3&&
 \multicolumn{1}{c}{--}&\multicolumn{1}{c}{--}&100.0&100.0&100.0&100.0\\
&FLY
&28.4&14.4&11.4&8.5&\multicolumn{1}{c}{--}&\multicolumn{1}{c}{--}&&
 100.0&100.0&100.0&100.0&\multicolumn{1}{c}{--}&\multicolumn{1}{c}{--}\\
&GOS
&6.5&7.1&5.9&7.0&\multicolumn{1}{c}{--}&\multicolumn{1}{c}{--}&&
 89.2&98.6&99.8&100.0&\multicolumn{1}{c}{--}&\multicolumn{1}{c}{--}\\
&PY
&10.0&8.1&5.7&7.0&5.2&5.5&&
 93.4&98.8&99.8&100.0&100.0&100.0\\
&AS
&6.9&7.3&6.4&7.4&4.9&4.7&&
 18.6&64.4&96.6&99.9&100.0&100.0\\
\midrule
500&Thm.\ 1
&3.4&5.3&3.3&3.2&4.2&4.0&&
 100.0&100.0&100.0&100.0&100.0&100.0\\
&FLLM
&4.9&3.7&2.8&3.2&2.6&3.4&&
 100.0&100.0&100.0&100.0&100.0&100.0\\
&GRS
&\multicolumn{1}{c}{--}&\multicolumn{1}{c}{--}&\multicolumn{1}{c}{--}&\multicolumn{1}{c}{--}&4.2&4.8&&
 \multicolumn{1}{c}{--}&\multicolumn{1}{c}{--}&\multicolumn{1}{c}{--}&\multicolumn{1}{c}{--}&100.0&100.0\\
&FLY
&38.5&19.3&13.8&8.6&\multicolumn{1}{c}{--}&\multicolumn{1}{c}{--}&&
 100.0&100.0&100.0&100.0&\multicolumn{1}{c}{--}&\multicolumn{1}{c}{--}\\
&GOS
&9.7&6.0&6.1&6.7&\multicolumn{1}{c}{--}&\multicolumn{1}{c}{--}&&
 97.2&99.9&100.0&100.0&\multicolumn{1}{c}{--}&\multicolumn{1}{c}{--}\\
&PY
&12.7&7.0&6.2&6.5&6.9&7.5&&
 99.2&100.0&100.0&100.0&100.0&100.0\\
&AS
&7.5&6.8&8.3&8.4&5.3&6.9&&
 21.5&78.2&99.7&100.0&100.0&100.0\\
\midrule
&&\multicolumn{6}{c}{$\phi_g=0.4$; $\alpha_i=0$}&&
  \multicolumn{6}{c}{$\phi_g=0.4$; $\alpha_i\sim N(0,1)$ (5\%)}\\
\midrule
$N$&$\text{Test }\backslash T$
&{100}&{200}&{300}&{500}&{1000}&{2000}&&
 {100}&{200}&{300}&{500}&{1000}&{2000}\\
 \midrule
100&Thm.\ 1
&5.3&3.9&3.2&3.6&3.3&3.2&&
 88.9&93.4&94.5&97.3&98.3&99.1\\
&FLLM
&11.6&9.0&10.1&9.5&11.5&11.1&&
 97.8&99.4&100.0&100.0&100.0&100.0\\
&GRS
&\multicolumn{1}{c}{--}&5.6&6.1&4.5&3.6&3.8&&
 \multicolumn{1}{c}{--}&100.0&100.0&100.0&100.0&100.0\\
&FLY
&33.2&18.3&14.0&10.9&\multicolumn{1}{c}{--}&\multicolumn{1}{c}{--}&&
 99.4&100.0&100.0&100.0&\multicolumn{1}{c}{--}&\multicolumn{1}{c}{--}\\
&GOS
&20.5&17.4&19.0&19.8&\multicolumn{1}{c}{--}&\multicolumn{1}{c}{--}&&
 80.4&93.8&97.7&98.9&\multicolumn{1}{c}{--}&\multicolumn{1}{c}{--}\\
&PY
&23.3&17.4&18.6&19.5&21.4&20.9&&
 84.5&93.4&97.4&98.7&99.9&100.0\\
&AS
&7.3&7.7&8.0&8.5&6.1&6.2&&
 15.9&46.7&79.9&96.9&99.7&100.0\\
\midrule
200&Thm.\ 1
&5.0&4.5&3.9&4.3&4.2&4.0&&
 98.4&99.4&99.8&100.0&99.9&100.0\\
&FLLM
&12.1&10.5&9.5&10.3&10.1&9.1&&
 100.0&100.0&100.0&100.0&100.0&100.0\\
&GRS
&\multicolumn{1}{c}{--}&\multicolumn{1}{c}{--}&6.4&4.6&3.0&4.3&&
 \multicolumn{1}{c}{--}&\multicolumn{1}{c}{--}&100.0&100.0&100.0&100.0\\
&FLY
&39.0&19.4&15.2&9.4&\multicolumn{1}{c}{--}&\multicolumn{1}{c}{--}&&
 100.0&100.0&100.0&100.0&\multicolumn{1}{c}{--}&\multicolumn{1}{c}{--}\\
&GOS
&19.3&20.2&17.9&20.2&\multicolumn{1}{c}{--}&\multicolumn{1}{c}{--}&&
 87.7&98.4&99.6&100.0&\multicolumn{1}{c}{--}&\multicolumn{1}{c}{--}\\
&PY
&24.0&20.1&17.8&20.0&18.5&18.7&&
 90.9&98.4&99.6&100.0&100.0&100.0\\
&AS
&8.7&9.4&9.0&10.9&5.1&5.6&&
 18.5&56.5&91.4&99.4&100.0&100.0\\
\midrule
500&Thm.\ 1
&6.0&5.4&4.0&3.4&4.2&4.1&&
 100.0&100.0&100.0&100.0&100.0&100.0\\
&FLLM
&13.7&10.2&10.5&9.4&9.1&10.5&&
 100.0&100.0&100.0&100.0&100.0&100.0\\
&GRS
&\multicolumn{1}{c}{--}&\multicolumn{1}{c}{--}&\multicolumn{1}{c}{--}&\multicolumn{1}{c}{--}&4.6&5.1&&
 \multicolumn{1}{c}{--}&\multicolumn{1}{c}{--}&\multicolumn{1}{c}{--}&\multicolumn{1}{c}{--}&100.0&100.0\\
&FLY
&48.0&21.3&16.3&11.6&\multicolumn{1}{c}{--}&\multicolumn{1}{c}{--}&&
 100.0&100.0&100.0&100.0&\multicolumn{1}{c}{--}&\multicolumn{1}{c}{--}\\
&GOS
&21.6&17.6&18.2&18.4&\multicolumn{1}{c}{--}&\multicolumn{1}{c}{--}&&
 94.6&99.9&100.0&100.0&\multicolumn{1}{c}{--}&\multicolumn{1}{c}{--}\\
&PY
&24.6&18.1&18.1&18.2&20.8&19.7&&
 97.7&99.9&100.0&100.0&100.0&100.0\\
&AS
&9.4&7.8&8.4&9.1&6.2&7.3&&
 21.1&65.3&96.1&99.9&100.0&100.0\\
\bottomrule
\end{tabular}
 }
\end{center}
\par
{\small \justifying
\textbf{Note: }{The nominal size is 5\% and powers are assessed at 5\% level
of significance; frequencies are computed across $M=1000$ Monte Carlo
samples and we set $\nu=5$ when computing $\psi_{i,NT}$. All the assets are
exposed to the first pricing factor, while only $\lfloor N^{0.8}\rfloor $ of
them load on the remaining ones. Results are in percentage points.} }
\end{table}

\newpage
\clearpage

\subsection{Non-tradable factors\label{MC-FF}}

We now study the finite sample properties of the testing procedure described
in Theorem \ref{ff-obs}. To do it, we consider a three-factor pricing model
similar to that of Section \ref{MC_omitted_factors_strength} but with no
factor structure in $u_{i,t}$, viz. 
\begin{equation*}
\begin{aligned} &y_{i,t} =\alpha _{i}+\beta^{\prime}\lambda +
\sum_{p=1}^{3}\beta _{i,p}v_{p,t}+u_{i,t}, \\ &v_{t} =\Phi
v_{t-1}+\zeta_{t}, \end{aligned}
\end{equation*}%
where we have constructed the DGP in terms of $v_{t}$ for coherence with
Section \ref{nontradable} (see equation \eqref{eq:DGP_thousands}, in
particular). Similarly to the main body, $\mathbf{u}_{t}=\left(
u_{1,t},\dots ,u_{N,t}\right) ^{\prime }$ follows one of the following three
specifications:

\begin{enumerate}
\item \textbf{The Gaussian case}: $\mathbf{u}_{t}\overset{i.i.d.}{\sim }%
\mathcal{N}_{N}(0,I_{N})$.

\item \textbf{The Student's $t$ case}: $u_{i,t}$ follows a Students's $t$
distribution with $d=5.5$ degrees of freedom, zero mean and unit scale,
independent across $i$. In this case, $u_{i,t}$ and $y_{i,t}$ have regularly
varying tails;

\item \textbf{The GARCH case}: we generate $\mathbf{u}_{t}=\mathbf{H}_{t}%
\mathbf{z}_{t}$, with: $\mathbf{z}_{t}=\left( z_{1,t},...,z_{N,t}\right)
^{\prime }$ and $z_{i,t}\overset{i.i.d.}{\sim }\mathcal{N}(0,1)$; and $%
\mathbf{H}_{t}=\mathrm{diag}\left\{ h_{1,t},\dots ,h_{N,t}\right\} $ with $%
h_{i,t}^{2}=\omega _{i}+\alpha _{i}\xi _{i,t}^{2}+\beta _{i}h_{i,t-1}^{2}$.
\end{enumerate}

Values of $\Phi$, $\beta = \left(\beta_1,\beta_2,\beta_3\right)^{\prime}$,
and of GARCH parameters are set as in Section \ref{simulations}, while $%
\lambda_p \overset{i.i.d.}{\sim}\mathcal{U}(0,1/2)$ for $p=1,2,3$.

Table \ref{tab:size_FM} reports empirical rejection frequencies under the
null (left panels) and under the alternative (right panel) for the test in
Theorem \ref{ff-obs}. As in Section \ref{simulations}, we set the nominal
size to $\tau = 5\%$ and study powers for the same significance level.
Because our test is the only one with a fully-fledged asymptotic theory for
non-tradabale factors, we do not report rejection frequencies for other
approaches. Our procedure satisfactorily controls the size for any sample
size and specification of the residuals $u_{i,t}$. Empirical powers are
consistently one. Table \ref{tab:derand_FM} presents rejection frequencies
for a derandomized procedure similar to that of Theorem \ref{strong-rule}
but based on Theorem \ref{ff-obs} rather than \ref{asy-max}. No matter the
residuals' properties, empirical rejection frequencies always converge to
zero as the sample size increases. As for the analysis in Section \ref%
{simulations}, this convergence is much faster using critical values based
on $f(B) = B^{-1/4}$. 
\begin{table}[!t]
\caption{Empirical rejection frequencies for the test in Theorem \protect\ref%
{ff-obs}.}
\label{tab:size_FM}
\begin{center}
{\scriptsize \centering
\begin{tabular}{l|
                S S S S S S
                c
                S S S S S S}
\toprule
&\multicolumn{13}{c}{Gaussian case}\\
&\multicolumn{6}{c}{$\alpha_i=0$ for all $i$}&&
 \multicolumn{6}{c}{$\alpha_i\sim N(0,1)$ for 5\% of units}\\
\midrule
$N\backslash T$
&{100}&{200}&{300}&{500}&{1000}&{2000}&&
 {100}&{200}&{300}&{500}&{1000}&{2000}\\
\midrule
100
&4.1&3.3&2.5&2.4&4.2&3.0&&
 100.0&100.0&100.0&100.0&100.0&100.0\\
200
&4.8&3.7&3.8&4.5&4.5&4.0&&
 100.0&100.0&100.0&100.0&100.0&100.0\\
500
&4.3&3.9&3.3&3.2&3.5&3.6&&
 100.0&100.0&100.0&100.0&100.0&100.0\\
\midrule
&\multicolumn{13}{c}{Student's $t$ case}\\
\midrule
$N\backslash T$
&{100}&{200}&{300}&{500}&{1000}&{2000}&&
 {100}&{200}&{300}&{500}&{1000}&{2000}\\
\midrule
100
&2.6&3.3&3.4&2.2&2.9&3.2&&
 100.0&100.0&100.0&100.0&100.0&100.0\\
200
&3.5&3.2&4.0&3.8&3.7&2.1&&
 100.0&100.0&100.0&100.0&100.0&100.0\\
500
&5.8&4.0&3.0&4.0&2.7&3.3&&
 100.0&100.0&100.0&100.0&100.0&100.0\\
\midrule
&\multicolumn{13}{c}{GARCH case}\\
\midrule
$N\backslash T$
&{100}&{200}&{300}&{500}&{1000}&{2000}&&
 {100}&{200}&{300}&{500}&{1000}&{2000}\\
\midrule
100
&5.5&4.5&3.2&4.3&4.0&2.5&&
 100.0&100.0&100.0&100.0&100.0&100.0\\
200
&8.1&3.7&4.2&1.9&2.7&2.6&&
 100.0&100.0&100.0&100.0&100.0&100.0\\
500
&9.1&3.8&4.0&4.5&2.9&3.4&&
 100.0&100.0&100.0&100.0&100.0&100.0\\
\bottomrule
\end{tabular}
 }
\end{center}
\par
{\small \justifying
\textbf{Note: }{The nominal size is 5\% and powers are assessed at 5\% level
of significance; frequencies are computed across $M=1000$ Monte Carlo
samples and we set $\nu=5$ when computing $\psi_{i,NT}^{FM}$. Results are in percentage points.} }
\end{table}

\begin{table}[!t]
\caption{Empirical rejection frequencies for a derandomized procedure based
on Theorem \protect\ref{ff-obs}.}
\label{tab:derand_FM}
\begin{center}
{\scriptsize \centering
\begin{tabular}{ll|
                S S S S S S
                c
                S S S S S S}
\toprule
&&\multicolumn{13}{c}{Gaussian case}\\
&&\multicolumn{6}{c}{$\alpha_i=0$ for all $i$}&&
  \multicolumn{6}{c}{$\alpha_i\sim N(0,1)$ for 1\% of units}\\
\midrule
$N$&$\text{C.V. }\backslash T$
&{100}&{200}&{300}&{500}&{1000}&{2000}&&
 {100}&{200}&{300}&{500}&{1000}&{2000}\\
\midrule
100&LIL
&0.1&0.0&0.0&0.0&0.0&0.0&&
 100.0&100.0&100.0&100.0&100.0&100.0\\
&$f(B)=B^{-1/4}$
&0.0&0.0&0.0&0.0&0.0&0.0&&
 100.0&100.0&100.0&100.0&100.0&100.0\\
200&LIL
&0.2&0.0&0.0&0.0&0.0&0.0&&
 100.0&100.0&100.0&100.0&100.0&100.0\\
&$f(B)=B^{-1/4}$
&0.1&0.0&0.0&0.0&0.0&0.0&&
 100.0&100.0&100.0&100.0&100.0&100.0\\
500&LIL
&27.7&8.6&5.3&1.7&0.0&0.0&&
 100.0&100.0&100.0&100.0&100.0&100.0\\
&$f(B)=B^{-1/4}$
&0.0&0.0&0.0&0.0&0.0&0.0&&
 100.0&100.0&100.0&100.0&100.0&100.0\\
\midrule
&&\multicolumn{13}{c}{Student's $t$ case}\\
\midrule
$N$&$\text{C.V. }\backslash T$
&{100}&{200}&{300}&{500}&{1000}&{2000}&&
 {100}&{200}&{300}&{500}&{1000}&{2000}\\
\midrule
100&LIL
&0.9&0.0&0.0&0.0&0.0&0.0&&
 100.0&100.0&100.0&100.0&100.0&100.0\\
&$f(B)=B^{-1/4}$
&0.0&0.0&0.0&0.0&0.0&0.0&&
 100.0&100.0&100.0&100.0&100.0&100.0\\
200&LIL
&0.5&0.0&0.0&0.0&0.0&0.0&&
 100.0&100.0&100.0&100.0&100.0&100.0\\
&$f(B)=B^{-1/4}$
&0.1&0.0&0.0&0.0&0.0&0.0&&
 100.0&100.0&100.0&100.0&100.0&100.0\\
500&LIL
&30.8&10.4&4.9&1.9&0.0&0.0&&
 100.0&100.0&100.0&100.0&100.0&100.0\\
&$f(B)=B^{-1/4}$
&0.0&0.0&0.0&0.0&0.0&0.0&&
 100.0&100.0&100.0&100.0&100.0&100.0\\
\midrule
&&\multicolumn{13}{c}{GARCH noise}\\
\midrule
$N$&$\text{C.V. }\backslash T$
&{100}&{200}&{300}&{500}&{1000}&{2000}&&
 {100}&{200}&{300}&{500}&{1000}&{2000}\\
\midrule
100&LIL
&3.9&0.5&0.1&0.0&0.0&0.0&&
 100.0&100.0&100.0&100.0&100.0&100.0\\
&$f(B)=B^{-1/4}$
&0.3&0.2&0.0&0.0&0.0&0.0&&
 100.0&100.0&100.0&100.0&100.0&100.0\\
200&LIL
&6.7&0.9&0.0&0.1&0.0&0.0&&
 100.0&100.0&100.0&100.0&100.0&100.0\\
&$f(B)=B^{-1/4}$
&1.0&0.4&0.0&0.0&0.0&0.0&&
 100.0&100.0&100.0&100.0&100.0&100.0\\
500&LIL
&45.3&17.2&8.3&3.4&0.0&0.0&&
 100.0&100.0&100.0&100.0&100.0&100.0\\
&$f(B)=B^{-1/4}$
&2.6&0.3&0.0&0.1&0.0&0.0&&
 100.0&100.0&100.0&100.0&100.0&100.0\\
\bottomrule
\end{tabular}
 }
\end{center}
\par
{\scriptsize 
\textbf{Note: }{Results using either LIL-based critical values (LIL) or
critical values based on $f(B)=B^{-1/4}$. The derandomized statistic is
based on nominal level $\tau=5\%$. We set $B=\left(\mathrm{\log }%
\,N\right)^{2}$ and $\nu=5$. Results are in percentage points.} }
\end{table}

\newpage

\subsection{Latent factors\label{MC-latent}}

We now investigate the finite sample properties of the test in Theorem \ref%
{ff-unobs} when considering the DGP of Section \ref{MC-FF} with unobserved
factors. Again, size and power are studied for nominal level $\tau = 5\%$.
Results for empirical rejections frequencies under the null (left panels)
and under the alternative (right panels) are reported in Table \ref%
{tab:size_PC}. Again, the testing procedure exhibits satisfactory size
control and excellent power properties. The test is a bit oversized in the
GARCH case, particularly when $T=100$. Given the results on the other
estimators (standard OLS and Fama-MacBeth), this over-rejection seem to be
due to problems in PCA-based estimation of factors and loadings under this
particular GARCH DGP. Similar conclusions hold when we look at results on
the derandomized procedure in Table \ref{tab:derandomized_PC}, where we see
that the approach works fine up to some over-rejection of the null in the
GARCH case.

\begin{table}[!t]
\caption{Empirical rejection frequencies for the test in Theorem \protect\ref%
{ff-unobs}.}
\label{tab:size_PC}
\begin{center}
{\scriptsize \centering
\begin{tabular}{l|
                S S S S S S
                c
                S S S S S S}
\toprule
&\multicolumn{13}{c}{Gaussian case}\\
&\multicolumn{6}{c}{$\alpha_i = 0$ for all $i$}&&
 \multicolumn{6}{c}{$\alpha_i \sim N(0,1)$ for 5\% of units}\\
\midrule
$N\backslash T$
&{100}&{200}&{300}&{500}&{1000}&{2000}&&
 {100}&{200}&{300}&{500}&{1000}&{2000}\\
\midrule
100
&4.7&3.3&2.5&2.4&4.2&3.0&&
 100.0&100.0&100.0&100.0&100.0&100.0\\
200
&5.2&3.7&3.8&4.5&4.5&3.8&&
 100.0&100.0&100.0&100.0&100.0&100.0\\
500
&5.0&4.0&3.3&3.2&3.5&3.6&&
 100.0&100.0&100.0&100.0&100.0&100.0\\
\midrule
&\multicolumn{13}{c}{Student's $t$ case}\\
\midrule
$N\backslash T$
&{100}&{200}&{300}&{500}&{1000}&{2000}&&
 {100}&{200}&{300}&{500}&{1000}&{2000}\\
\midrule
100
&3.3&3.2&3.3&2.2&2.9&3.2&&
 100.0&100.0&100.0&100.0&100.0&100.0\\
200
&4.2&3.3&4.1&3.8&3.7&2.1&&
 100.0&100.0&100.0&100.0&100.0&100.0\\
500
&6.5&4.0&3.0&4.1&2.7&3.3&&
 100.0&100.0&100.0&100.0&100.0&100.0\\
\midrule
&\multicolumn{13}{c}{GARCH case}\\
\midrule
$N\backslash T$
&{100}&{200}&{300}&{500}&{1000}&{2000}&&
 {100}&{200}&{300}&{500}&{1000}&{2000}\\
\midrule
100
&8.0&4.7&3.0&4.3&4.2&3.0&&
 100.0&100.0&100.0&100.0&100.0&100.0\\
200
&13.1&5.1&4.3&1.9&4.5&3.8&&
 100.0&100.0&100.0&100.0&100.0&100.0\\
500
&19.3&6.4&4.9&4.6&3.5&3.6&&
 100.0&100.0&100.0&100.0&100.0&100.0\\
\bottomrule
\end{tabular}
 }
\end{center}
\par
{\small \justifying
\textbf{Note: }{The nominal size is 5\% and powers are assessed at 5\% level
of significance; frequencies are computed across $M=1000$ Monte Carlo
samples and we set $\nu=5$ when computing $\psi_{i,NT}^{PC}$.  Results are in percentage points.} }
\end{table}

\begin{table}[!t]
\caption{Empirical rejection frequencies for a derandomized procedure based
on Theorem \protect\ref{ff-unobs}.}
\label{tab:derandomized_PC}
\begin{center}
{\scriptsize \centering
\begin{tabular}{ll|
                S S S S S S
                c
                S S S S S S}
\toprule
&&\multicolumn{13}{c}{Gaussian case} \\
&&\multicolumn{6}{c}{$\phi_{\nu}=0$; $\alpha_i=0$ for all $i$}&&
  \multicolumn{6}{c}{$\phi_{\nu}=0$; $\alpha_i\sim N(0,1)$ for 1\% of units} \\
\midrule
$N$&$\text{C.V. }\backslash T$
&{100}&{200}&{300}&{500}&{1000}&{2000}&&
 {100}&{200}&{300}&{500}&{1000}&{2000}\\
\midrule
100&LIL
&0.5&0.0&0.0&0.0&0.0&0.0&&
 100.0&100.0&100.0&100.0&100.0&100.0\\
&$f(B)=B^{-1/4}$
&0.0&0.0&0.0&0.0&0.0&0.0&&
 100.0&100.0&100.0&100.0&100.0&100.0\\
200&LIL
&0.7&0.0&0.0&0.0&0.0&0.0&&
 100.0&100.0&100.0&100.0&100.0&100.0\\
&$f(B)=B^{-1/4}$
&0.0&0.0&0.0&0.0&0.0&0.0&&
 100.0&100.0&100.0&100.0&100.0&100.0\\
500&LIL
&39.9&13.4&8.7&3.2&0.0&0.0&&
 100.0&100.0&100.0&100.0&100.0&100.0\\
&$f(B)=B^{-1/4}$
&0.0&0.0&0.0&0.0&0.0&0.0&&
 100.0&100.0&100.0&100.0&100.0&100.0\\
\midrule
&&\multicolumn{13}{c}{Student's $t$ case} \\
\midrule
$N$&$\text{C.V. }\backslash T$
&{100}&{200}&{300}&{500}&{1000}&{2000}&&
 {100}&{200}&{300}&{500}&{1000}&{2000}\\
\midrule
100&LIL
&0.7&0.0&0.0&0.0&0.0&0.0&&
 100.0&100.0&100.0&100.0&100.0&100.0\\
&$f(B)=B^{-1/4}$
&0.0&0.0&0.0&0.0&0.0&0.0&&
 100.0&100.0&100.0&100.0&100.0&100.0\\
200&LIL
&1.5&0.0&0.1&0.0&0.0&0.0&&
 100.0&100.0&100.0&100.0&100.0&100.0\\
&$f(B)=B^{-1/4}$
&0.2&0.0&0.1&0.0&0.0&0.0&&
 100.0&100.0&100.0&100.0&100.0&100.0\\
500&LIL
&37.8&13.8&8.8&3.6&0.0&0.0&&
 100.0&100.0&100.0&100.0&100.0&100.0\\
&$f(B)=B^{-1/4}$
&0.3&0.0&0.0&0.0&0.0&0.0&&
 100.0&100.0&100.0&100.0&100.0&100.0\\
\midrule
&&\multicolumn{13}{c}{GARCH case} \\
\midrule
$N$&$\text{C.V. }\backslash T$
&{100}&{200}&{300}&{500}&{1000}&{2000}&&
 {100}&{200}&{300}&{500}&{1000}&{2000}\\
\midrule
100&LIL
&9.6&0.5&0.2&0.0&0.0&0.0&&
 100.0&100.0&100.0&100.0&100.0&100.0\\
&$f(B)=B^{-1/4}$
&2.5&0.0&0.0&0.0&0.0&0.0&&
 100.0&100.0&100.0&100.0&100.0&100.0\\
200&LIL
&17.4&3.9&0.4&0.0&0.0&0.0&&
 100.0&100.0&100.0&100.0&100.0&100.0\\
&$f(B)=B^{-1/4}$
&4.5&0.8&0.0&0.0&0.0&0.0&&
 100.0&100.0&100.0&100.0&100.0&100.0\\
500&LIL
&72.7&34.2&17.4&7.2&0.0&0.0&&
 100.0&100.0&100.0&100.0&100.0&100.0\\
&$f(B)=B^{-1/4}$
&12.0&1.7&0.2&0.1&0.0&0.0&&
 100.0&100.0&100.0&100.0&100.0&100.0\\
\bottomrule
\end{tabular}
 }
\end{center}
\par
{\scriptsize 
\textbf{Note: }{Results using either LIL-based critical values (LIL) or
critical values based on $f(B)=B^{-1/4}$. The derandomized statistic is
based on nominal level $\tau=5\%$. We set $B=\left(\mathrm{\log }%
\,N\right)^{2} $ and $\nu=5$. Results are in percentage points.} }
\end{table}

\newpage
\subsection{Alternative critical values\label{app:MC_newCV}}

The results in Section \ref{simulations} suggest that a test based on $c_\tau$ becomes slightly  undersized as $T$ grows large (for any given $N$). In this section, we consider an alternative, fixed $N$ family of critical values 
\begin{equation}
c_{\tau}^{(1)}  = \Phi^{-1}\left((1-\tau)^{1/N}\right). \label{crv_fixed_N}
\end{equation}
These critical values are theoretically supported by the fact that the Gumbel cumulative distribution function  appears as the limit of the term $\Phi^{N}\left(a_{n}x +b_{n}\right)$ in the proof of Theorem \ref{asy-max}. Hence, these Gaussian quantiles area a natural finite sample counterpart to the Gumbel ones.

Results for the DGPs of the main body, even when $\phi_g = 0$, are in Table \ref{tab:size_newCV}. Empirical rejection frequencies are always larger than those based on asymptotic critical values. This is beneficial when $T$ gets larger, where asymptotic critical values returned empirical sizes lower than the nominal one (i.e.\ 5\%). Notably, large $T$ results (say, $T$ at most 300), are now on par with, if not better than, those for the GRS approach. Hence, we suggest using our test with asymptotic critical values when $T$ is small. When $T$ is larger, the user should employ either our test with adjusted critical values or, if the sample size permits it,  the GRS one.

\begin{table}[!t]
\caption{Empirical rejection frequencies for the test in Theorem \ref{asy-max}}
\label{tab:size_newCV}
\begin{center}
{\scriptsize \centering
\begin{tabular}{l|
                S S S S S S
                c
                S S S S S S}
\toprule
&\multicolumn{13}{c}{Gaussian Case, $\phi_g = 0$}\\
&\multicolumn{6}{c}{$\alpha_i = 0$ for all $i$}&&
 \multicolumn{6}{c}{$\alpha_i \sim N(0,1)$ for 5\% of units}\\
\midrule
$N\backslash T$
&{100}&{200}&{300}&{500}&{1000}&{2000}&&
 {100}&{200}&{300}&{500}&{1000}&{2000}\\
\midrule
100&7.1  &5.4 & 5.5  &5.5  &5.3  &5.2&&95.2   &97.0 &98.1& 98.9 &99.4 &99.8\\
200 &6.3&  5.4  &5.4  &5.1  &5.1&  5.1&&99.8 &99.9& 99.9& 99.9& 100.0 &100.0\\
500&7.1  &5.2&  5.4 & 5.2  &5.0  &4.7&&100.0&100.0&100.0&100.0&100.0&100.0\\
\midrule
&\multicolumn{13}{c}{Gaussian Case, $\phi_g = 0.4$}\\
&\multicolumn{6}{c}{$\alpha_i = 0$ for all $i$}&&
 \multicolumn{6}{c}{$\alpha_i \sim N(0,1)$ for 5\% of units}\\
\midrule
$N\backslash T$
&{100}&{200}&{300}&{500}&{1000}&{2000}&&
 {100}&{200}&{300}&{500}&{1000}&{2000}\\
\midrule
100&8.9  &6.1  &5.7 & 5.6&  5.3  &5.3&&95.1& 96.1& 97.9& 98.7 &99.3 &99.6\\
200&8.4 & 6.0  &5.5 & 5.3 & 5.2 & 5.3&& 99.8& 99.9& 99.9& 99.9 &99.9& 100.0\\
500& 9.7 & 6.4 & 5.7  &5.4 & 5.0  &4.8&&100.0&100.0&100.0&100.0&100.0&100.0\\
\midrule
&\multicolumn{13}{c}{Student's $t$ case, $\phi_g = 0$}\\
&\multicolumn{6}{c}{$\alpha_i = 0$ for all $i$}&&
 \multicolumn{6}{c}{$\alpha_i \sim N(0,1)$ for 5\% of units}\\
\midrule
$N\backslash T$
&{100}&{200}&{300}&{500}&{1000}&{2000}&&
 {100}&{200}&{300}&{500}&{1000}&{2000}\\
\midrule
100&7.0&  5.6&  5.3&  5.3&  5.3&  5.2&&90.4& 94.9 &96.7 &97.6& 98.8 &99.6\\
200& 6.4&  6.1&  5.2 & 4.9&  5.1  &4.9&&98.9& 99.7 &99.9 &99.9 &99.9 &100.0\\
500&  6.1 & 5.3 & 5.1&  5.0  &4.9  &4.8&& 100.0&100.0&100.0&100.0&100.0&100.0\\
\midrule
&\multicolumn{13}{c}{Student's $t$ case, $\phi_g = 0.4$}\\
&\multicolumn{6}{c}{$\alpha_i = 0$ for all $i$}&&
 \multicolumn{6}{c}{$\alpha_i \sim N(0,1)$ for 5\% of units}\\
\midrule
$N\backslash T$
&{100}&{200}&{300}&{500}&{1000}&{2000}&&
 {100}&{200}&{300}&{500}&{1000}&{2000}\\
\midrule
100&7.9 & 5.9  &5.6 & 5.4  &5.5  &5.3&&89.4 &93.6   &96.0 &97.8 &98.7 &99.5\\
200&8.2 & 6.2 & 5.9 & 5.1  &5.3 & 5.0&&98.9& 99.3& 99.9 &99.9 &99.9 &100.0\\
500& 8.3 & 5.1  &5.3 & 5.2 & 5.0 & 5.0&& 100.0&100.0&100.0&100.0&100.0&100.0\\
\midrule
&\multicolumn{13}{c}{GARCH case, $\phi_g = 0$}\\
&\multicolumn{6}{c}{$\alpha_i = 0$ for all $i$}&&
 \multicolumn{6}{c}{$\alpha_i \sim N(0,1)$ for 5\% of units}\\
\midrule
$N\backslash T$
&{100}&{200}&{300}&{500}&{1000}&{2000}&&
 {100}&{200}&{300}&{500}&{1000}&{2000}\\
\midrule
100&7.2 &5.7 &5.6& 5.3&5.2& 5.3&&97.9 &98.5 &98.7& 99.3 &99.7 &99.8\\
 200&7.9&5.7&5.2&5.1&5.1&5.0&&100.0&100.0&100.0&100.0&100.0&100.0\\
500&8.9&5.7&5.4&5.1&4.9&4.8&&100.0&100.0&100.0&100.0&100.0&100.0\\
\midrule
&\multicolumn{13}{c}{GARCH case, $\phi_g = 0.4$}\\
&\multicolumn{6}{c}{$\alpha_i = 0$ for all $i$}&&
 \multicolumn{6}{c}{$\alpha_i \sim N(0,1)$ for 5\% of units}\\
\midrule
$N\backslash T$
&{100}&{200}&{300}&{500}&{1000}&{2000}&&
 {100}&{200}&{300}&{500}&{1000}&{2000}\\
\midrule
100&10.0&6.7&5.9&5.3&5.3&5.4&&97.4&98.1&98.5&99.2&99.6&99.7\\
 200&10.8&6.9&5.8&5.2&5.2&5.1&&100.0&100.0&100.0&100.0&100.0&100.0\\
500&12.7&6.9&6.1&5.7&5.2&4.9&&100.0&100.0&100.0&100.0&100.0&100.0\\
\bottomrule
\end{tabular}
 }
\end{center}
\par
{\small \justifying
\textbf{Note: }{The nominal size is 5\% and powers are assessed at 5\% level
of significance; frequencies are computed across $M=1000$ Monte Carlo
samples and we set $\nu=5$ when computing $\psi_{i,NT}$.  Results are in percentage points.} }
\end{table}

\clearpage\newpage

\setcounter{equation}{0} \setcounter{lemma}{0} \setcounter{theorem}{0} %
\renewcommand{\theassumption}{B.\arabic{assumption}} 
\renewcommand{\thetheorem}{B.\arabic{theorem}} \renewcommand{\thelemma}{B.%
\arabic{lemma}} \renewcommand{\theproposition}{B.\arabic{proposition}} %
\renewcommand{\thecorollary}{A.\arabic{corollary}} \renewcommand{%
\theequation}{B.\arabic{equation}}

\section{Extensions\label{extend}}

As mentioned in the introduction to the main paper, our main contribution is
a methodology to test for no pricing errors, and we have primarily 
focused on (\ref{eq:fac_model}) and assumed factors are observable and
tradable only for simplicity. 

In this extension, we show that our methods can be readily extended to more
complex settings. In essence, we show that - as long as a consistent
estimator of the $\alpha _{i}$s is available - our methodology can be
applied even to the case of latent or non-tradable, obtaining the same results as in the case of observable and tradable ones. As illustrative examples, we consider  Fama-MacBeth estimation for the case of observable but non-tradable factors (Section \ref{nontradable}), and principal component analysis (PCA) estimation for latent factors (Section \ref{latent}); in both
cases, we directly use the estimation techniques proposed by %
\citet{giglio2021thousands} for the $\alpha _{i}$s. We only report the main
results on the \textquotedblleft one shot\textquotedblright\ tests; the extension to derandomization can be done by
following \textit{verbatim} Section \ref{derandom}. We note, however, that
strong factors are required in these cases. \\
We also consider the extension of the derandomized decision rule in the case of multiple testing (Section \ref{mht}), and of Assumption \ref{asymptotics} (Section \ref{complements}).

\smallskip

Henceforth, we define $\mathbb{M}_{1_{N}}=\mathbb{I}_{N}-N^{-1}\mathbf{\iota 
}_{N}\mathbf{\iota }_{N}^{\prime }$, where $\mathbf{\iota }_{N}$ is an $%
N\times 1$ vector of ones.

\subsection{Non-tradable factors and Fama-MacBeth estimation\label%
{nontradable}}

Consider first the case of a linear factor pricing model based on $K$
observable, \textit{non-tradable} factors 
\begin{equation}
y_{i,t}=\alpha _{i}+\beta _{i}^{\prime }\lambda +\beta _{i}^{\prime
}v_{t}+u_{i,t},  \label{eq:DGP_thousands}
\end{equation}%
where $v_{t}=f_{t}-\mathbb{E}\left( f_{t}\right) $ and $\lambda \in \mathbb{R%
}^{K}$ is the vector of risk premia for the $K$ factors $f_{t}$. 

Estimation of $\alpha _{i}$ is based on Algorithm 3 in %
\citet{giglio2021thousands}.

\smallskip

\begin{description}
\item[Step 1] Estimate $\beta _{i}$ by OLS in the time-series regressions $%
y_{i,t}=\pi _{i}+\beta _{i}'f_{t}+u_{i,t}$, 
\begin{equation}
\widehat{\beta }_{i}=\left[ \sum_{t=1}^{T}\left( f_{t}-\overline{f}\right)
\left( f_{t}-\overline{f}\right) ^{\prime }\right] ^{-1}\left[
\sum_{t=1}^{T}\left( f_{t}-\overline{f}\right) \left( y_{i,t}-\overline{y}%
_{i}\right) \right] ,  \label{beta_i_ff}
\end{equation}%
with $\overline{y}_{i}=T^{-1}\sum_{t=1}^{T}y_{i,t}$, and define $\widehat{%
\mathbf{\beta }}=\left( \widehat{\beta }_{1},...,\widehat{\beta }_{N}\right)
^{\prime }$.

\item[Step 2] Define $\widehat{\lambda }$ as the OLS estimates of a
regression of $\overline{\mathbf{y}}=\left( \overline{y}_{1},...,\overline{y}%
_{N}\right) ^{\prime }$ onto a vector of ones and $\widehat{\mathbf{\beta }}$
\begin{equation}
\widehat{\lambda }=\left( \widehat{\mathbf{\beta }}^{\prime }\mathbb{M}%
_{1_{N}}\widehat{\mathbf{\beta }}\right) ^{-1}\left( \widehat{\mathbf{\beta }%
}^{\prime }\mathbb{M}_{1_{N}}\bar{\mathbf{y}}\right) .  \label{lambda_ff}
\end{equation}

\item[Step 3] The estimator of $\alpha _{i}$ is given by 
\begin{equation}
\widehat{\alpha }_{i}^{FM}=\overline{y}_{i}-\widehat{\beta }_{i}^{\prime }%
\widehat{\lambda }.  \label{eq:alpha_FM}
\end{equation}
\end{description}

\smallskip

Based on $\widehat{\alpha }_{i}^{FM}$ defined in (\ref{eq:alpha_FM}), we can
construct the same test statistic as before, based on 
\begin{equation*}
\psi _{i,NT}^{FM}=\left\vert \frac{T^{1/\nu }\widehat{\alpha }_{i}^{FM}}{%
\widehat{s}_{NT}^{FM}}\right\vert ^{\nu /2},
\end{equation*}%
where the rescaling sequence $\widehat{s}_{NT}^{FM}$ is constructed as in (%
\ref{s_NT}), using the residuals $\widehat{u}_{i,t}^{FM}=y_{i,t}-\left( 
\widehat{\alpha }_{i}^{FM}+\widehat{\beta }_{i}^{\prime }f_{t}\right) $.
Defining $z_{i,NT}^{FM}=\psi _{i,NT}^{FM}+\omega _{i}$,\footnote{%
As before, $\omega _{i}\overset{i.i.d.}{\sim }\mathcal{N}\left( 0,1\right) $
generated independently of the sample $\left\{ \left( u_{i,t},f_{t}^{\prime
}\right) ^{\prime },1\leq i\leq N,1\leq t\leq T\right\} $.} our test can be
based on 
\begin{equation}
Z_{N,T}^{FM}=\max_{1\leq i\leq N}z_{i,NT}^{FM}.  \label{z_FM}
\end{equation}

\bigskip 

In order to derive our asymptotic theory, we consider the following
assumptions. Let%
\begin{equation}
\mathbf{S}_{\beta }=\frac{1}{N}\sum_{i=1}^{N}\left( \beta _{i}-\overline{%
\beta }\right) \left( \beta _{i}-\overline{\beta }\right) ^{\prime },
\label{s-beta}
\end{equation}%
where we define 
\begin{equation}
\overline{\beta }=\frac{1}{N}\sum_{i=1}^{N}\beta _{i}.  \label{beta-bar}
\end{equation}

\begin{assumption}
\label{cov_beta} It holds that: \textit{(i)} $\lambda $ and $\beta _{i}$ are
fixed with $\left\Vert \lambda \right\Vert <\infty $ and $\max_{1\leq i\leq
N}\left\Vert \beta _{i}\right\Vert <\infty $; and \textit{(ii)} $\mathbf{S}%
_{\beta }$ is positive definite for all values of $N$.
\end{assumption}

\begin{assumption}
\label{weak_CS_dep} It holds that: \textit{(i)} $\left\{ u_{i,t},1\leq t\leq
T\right\} $ and $\left\{ f_{t},1\leq t\leq T\right\} $ are two mutually
independent groups; and \textit{(ii)} $\sum_{i=1}^{N}\sum_{j=1}^{N}%
\sum_{s=1}^{T}\sum_{t=1}^{T}\left\vert \mathbb{E}\left(
u_{i,t}u_{j,s}\right) \right\vert \leq c_{0}NT$.
\end{assumption}

\bigskip 

It holds that

\begin{theorem}
\label{ff-obs}We assume that the assumptions of Theorem \ref{asy-max} are
satisfied, and that Assumptions \ref{cov_beta} and \ref{weak_CS_dep} also hold. Then, the same
result as in Theorem \ref{asy-max} holds.
\end{theorem}

\subsection{Latent factors\label{latent}}

Consider now a linear factor pricing model based on $K$ latent factors 
\begin{equation}
y_{i,t}=\alpha _{i}+\beta _{i}^{\prime }\lambda +\beta _{i}^{\prime
}v_{t}+u_{i,t},  \label{latent-f}
\end{equation}%
where $v_{t}=f_{t}-\mathbb{E}\left( f_{t}\right) $ is \textit{not observable}%
, $\beta _{i}$ is a $K\times 1$ vector of loadings and, as above, $\lambda $
is the vector of risk premia for the $K$ latent factors $f_{t}$.

Write $\widetilde{\mathbf{y}}_{t}=\mathbf{\beta }\widetilde{v}_{t}+%
\widetilde{\mathbf{u}}_{t}$, where $\widetilde{\mathbf{y}}_{t}=\mathbf{y}%
_{t}-\overline{\mathbf{y}}$ with $\mathbf{y}_{t}=\left(
y_{1,t},...,y_{N,t}\right) ^{\prime }$, $\widetilde{v}_{t}=v_{t}-\left(
T^{-1}\sum_{t=1}^{T}v_{t}\right) $, $\widetilde{\mathbf{u}}_{t}$ is defined
analogously, and $\mathbf{\beta }=\left( \beta _{1},...,\beta _{N}\right)
^{\prime }$. Define also the $N\times N$ sample second moment matrix%
\begin{equation}
\widehat{\mathbf{\Sigma }}_{y}=\frac{1}{NT}\sum_{t=1}^{T}\widetilde{\mathbf{y%
}}_{t}\widetilde{\mathbf{y}}_{t}^{\prime }.  \label{sig-y}
\end{equation}%
The estimation of $\alpha _{i}$ follows Algorithm 4 in %
\citet{giglio2021thousands}.

\smallskip

\begin{description}
\item[Step 1] Estimate $\mathbf{\beta }$ using PCA, with estimator $\widehat{%
\mathbf{\beta }}^{PC}$ given by the eigenvectors corresponding to the first $%
K$ eigenvalues of $\widehat{\mathbf{\Sigma }}_{y}$ under the constraint $%
\left( \widehat{\mathbf{\beta }}^{PC}\right) ^{\prime }\widehat{\mathbf{%
\beta }}^{PC}=N\mathbb{I}_{K}$.\footnote{%
Our discussion implicitly assumes that $K$ is known. Of course, this is not
the case in practice, where $K$ has to be determined by the user. This is
ordinarily done using consistent estimators such as those of \cite{baing02}, 
\cite{ahnhorenstein13}, and \cite{trapani2018randomized}. Note that the
result in Theorem \ref{ff-unobs} holds unchanged when $K$ is estimated using
these consistent estimators.}

\item[Steps 2 and 3] These steps are the same as in the previous section,
with $\widehat{\lambda }^{PC}=\left( \widehat{\mathbf{\beta }}^{PC\prime }%
\mathbb{M}_{1_{N}}\widehat{\mathbf{\beta }}^{PC}\right) ^{-1}\left( \widehat{%
\mathbf{\beta }}^{PC\prime }\mathbb{M}_{1_{N}}\bar{\mathbf{y}}\right) $, and 
\begin{equation}
\widehat{\alpha }_{i}^{PC}=\overline{y}_{i}-\left( \widehat{\beta }%
_{i}^{PC}\right) ^{\prime }\widehat{\lambda }^{PC}.  \label{alpha-unobs}
\end{equation}
\end{description}

\smallskip

Let $C_{N,T}=\min \left\{ N,T\right\} $. Based on $\widehat{\alpha }_{i}^{PC}
$ defined in (\ref{alpha-unobs}), we define 
\begin{equation*}
\psi _{i,NT}^{PC}=\left\vert \frac{C_{N,T}^{1/\nu }\widehat{\alpha }_{i}^{PC}%
}{\widehat{s}_{NT}^{PC}}\right\vert ^{\nu /2},
\end{equation*}%
where $\widehat{s}_{NT}^{PC}$ is constructed as in (\ref{s_NT}), using $%
\widehat{u}_{i,t}^{PC}=y_{i,t}-\left( \widehat{\alpha }_{i}^{PC}+\widehat{%
\beta }_{i}^{PC\prime }\widehat{f}_{t}^{PC}\right) $ and $\widehat{f}%
_{t}^{PC}=N^{-1}\widehat{\mathbf{\beta }}^{PC\prime }\widetilde{\mathbf{y}}%
_{t}$. Letting $z_{i,NT}^{PC}=\psi _{i,NT}^{PC}+\omega _{i},$ with $\omega
_{i}$ defined as above, the test is based on 
\begin{equation}
Z_{N,T}^{PC}=\max_{1\leq i\leq N}z_{i,NT}^{PC}.  \label{z_PC}
\end{equation}%
We consider the following assumptions - which complement and extend the
assumptions in the main paper - for the case of \textit{latent factors}.

\begin{assumption}
\label{loadings} It holds that: \textit{(i)} $\lambda$ and $\mathbf{\beta }$
are nonrandom with $\left\Vert\lambda\right\Vert <\infty$ and $\left\Vert 
\mathbf{\beta }\right\Vert <\infty $; \textit{(ii)} $\mathbf{\beta }^{\prime
}\mathbf{\beta }=N\mathbb{I}_{K}$; and \textit{(iii)} $\mathbf{S}_{\beta }$
is positive definite for all values of $N$.
\end{assumption}

\begin{assumption}
\label{factors} It holds that $\mathbb{E}\left( \widetilde{v}_{t}\widetilde{v%
}_{t}^{\prime }\right) $ is a positive definite matrix.
\end{assumption}

\begin{assumption}
\label{idiosyncratic} Let $\gamma _{s,t}=\sum_{i=1}^{N}\mathbb{E}\left(
u_{i,t}u_{i,s}\right) /N$. It holds that: \textit{(i)} $\sum_{s=1}^{T}\left%
\vert \gamma _{s,t}\right\vert <c_{0}$ for $1\leq t\leq T$; \textit{(ii)} $%
\mathbb{E}\left\vert \sum_{i=1}^{N}\left( u_{i,s}u_{i,t}-\gamma
_{s,t}\right) \right\vert ^{2}<c_{0}N$ for $1\leq s,t\leq T$; \textit{(iii)} 
$\mathbb{E}\left\Vert \sum_{i=1}^{N}\beta _{i}u_{i,t}\right\Vert
^{4}<c_{0}N^{2}$ for $1\leq t\leq T$; \textit{(iv)} $\sum_{i=1}^{N}\left%
\vert \mathbb{E}\left( u_{i,t}u_{j,s}\right) \right\vert \leq c_{0}$ for all 
$1\leq t,s\leq T$ and $1\leq j\leq N$; \textit{(v)} $\sum_{i=1}^{N}%
\sum_{s=1}^{T}\left\vert \mathbb{E}\left( u_{i,t}u_{j,s}\right) \right\vert
\leq c_{0}$ for all $1\leq t\leq T$ and $1\leq j\leq N$; \textit{(vi)} 
\begin{equation*}
\mathbb{E}\left\Vert \sum_{t=1}^{T}w^{\prime }\left( \mathbf{u}_{t}u_{i,t}-%
\mathbb{E}\left( \mathbf{u}_{0}u_{i,0}\right) \right) \right\Vert ^{\nu
/2}\leq c_{0}\left( NT\right) ^{\nu /4},
\end{equation*}%
for any $w$ such that $\left\Vert w\right\Vert =O\left( N^{1/2}\right) $.
\end{assumption}

\begin{assumption}
\label{fact-idios} It holds that $\left\{ v_{t},1\leq t\leq T\right\} $ and $%
\left\{ u_{i,t},1\leq t\leq T\right\} $ are two mutually independent groups,
for $1\leq i\leq N$.
\end{assumption}

\bigskip 

It holds that

\begin{theorem}
\label{ff-unobs}We assume that the assumptions of Theorem \ref{asy-max} are
satisfied, and that Assumptions \ref{loadings}-\ref{fact-idios} also hold. Then, the same result as in
Theorem \ref{asy-max} holds.
\end{theorem}

We conjecture that Theorem \ref{ff-unobs} holds with minor modifications to
Assumptions \ref{loadings}-\ref{fact-idios} when one estimates factors and
loadings with the Risk-Premium PCA (RP-PCA) approach of \cite%
{lettau2020estimating, lettau2020factors}. In fact, just like our Theorem %
\ref{ff-unobs}, their theory for the strong factors case relies on
conditions that are extremely similar to those of \cite{bai03}. Similar
considerations also hold for the Projected PCA approach of \cite%
{fan2016projected}, which is increasingly being used in asset pricing
studies (see \citealp{kim2021arbitrage} and \citealp{hong2025dynamic}, among
others).

\subsection{On multiple testing and the derandomized confidence function\label{mht}}

We consider the use of the derandomized rule discussed in Section \ref{derandom} in
the main paper in the presence of multiple testing. Whilst this extension is
specifically designed to address the presence of multiple tests carried out
across multiple windows as in Section \ref{empirical}, it can be extended beyond this
specific case.

\bigskip 

Consider the case where the test is repeated across $1\leq v\leq W$ windows,
whose data may overlap fully, partly, or not at all, for the null hypotheses
that%
\[
\mathbb{H}_{0}^{\left( v\right) }:\max_{1\leq i\leq N}\left\vert \alpha
_{i}^{\left( v\right) }\right\vert =0,
\]%
where $\alpha _{i}^{\left( v\right) }$ is the intercept estimated for unit $i
$ with the dataset pertaining to window $v$. A natural question is how
frequently - in the case whereby $\mathbb{H}_{0}^{\left( v\right) }$ is
satisfied across all $1\leq v\leq W$ - can one null $\mathbb{H}_{0}^{\left(
v\right) }$ be rejected. This question corresponds to the idea of
family-wise size control. 

Suppose that, at each window $v$, the test is carried out using the
randomized confidence function defined in equation (\ref{q}), which we denote with
the short-hand notation $Q_{v}\left( \tau \right) $. At each $v$, the
randomness added to the statistic $\psi _{i,NT}^{\left( v\right) }$, say $%
\omega _{i}^{\left( b\right) ,\left( v\right) }$, is independent across $%
1\leq b\leq B$ and across $v$. Then, conditional on the whole sample, $%
Q_{v}\left( \tau \right) $ is independent across $v$. Our question can now
be formalised by calculating%
\[
\mathbb{P}^{\ast }\left[ \text{reject any null hypothesis}|\left\{ \mathbb{H}%
_{0}^{\left( v\right) }\text{ is true, }1\leq v\leq W\right\} \right] .
\]%
The following theorem stipulates that, in essence, family-wise rejection
probability is controlled, as long as there are not too many windows.

\begin{theorem}
\label{family}We assume that Assumptions \ref{error}-\ref{asymptotics} are satisfied. Then, if 
\begin{equation}
\log W-\tau ^{-1}B\left\vert f\left( B\right) \right\vert ^{2}-\tau
^{-3/2}B\left\vert f\left( B\right) \right\vert ^{3}\rightarrow -\infty ,
\label{suff-family}
\end{equation}%
it holds that, as $\min \left\{ N,T,B\right\} \rightarrow \infty $ with $%
B=O\left( \left( \log N\right) ^{2}\right) $%
\[
\mathbb{P}^{\ast }\left[ \text{reject any null hypothesis}|\left\{ \mathbb{H}%
_{0}^{\left( v\right) }\text{ is true, }1\leq v\leq W\right\} \right] =0,
\]%
a.s. conditionally on the sample. 
\end{theorem}

The theorem states that family-wise size control is guaranteed, with the
probability of rejecting any null tending to zero. This requires the
condition in (\ref{suff-family}), which is stated generally, and in which,
for example, one could choose a sample size-adjusted nominal level $\tau $,
akin to a Bonferroni condition. On the other hand, keeping $\tau $ fixed, as
we do in Section \ref{empirical}, and using, as suggested in Section \ref{practice}, $f\left(
B\right) =B^{-1/4}$, (\ref{suff-family}) would be satisfied as long as%
\[
\log W-B^{1/2}\rightarrow -\infty ;
\]%
seeing as we recommend using $B=\left\lfloor \left( \log N\right)
^{2}\right\rfloor $, the above would be equivalent to 
\[
W=o\left( N\right) .
\]

\subsection{Extension of the asymptotic regime in Assumption \protect\ref{asymptotics}\label{complements}}

Here and henceforth, the Euclidean norm of a vector is denoted as $%
\left\Vert \mathbf{\cdot }\right\Vert $. Given an $m\times n$ matrix $%
\mathbf{A}$ with element $a_{ij}$\ we use the following notation for its
norms: $\left\Vert \mathbf{A}\right\Vert $ is the Euclidean/spectral norm,
defined as $\left\Vert \mathbf{A}\right\Vert \leq \sqrt{\lambda _{\max
}\left( \mathbf{A}^{\prime }\mathbf{A}\right) }$; $\left\Vert \mathbf{A}%
\right\Vert_{F}$ is the Frobenious norm; $\left\Vert \mathbf{A}%
\right\Vert_{1}$ is the $\mathcal{L}_{1}$-norm defined as $\left\Vert 
\mathbf{A}\right\Vert _{1}=\max_{1\leq j\leq n}\sum_{i=1}^{m}\left\vert
a_{ij}\right\vert $; the $\mathcal{L}_{\infty }$-norm$\ \left\Vert \mathbf{A}%
\right\Vert _{\infty \text{ }}$is defined as $\left\Vert \mathbf{A}%
\right\Vert _{\infty \text{ }}=\max_{1\leq i\leq m}\sum_{j=1}^{n}\left\vert
a_{ij}\right\vert $.

\bigskip 

We discuss how to extend the asymptotic regime required in Assumption \ref%
{asymptotics}. In particular, we show how we can further relax the
assumption to $N=O\left( T^{\nu /4-\varepsilon }\right) $, which allows for
larger values of $N$ compared to Assumption \ref{asymptotics}, and to the
corresponding Assumption A1(iii) in \citet{feng2022high}. In such a case, we
would need to redefine $\psi _{i,NT}$ as 
\begin{equation}
\psi _{i,NT}=\left\vert \frac{T^{\delta }\widehat{\alpha }_{i,T}}{\widehat{s}%
_{NT}}\right\vert ^{\nu /2},  \label{psi-delta}
\end{equation}%
where $\delta $ is a user-chosen quantity such that 
\begin{equation}
_{{}}0<\delta <\frac{1}{2}-\frac{2}{\nu }\frac{\log N}{\log T}.
\label{delta}
\end{equation}%
Equation (\ref{delta}) does not suggest a decision rule \textit{per se}, but
only an upper bound for $\delta $, which is a tuning parameter. The
rationale underpinning (\ref{psi-delta}) is based on the fact that - upon
inspecting the proofs of Theorem \ref{asy-max} and Lemma \ref{psi} - we
require that, under the null, $\sum_{i=1}^{N}\psi _{i,NT}=o_{a.s.}\left(
1\right) $. In turn, this follows as long as $N\left\vert T^{\delta }%
\widehat{\alpha }_{i,T}\right\vert ^{\nu /2}$ drifts to zero; intuitively,
under the null $\widehat{\alpha }_{i,T}$ drifts to zero at a rate $%
O_{a.s.}\left( T^{-1/2}\right) $, and therefore $N\left\vert T^{\delta }%
\widehat{\alpha }_{i,T}\right\vert ^{\nu /2}=O_{a.s.}\left( NT^{\left(
\delta -1/2\right) \nu /2}\right) =o_{a.s.}\left( 1\right) $ by the
definition of $\delta $\ in (\ref{delta}).

The same arguments hold in the case of nontradable and latent factors, upon
replacing $T^{1/\nu }$ with $T^{\delta }$ in the definition of $\psi
_{i,NT}^{FM}$, and $C_{NT}^{1/\nu }$ with $C_{NT}^{\delta }$\ in the
definition of $\psi _{i,NT}^{PC}$, respectively.

\newpage

\newpage 

\renewcommand*{\thesection}{\Alph{section}}

\setcounter{equation}{0} \setcounter{lemma}{0} \setcounter{theorem}{0} %
\renewcommand{\theassumption}{C.\arabic{assumption}} 
\renewcommand{\thetheorem}{C.\arabic{theorem}} \renewcommand{\thelemma}{C.%
\arabic{lemma}} \renewcommand{\theproposition}{C.\arabic{proposition}} %
\renewcommand{\thecorollary}{C.\arabic{corollary}} \renewcommand{%
\theequation}{C.\arabic{equation}} \renewcommand{\theremark}{C.%
\arabic{remark}}

\section{Technical lemmas\label{lemmas}}

Henceforth, we denote the distribution function of the standard normal
calculated at $-\infty <x<\infty $ as $\Phi \left( x\right) $.

\subsection{Preliminary lemmas}

We begin with a Baum-Katz-type theorem which is also reported in %
\citet{massacci2022}.

\begin{lemma}
\label{stout}Consider a multi-index partial sum process $U_{S_{1},..,S_{h}}=%
\sum_{i_{2}=1}^{S_{2}}\cdot \cdot \cdot \sum_{i_{h}=1}^{S_{h}}\xi
_{i_{1},...,i_{h}}$, and assume that, for some $q\geq 1$ 
\begin{equation*}
\mathbb{E}\sum_{i_{1}=1}^{S_{1}}\left\vert U_{S_{1},..,S_{h}}\right\vert
^{q}\leq c_{0}S_{1}\prod\limits_{j=2}^{h}S_{j}^{d_{j}},
\end{equation*}%
where $d_{j}\geq 1$ for all $1\leq j\leq h$. Then it holds that 
\begin{equation*}
\limsup_{\min \left\{ S_{1},...,S_{h}\right\} \rightarrow \infty }\frac{%
\sum_{i_{1}=1}^{S_{1}}\left\vert U_{S_{1},..,S_{h}}\right\vert ^{q}}{%
S_{1}\prod\limits_{j=2}^{h}S_{j}^{d_{j}}\left( \prod\limits_{j=1}^{h}\log
S_{j}\right) ^{2+\epsilon }}=0\text{ a.s.,}
\end{equation*}%
for all $\epsilon >0$.

\begin{proof}
We begin by noting that the function%
\begin{equation*}
g\left( x_{1},....,x_{h}\right) =x_{1}\prod\limits_{j=2}^{h}x_{j}^{d_{j}},
\end{equation*}%
is superadditive. Consider the vector $\left( y_{1},....,y_{h}\right) $ such
that $y_{i}\geq x_{i}$ for all $1\leq i\leq h$. Then, for any two $s$ and $t$
such that $x_{1}+s\leq y_{1}+t$%
\begin{eqnarray*}
\frac{1}{s}\left[ g\left( x_{1}+s,....,x_{h}\right) -g\left(
x_{1},....,x_{h}\right) \right] &=&\prod\limits_{j=2}^{h}x_{j}^{d_{j}}, \\
\frac{1}{t}\left[ g\left( y_{1}+t,....,y_{h}\right) -g\left(
y_{1},....,y_{h}\right) \right] &=&\prod\limits_{j=2}^{h}x_{j}^{d_{j}},
\end{eqnarray*}%
whence it trivially follows that%
\begin{equation*}
\frac{1}{s}\left[ g\left( x_{1}+s,....,x_{h}\right) -g\left(
x_{1},....,x_{h}\right) \right] =\frac{1}{t}\left[ g\left(
y_{1}+t,....,y_{h}\right) -g\left( y_{1},....,y_{h}\right) \right] .
\end{equation*}%
\citet{he2023one} also showed that, for any two nonzero $s$ and $t$ such
that $x_{i}+s\leq y_{i}+t$, $2\leq i\leq h$%
\begin{equation*}
\frac{1}{s}\left[ g\left( x_{1}+s,....,x_{h}\right) -g\left(
x_{1},....,x_{h}\right) \right] \leq \frac{1}{t}\left[ g\left(
y_{1}+t,....,y_{h}\right) -g\left( y_{1},....,y_{h}\right) \right] .
\end{equation*}%
Thus, $g\left( x_{1},....,x_{h}\right) $ is an S-convex function (see
Definition 2.1 and Proposition 2.3 in \citealp{potra}), and therefore it is
superadditive (by Proposition 2.9 in \citealp{potra}). Hence we can apply
the maximal inequality in Corollary 4 in \citet{moricz1983} with - (in his
notation) $f\left( R\right) =S_{1}\prod\limits_{j=2}^{h}S_{j}^{d_{j}}$ and $%
\phi \left( \cdot \right) =c_{0}$. Letting%
\begin{equation*}
V_{i_{1},...,i_{h}}=\sum_{j_{1}=1}^{i_{1}}\left\vert
\sum_{j_{2}=1}^{i_{2}}\cdot \cdot \cdot \sum_{j_{h}=1}^{i_{h}}\xi
_{j_{1},...,j_{h}}\right\vert ^{q},
\end{equation*}%
it follows that%
\begin{equation*}
\mathbb{E}\max_{1\leq i_{1}\leq S_{1},....,1\leq i_{h}\leq
S_{h}}V_{i_{1},...,i_{h}}\leq
c_{0}S_{1}\prod\limits_{j=2}^{h}S_{j}^{d_{j}}\left(
\prod\limits_{j=1}^{h}\log S_{j}\right) .
\end{equation*}%
Hence we have%
\begin{eqnarray*}
&&\sum_{S_{1}=1}^{\infty }\cdot \cdot \cdot \sum_{S_{h}=1}^{\infty }\frac{1}{%
\prod\limits_{j=1}^{h}S_{j}}\mathbb{P}\left( \max_{1\leq i_{1}\leq
S_{1},....,1\leq i_{h}\leq S_{h}}V_{i_{1},...,i_{h}}\geq \varepsilon
S_{1}\prod\limits_{j=2}^{h}S_{j}^{d_{j}}\left( \prod\limits_{j=1}^{h}\log
S_{j}\right) ^{2+\epsilon }\right) \\
&\leq &\varepsilon ^{-1}\sum_{S_{1}=1}^{\infty }\cdot \cdot \cdot
\sum_{S_{h}=1}^{\infty }\frac{1}{S_{1}^{2}\prod%
\limits_{j=2}^{h}S_{j}^{d_{j}+1}\left( \prod\limits_{j=1}^{h}\log
S_{j}\right) ^{2+\epsilon }}\mathbb{E}\max_{1\leq i_{1}\leq S_{1},....,1\leq
i_{h}\leq S_{h}}V_{i_{1},...,i_{h}} \\
&\leq &c_{0}\varepsilon ^{-1}\sum_{S_{1}=1}^{\infty }\cdot \cdot \cdot
\sum_{S_{h}=1}^{\infty }\frac{1}{\prod\limits_{j=1}^{h}S_{j}\left(
\prod\limits_{j=1}^{h}\log S_{j}\right) ^{1+\epsilon }}\leq c_{1}\varepsilon
^{-1}.
\end{eqnarray*}%
The desired result now follows by repeating the proof of Lemma A.1 in %
\citet{BT2}.
\end{proof}
\end{lemma}

\medskip

The following estimate on the growth rate of moments of partial sums will be
used throughout the paper, and it can be contrasted with Proposition 4.1 in %
\citet{berkeshormann}.

\begin{lemma}
\label{summability}Let $w_{t}$ be a centered, $\mathcal{L}_{q}$-decomposable
Bernoulli shift with $q\geq 2$. Then, if $a>1$, it holds that%
\begin{equation}
E\left( \sum_{t=1}^{m}w_{t}\right) ^{2}\leq c_{0}m,  \label{sum41}
\end{equation}%
Further, for all $2<p\leq q$, if $a>\left( q-1\right) /\left( q-2\right) $,
it holds that 
\begin{equation}
E\left( \sum_{t=1}^{m}w_{t}\right) ^{p}\leq c_{0}m^{p/2}.  \label{sum4}
\end{equation}

\begin{proof}
We begin by showing (\ref{sum41})%
\begin{equation}
E\left( \sum_{t=1}^{m}w_{t}\right) ^{2}\leq c_{0}m.  \label{sumsq}
\end{equation}%
By stationarity, we can write%
\begin{eqnarray*}
E\left( \sum_{t=1}^{m}w_{t}\right) ^{2} &=&E\left(
\sum_{t=1}^{m}\sum_{s=1}^{m}w_{t}w_{s}\right) =mE\left( w_{0}^{2}\right)
+2\sum_{t=1}^{m}\left( m-t\right) E\left( w_{t}w_{0}\right)  \\
&\leq &mE\left( w_{0}^{2}\right) +2\sum_{t=1}^{m}\left\vert E\left(
w_{t}w_{0}\right) \right\vert .
\end{eqnarray*}%
Consider now the coupling $\widetilde{w}_{t,t}$, and note that%
\begin{equation*}
E\left( w_{t}w_{0}\right) =E\left( \left( w_{t}-\widetilde{w}_{t,t}\right)
w_{0}\right) +E\left( \widetilde{w}_{t,t}w_{0}\right) =E\left( \left( w_{t}-%
\widetilde{w}_{t,t}\right) w_{0}\right) ,
\end{equation*}%
on account of the independence between $\widetilde{w}_{t,t}$ and $w_{0}$.
Further%
\begin{equation*}
\left\vert E\left( \left( w_{t}-\widetilde{w}_{t,t}\right) w_{0}\right)
\right\vert \leq \left\vert w_{0}\right\vert _{2}\left\vert w_{t}-\widetilde{%
w}_{t,t}\right\vert _{2}\leq c_{0}t^{-a},
\end{equation*}%
and therefore 
\begin{equation*}
\sum_{t=1}^{m}\left\vert E\left( w_{t}w_{0}\right) \right\vert =O\left(
m\right) .
\end{equation*}%
The desired result now follows by putting everything together. We now show
that%
\begin{equation*}
E\left( \sum_{t=1}^{m}w_{t}\right) ^{q}\leq c_{0}m^{q/2}.
\end{equation*}
Define $\widetilde{w}_{t,\ell }$ with $\ell =\left\lfloor m^{\varsigma
}\right\rfloor $, where 
\begin{equation}
\frac{1}{2a}<\varsigma <\frac{q-2}{2\left( q-1\right) }.  \label{funnyletter}
\end{equation}%
It holds that%
\begin{eqnarray*}
E\left( \sum_{t=1}^{m}w_{t}\right) ^{q} &\leq &2^{q-1}\left( E\left(
\sum_{t=1}^{m}\widetilde{w}_{t,\ell }\right) ^{q}+E\left(
\sum_{t=1}^{m}\left( w_{t}-\widetilde{w}_{t,\ell }\right) \right)
^{q}\right)  \\
&\leq &2^{q-1}\left( E\left( \sum_{t=1}^{m}\widetilde{w}_{t,\ell }\right)
^{q}+E\left( \sum_{t=1}^{m}\left\vert w_{t}-\widetilde{w}_{t,\ell
}\right\vert \right) ^{q}\right) .
\end{eqnarray*}%
We have%
\begin{equation*}
E\left( \sum_{t=1}^{m}\left\vert w_{t}-\widetilde{w}_{t,\varsigma
}\right\vert \right) ^{q}\leq m^{q-1}\sum_{t=1}^{m}E\left\vert w_{t}-%
\widetilde{w}_{t,\varsigma }\right\vert ^{q}\leq c_{0}m^{q-1}m\ell
^{-qa}\leq c_{1}m^{q/2},
\end{equation*}%
on account of (\ref{funnyletter}). We now estimate $E\left( \sum_{t=1}^{m}%
\widetilde{w}_{t,\ell }\right) ^{q}$; consider the $\left\lfloor m/\ell
\right\rfloor +1$ blocks%
\begin{equation*}
\mathcal{B}_{i}=\sum_{t=\ell \left( i-1\right) +1}^{\ell i}\widetilde{w}%
_{t,\ell }\text{, }1\leq i\leq \left\lfloor m/\ell \right\rfloor \text{ \ \
and \ \ }\mathcal{B}_{\left\lfloor m/\ell \right\rfloor
+1}=\sum_{t=\left\lfloor m/\ell \right\rfloor +1}^{m}\widetilde{w}_{t,\ell }.
\end{equation*}%
Note that, by construction, the sequence of blocks $\mathcal{B}_{i}$ with $i$
even is an independent sequence, and so is the sequence of the $\mathcal{B}%
_{i}$s with odd $i$. Hence we can write%
\begin{equation*}
\sum_{t=1}^{m}\widetilde{w}_{t,\ell }=\sum_{i=1}^{\left\lfloor m/\ell
\right\rfloor /2}\mathcal{B}_{2i}+\sum_{i=1}^{\left\lfloor m/\ell
\right\rfloor /2}\mathcal{B}_{2\left( i-1\right) +1}+\mathcal{B}%
_{\left\lfloor m/\ell \right\rfloor +1}.
\end{equation*}%
Thus%
\begin{eqnarray*}
E\left( \sum_{t=1}^{m}w_{t}\right) ^{q} &\leq &3^{q-1}\left( E\left(
\sum_{i=1}^{\left\lfloor m/\ell \right\rfloor /2}\mathcal{B}_{2i}\right)
^{q}+E\left( \sum_{i=1}^{\left\lfloor m/\ell \right\rfloor /2}\mathcal{B}%
_{2\left( i-1\right) +1}\right) ^{q}+E\left( \mathcal{B}_{\left\lfloor
m/\ell \right\rfloor +1}\right) ^{q}\right)  \\
&\leq &3^{p-1}\left( E\left\vert \sum_{i=1}^{\left\lfloor m/\ell
\right\rfloor /2}\mathcal{B}_{2i}\right\vert ^{q}+E\left\vert
\sum_{i=1}^{\left\lfloor m/\ell \right\rfloor /2}\mathcal{B}_{2\left(
i-1\right) +1}\right\vert ^{q}+E\left\vert \mathcal{B}_{\left\lfloor m/\ell
\right\rfloor +1}\right\vert ^{q}\right) 
\end{eqnarray*}%
On account of the independence of the $\mathcal{B}_{2i}$s across $i$, we can
use Rosenthal's inequality (see e.g. Theorem 2.9 in \citealp{petrov1995}),
whence%
\begin{equation}
E\left\vert \sum_{i=1}^{\left\lfloor m/\ell \right\rfloor /2}\mathcal{B}%
_{2i}\right\vert ^{q}\leq c\left( q\right) \left( E\sum_{i=1}^{\left\lfloor
m/\ell \right\rfloor /2}\left\vert \mathcal{B}_{2i}\right\vert
^{q}+\left\vert \sum_{i=1}^{\left\lfloor m/\ell \right\rfloor /2}E\left( 
\mathcal{B}_{2i}^{2}\right) \right\vert ^{q/2}\right) ,  \label{rosenthal}
\end{equation}%
where $c\left( q\right) $ is a positive, finite constant that depends only
on $p$. We already know from (\ref{sumsq}) that%
\begin{equation*}
E\left( \mathcal{B}_{2i}^{2}\right) \leq c_{0}\ell ,
\end{equation*}%
and therefore%
\begin{equation*}
\left\vert \sum_{i=1}^{\left\lfloor m/\ell \right\rfloor /2}E\left( \mathcal{%
B}_{2i}^{2}\right) \right\vert ^{q/2}\leq c_{0}m^{q/2},
\end{equation*}%
for some $c_{0}$. Further%
\begin{eqnarray*}
&&E\sum_{i=1}^{\left\lfloor m/\ell \right\rfloor /2}\left\vert \mathcal{B}%
_{2i}\right\vert ^{q} \\
&=&E\sum_{i=1}^{\left\lfloor m/\ell \right\rfloor /2}\left\vert \sum_{t=\ell
\left( 2i-1\right) +1}^{2\ell i}\widetilde{w}_{t,\ell }\right\vert ^{q}\leq
c_{0}\left\lfloor \frac{m}{\ell }\right\rfloor \ell ^{q-1}\sum_{t=\ell
\left( 2i-1\right) +1}^{2\ell i}E\left\vert \widetilde{w}_{t,\ell
}\right\vert ^{q} \\
&\leq &c_{1}\frac{m}{\ell }\ell ^{q}\leq c_{2}m^{\varsigma \left( q-1\right)
+1}\leq c_{3}m^{q/2},
\end{eqnarray*}%
by the definition of $\varsigma $\ in (\ref{funnyletter}). Putting all
together, (\ref{rosenthal}) now yields%
\begin{equation*}
E\left\vert \sum_{i=1}^{\left\lfloor m/\ell \right\rfloor /2}\mathcal{B}%
_{2i}\right\vert ^{q}\leq c_{0}m^{q/2},
\end{equation*}%
and the same holds for the odd blocks $\mathcal{B}_{2\left( i-1\right) +1}$,
and, similarly, for $E\left\vert \mathcal{B}_{\left\lfloor m/\ell
\right\rfloor +1}\right\vert ^{q}$. Equation (\ref{sum4}) now follows from
Lyapunov's inequality. 
\end{proof}
\end{lemma}

\subsection{Lemmas for Section \protect\ref{tests}}

\begin{lemma}
\label{average-ft}We assume that Assumption \ref{regressor} is satisfied.
Then it holds that%
\begin{equation*}
\overline{f}=\frac{1}{T}\sum_{t=1}^{T}f_{t}=\mathbb{E}f_{t}+o_{a.s.}\left(
1\right) .
\end{equation*}

\begin{proof}
We report the proof for the case $K=1$, for simplicity and without loss of
generality. The proof follows from standard arguments; indeed%
\begin{equation*}
\frac{1}{T}\sum_{t=1}^{T}f_{t}=\mathbb{E}f_{t}+\frac{1}{T}%
\sum_{t=1}^{T}\left( f_{t}-\mathbb{E}f_{t}\right) .
\end{equation*}%
Recall that, by Assumption \ref{regressor}, $f_{t}-\mathbb{E}f_{t}$ is a
centered, $\mathcal{L}_{\nu }$-decomposable Bernoulli shift; thus, by Lemma %
\ref{summability}%
\begin{equation*}
\mathbb{E}\left\vert \sum_{t=1}^{T}\left( f_{t}-\mathbb{E}f_{t}\right)
\right\vert ^{p}\leq c_{0}T^{p/2},
\end{equation*}%
for all $2\leq p\leq \nu $, whence Lemma \ref{stout} readily entails that 
\begin{equation*}
\frac{1}{T}\sum_{t=1}^{T}\left( f_{t}-\mathbb{E}f_{t}\right) =o_{a.s.}\left(
1\right) .
\end{equation*}
\end{proof}
\end{lemma}

\begin{lemma}
\label{average-error}We assume that Assumption \ref{error} is satisfied.
Then it holds that%
\begin{equation*}
\sum_{i=1}^{N}\left\vert \sum_{t=1}^{T}u_{i,t}\right\vert ^{\gamma
}=o_{a.s.}\left( NT^{\gamma /2}\left( \log N\log T\right) ^{2+\epsilon
}\right) ,
\end{equation*}%
for all $\epsilon >0$ and all $2\leq \gamma \leq \nu $.

\begin{proof}
We estimate convergence rate of 
\begin{equation*}
\sum_{i=1}^{N}\mathbb{E}\left\vert \sum_{t=1}^{T}u_{i,t}\right\vert ^{\gamma
}.
\end{equation*}%
By Assumption \ref{error}, we can use Lemma \ref{summability}, which entails
that, for all $1\leq i\leq N$%
\begin{equation*}
\mathbb{E}\left\vert \sum_{t=1}^{T}u_{i,t}\right\vert ^{\gamma }\leq c_{\nu
}T^{\gamma /2},
\end{equation*}%
where $c_{\nu }$ is a positive, finite constant which depends only on $\nu $%
, whence 
\begin{equation*}
\sum_{i=1}^{N}\mathbb{E}\left\vert \sum_{t=1}^{T}u_{i,t}\right\vert ^{\gamma
}\leq c_{0}NT^{\gamma /2}.
\end{equation*}%
The desired result now readily obtains from Lemma \ref{stout}.
\end{proof}
\end{lemma}

\begin{lemma}
\label{denominator-LS}We assume that Assumption \ref{regressor} is
satisfied. Then it holds that%
\begin{equation*}
\frac{1}{T}\sum_{t=1}^{T}\left( f_{t}-\overline{f}\right) \left( f_{t}-%
\overline{f}\right) ^{\prime }=\mathcal{V}\left( f\right) +o_{a.s.}\left(
1\right) .
\end{equation*}

\begin{proof}
As above, we report the proof for the case $K=1$, for simplicity and without
loss of generality. It holds that%
\begin{equation*}
\frac{1}{T}\sum_{t=1}^{T}\left( f_{t}-\overline{f}\right) ^{2}=\frac{1}{T}%
\sum_{t=1}^{T}f_{t}^{2}-\overline{f}^{2}.
\end{equation*}%
Consider $f_{t}^{2}$; Assumption \ref{regressor}\textit{(i)} immediately
entails that $\left\{ f_{t}^{2},-\infty <t<\infty \right\} $ is an $\mathcal{%
L}_{\nu /2}$-decomposable Bernoulli shift with rate $a>\left( \nu -1\right)
/\left( \nu -2\right) $. Indeed, letting 
\begin{equation*}
f_{t}=g^{\left( f\right) }\left( \eta _{t}^{\left( f\right) },\eta
_{t-1}^{\left( f\right) },...\right) ,
\end{equation*}%
where $g^{\left( f\right) }:S^{\infty }\rightarrow 
\mathbb{R}
^{K}$ is a non random measurable function and $\left\{ \eta _{t}^{\left(
f\right) },-\infty <t<\infty \right\} $ is an \textit{i.i.d.} sequence with
values in a measurable space $S$, and consider the coupling construction 
\begin{equation*}
f_{t}^{\prime }=g^{\left( f\right) }\left( \eta _{t}^{\left( f\right)
},...,\eta _{t-\ell +1}^{\left( f\right) },\eta _{t-\ell ,t,\ell }^{\ast
\left( f\right) },\eta _{t-\ell -1,t,\ell }^{\ast \left( f\right)
}...\right) ,
\end{equation*}%
with $\left\{ \eta _{s,t,\ell }^{\ast \left( f\right) },-\infty <s,\ell
,t<\infty \right\} $ \textit{i.i.d.} copies of $\eta _{0}^{\left( f\right) }$
independent of $\left\{ \eta _{t}^{\left( f\right) },-\infty <t<\infty
\right\} $. Then we have%
\begin{eqnarray*}
&&\left\vert f_{t}^{2}-\left( f_{t}^{\prime }\right) ^{2}\right\vert _{\nu
/2} \\
&=&\left\vert \left( f_{t}+f_{t}^{\prime }\right) \left( f_{t}-f_{t}^{\prime
}\right) \right\vert _{\nu /2}\leq \left\vert f_{t}+f_{t}^{\prime
}\right\vert _{\nu }\left\vert f_{t}-f_{t}^{\prime }\right\vert _{\nu } \\
&\leq &2\left\vert f_{t}\right\vert _{\nu }\left\vert f_{t}-f_{t}^{\prime
}\right\vert _{\nu }\leq c_{0}\ell ^{-a},
\end{eqnarray*}%
having used the Cauchy-Schwartz inequality, Minkowski's inequality, and the
facts that - by Assumption \ref{regressor} - $\left\vert f_{t}\right\vert
_{\nu }=\left\vert f_{t}^{\prime }\right\vert _{\nu }<\infty $ and $%
\left\vert f_{t}-f_{t}^{\prime }\right\vert _{\nu }\leq c_{0}\ell ^{-a}$,
with $a>\left( \nu -1\right) /\left( \nu -2\right) $. Hence%
\begin{equation*}
T^{-\nu /2}\mathbb{E}\left\vert \sum_{t=1}^{T}\left( f_{t}^{2}-\mathbb{E}%
f_{t}^{2}\right) \right\vert ^{\nu /2}\leq c_{\nu /2}T^{-\nu /4},
\end{equation*}%
by Lemma \ref{summability}, from which it follows from standard arguments
that%
\begin{equation*}
\frac{1}{T}\sum_{t=1}^{T}f_{t}^{2}=\mathbb{E}f_{t}^{2}+o_{a.s.}\left(
1\right) .
\end{equation*}%
By the same token, it is not hard to see that%
\begin{equation*}
\overline{f}=\mathbb{E}\left( f_{t}\right) +o_{a.s.}\left( 1\right) .
\end{equation*}%
Thus we have%
\begin{equation*}
\frac{1}{T}\sum_{t=1}^{T}\left( f_{t}-\overline{f}\right) ^{2}=\mathbb{E}%
f_{t}^{2}-\left( \mathbb{E}f_{t}\right) ^{2}+o_{a.s.}\left( 1\right) ,
\end{equation*}%
and the desired result obtains from Assumption \ref{regressor}\textit{(ii)}.
\end{proof}
\end{lemma}

\begin{lemma}
\label{cross-prod}We assume that Assumptions \ref{error}-\ref{exogeneity}
are satisfied. Then it holds that%
\begin{equation*}
\sum_{i=1}^{N}\left\Vert \sum_{t=1}^{T}f_{t}u_{i,t}\right\Vert ^{\gamma
}=o_{a.s.}\left( NT^{\gamma /2}\left( \log N\log T\right) ^{2+\epsilon
}\right) ,
\end{equation*}%
for every $\epsilon >0$ and $2\leq \gamma \leq \nu /2$.

\begin{proof}
Let $K=1$ with no loss of generality. We begin by showing that $\left\{
f_{t}u_{i,t},-\infty <t<\infty \right\} $ is an $\mathcal{L}_{\nu /2}$%
-decomposable Bernoulli shift with rate $a>\left( \nu -1\right) /\left( \nu
-2\right) $. Recall that, by Assumption \ref{exogeneity}, $E\left(
f_{t}u_{i,t}\right) =0$, and 
\begin{eqnarray*}
f_{t} &=&g^{\left( f\right) }\left( \eta _{t}^{\left( f\right) },\eta
_{t-1}^{\left( f\right) },...\right) , \\
u_{i,t} &=&g^{\left( u_{i}\right) }\left( \eta _{t}^{\left( i\right) },\eta
_{t-1}^{\left( i\right) },...\right) ,
\end{eqnarray*}%
where $g^{\left( f\right) }:S^{\infty }\rightarrow 
\mathbb{R}
^{K}$ and $g^{\left( u_{i}\right) }:S^{\infty }\rightarrow 
\mathbb{R}
$, $1\leq i\leq N$, are non random measurable function and $\left\{ \eta
_{t}^{\left( f\right) },-\infty <t<\infty \right\} $ and $\left\{ \eta
_{t}^{\left( i\right) },-\infty <t<\infty \right\} $\ are \textit{i.i.d.}
sequences with values in a measurable space $S$, and consider the coupling
constructions%
\begin{eqnarray*}
f_{t}^{\prime } &=&g^{\left( f\right) }\left( \eta _{t}^{\left( f\right)
},...,\eta _{t-\ell +1}^{\left( f\right) },\eta _{t-\ell ,t,\ell }^{\ast
\left( f\right) },\eta _{t-\ell -1,t,\ell }^{\ast \left( f\right)
}...\right) , \\
u_{i,t}^{\prime } &=&g^{\left( u_{i}\right) }\left( \eta _{t}^{\left(
i\right) },...,\eta _{t-\ell +1}^{\left( i\right) },\eta _{t-\ell ,t,\ell
}^{\ast \left( i\right) },\eta _{t-\ell -1,t,\ell }^{\ast \left( i\right)
}...\right) ,
\end{eqnarray*}%
where $\left\{ \eta _{s,t,\ell }^{\ast \left( f\right) },-\infty <s,\ell
,t<\infty \right\} $ and $\left\{ \eta _{s,t,\ell }^{\ast \left( i\right)
},-\infty <s,\ell ,t<\infty \right\} $\ are \textit{i.i.d.} copies of $\eta
_{0}^{\left( f\right) }$ and $\eta _{0}^{\left( i\right) }$\ respectively,
independent of $\left\{ \eta _{t}^{\left( f\right) },-\infty <t<\infty
\right\} $ and $\left\{ \eta _{t}^{\left( i\right) },-\infty <t<\infty
\right\} $. Then we have, by elementary arguments%
\begin{eqnarray*}
&&\left\vert f_{t}u_{i,t}-f_{t}^{\prime }u_{i,t}^{\prime }\right\vert
_{\gamma } \\
&\leq &\left\vert \left( f_{t}-f_{t}^{\prime }\right) u_{i,t}^{\prime
}\right\vert _{\gamma }+\left\vert f_{t}^{\prime }\left(
u_{i,t}-u_{i,t}^{\prime }\right) \right\vert _{\gamma }+\left\vert \left(
f_{t}-f_{t}^{\prime }\right) \left( u_{i,t}-u_{i,t}^{\prime }\right)
\right\vert _{\gamma } \\
&\leq &\left\vert u_{i,t}\right\vert _{2\gamma }\left\vert
f_{t}-f_{t}^{\prime }\right\vert _{2\gamma }+\left\vert f_{t}\right\vert
_{2\gamma }\left\vert u_{i,t}-u_{i,t}^{\prime }\right\vert _{2\gamma
}+\left\vert f_{t}-f_{t}^{\prime }\right\vert _{2\gamma }\left\vert
u_{i,t}-u_{i,t}^{\prime }\right\vert _{2\gamma } \\
&\leq &\left\vert u_{i,t}\right\vert _{\nu }\left\vert f_{t}-f_{t}^{\prime
}\right\vert _{\nu }+\left\vert f_{t}\right\vert _{\nu }\left\vert
u_{i,t}-u_{i,t}^{\prime }\right\vert _{\nu }+\left\vert f_{t}-f_{t}^{\prime
}\right\vert _{\nu }\left\vert u_{i,t}-u_{i,t}^{\prime }\right\vert _{\nu }
\\
&\leq &c_{0}\ell ^{-a}+c_{1}\ell ^{-a}+c_{2}\ell ^{-2a}\leq c_{3}\ell ^{-a}.
\end{eqnarray*}%
Then it holds that%
\begin{equation*}
\sum_{i=1}^{N}\mathbb{E}\left\vert \sum_{t=1}^{T}f_{t}u_{i,t}\right\vert
^{\gamma }\leq c_{0}NT^{\gamma /2},
\end{equation*}%
having used Lemma \ref{summability}. The desired result now follows from
Lemma \ref{stout}.
\end{proof}
\end{lemma}

\begin{lemma}
\label{normalising}We assume that Assumptions \ref{error}-\ref{exogeneity}
are satisfied. Then it holds that%
\begin{eqnarray*}
\liminf_{\min \left\{ N,T\right\} \rightarrow \infty }\frac{1}{NT}%
\sum_{i=1}^{N}\sum_{t=1}^{T}\widehat{u}_{i,t}^{2} &>&0, \\
\limsup_{\min \left\{ N,T\right\} \rightarrow \infty }\frac{1}{NT}%
\sum_{i=1}^{N}\sum_{t=1}^{T}\widehat{u}_{i,t}^{2} &<&\infty .
\end{eqnarray*}

\begin{proof}
The proof uses several arguments used also elsewhere, so we omit passages
when possible to avoid repetitions. It holds that%
\begin{eqnarray*}
&&\frac{1}{NT}\sum_{i=1}^{N}\sum_{t=1}^{T}\widehat{u}_{i,t}^{2} \\
&=&\frac{1}{NT}\sum_{i=1}^{N}\sum_{t=1}^{T}u_{i,t}^{2}+\frac{1}{NT}%
\sum_{i=1}^{N}\sum_{t=1}^{T}\left( \widehat{\alpha }_{i}-\alpha _{i}\right)
^{2}+\frac{1}{NT}\sum_{i=1}^{N}\sum_{t=1}^{T}\left( \widehat{\beta }%
_{i}-\beta _{i}\right) ^{2}f_{t}^{2} \\
&&+\frac{2}{NT}\sum_{i=1}^{N}\sum_{t=1}^{T}\left( \widehat{\alpha }%
_{i}-\alpha _{i}\right) u_{i,t}+\frac{2}{NT}\sum_{i=1}^{N}\sum_{t=1}^{T}%
\left( \widehat{\beta }_{i}-\beta _{i}\right) f_{t}u_{i,t} \\
&&+\frac{2}{NT}\sum_{i=1}^{N}\sum_{t=1}^{T}\left( \widehat{\alpha }%
_{i}-\alpha _{i}\right) \left( \widehat{\beta }_{i}-\beta _{i}\right) f_{t}
\\
&=&I+II+III+IV+V+VI.
\end{eqnarray*}%
It holds that%
\begin{equation*}
I=\frac{1}{NT}\sum_{i=1}^{N}\sum_{t=1}^{T}\mathbb{E}u_{i,t}^{2}+\frac{1}{NT}%
\sum_{i=1}^{N}\sum_{t=1}^{T}\left( u_{i,t}^{2}-\mathbb{E}u_{i,t}^{2}\right)
=I_{a}+I_{b}.
\end{equation*}%
By Assumption \ref{error}, it follows immediately that $0<I_{a}<\infty $;
also, it is easy to see that $u_{i,t}^{2}-\mathbb{E}u_{i,t}^{2}$ is a
centered, $\mathcal{L}_{\nu /2}$-decomposable Bernoulli shift (see the
arguments in the proof of Lemma \ref{denominator-LS}), and therefore, by
Lemma \ref{summability}%
\begin{eqnarray*}
&&\mathbb{E}\left\vert \frac{1}{NT}\sum_{i=1}^{N}\sum_{t=1}^{T}\left(
u_{i,t}^{2}-\mathbb{E}u_{i,t}^{2}\right) \right\vert ^{2} \\
&\leq &\frac{1}{NT^{2}}\sum_{i=1}^{N}\mathbb{E}\left\vert
\sum_{t=1}^{T}\left( u_{i,t}^{2}-\mathbb{E}u_{i,t}^{2}\right) \right\vert
^{2}\leq c_{0}T^{-1},
\end{eqnarray*}%
whence Lemma \ref{stout} yields $I_{b}=o_{a.s.}\left( 1\right) $. Note also
that%
\begin{equation*}
\frac{1}{N}\sum_{i=1}^{N}\left( \widehat{\beta }_{i}-\beta _{i}\right)
^{2}\left( \frac{1}{T}\sum_{t=1}^{T}f_{t}^{2}\right) ,
\end{equation*}%
with $T^{-1}\sum_{t=1}^{T}f_{t}^{2}=O_{a.s.}\left( 1\right) $ by Lemma \ref%
{average-ft} and%
\begin{equation*}
\frac{1}{N}\sum_{i=1}^{N}\left( \widehat{\beta }_{i}-\beta _{i}\right) ^{2}=%
\frac{\frac{1}{N}\sum_{i=1}^{N}\left( \frac{1}{T}\sum_{t=1}^{T}\left( f_{t}-%
\overline{f}\right) u_{i,t}\right) ^{2}}{\left( \frac{1}{T}%
\sum_{t=1}^{T}\left( f_{t}-\overline{f}\right) ^{2}\right) ^{2}}.
\end{equation*}%
We know from Lemma \ref{denominator-LS} that 
\begin{equation*}
\frac{1}{T}\sum_{t=1}^{T}\left( f_{t}-\overline{f}\right)
^{2}=c_{0}+o_{a.s.}\left( 1\right) ,
\end{equation*}%
with $c_{0}>0$. Further, using Lemma \ref{cross-prod}, it follows that%
\begin{equation*}
\frac{1}{N}\sum_{i=1}^{N}\left\Vert \frac{1}{T}\sum_{t=1}^{T}\left( f_{t}-%
\overline{f}\right) u_{i,t}\right\Vert ^{2}=o_{a.s.}\left( 1\right) ,
\end{equation*}%
whence $III=o_{a.s.}\left( 1\right) $. The same arguments as in the proof of
Lemma \ref{psi} entail that $II=o_{a.s.}\left( 1\right) $. Finally, a
routine application of H\"{o}lder's inequality yields that $%
IV-VI=o_{a.s.}\left( 1\right) $.
\end{proof}
\end{lemma}

\begin{lemma}
\label{psi}We assume that Assumptions \ref{error}-\ref{asymptotics} are
satisfied. Then, under the null in \eqref{null} it holds that%
\begin{equation*}
\sum_{i=1}^{N}\psi_{i,NT}=o_{a.s.}\left( 1\right) .
\end{equation*}

\begin{proof}
Let - for simplicity and with no loss of generality - $K=1$. Recall (\ref%
{alpha-hat}), whence also%
\begin{equation*}
\widehat{\alpha }_{i}=\alpha _{i}-\left( \widehat{\beta }_{i}-\beta
_{i}\right) \overline{f}+\overline{u}_{i}=-\left( \widehat{\beta }_{i}-\beta
_{i}\right) \overline{f}+\overline{u}_{i},
\end{equation*}%
under $\mathbb{H}_{0}$. Hence we have 
\begin{eqnarray*}
\sum_{i=1}^{N}\psi _{i,NT} &=&\frac{T^{1/2}}{\left\vert \hat{s}%
_{NT}\right\vert ^{\nu /2}}\sum_{i=1}^{N}\left\vert \widehat{\alpha }%
_{i}\right\vert ^{\nu /2} \\
&\leq &\frac{T^{1/2}}{\left\vert \hat{s}_{NT}\right\vert ^{\nu /2}}%
\sum_{i=1}^{N}\left\vert \overline{u}_{i}\right\vert ^{\nu /2}+\frac{T^{1/2}%
}{\left\vert \hat{s}_{NT}\right\vert ^{\nu /2}}\sum_{i=1}^{N}\left\vert
\left( \widehat{\beta }_{i}-\beta _{i}\right) \overline{f}\right\vert ^{\nu
/2} \\
&=&\frac{T^{1/2}}{\left\vert \hat{s}_{NT}\right\vert ^{\nu /2}}%
\sum_{i=1}^{N}\left\vert \overline{u}_{i}\right\vert ^{\nu /2}+\frac{T^{1/2}%
}{\left\vert \hat{s}_{NT}\right\vert ^{\nu /2}}\sum_{i=1}^{N}\left\vert 
\frac{\sum_{t=1}^{T}\left( f_{t}\overline{f}-\overline{f}^{2}\right) u_{i,t}%
}{\sum_{t=1}^{T}\left( f_{t}-\overline{f}\right) ^{2}}\right\vert ^{\nu /2},
\end{eqnarray*}%
We know from Lemma \ref{normalising} that there exists a positive, finite
constant $c_{0}$ and a couple of random variables $\left( N_{0},T_{0}\right) 
$ such that, for all $N\geq N_{0}$ and $T\geq T_{0}$%
\begin{equation*}
\frac{T^{1/2}}{\left\vert \hat{s}_{NT}\right\vert ^{\nu /2}}%
\sum_{i=1}^{N}\left\vert \overline{u}_{i}\right\vert ^{\nu /2}\leq
c_{0}T^{1/2}\sum_{i=1}^{N}\left\vert \overline{u}_{i}\right\vert ^{\nu /2};
\end{equation*}%
using Lemma \ref{average-error}, it follows that%
\begin{equation*}
\frac{T^{1/2}}{\left\vert \hat{s}_{NT}\right\vert ^{\nu /2}}%
\sum_{i=1}^{N}\left\vert \overline{u}_{i}\right\vert ^{\nu
/2}=o_{a.s.}\left( NT^{1/2}T^{-\nu /4}\left( \log N\log T\right)
^{2+\epsilon }\right) .
\end{equation*}%
Also%
\begin{eqnarray*}
&&\frac{T^{1/2}}{\left\vert \hat{s}_{NT}\right\vert ^{\nu /2}}%
\sum_{i=1}^{N}\left\vert \frac{\sum_{t=1}^{T}\left( f_{t}\overline{f}-%
\overline{f}^{2}\right) u_{i,t}}{\sum_{t=1}^{T}\left( f_{t}-\overline{f}%
\right) ^{2}}\right\vert ^{\nu /2} \\
&\leq &\frac{T^{1/2}}{\left\vert \hat{s}_{NT}\right\vert ^{\nu /2}}\frac{%
\sum_{i=1}^{N}\left\vert \overline{f}\frac{1}{T}\sum_{t=1}^{T}f_{t}u_{i,t}%
\right\vert ^{\nu /2}}{\left\vert \frac{1}{T}\sum_{t=1}^{T}\left( f_{t}-%
\overline{f}\right) ^{2}\right\vert ^{\nu /2}}+\frac{T^{1/2}}{\left\vert 
\hat{s}_{NT}\right\vert ^{\nu /2}}\frac{\sum_{i=1}^{N}\left\vert \overline{f}%
^{2}\frac{1}{T}\sum_{t=1}^{T}u_{i,t}\right\vert ^{\nu /2}}{\left\vert \frac{1%
}{T}\sum_{t=1}^{T}\left( f_{t}-\overline{f}\right) ^{2}\right\vert ^{\nu /2}}%
.
\end{eqnarray*}%
Lemmas \ref{average-ft}, \ref{denominator-LS} and \ref{normalising} entail
that there exists a positive, finite constant $c_{0}$ and a random variable $%
T_{0}$ such that, for all $T\geq T_{0}$%
\begin{eqnarray*}
\frac{T^{1/2}}{\left\vert \hat{s}_{NT}\right\vert ^{\nu /2}}\frac{%
\sum_{i=1}^{N}\left\vert \overline{f}\frac{1}{T}\sum_{t=1}^{T}f_{t}u_{i,t}%
\right\vert ^{\nu /2}}{\left\vert \frac{1}{T}\sum_{t=1}^{T}\left( f_{t}-%
\overline{f}\right) ^{2}\right\vert ^{\nu /2}} &\leq
&c_{0}T^{1/2}\sum_{i=1}^{N}\left\vert \frac{1}{T}\sum_{t=1}^{T}f_{t}u_{i,t}%
\right\vert ^{\nu /2}, \\
\frac{T^{1/2}}{\left\vert \hat{s}_{NT}\right\vert ^{\nu /2}}\frac{%
\sum_{i=1}^{N}\left\vert \overline{f}^{2}\frac{1}{T}\sum_{t=1}^{T}u_{i,t}%
\right\vert ^{\nu /2}}{\left\vert \frac{1}{T}\sum_{t=1}^{T}\left( f_{t}-%
\overline{f}\right) ^{2}\right\vert ^{\nu /2}} &\leq
&c_{0}T^{1/2}\sum_{i=1}^{N}\left\vert \frac{1}{T}\sum_{t=1}^{T}u_{i,t}\right%
\vert ^{\nu /2}.
\end{eqnarray*}%
We already know from the above that the second term is $o_{a.s.}\left(
1\right) $. Using Lemma \ref{cross-prod}, it finally follows that%
\begin{equation*}
\frac{T^{1/2}}{\left\vert \hat{s}_{NT}\right\vert ^{\nu /2}}\frac{%
\sum_{i=1}^{N}\left\vert \overline{f}\frac{1}{T}\sum_{t=1}^{T}f_{t}u_{i,t}%
\right\vert ^{\nu /2}}{\left\vert \frac{1}{T}\sum_{t=1}^{T}\left( f_{t}-%
\overline{f}\right) ^{2}\right\vert ^{\nu /2}}=o_{a.s.}\left(
NT^{1/2}T^{-\nu /4}\left( \log N\log T\right) ^{2+\epsilon }\right) .
\end{equation*}%
By combining the results above, we receive%
\begin{equation}
\sum_{i=1}^{N}\psi _{i,NT}=o_{a.s.}\left( NT^{1/2}T^{-\nu /4}\left( \log
N\log T\right) ^{2+\epsilon }\right) .  \label{psi_rateas}
\end{equation}%
The desired result now obtains by Assumption \ref{asymptotics}.
\end{proof}
\end{lemma}

\subsection{Lemmas for Section \protect\ref{nontradable}\label{lemma-ff}}

We now report a series of lemmas for the case, discussed in Section \ref%
{nontradable}, of nontradable factors. In order for the notation not to be
overly burdensome, we will assume - unless otherwise stated - $K=1$ whenever
possible and with no loss of generality.

\smallskip

Recall the short-hand notation $\mathbf{S}_{\beta }$ defined in (\ref{s-beta}%
), let $\mathbf{\alpha }=\left( \alpha _{1},...,\alpha _{N}\right) ^{\prime
} $ and 
\begin{equation}
\widehat{\mathbf{S}}_{\beta }=\frac{1}{N}\sum_{i=1}^{N}\left( \widehat{\beta 
}_{i}-\frac{1}{N}\sum_{i=1}^{N}\widehat{\beta }_{i}\right) \left( \widehat{%
\beta }_{i}-\frac{1}{N}\sum_{i=1}^{N}\widehat{\beta }_{i}\right) ^{\prime },
\label{s-beta-hat}
\end{equation}%
and%
\begin{equation}
\mathbf{S}_{f}=\frac{1}{T}\sum_{t=1}^{T}\left( f_{t}-\overline{f}\right)
\left( f_{t}-\overline{f}\right) ^{\prime },  \label{s-f}
\end{equation}%
and note that, after standard passages%
\begin{eqnarray}
\widehat{\lambda } &=&\lambda +\frac{1}{N}\mathbf{S}_{\beta }^{-1}\mathbf{%
\beta }^{\prime }\mathbb{M}_{1_{N}}\mathbf{\alpha +}\overline{v}\mathbf{+}%
\frac{1}{N}\left( \widehat{\mathbf{S}}_{\beta }^{-1}-\mathbf{S}_{\beta
}^{-1}\right) \mathbf{\beta }^{\prime }\mathbb{M}_{1_{N}}\mathbf{\alpha }
\label{lambda-ff-expansion} \\
&&+\frac{1}{N}\widehat{\mathbf{S}}_{\beta }^{-1}\left( \widehat{\mathbf{%
\beta }}\mathbf{-\beta }\right) ^{\prime }\mathbb{M}_{1_{N}}\mathbf{\alpha +}%
\frac{1}{N}\widehat{\mathbf{S}}_{\beta }^{-1}\widehat{\mathbf{\beta }}%
^{\prime }\mathbb{M}_{1_{N}}\left( \widehat{\mathbf{\beta }}\mathbf{-\beta }%
\right) \lambda  \notag \\
&&+\frac{1}{N}\widehat{\mathbf{S}}_{\beta }^{-1}\widehat{\mathbf{\beta }}%
^{\prime }\mathbb{M}_{1_{N}}\left( \mathbf{\beta -}\widehat{\mathbf{\beta }}%
\right) \overline{v}+\frac{1}{N}\widehat{\mathbf{S}}_{\beta }^{-1}\mathbf{%
\beta }^{\prime }\mathbb{M}_{1_{N}}\overline{\mathbf{u}}  \notag \\
&&+\frac{1}{N}\widehat{\mathbf{S}}_{\beta }^{-1}\left( \widehat{\mathbf{%
\beta }}-\mathbf{\beta }\right) ^{\prime }\mathbb{M}_{1_{N}}\overline{%
\mathbf{u}}.  \notag
\end{eqnarray}

Finally, under both $\mathbb{H}_{0}$ and $\mathbb{H}_{A}$ we have%
\begin{equation}
\widehat{\alpha }_{i}^{FM}=\alpha _{i}+\beta _{i}^{\prime }\overline{v}+%
\overline{u}_{i}-\left( \widehat{\beta }_{i}-\beta _{i}\right) ^{\prime
}\lambda -\beta _{i}^{\prime }\left( \widehat{\lambda }-\lambda \right)
-\left( \widehat{\beta }_{i}-\beta _{i}\right) ^{\prime }\left( \widehat{%
\lambda }-\lambda \right) .  \label{alpha-ff}
\end{equation}

\smallskip

\begin{lemma}
\label{beta-hat-ff}We assume that Assumptions \ref{error}-\ref{exogeneity}
and \ref{cov_beta} and \ref{weak_CS_dep} are satisfied. Then it holds that%
\begin{eqnarray}
N^{-1}\left\Vert \widehat{\mathbf{\beta }}-\mathbf{\beta }\right\Vert ^{2}
&=&o_{a.s.}\left( \frac{\left( \log N\log T\right) ^{2+\epsilon }}{T}\right)
,  \label{beta-hat-ff-norm} \\
\left\Vert \widehat{\mathbf{S}}_{\beta }-\mathbf{S}_{\beta }\right\Vert
&=&o_{a.s.}\left( \frac{\left( \log N\log T\right) ^{2+\epsilon }}{T}\right)
+o_{a.s.}\left( \frac{\left( \log N\log T\right) ^{1+\epsilon }}{\sqrt{NT}}%
\right) ,  \label{s-beta-hat-error}
\end{eqnarray}%
for all $\epsilon >0$.

\begin{proof}
Recall that we use $K=1$, and note that%
\begin{equation*}
\widehat{\mathbf{\beta }}-\mathbf{\beta =S}_{f}^{-1}\left( \frac{1}{T}%
\sum_{t=1}^{T}\mathbf{u}_{t}\left( v_{t}-\overline{v}\right) \right) ,
\end{equation*}%
having defined $\mathbf{u}_{t}=\left( u_{1,t},...,u_{N,t}\right) ^{\prime }$
and $\overline{v}=T^{-1}\sum_{t=1}^{T}v_{t}$, whence also%
\begin{equation*}
\widehat{\beta }_{i}-\beta _{i}\mathbf{=S}_{f}^{-1}\left( \frac{1}{T}%
\sum_{t=1}^{T}\left( v_{t}-\overline{v}\right) u_{i,t}\right) .
\end{equation*}%
Note that%
\begin{eqnarray*}
&&N^{-1}\left\Vert \widehat{\mathbf{\beta }}-\mathbf{\beta }\right\Vert ^{2}
\\
&\leq &N^{-1}\sum_{i=1}^{N}\left\vert \widehat{\beta }_{i}-\beta
_{i}\right\vert ^{2}=\mathbf{S}_{f}^{-2}\times N^{-1}\sum_{i=1}^{N}\left( 
\frac{1}{T}\sum_{t=1}^{T}\left( v_{t}-\overline{v}\right) u_{i,t}\right)
^{2},
\end{eqnarray*}%
and%
\begin{eqnarray*}
&&\mathbb{E}\left[ N^{-1}\sum_{i=1}^{N}\left( \frac{1}{T}\sum_{t=1}^{T}%
\left( v_{t}-\overline{v}\right) u_{i,t}\right) ^{2}\right] \\
&=&\frac{1}{N}\sum_{i=1}^{N}\frac{1}{T^{2}}\sum_{t=1}^{T}\sum_{s=1}^{T}%
\mathbb{E}\left[ \left( v_{t}-\overline{v}\right) \left( v_{s}-\overline{v}%
\right) \right] E\left( u_{i,t}u_{i,s}\right) \\
&\leq &\frac{1}{N}\sum_{i=1}^{N}\frac{1}{T^{2}}\sum_{t=1}^{T}\sum_{s=1}^{T}%
\left( \mathbb{E}\left\vert v_{t}-\overline{v}\right\vert ^{2}\right)
\left\vert E\left( u_{i,t}u_{i,s}\right) \right\vert \leq c_{0}T,
\end{eqnarray*}%
now the desired result follows by using Lemma \ref{denominator-LS}. Turning
to (\ref{s-beta-hat-error}), note that 
\begin{equation*}
\widehat{\mathbf{S}}_{\beta }=\frac{1}{N}\sum_{i=1}^{N}\widehat{\beta }%
_{i}^{2}-\left( \frac{1}{N}\sum_{i=1}^{N}\widehat{\beta }_{i}\right)
^{2},\qquad \mathbf{S}_{\beta }=\frac{1}{N}\sum_{i=1}^{N}\beta
_{i}^{2}-\left( \frac{1}{N}\sum_{i=1}^{N}\beta _{i}\right) ^{2}
\end{equation*}%
with 
\begin{equation}
\widehat{\beta }_{i}^{2}=\left( \beta _{i}+\mathbf{S}_{f}^{-1}\frac{1}{T}%
\sum_{t=1}^{T}u_{i,t}(v_{t}-\overline{v})\right) ^{2}.  \label{beta-i-hat}
\end{equation}%
Hence%
\begin{eqnarray*}
\left\Vert \widehat{\mathbf{S}}_{\beta }-\mathbf{S}_{\beta }\right\Vert
&\leq &\left\Vert \frac{1}{N}\sum_{i=1}^{N}\widehat{\beta }_{i}^{2}-\frac{1}{%
N}\sum_{i=1}^{N}\beta _{i}^{2}\right\Vert +\left\Vert \frac{1}{N}%
\sum_{i=1}^{N}\left( \widehat{\beta }_{i}+\beta _{i}\right) \right\Vert
\left\Vert \frac{1}{N}\sum_{i=1}^{N}\left( \widehat{\beta }_{i}-\beta
_{i}\right) \right\Vert \\
&=&I+II.
\end{eqnarray*}%
Using (\ref{beta-i-hat})%
\begin{eqnarray*}
\frac{1}{N}\sum_{i=1}^{N}\widehat{\beta }_{i}^{2} &=&\frac{1}{N}%
\sum_{i=1}^{N}\beta _{i}^{2}+\frac{1}{NT^{2}}\sum_{i=1}^{N}\mathbf{S}%
_{f}^{-2}\left( \sum_{t=1}^{T}u_{i,t}\left( v_{t}-\overline{v}\right)
\right) ^{2}+\mathbf{S}_{f}^{-1}\frac{2}{NT}\sum_{i=1}^{N}\beta
_{i}\sum_{t=1}^{T}u_{i,t}\left( v_{t}-\overline{v}\right) \\
&=&I_{a}+I_{b}+I_{c};
\end{eqnarray*}%
the same passages as above readily yield 
\begin{equation*}
I_{b}=o_{a.s.}\left( \frac{\left( \log N\log T\right) ^{2+\epsilon }}{T}%
\right) ;
\end{equation*}%
also, using Assumptions \ref{cov_beta} and \ref{weak_CS_dep}%
\begin{eqnarray*}
&&\mathbb{E}\left( \frac{1}{NT}\sum_{i=1}^{N}\beta
_{i}\sum_{t=1}^{T}u_{i,t}\left( v_{t}-\overline{v}\right) \right) ^{2} \\
&=&\frac{1}{N^{2}T^{2}}\mathbb{E}\left( \sum_{i=1}^{N}\sum_{j=1}^{N}\beta
_{i}\beta _{j}\sum_{t=1}^{T}\sum_{s=1}^{T}\left( v_{t}-\overline{v}\right)
\left( v_{s}-\overline{v}\right) u_{i,t}u_{j,s}\right) \\
&\leq &\frac{1}{N^{2}T^{2}}\sum_{i=1}^{N}\sum_{j=1}^{N}\sum_{t=1}^{T}%
\sum_{s=1}^{T}\left( \mathbb{E}\left( v_{t}-\overline{v}\right) ^{2}\right)
\left\vert \mathbb{E}\left( u_{i,t}u_{j,s}\right) \right\vert \leq
c_{0}\left( NT\right) ^{-1},
\end{eqnarray*}%
so that%
\begin{equation*}
I_{c}=o_{a.s.}\left( \frac{\left( \log N\log T\right) ^{1+\epsilon }}{\left(
NT\right) ^{1/2}}\right) .
\end{equation*}%
By the same token, turning to $II$ we have 
\begin{equation*}
\frac{1}{N}\sum_{i=1}^{N}\left( \widehat{\beta }_{i}-\beta _{i}\right) =%
\mathbf{S}_{f}^{-1}\frac{1}{N}\sum_{i=1}^{N}\left( \frac{1}{T}%
\sum_{t=1}^{T}\left( v_{t}-\overline{v}\right) u_{i,t}\right) ,
\end{equation*}%
where, by the same logic as above, it follows that%
\begin{equation*}
\left\Vert \frac{1}{N}\sum_{i=1}^{N}\left( \widehat{\beta }_{i}-\beta
_{i}\right) \right\Vert =o_{a.s.}\left( \frac{\left( \log N\log T\right)
^{1+\epsilon }}{\left( NT\right) ^{1/2}}\right) .
\end{equation*}%
The desired result follows from noting%
\begin{equation*}
\left\Vert \frac{1}{N}\sum_{i=1}^{N}\left( \widehat{\beta }_{i}+\beta
_{i}\right) \right\Vert \leq \left\Vert \frac{2}{N}\sum_{i=1}^{N}\beta
_{i}\right\Vert +\left\Vert \frac{1}{N}\sum_{i=1}^{N}\left( \widehat{\beta }%
_{i}-\beta _{i}\right) \right\Vert .
\end{equation*}
\end{proof}
\end{lemma}

\begin{lemma}
\label{beta-sum}We assume that Assumptions \ref{error}-\ref{exogeneity} and %
\ref{cov_beta} and \ref{weak_CS_dep} are satisfied. Then it holds that%
\begin{equation*}
\sum_{i=1}^{N}\left\Vert \widehat{\beta }_{i}-\beta _{i}\right\Vert ^{\nu
/2}=o_{a.s.}\left( NT^{-\nu /4}\left( \log N\log T\right) ^{2+\epsilon
}\right) ,
\end{equation*}%
for all $\epsilon >0$.

\begin{proof}
Recall that 
\begin{equation*}
\widehat{\beta }_{i}-\beta _{i}\mathbf{=S}_{f}^{-1}\left( \frac{1}{T}%
\sum_{t=1}^{T}\left( v_{t}-\overline{v}\right) u_{i,t}\right) .
\end{equation*}%
Then 
\begin{equation*}
\sum_{i=1}^{N}\left\Vert \widehat{\beta }_{i}-\beta _{i}\right\Vert ^{\nu
/2}\leq c_{0}\left( \sum_{i=1}^{N}\left\Vert \frac{1}{T}%
\sum_{t=1}^{T}v_{t}u_{i,t}\right\Vert ^{\nu /2}+\sum_{i=1}^{N}\left\Vert 
\frac{1}{T}\sum_{t=1}^{T}\overline{v}u_{i,t}\right\Vert ^{\nu /2}\right) ,
\end{equation*}%
and we can readily show - by following the arguments above - that 
\begin{equation*}
\mathbb{E}\sum_{i=1}^{N}\left\Vert \frac{1}{T}\sum_{t=1}^{T}v_{t}u_{i,t}%
\right\Vert ^{\nu /2}\leq c_{0}NT^{-\nu /4},
\end{equation*}%
so that%
\begin{equation*}
\sum_{i=1}^{N}\left\Vert \frac{1}{T}\sum_{t=1}^{T}v_{t}u_{i,t}\right\Vert
^{\nu /2}=o_{a.s.}\left( NT^{-\nu /4}\left( \log N\log T\right) ^{2+\epsilon
}\right) .
\end{equation*}%
Recall that $\overline{v}=o_{a.s.}T^{-1/2}\left( \log T\right) ^{1+\epsilon
} $, and note%
\begin{equation*}
\mathbb{E}\sum_{i=1}^{N}\left\Vert \frac{1}{T}\sum_{t=1}^{T}u_{i,t}\right%
\Vert ^{\nu /2}\leq c_{0}NT^{-\nu /4},
\end{equation*}%
so that ultimately we receive the desired result by putting all together.
\end{proof}
\end{lemma}

\smallskip

We now report two lemmas on the rates of $\widehat{\lambda }-\lambda $ under
the null and under the alternative.

\begin{lemma}
\label{lambda-h0}We assume that $\mathbb{H}_{0}$ of (\ref{null}) holds, and
that Assumptions \ref{error}-\ref{exogeneity} and \ref{cov_beta} and \ref%
{weak_CS_dep} are satisfied. Then it holds that%
\begin{equation*}
\widehat{\lambda }-\lambda =o_{a.s.}\left( T^{-1/2}\left( \log T\right)
^{1+\epsilon }\right) ,
\end{equation*}%
for all $\epsilon >0$.

\begin{proof}
Recall (\ref{lambda-ff-expansion}). Under $\mathbb{H}_{0}$, it holds that $%
\mathbf{\alpha }=0$ and therefore%
\begin{eqnarray*}
\widehat{\lambda } &=&\lambda \mathbf{+}\overline{v}+\frac{1}{N}\widehat{%
\mathbf{S}}_{\beta }^{-1}\widehat{\mathbf{\beta }}^{\prime }\mathbb{M}%
_{1_{N}}\left( \widehat{\mathbf{\beta }}\mathbf{-\beta }\right) \lambda \\
&&+\frac{1}{N}\widehat{\mathbf{S}}_{\beta }^{-1}\widehat{\mathbf{\beta }}%
^{\prime }\mathbb{M}_{1_{N}}\left( \mathbf{\beta -}\widehat{\mathbf{\beta }}%
\right) \overline{v}+\frac{1}{N}\widehat{\mathbf{S}}_{\beta }^{-1}\mathbf{%
\beta }^{\prime }\mathbb{M}_{1_{N}}\overline{\mathbf{u}} \\
&&+\frac{1}{N}\widehat{\mathbf{S}}_{\beta }^{-1}\left( \widehat{\mathbf{%
\beta }}-\mathbf{\beta }\right) ^{\prime }\mathbb{M}_{1_{N}}\overline{%
\mathbf{u}} \\
&=&\lambda +I+II+III+IV+V.
\end{eqnarray*}%
We begin by noting that, from standard passages, $I=o_{a.s.}\left(
T^{-1/2}\left( \log T\right) ^{1+\epsilon }\right) $. Note that, combining (%
\ref{s-beta-hat-error}) and Assumption \ref{cov_beta}\textit{(ii)}%
\begin{eqnarray*}
\left\Vert \widehat{\mathbf{S}}_{\beta }^{-1}-\mathbf{S}_{\beta
}^{-1}\right\Vert &\leq &\left\Vert \left( \widehat{\mathbf{S}}_{\beta }\pm 
\mathbf{S}_{\beta }\right) ^{-1}\right\Vert \left\Vert \mathbf{S}_{\beta
}^{-1}\right\Vert \left\Vert \widehat{\mathbf{S}}_{\beta }-\mathbf{S}_{\beta
}\right\Vert \\
&=&o_{a.s.}\left( \frac{\left( \log N\log T\right) ^{2+\epsilon }}{T}\right)
+o_{a.s.}\left( \frac{\left( \log N\log T\right) ^{1+\epsilon }}{\sqrt{NT}}%
\right) ,
\end{eqnarray*}%
and therefore we have 
\begin{equation*}
\left\Vert \widehat{\mathbf{S}}_{\beta }^{-1}\right\Vert =O_{a.s.}\left(
1\right) .
\end{equation*}

Consider now%
\begin{eqnarray*}
&&\left\Vert \frac{1}{N}\widehat{\mathbf{S}}_{\beta }^{-1}\widehat{\mathbf{%
\beta }}^{\prime }\mathbb{M}_{1_{N}}\left( \widehat{\mathbf{\beta }}\mathbf{%
-\beta }\right) \lambda \right\Vert \\
&\leq &\frac{1}{N}\left\Vert \widehat{\mathbf{S}}_{\beta }^{-1}\pm \mathbf{S}%
_{\beta }^{-1}\right\Vert \left\Vert \left( \widehat{\mathbf{\beta }}\pm 
\mathbf{\beta }\right) ^{\prime }\mathbb{M}_{1_{N}}\left( \widehat{\mathbf{%
\beta }}\mathbf{-\beta }\right) \right\Vert \left\Vert \lambda \right\Vert .
\end{eqnarray*}%
We have%
\begin{eqnarray*}
&&\left\Vert \left( \widehat{\mathbf{\beta }}\pm \mathbf{\beta }\right)
^{\prime }\mathbb{M}_{1_{N}}\left( \widehat{\mathbf{\beta }}\mathbf{-\beta }%
\right) \right\Vert \\
&\leq &\left\Vert \mathbf{\beta }^{\prime }\mathbb{M}_{1_{N}}\left( \widehat{%
\mathbf{\beta }}\mathbf{-\beta }\right) \right\Vert +\left\Vert \left( 
\widehat{\mathbf{\beta }}-\mathbf{\beta }\right) ^{\prime }\mathbb{M}%
_{1_{N}}\left( \widehat{\mathbf{\beta }}\mathbf{-\beta }\right) \right\Vert
\\
&=&a+b.
\end{eqnarray*}%
Consider $a$, and let $\mathbf{\beta }^{\prime }\mathbb{M}_{1_{N}}=\mathbf{w}%
^{\prime }$ for short; we have%
\begin{equation*}
\frac{1}{N}\mathbf{w}^{\prime }\left( \widehat{\mathbf{\beta }}\mathbf{%
-\beta }\right) =\frac{1}{N}\sum_{i=1}^{N}w_{i}\left( \widehat{\beta }%
_{i}-\beta _{i}\right) =\frac{1}{N}\sum_{i=1}^{N}w_{i}\frac{1}{T\mathbf{S}%
_{f}}\sum_{t=1}^{T}u_{i,t}\left( v_{t}-\overline{v}\right) ,
\end{equation*}%
and therefore%
\begin{eqnarray*}
&&\mathbb{E}\left\vert \frac{1}{N}\mathbf{w}^{\prime }\left( \widehat{%
\mathbf{\beta }}\mathbf{-\beta }\right) \right\vert ^{2} \\
&=&\frac{1}{N^{2}T^{2}\mathbf{S}_{f}^{2}}\sum_{i=1}^{N}%
\sum_{j=1}^{N}w_{i}w_{j}\mathbb{E}\left[ \sum_{t=1}^{T}%
\sum_{s=1}^{T}u_{i,t}u_{i,s}\left( v_{t}-\overline{v}\right) \left( v_{s}-%
\overline{v}\right) \right] \\
&\leq &c_{0}\frac{1}{N^{2}T^{2}\mathbf{S}_{f}^{2}}\mathbb{E}\left[ \left(
v_{t}-\overline{v}\right) ^{2}\right] \sum_{i=1}^{N}\sum_{j=1}^{N}%
\sum_{t=1}^{T}\sum_{s=1}^{T}\left\vert \mathbb{E}\left(
u_{i,t}u_{i,s}\right) \right\vert \leq c_{1}\frac{1}{NT},
\end{eqnarray*}%
so that%
\begin{equation*}
a=o_{a.s.}\left( \frac{\left( \log N\log T\right) ^{1+\epsilon }}{\sqrt{NT}}%
\right) .
\end{equation*}%
Also%
\begin{eqnarray*}
\frac{1}{N}\left\Vert \left( \widehat{\mathbf{\beta }}-\mathbf{\beta }%
\right) ^{\prime }\mathbb{M}_{1_{N}}\left( \widehat{\mathbf{\beta }}\mathbf{%
-\beta }\right) \right\Vert &\leq &\frac{1}{N}\left\Vert \widehat{\mathbf{%
\beta }}-\mathbf{\beta }\right\Vert ^{2}+\frac{1}{N^{2}}\left[
\sum_{i=1}^{N}\left( \widehat{\beta }_{i}-\beta _{i}\right) \right] ^{2} \\
&=&o_{a.s.}\left( \frac{\left( \log N\log T\right) ^{2+\epsilon }}{T}\right)
,
\end{eqnarray*}%
following the proof of Lemma \ref{beta-hat-ff}. Hence, we obtain that 
\begin{equation*}
II=o_{a.s.}\left( \frac{\left( \log N\log T\right) ^{1+\epsilon }}{\sqrt{NT}}%
\right) +o_{a.s.}\left( \frac{\left( \log N\log T\right) ^{2+\epsilon }}{T}%
\right) .
\end{equation*}%
The same logic yields that $III$ is dominated by $II$. Turning to $IV$, it
holds that%
\begin{equation*}
\frac{1}{N}\mathbf{\beta }^{\prime }\mathbb{M}_{1_{N}}\overline{\mathbf{u}}=%
\frac{1}{N}\sum_{i=1}^{N}\beta _{i}\frac{1}{T}\sum_{t=1}^{T}u_{i,t}-\frac{1}{%
N}\sum_{i=1}^{N}\beta _{i}\frac{1}{NT}\sum_{i=1}^{N}\sum_{t=1}^{T}u_{i,t};
\end{equation*}%
we have%
\begin{eqnarray*}
&&\mathbb{E}\left[ \left( \frac{1}{N}\sum_{i=1}^{N}\beta _{i}\frac{1}{T}%
\sum_{t=1}^{T}u_{i,t}\right) ^{2}\right] \\
&=&\frac{1}{N^{2}T^{2}}\sum_{i=1}^{N}\sum_{i=1}^{N}\beta _{i}\beta
_{j}\sum_{t=1}^{T}\sum_{t=1}^{T}\mathbb{E}\left( u_{i,t}u_{j,s}\right) \\
&\leq &\frac{1}{N^{2}T^{2}}\sum_{i=1}^{N}\sum_{i=1}^{N}\sum_{t=1}^{T}%
\sum_{t=1}^{T}\left\vert \mathbb{E}\left( u_{i,t}u_{j,s}\right) \right\vert
\leq c_{0}\frac{1}{NT},
\end{eqnarray*}%
and 
\begin{eqnarray*}
&&\mathbb{E}\left[ \left( \frac{1}{NT}\sum_{i=1}^{N}\sum_{t=1}^{T}u_{i,t}%
\right) ^{2}\right] \\
&\leq &\frac{1}{N^{2}T^{2}}\sum_{i,j=1}^{N}\sum_{t,s=1}^{T}\mathbb{E}\left(
u_{i,t}u_{j,s}\right) \leq c_{0}\frac{1}{NT},
\end{eqnarray*}%
by Assumption \ref{weak_CS_dep}\textit{(ii)}, so that%
\begin{equation}
IV=o_{a.s.}\left( \frac{\left( \log N\log T\right) ^{1+\epsilon }}{\sqrt{NT}}%
\right) .  \label{proof-a-caso}
\end{equation}%
We conclude by only sketching the arguments for $V$; seeing as%
\begin{eqnarray*}
&&\frac{1}{N}\left( \widehat{\mathbf{\beta }}-\mathbf{\beta }\right)
^{\prime }\mathbb{M}_{1_{N}}\overline{\mathbf{u}} \\
&=&\frac{1}{N}\left( \widehat{\mathbf{\beta }}-\mathbf{\beta }\right)
^{\prime }\overline{\mathbf{u}}=\frac{1}{N}\sum_{i=1}^{N}\left( \widehat{%
\beta }_{i}-\beta _{i}\right) \frac{1}{T}\sum_{t=1}^{T}u_{i,t} \\
&\leq &\left( \frac{1}{N}\sum_{i=1}^{N}\left( \widehat{\beta }_{i}-\beta
_{i}\right) ^{2}\right) ^{1/2}\left( \frac{1}{N}\sum_{i=1}^{N}\left( \frac{1%
}{T}\sum_{t=1}^{T}u_{i,t}\right) ^{2}\right) ^{1/2},
\end{eqnarray*}%
and noting%
\begin{equation*}
\mathbb{E}\left[ \frac{1}{N}\sum_{i=1}^{N}\left( \frac{1}{T}%
\sum_{t=1}^{T}u_{i,t}\right) ^{2}\right] \leq c_{0}T^{-1},
\end{equation*}%
using Lemma \ref{beta-hat-ff} it follows that%
\begin{equation*}
V=o_{a.s.}\left( \frac{\left( \log N\log T\right) ^{2+\epsilon }}{T}\right) .
\end{equation*}%
The desired result now follows.
\end{proof}
\end{lemma}

\begin{lemma}
\label{lambda-hA}We assume that $\mathbb{H}_{A}$ of (\ref{hA}) holds, and
that Assumptions \ref{error}-\ref{exogeneity} and \ref{cov_beta} and \ref%
{weak_CS_dep} are satisfied. Then it holds that%
\begin{equation*}
\widehat{\lambda }-\lambda =\frac{1}{N}\mathbf{S}_{\beta }^{-1}\mathbf{\beta 
}^{\prime }\mathbb{M}_{1_{N}}\mathbf{\alpha }+o_{a.s.}\left( T^{-1/2}\left(
\log T\right) ^{1+\epsilon }\right) ,
\end{equation*}%
for all $\epsilon >0$.

\begin{proof}
Considering 
\begin{eqnarray*}
\widehat{\lambda } &=&\lambda +\frac{1}{N}\mathbf{S}_{\beta }^{-1}\mathbf{%
\beta }^{\prime }\mathbb{M}_{1_{N}}\mathbf{\alpha +}\overline{v}\mathbf{+}%
\frac{1}{N}\left( \widehat{\mathbf{S}}_{\beta }^{-1}-\mathbf{S}_{\beta
}^{-1}\right) \mathbf{\beta }^{\prime }\mathbb{M}_{1_{N}}\mathbf{\alpha } \\
&&+\frac{1}{N}\widehat{\mathbf{S}}_{\beta }^{-1}\left( \widehat{\mathbf{%
\beta }}\mathbf{-\beta }\right) ^{\prime }\mathbb{M}_{1_{N}}\mathbf{\alpha +}%
\frac{1}{N}\widehat{\mathbf{S}}_{\beta }^{-1}\widehat{\mathbf{\beta }}%
^{\prime }\mathbb{M}_{1_{N}}\left( \widehat{\mathbf{\beta }}\mathbf{-\beta }%
\right) \lambda \\
&&+\frac{1}{N}\widehat{\mathbf{S}}_{\beta }^{-1}\widehat{\mathbf{\beta }}%
^{\prime }\mathbb{M}_{1_{N}}\left( \mathbf{\beta -}\widehat{\mathbf{\beta }}%
\right) \overline{v}+\frac{1}{N}\widehat{\mathbf{S}}_{\beta }^{-1}\mathbf{%
\beta }^{\prime }\mathbb{M}_{1_{N}}\overline{\mathbf{u}} \\
&&+\frac{1}{N}\widehat{\mathbf{S}}_{\beta }^{-1}\left( \widehat{\mathbf{%
\beta }}-\mathbf{\beta }\right) ^{\prime }\mathbb{M}_{1_{N}}\overline{%
\mathbf{u}} \\
&=&\lambda +I+II+III+IV+V+VI+VII+VIII,
\end{eqnarray*}%
the only terms that require some analysis are $III$ and $IV$. However, we
already know that%
\begin{equation*}
\left\Vert \widehat{\mathbf{S}}_{\beta }^{-1}-\mathbf{S}_{\beta
}^{-1}\right\Vert =o_{a.s.}\left( \frac{\left( \log N\log T\right)
^{2+\epsilon }}{T}\right) +o_{a.s.}\left( \frac{\left( \log N\log T\right)
^{1+\epsilon }}{\sqrt{NT}}\right) ;
\end{equation*}%
further%
\begin{eqnarray*}
&&\frac{1}{N}\left( \widehat{\mathbf{\beta }}\mathbf{-\beta }\right)
^{\prime }\mathbb{M}_{1_{N}}\mathbf{\alpha } \\
&\mathbf{=}&\frac{1}{N}\sum_{i=1}^{N}\left( \alpha _{i}-\frac{1}{N}%
\sum_{i=1}^{N}\alpha _{i}\right) \left( \widehat{\beta }_{i}-\beta
_{i}\right) =\frac{1}{N}\sum_{i=1}^{N}\widetilde{w}_{i}\frac{1}{T}%
\sum_{t=1}^{T}\left( v_{t}-\overline{v}\right) u_{i,t},
\end{eqnarray*}%
which can be shown to be $o_{a.s.}\left( \left( NT\right) ^{-1/2}\left( \log
N\log T\right) ^{1+\epsilon }\right) $.
\end{proof}
\end{lemma}

\begin{lemma}
\label{normalising_FM} We assume that Assumptions \ref{error}-\ref%
{exogeneity}, \ref{cov_beta} and \ref{weak_CS_dep} are satisfied. Then it
holds that%
\begin{eqnarray*}
\liminf_{\min \left\{ N,T\right\} \rightarrow \infty }\frac{1}{NT}%
\sum_{i=1}^{N}\sum_{t=1}^{T}\left( \widehat{u}_{i,t}^{FM}\right) ^{2} &>&0,
\\
\limsup_{\min \left\{ N,T\right\} \rightarrow \infty }\frac{1}{NT}%
\sum_{i=1}^{N}\sum_{t=1}^{T}\left( \widehat{u}_{i,t}^{FM}\right) ^{2}
&<&\infty .
\end{eqnarray*}

\begin{proof}
The proof is very similat to that of Lemma \ref{normalising}, and we report
only the main passages to save space. It holds that 
\begin{eqnarray*}
&&\frac{1}{NT}\sum_{i=1}^{N}\sum_{t=1}^{T}\left( \widehat{u}%
_{i,t}^{FM}\right) ^{2} \\
&=&\frac{1}{NT}\sum_{i=1}^{N}\sum_{t=1}^{T}u_{i,t}^{2}+\frac{1}{N}%
\sum_{i=1}^{N}\left( \widehat{\alpha }_{i}^{FM}-\alpha _{i}\right) ^{2}+%
\frac{1}{NT}\sum_{i=1}^{N}\sum_{t=1}^{T}\left( \widehat{\beta }_{i}-\beta
_{i}\right) ^{2}f_{t}^{2} \\
&&+\frac{2}{NT}\sum_{i=1}^{N}\sum_{t=1}^{T}\left( \widehat{\alpha }%
_{i}^{FM}-\alpha _{i}\right) u_{i,t}+\frac{2}{NT}\sum_{i=1}^{N}%
\sum_{t=1}^{T}\left( \widehat{\beta }_{i}-\beta _{i}\right) f_{t}u_{i,t} \\
&&+\frac{2}{NT}\sum_{i=1}^{N}\sum_{t=1}^{T}\left( \widehat{\alpha }%
_{i}^{FM}-\alpha _{i}\right) \left( \widehat{\beta }_{i}-\beta _{i}\right)
f_{t} \\
&=&I+II+III+IV+V+VI.
\end{eqnarray*}%
The rates of terms $I$, $III$ and $V$ are the same as in the proof of Lemma %
\ref{normalising}. Note that $II\geq 0$; we derive an upper bound for it
using (\ref{alpha-ff}). Noting that%
\begin{eqnarray*}
&&\frac{1}{N}\sum_{i=1}^{N}\left( \widehat{\alpha }_{i}^{FM}-\alpha
_{i}\right) ^{2} \\
&\leq &c_{0}\left[ \frac{1}{N}\sum_{i=1}^{N}\beta _{i}^{2}\overline{v}^{2}+%
\frac{1}{N}\sum_{i=1}^{N}\overline{u}_{i}^{2}+\left( \frac{1}{N}%
\sum_{i=1}^{N}\left( \widehat{\beta }_{i}-\beta _{i}\right) ^{2}\right)
\lambda ^{2}\right. \\
&&\left. \left( \frac{1}{N}\sum_{i=1}^{N}\beta _{i}^{2}\right) \left( 
\widehat{\lambda }-\lambda \right) ^{2}+\left( \frac{1}{N}%
\sum_{i=1}^{N}\left( \widehat{\beta }_{i}-\beta _{i}\right) ^{2}\right)
\left( \widehat{\lambda }-\lambda \right) ^{2}\right] ,
\end{eqnarray*}%
the results above readily entail that $II=o_{a.s.}\left( 1\right) $.
Similarly%
\begin{eqnarray*}
IV &=&\frac{2}{NT}\left( \sum_{i=1}^{N}\sum_{t=1}^{T}\beta
_{i}u_{i,t}\right) \overline{v}+\frac{2}{N}\sum_{i=1}^{N}\overline{u}%
_{i}\left( \frac{1}{T}\sum_{t=1}^{T}u_{i,t}\right) -\frac{2}{NT}\left(
\sum_{i=1}^{N}\left( \widehat{\beta }_{i}-\beta _{i}\right)
\sum_{t=1}^{T}u_{i,t}\right) \lambda \\
&&-\frac{2}{NT}\left( \sum_{i=1}^{N}\sum_{t=1}^{T}\beta _{i}u_{i,t}\right)
\left( \widehat{\lambda }-\lambda \right) -\frac{2}{NT}\left(
\sum_{i=1}^{N}\left( \widehat{\beta }_{i}-\beta _{i}\right) \left(
\sum_{t=1}^{T}u_{i,t}\right) \right) \left( \widehat{\lambda }-\lambda
\right) \\
&=&IV_{a}+IV_{b}+IV_{c}+IV_{d}+IV_{e}.
\end{eqnarray*}%
Since it is immediate to see that%
\begin{equation*}
\mathbb{E}\left( \sum_{i=1}^{N}\sum_{t=1}^{T}\beta _{i}u_{i,t}\right)
^{2}\leq c_{0}NT,
\end{equation*}%
we have 
\begin{equation*}
IV_{a}=o_{a.s.}\left( \frac{\left( \log N\log ^{2}T\right) ^{1+\epsilon }}{%
N^{1/2}T}\right) ;
\end{equation*}%
by the same arguments%
\begin{equation*}
\mathbb{E}\left( \frac{2}{N}\sum_{i=1}^{N}\overline{u}_{i}\left( \frac{1}{T}%
\sum_{t=1}^{T}u_{i,t}\right) \right) =\frac{2}{N}\sum_{i=1}^{N}\mathbb{E}%
\left( \overline{u}_{i}^{2}\right) \leq c_{0}T^{-1},
\end{equation*}%
and therefore%
\begin{equation*}
IV_{b}=o_{a.s.}\left( \frac{\left( \log T\right) ^{2+\epsilon }}{T}\right) .
\end{equation*}%
The other terms can be shown to be dominated by using the arguments above.
Noting that%
\begin{equation*}
\frac{1}{N}\sum_{i=1}^{N}\left( \widehat{\beta }_{i}-\beta _{i}\right)
\left( \frac{1}{T}\sum_{t=1}^{T}u_{i,t}\right) \leq \left( \frac{1}{N}%
\sum_{i=1}^{N}\left( \widehat{\beta }_{i}-\beta _{i}\right) ^{2}\right)
^{1/2}\left( \frac{1}{N}\sum_{i=1}^{N}\left( \frac{1}{T}%
\sum_{t=1}^{T}u_{i,t}\right) ^{2}\right) ^{1/2},
\end{equation*}%
it is easy to see that $VI=o_{a.s.}\left( 1\right) $. The desired result now
follows.
\end{proof}
\end{lemma}

\begin{lemma}
\label{psi-fm} We assume that Assumptions \ref{error}-\ref{exogeneity}, \ref%
{cov_beta} and \ref{weak_CS_dep} are satisfied. Then it holds that, under $%
\mathbb{H}_{0}$%
\begin{equation*}
\sum_{i=1}^{N}\psi _{i,NT}^{FM}=o_{a.s.}\left( 1\right) .
\end{equation*}

\begin{proof}
Consider the case $K=1$, and recall (\ref{alpha-ff}), which under $\mathbb{H}%
_{0}$ becomes%
\begin{equation*}
\widehat{\alpha }_{i}^{FM}=\beta _{i}\overline{v}+\overline{u}_{i}-\left( 
\widehat{\beta }_{i}-\beta _{i}\right) \lambda -\beta _{i}\left( \widehat{%
\lambda }-\lambda \right) -\left( \widehat{\beta }_{i}-\beta _{i}\right)
\left( \widehat{\lambda }-\lambda \right) .
\end{equation*}%
Hence we have%
\begin{eqnarray*}
&&\sum_{i=1}^{N}\psi _{i,NT}^{FM} \\
&=&\frac{T^{1/2}}{\left\vert \widehat{s}_{NT}^{FM}\right\vert ^{\nu /2}}%
\sum_{i=1}^{N}\left\vert \widehat{\alpha }_{i}^{FM}\right\vert ^{\nu /2} \\
&\leq &c_{0}\left[ \frac{T^{1/2}}{\left\vert \widehat{s}_{NT}^{FM}\right%
\vert ^{\nu /2}}\left( \sum_{i=1}^{N}\left\vert \beta _{i}\right\vert ^{\nu
/2}\right) \left\vert \overline{v}\right\vert ^{\nu /2}+\frac{T^{1/2}}{%
\left\vert \widehat{s}_{NT}^{FM}\right\vert ^{\nu /2}}\sum_{i=1}^{N}\left%
\vert \overline{u}_{i}\right\vert ^{\nu /2}+\frac{T^{1/2}}{\left\vert 
\widehat{s}_{NT}^{FM}\right\vert ^{\nu /2}}\left( \sum_{i=1}^{N}\left\vert 
\widehat{\beta }_{i}-\beta _{i}\right\vert ^{\nu /2}\right) \left\vert
\lambda \right\vert ^{\nu /2}\right. \\
&&\left. +\frac{T^{1/2}}{\left\vert \widehat{s}_{NT}^{FM}\right\vert ^{\nu
/2}}\left( \sum_{i=1}^{N}\left\vert \beta _{i}\right\vert ^{\nu /2}\right)
\left\vert \widehat{\lambda }-\lambda \right\vert ^{\nu /2}+\frac{T^{1/2}}{%
\left\vert \widehat{s}_{NT}^{FM}\right\vert ^{\nu /2}}\left(
\sum_{i=1}^{N}\left\vert \widehat{\beta }_{i}-\beta _{i}\right\vert ^{\nu
/2}\right) \left\vert \widehat{\lambda }-\lambda \right\vert ^{\nu /2}\right]
\\
&=&I+II+III+IV+V.
\end{eqnarray*}%
By Assumption \ref{cov_beta}\textit{(i)}, $\sum_{i=1}^{N}\left\vert \beta
_{i}\right\vert ^{\nu /2}=O(N) $, so that 
\begin{equation*}
I=o_{a.s.}\left(NT^{1/2-\nu /4}\left( \log T\right) ^{\left( 1+\epsilon
\right) \nu /2}\right) =o_{a.s.}\left( 1\right) ,
\end{equation*}%
seeing as $\nu \geq 4$. Also, we have already shown in the proof of Lemma \ref%
{psi} that $II=o_{a.s.}\left( 1\right) $. Moreover, Lemma \ref{beta-sum}
yields%
\begin{equation*}
\frac{T^{1/2}}{\left\vert \widehat{s}_{NT}^{FM}\right\vert ^{\nu /2}}\left(
\sum_{i=1}^{N}\left\vert \widehat{\beta }_{i}-\beta _{i}\right\vert ^{\nu
/2}\right) \left\vert \lambda \right\vert ^{\nu /2}=o_{a.s.}\left(
NT^{1/2}T^{-\nu /4}\left( \log N\log T\right) ^{2+\epsilon }\right)
=o_{a.s.}\left( 1\right) .
\end{equation*}%
Finally, by Assumption \ref{cov_beta}\textit{(i)} and Lemma \ref{lambda-h0}
entail%
\begin{equation*}
IV=o_{a.s.}\left( NT^{1/2-\nu /4}\left( \log T\right) ^{\left( 1+\epsilon
\right) \nu /2}\right) =o_{a.s.}\left( 1\right) .
\end{equation*}%
Finally, it is easy to see that $V$ is dominated by $III$. The desired
result now follows by putting all together.
\end{proof}
\end{lemma}

\subsection{Lemmas for Section \protect\ref{latent}\label{lemma-latent}}

We now report a series of lemmas for the case, discussed in Section \ref%
{latent}, of latent factors. As in the previous subsection, in the proofs we
will assume $K=1$ whenever possible.

\smallskip

Recall that - with reference to (\ref{latent-f}) - $\mathbf{\beta }=\left(
\beta _{1},...,\beta _{N}\right) ^{\prime }$. Let $\widehat{\Phi }$ be the
diagonal matrix containing, in descending order, the $K$ largest eigenvalues
of $\widehat{\mathbf{\Sigma }}_{y}$. Then, by definition%
\begin{equation*}
\widehat{\mathbf{\Sigma }}_{y}\widehat{\mathbf{\beta }}^{PC}=\widehat{%
\mathbf{\beta }}^{PC}\widehat{\Phi },
\end{equation*}%
which implies the following expansion 
\begin{eqnarray}
\widehat{\mathbf{\beta }}^{PC} &=&\widehat{\mathbf{\Sigma }}_{y}\widehat{%
\mathbf{\beta }}^{PC}\widehat{\Phi }^{-1}  \label{beta-pc-dec} \\
&=&\left[ \frac{1}{NT}\sum_{t=1}^{T}\left( \mathbf{\beta }\widetilde{v}_{t}+%
\widetilde{\mathbf{u}}_{t}\right) \left( \mathbf{\beta }\widetilde{v}_{t}+%
\widetilde{\mathbf{u}}_{t}\right) ^{\prime }\right] \widehat{\mathbf{\beta }}%
^{PC}\widehat{\Phi }^{-1}  \notag \\
&=&\mathbf{\beta H+}\frac{1}{N}\mathbf{\beta }\left( \frac{1}{T}%
\sum_{t=1}^{T}\widetilde{v}_{t}\widetilde{\mathbf{u}}_{t}^{\prime }\right) 
\widehat{\mathbf{\beta }}^{PC}\widehat{\Phi }^{-1}+\frac{1}{N}\left( \frac{1%
}{T}\sum_{t=1}^{T}\widetilde{\mathbf{u}}_{t}\widetilde{v}_{t}^{\prime
}\right) \mathbf{\beta }^{\prime }\widehat{\mathbf{\beta }}^{PC}\widehat{%
\Phi }^{-1}  \notag \\
&&+\left[ \frac{1}{NT}\sum_{t=1}^{T}\widetilde{\mathbf{u}}_{t}\widetilde{%
\mathbf{u}}_{t}^{\prime }\right] \widehat{\mathbf{\beta }}^{PC}\widehat{\Phi 
}^{-1},  \notag
\end{eqnarray}%
with the constraint $\left( \widehat{\mathbf{\beta }}^{PC}\right) ^{\prime }%
\widehat{\mathbf{\beta }}^{PC}=N\times \mathbb{I}_{K}$ and having defined%
\begin{equation*}
\mathbf{H=}\left( \frac{1}{T}\sum_{t=1}^{T}\widetilde{v}_{t}\widetilde{v}%
_{t}^{\prime }\right) \frac{\mathbf{\beta }^{\prime }\widehat{\mathbf{\beta }%
}^{PC}}{N}\widehat{\Phi }^{-1}.
\end{equation*}%
Similarly, we have also the unit-by-unit version of (\ref{beta-pc-dec})%
\begin{eqnarray}
\widehat{\beta }_{i}^{PC} &=&\mathbf{H}^{\prime }\beta _{i}+\widehat{\Phi }%
^{-1}\widehat{\mathbf{\beta }}^{PC\prime }\left( \frac{1}{NT}\sum_{t=1}^{T}%
\widetilde{\mathbf{u}}_{t}\widetilde{v}_{t}^{\prime }\right) \beta _{i}
\label{beta-pc-scalar} \\
&&+\widehat{\Phi }^{-1}\widehat{\mathbf{\beta }}^{PC\prime }\mathbf{\beta }%
\left( \frac{1}{NT}\sum_{t=1}^{T}\widetilde{v}_{t}\widetilde{u}_{i,t}\right)
+\widehat{\Phi }^{-1}\widehat{\mathbf{\beta }}^{PC\prime }\left( \frac{1}{NT}%
\sum_{t=1}^{T}\widetilde{\mathbf{u}}_{t}\widetilde{u}_{i,t}\right) .  \notag
\end{eqnarray}

Then, considering 
\begin{equation*}
\widehat{\lambda }^{PC}=\left( \widehat{\mathbf{\beta }}^{PC\prime }\mathbb{M%
}_{1_{N}}\widehat{\mathbf{\beta }}^{PC}\right) ^{-1}\left( \widehat{\mathbf{%
\beta }}^{PC\prime }\mathbb{M}_{1_{N}}\bar{\mathbf{y}}\right) ,
\end{equation*}%
we have%
\begin{align}
\widehat{\lambda }^{PC}=& \left( \widehat{\mathbf{\beta }}^{PC\prime }%
\mathbb{M}_{1_{N}}\widehat{\mathbf{\beta }}^{PC}\right) ^{-1}\left( \widehat{%
\mathbf{\beta }}^{PC\prime }\mathbb{M}_{1_{N}}\bar{\mathbf{y}}\right)
\label{lambda-pc} \\
=& \left( \widehat{\mathbf{\beta }}^{PC\prime }\mathbb{M}_{1_{N}}\widehat{%
\mathbf{\beta }}^{PC}\right) ^{-1}\left( \widehat{\mathbf{\beta }}^{PC\prime
}\mathbb{M}_{1_{N}}\mathbf{\alpha }\right) +\left( \widehat{\mathbf{\beta }}%
^{PC\prime }\mathbb{M}_{1_{N}}\widehat{\mathbf{\beta }}^{PC}\right)
^{-1}\left( \widehat{\mathbf{\beta }}^{PC\prime }\mathbb{M}_{1_{N}}\mathbf{%
\beta }\right) \lambda  \notag \\
& +\left( \widehat{\mathbf{\beta }}^{PC\prime }\mathbb{M}_{1_{N}}\widehat{%
\mathbf{\beta }}^{PC}\right) ^{-1}\left( \widehat{\mathbf{\beta }}^{PC\prime
}\mathbb{M}_{1_{N}}\mathbf{\beta }\right) \overline{v}+\left( \widehat{%
\mathbf{\beta }}^{PC\prime }\mathbb{M}_{1_{N}}\widehat{\mathbf{\beta }}%
^{PC}\right) ^{-1}\left( \widehat{\mathbf{\beta }}^{PC\prime }\mathbb{M}%
_{1_{N}}\overline{\mathbf{u}}\right) .  \notag
\end{align}

\begin{lemma}
\label{inversion}We assume that Assumptions \ref{error}-\ref{exogeneity},
and \ref{loadings}-\ref{fact-idios} are satisfied. Then it holds that $%
\left\Vert \widehat{\Phi }^{-1}\right\Vert =O_{a.s.}\left( 1\right) $.

\begin{proof}
Let $\Phi =diag\left\{ \Phi _{1},...,\Phi _{K}\right\} $ denote the diagonal
matrix containing the $K$ largest eigenvalues of $\mathbf{\beta }\mathbb{E}%
\left( \widetilde{v}_{t}\widetilde{v}_{t}^{\prime }\right) \mathbf{\beta }%
^{\prime }/N$. Then, it is immediate to verify that the assumptions of Lemma
2.2 in \cite{trapani2018randomized} hold, and 
\begin{equation*}
\widehat{\Phi }_{j}=\Phi _{j}+o_{a.s.}\left( 1\right) ,
\end{equation*}%
for all $1\leq j\leq K$. Seeing as, using the multiplicative Weyl's
inequality (theorem 7 in \citealt{merikoski2004inequalities}), and
Assumptions \ref{loadings} and \ref{factors}%
\begin{equation*}
\Phi _{K}\geq \rho _{K}\left( \mathbf{\beta }^{\prime }\mathbf{\beta }%
/N\right) \rho _{\min }\left( \mathbb{E}\left( \widetilde{v}_{t}\widetilde{v}%
_{t}^{\prime }\right) \right) \geq c_{0}>0,
\end{equation*}%
where $\rho _{k}\left( A\right) $ denotes the $k$-th largest eigenvalue of
matrix $A$, the desired result follows.
\end{proof}
\end{lemma}

\begin{lemma}
\label{l2norm}We assume that Assumptions \ref{error}-\ref{exogeneity}, and %
\ref{loadings}-\ref{fact-idios} are satisfied. Then it holds that 
\begin{equation*}
\frac{1}{N}\left\Vert \widehat{\mathbf{\beta }}^{PC}-\mathbf{\beta H}%
\right\Vert ^{2}=o_{a.s.}\left( \frac{\left( \log N\log T\right)
^{2+\epsilon }}{T}\right) +O\left( \frac{1}{N^{2}}\right) .
\end{equation*}

\begin{proof}
By (\ref{beta-pc-dec}), we have%
\begin{eqnarray*}
&&\left\Vert \widehat{\mathbf{\beta }}^{PC}-\mathbf{\beta H}\right\Vert ^{2}
\\
&\leq &c_{0}\left( \frac{1}{N^{2}}\left\Vert \mathbf{\beta }\left( \frac{1}{T%
}\sum_{t=1}^{T}\widetilde{v}_{t}\widetilde{\mathbf{u}}_{t}^{\prime }\right) 
\widehat{\mathbf{\beta }}^{PC}\widehat{\Phi }^{-1}\right\Vert ^{2}+\frac{1}{%
N^{2}}\left\Vert \left( \frac{1}{T}\sum_{t=1}^{T}\widetilde{\mathbf{u}}_{t}%
\widetilde{v}_{t}^{\prime }\right) \mathbf{\beta }^{\prime }\widehat{\mathbf{%
\beta }}^{PC}\widehat{\Phi }^{-1}\right\Vert ^{2}\right. \\
&&\qquad +\left. \left\Vert \left( \frac{1}{NT}\sum_{t=1}^{T}\widetilde{%
\mathbf{u}}_{t}\widetilde{\mathbf{u}}_{t}^{\prime }\right) \widehat{\mathbf{%
\beta }}^{PC}\widehat{\Phi }^{-1}\right\Vert ^{2}\right) \\
&=&I+I^{\prime }+II.
\end{eqnarray*}%
It holds that%
\begin{equation*}
I\leq \frac{1}{N^{2}}\left\Vert \mathbf{\beta }\right\Vert ^{2}\left\Vert 
\widehat{\mathbf{\beta }}^{PC}\right\Vert ^{2}\left\Vert \widehat{\Phi }%
^{-1}\right\Vert ^{2}\left\Vert \frac{1}{T}\sum_{t=1}^{T}\widetilde{v}_{t}%
\widetilde{\mathbf{u}}_{t}^{\prime }\right\Vert ^{2}\leq c_{0}\left\Vert 
\frac{1}{T}\sum_{t=1}^{T}\widetilde{v}_{t}\widetilde{\mathbf{u}}_{t}^{\prime
}\right\Vert ^{2},
\end{equation*}%
on account of Assumption \ref{loadings} and Lemma \ref{inversion}. Also note
that%
\begin{equation*}
\left\Vert \frac{1}{T}\sum_{t=1}^{T}\widetilde{v}_{t}\widetilde{\mathbf{u}}%
_{t}^{\prime }\right\Vert ^{2}\leq \left\Vert \frac{1}{T}\sum_{t=1}^{T}v_{t}%
\mathbf{u}_{t}^{\prime }\right\Vert ^{2}+\left\Vert \overline{v}\right\Vert
^{2}\left\Vert \overline{\mathbf{u}}\right\Vert ^{2}.
\end{equation*}%
It holds that%
\begin{equation*}
\left\Vert \frac{1}{T}\sum_{t=1}^{T}v_{t}\mathbf{u}_{t}^{\prime }\right\Vert
^{2}=\sum_{i=1}^{N}\left( \frac{1}{T}\sum_{t=1}^{T}v_{t}u_{i,t}\right) ^{2},
\end{equation*}%
and%
\begin{eqnarray*}
&&\mathbb{E}\left[ \sum_{i=1}^{N}\left( \frac{1}{T}%
\sum_{t=1}^{T}v_{t}u_{i,t}\right) ^{2}\right] \\
&=&\frac{1}{T^{2}}\sum_{i=1}^{N}\sum_{t,s=1}^{T}\mathbb{E}\left(
v_{t}v_{s}u_{i,t}u_{i,s}\right) \leq \frac{1}{T^{2}}\sum_{i=1}^{N}%
\sum_{t,s=1}^{T}\mathbb{E}\left( v_{t}v_{s}\right) \mathbb{E}\left(
u_{i,t}u_{i,s}\right) \\
&\leq &\frac{1}{T^{2}}\mathbb{E}\left( v_{0}^{2}\right)
\sum_{i=1}^{N}\sum_{t,s=1}^{T}\left\vert \mathbb{E}\left(
u_{i,t}u_{i,s}\right) \right\vert \leq c_{0}\frac{N}{T}.
\end{eqnarray*}%
Also, we know that $\left\Vert \overline{v}\right\Vert =o_{a.s.}\left(
T^{-1/2}\left( \log T\right) ^{1+\epsilon }\right) $, and%
\begin{equation*}
\left\Vert \overline{\mathbf{u}}\right\Vert ^{2}=\sum_{i=1}^{N}\left( \frac{1%
}{T}\sum_{t=1}^{T}u_{i,t}\right) ^{2}=\frac{1}{T^{2}}\sum_{i=1}^{N}%
\sum_{t,s=1}^{T}u_{i,t}u_{i,s},
\end{equation*}%
with%
\begin{equation*}
\mathbb{E}\left( \frac{1}{T^{2}}\sum_{i=1}^{N}\sum_{t,s=1}^{T}u_{i,t}u_{i,s}%
\right) \leq \frac{1}{T^{2}}\sum_{i=1}^{N}\sum_{t,s=1}^{T}\left\vert \mathbb{%
E}\left( u_{i,t}u_{i,s}\right) \right\vert \leq c_{0}\frac{N}{T},
\end{equation*}%
so that $\left\Vert \overline{\mathbf{u}}\right\Vert =o_{a.s.}\left(
N^{1/2}T^{-1/2}\left( \log T\log N\right) ^{1+\epsilon }\right) $. Putting
all together, we ultimately receive%
\begin{equation*}
I=o_{a.s.}\left( \frac{N}{T}\left( \log N\log T\right) ^{2+\epsilon }\right)
.
\end{equation*}%
Turning to $II$, note%
\begin{eqnarray*}
&&\left\Vert \left( \frac{1}{NT}\sum_{t=1}^{T}\widetilde{\mathbf{u}}_{t}%
\widetilde{\mathbf{u}}_{t}^{\prime }\right) \widehat{\mathbf{\beta }}^{PC}%
\widehat{\Phi }^{-1}\right\Vert ^{2} \\
&\leq &c_{0}\left\Vert \left( \frac{1}{NT}\sum_{t=1}^{T}\mathbf{u}_{t}%
\mathbf{u}_{t}^{\prime }-\mathbb{E}\left( \mathbf{u}_{0}\mathbf{u}%
_{0}^{\prime }\right) \right) \widehat{\mathbf{\beta }}^{PC}\widehat{\Phi }%
^{-1}\right\Vert ^{2}+c_{0}\left\Vert \frac{1}{N}\mathbb{E}\left( \mathbf{u}%
_{0}\mathbf{u}_{0}^{\prime }\right) \widehat{\mathbf{\beta }}^{PC}\widehat{%
\Phi }^{-1}\right\Vert ^{2}+c_{0}\left\Vert \frac{1}{N}\overline{\mathbf{u}}%
\overline{\mathbf{u}}^{\prime }\widehat{\mathbf{\beta }}^{PC}\widehat{\Phi }%
^{-1}\right\Vert ^{2} \\
&=&II_{a}+II_{b}+II_{c}.
\end{eqnarray*}%
We already know from the above that $\left\Vert \overline{\mathbf{u}}%
\right\Vert ^{2}=o_{a.s.}\left( NT^{-1}\left( \log N\log T\right)
^{2+\epsilon }\right) $, which readily yields 
\begin{equation*}
II_{c}=o_{a.s.}\left( NT^{-2}\left( \log N\log T\right) ^{4+\epsilon
}\right) .
\end{equation*}%
Also%
\begin{equation*}
\left\Vert \mathbb{E}\left( \mathbf{u}_{0}\mathbf{u}_{0}^{\prime }\right)
\right\Vert ^{2}\leq \left\Vert \mathbb{E}\left( \mathbf{u}_{0}\mathbf{u}%
_{0}^{\prime }\right) \right\Vert _{1}^{2}=\left( \max_{1\leq i\leq
N}\sum_{j=1}^{N}\left\vert \mathbb{E}\left( u_{i,0}u_{j,0}\right)
\right\vert \right) ^{2}\leq c_{0},
\end{equation*}%
so that%
\begin{equation*}
II_{b}=O\left( \frac{1}{N}\right) .
\end{equation*}%
Finally%
\begin{equation*}
\left\Vert \frac{1}{NT}\sum_{t=1}^{T}\mathbf{u}_{t}\mathbf{u}_{t}^{\prime }-%
\mathbb{E}\left( \mathbf{u}_{0}\mathbf{u}_{0}^{\prime }\right) \right\Vert
_{F}^{2}=\frac{1}{N^{2}T^{2}}\sum_{i,j=1}^{N}\left(
\sum_{t,s=1}^{T}Cov\left( u_{i,t}u_{j,t},u_{i,s}u_{j,s}\right) \right) \leq
c_{0}\frac{1}{T},
\end{equation*}%
and therefore%
\begin{equation*}
II_{a}=o_{a.s.}\left( \frac{N}{T}\left( \log N\log T\right) ^{2+\epsilon
}\right) .
\end{equation*}

The desired result obtains by putting all together.
\end{proof}
\end{lemma}

\begin{lemma}
\label{matrix-h}We assume that Assumptions \ref{error}-\ref{exogeneity}, and %
\ref{loadings}-\ref{fact-idios} are satisfied. Then it holds that $%
\left\Vert \mathbf{H}\right\Vert =O_{a.s.}\left( 1\right) $, and $\left\Vert 
\mathbf{H}^{-1}\right\Vert =O_{a.s.}\left( 1\right) $.

\begin{proof}
Note%
\begin{equation*}
\left\Vert \mathbf{H}\right\Vert \leq \left\Vert \frac{1}{T}\sum_{t=1}^{T}%
\widetilde{v}_{t}\widetilde{v}_{t}^{\prime }\right\Vert \frac{\left\Vert 
\mathbf{\beta }\right\Vert \left\Vert \widehat{\mathbf{\beta }}%
^{PC}\right\Vert }{N}\left\Vert \widehat{\Phi }^{-1}\right\Vert .
\end{equation*}%
Standard arguments yield $\left\Vert T^{-1}\sum_{t=1}^{T}\widetilde{v}_{t}%
\widetilde{v}_{t}^{\prime }\right\Vert =O_{a.s.}\left( 1\right) $; further, $%
\left\Vert \widehat{\mathbf{\beta }}^{PC}\right\Vert =N^{1/2}$ by
construction, and $\left\Vert \mathbf{\beta }\right\Vert =N^{1/2}$ by
Assumption \ref{loadings}. The desired result now follows by Lemma \ref%
{inversion}. As far as the second part of the lemma is concerned, recall the
identification restriction $\mathbf{\beta }^{\prime }\mathbf{\beta }=N%
\mathbb{I}_{K}$, and that, by construction $\left( \widehat{\mathbf{\beta }}%
^{PC}\right) ^{\prime }\widehat{\mathbf{\beta }}^{PC}=N\mathbb{I}_{K}$. Then
we have 
\begin{eqnarray*}
\mathbb{I}_{K} &=&\frac{1}{N}\left( \widehat{\mathbf{\beta }}^{PC}\right)
^{\prime }\widehat{\mathbf{\beta }}^{PC} \\
&=&\frac{1}{N}\left( \widehat{\mathbf{\beta }}^{PC}-\mathbf{\beta H+\beta H}%
\right) ^{\prime }\left( \widehat{\mathbf{\beta }}^{PC}-\mathbf{\beta
H+\beta H}\right) \\
&=&\mathbf{H}^{\prime }\left( \frac{1}{N}\mathbf{\beta }^{\prime }\mathbf{%
\beta }\right) \mathbf{H}+\mathbf{H}^{\prime }\frac{1}{N}\mathbf{\beta }%
^{\prime }\left( \widehat{\mathbf{\beta }}^{PC}-\mathbf{\beta H}\right) +%
\left[ \mathbf{H}^{\prime }\frac{1}{N}\mathbf{\beta }^{\prime }\left( 
\widehat{\mathbf{\beta }}^{PC}-\mathbf{\beta H}\right) \right] ^{\prime } \\
&&+\frac{1}{N}\left( \widehat{\mathbf{\beta }}^{PC}-\mathbf{\beta H}\right)
^{\prime }\left( \widehat{\mathbf{\beta }}^{PC}-\mathbf{\beta H}\right) \\
&=&\mathbf{H}^{\prime }\mathbf{H}+\mathbf{H}^{\prime }\frac{1}{N}\mathbf{%
\beta }^{\prime }\left( \widehat{\mathbf{\beta }}^{PC}-\mathbf{\beta H}%
\right) +\left[ \mathbf{H}^{\prime }\frac{1}{N}\mathbf{\beta }^{\prime
}\left( \widehat{\mathbf{\beta }}^{PC}-\mathbf{\beta H}\right) \right]
^{\prime } \\
&&+\frac{1}{N}\left( \widehat{\mathbf{\beta }}^{PC}-\mathbf{\beta H}\right)
^{\prime }\left( \widehat{\mathbf{\beta }}^{PC}-\mathbf{\beta H}\right) .
\end{eqnarray*}%
Using Lemma \ref{l2norm} repeatedly, it is easy to see that 
\begin{equation*}
\mathbf{H}^{\prime }\frac{1}{N}\mathbf{\beta }^{\prime }\left( \widehat{%
\mathbf{\beta }}^{PC}-\mathbf{\beta H}\right) +\left[ \mathbf{H}^{\prime }%
\frac{1}{N}\mathbf{\beta }^{\prime }\left( \widehat{\mathbf{\beta }}^{PC}-%
\mathbf{\beta H}\right) \right] ^{\prime }+\frac{1}{N}\left( \widehat{%
\mathbf{\beta }}^{PC}-\mathbf{\beta H}\right) ^{\prime }\left( \widehat{%
\mathbf{\beta }}^{PC}-\mathbf{\beta H}\right) =o_{a.s.}\left( 1\right) ,
\end{equation*}%
and therefore%
\begin{equation*}
\mathbf{H}^{\prime }\mathbf{H=}\mathbb{I}_{K}+o_{a.s.}\left( 1\right) ,
\end{equation*}%
so that $\mathbf{H}^{-1}=\mathbf{H}^{\prime }$. This proves that $\mathbf{H}$
is invertible.
\end{proof}
\end{lemma}

\begin{lemma}
\label{l2norm_lambda}We assume that Assumptions \ref{error}-\ref{exogeneity}%
, and \ref{loadings}-\ref{fact-idios} are satisfied. Then it holds that,
under $\mathbb{H}_{0}$ 
\begin{equation*}
\widehat{\lambda }^{PC}-\mathbf{H}^{-1}\lambda =o_{a.s.}\left( \frac{\left(
\log T\log N\right) ^{1+\epsilon }}{T^{1/2}}\right) +O\left( \frac{1}{N}%
\right) ,
\end{equation*}%
for all $\epsilon >0$.

\begin{proof}
By (\ref{lambda-pc}), under the null it holds that%
\begin{eqnarray}
\widehat{\lambda }^{PC} &=&\left( \frac{\widehat{\mathbf{\beta }}^{PC\prime }%
\mathbb{M}_{1_{N}}\widehat{\mathbf{\beta }}^{PC}}{N}\right) ^{-1}\left( 
\frac{\widehat{\mathbf{\beta }}^{PC\prime }\mathbb{M}_{1_{N}}\mathbf{\beta }%
}{N}\right) \lambda  \label{lam-pc} \\
&&+\left( \frac{\widehat{\mathbf{\beta }}^{PC\prime }\mathbb{M}_{1_{N}}%
\widehat{\mathbf{\beta }}^{PC}}{N}\right) ^{-1}\left( \frac{\widehat{\mathbf{%
\beta }}^{PC\prime }\mathbb{M}_{1_{N}}\mathbf{\beta }}{N}\right) \overline{v}
\notag \\
&&+\left( \frac{\widehat{\mathbf{\beta }}^{PC\prime }\mathbb{M}_{1_{N}}%
\widehat{\mathbf{\beta }}^{PC}}{N}\right) ^{-1}\left( \frac{\widehat{\mathbf{%
\beta }}^{PC\prime }\mathbb{M}_{1_{N}}\overline{\mathbf{u}}}{N}\right) 
\notag \\
&=&I+II+III.  \notag
\end{eqnarray}%
Note, to begin with, that%
\begin{eqnarray*}
&&\widehat{\mathbf{\beta }}^{PC\prime }\mathbb{M}_{1_{N}}\widehat{\mathbf{%
\beta }}^{PC} \\
&=&\mathbf{H}^{\prime }\mathbf{S}_{\beta }\mathbf{H}+\mathbf{H}^{\prime }%
\mathbf{\beta }^{\prime }\mathbb{M}_{1_{N}}\left( \widehat{\mathbf{\beta }}%
^{PC}-\mathbf{\beta H}\right) +\left( \widehat{\mathbf{\beta }}^{PC}-\mathbf{%
\beta H}\right) ^{\prime }\mathbb{M}_{1_{N}}\mathbf{\beta H} \\
&&+\left( \widehat{\mathbf{\beta }}^{PC}-\mathbf{\beta H}\right) ^{\prime }%
\mathbb{M}_{1_{N}}\left( \widehat{\mathbf{\beta }}^{PC}-\mathbf{\beta H}%
\right) ;
\end{eqnarray*}%
thus, using Lemma \ref{l2norm}, it is easy to see that%
\begin{equation*}
\frac{1}{N}\left\Vert \widehat{\mathbf{\beta }}^{PC\prime }\mathbb{M}_{1_{N}}%
\widehat{\mathbf{\beta }}^{PC}-\mathbf{H}^{\prime }\mathbf{S}_{\beta }%
\mathbf{H}\right\Vert =o_{a.s.}\left( \frac{\left( \log N\log T\right)
^{1+\epsilon }}{T^{1/2}}\right) +O\left( \frac{1}{N}\right) .
\end{equation*}%
Assumption \ref{loadings}\textit{(iii)} and Lemma \ref{matrix-h} guarantee
that $\mathbf{H}^{\prime }\mathbf{S}_{\beta }\mathbf{H}$ is invertible, and
therefore we may write 
\begin{equation}
\left\Vert \left( \frac{1}{N}\widehat{\mathbf{\beta }}^{PC\prime }\mathbb{M}%
_{1_{N}}\widehat{\mathbf{\beta }}^{PC}\right) ^{-1}\right\Vert
=O_{a.s.}\left( 1\right) .  \label{inv-den-pc}
\end{equation}%
Note now that, using (\ref{beta-pc-dec}) \textquotedblleft in
reverse\textquotedblright , viz.%
\begin{eqnarray}
\mathbf{\beta } &=&\widehat{\mathbf{\beta }}^{PC}\mathbf{H}^{-1}\mathbf{-}%
\frac{1}{N}\mathbf{\beta }\left( \frac{1}{T}\sum_{t=1}^{T}\widetilde{v}_{t}%
\widetilde{\mathbf{u}}_{t}^{\prime }\right) \widehat{\mathbf{\beta }}^{PC}%
\widehat{\Phi }^{-1}\mathbf{H}^{-1}  \label{b-pc-rev} \\
&&-\frac{1}{N}\left( \frac{1}{T}\sum_{t=1}^{T}\widetilde{\mathbf{u}}_{t}%
\widetilde{v}_{t}^{\prime }\right) \mathbf{\beta }^{\prime }\widehat{\mathbf{%
\beta }}^{PC}\widehat{\Phi }^{-1}\mathbf{H}^{-1}  \notag \\
&&-\left[ \frac{1}{NT}\sum_{t=1}^{T}\widetilde{\mathbf{u}}_{t}\widetilde{%
\mathbf{u}}_{t}^{\prime }\right] \widehat{\mathbf{\beta }}^{PC}\widehat{\Phi 
}^{-1}\mathbf{H}^{-1},  \notag
\end{eqnarray}%
we have%
\begin{eqnarray*}
&&\widehat{\mathbf{\beta }}^{PC\prime }\mathbb{M}_{1_{N}}\mathbf{\beta } \\
&\mathbf{=}&\widehat{\mathbf{\beta }}^{PC\prime }\mathbb{M}_{1_{N}}\widehat{%
\mathbf{\beta }}^{PC}\mathbf{H}^{-1}-\widehat{\mathbf{\beta }}^{PC\prime }%
\mathbb{M}_{1_{N}}\frac{1}{N}\mathbf{\beta }\left( \frac{1}{T}\sum_{t=1}^{T}%
\widetilde{v}_{t}\widetilde{\mathbf{u}}_{t}^{\prime }\right) \widehat{%
\mathbf{\beta }}^{PC}\widehat{\Phi }^{-1}\mathbf{H}^{-1} \\
&&-\widehat{\mathbf{\beta }}^{PC\prime }\mathbb{M}_{1_{N}}\frac{1}{N}\left( 
\frac{1}{T}\sum_{t=1}^{T}\widetilde{\mathbf{u}}_{t}\widetilde{v}_{t}^{\prime
}\right) \mathbf{\beta }^{\prime }\widehat{\mathbf{\beta }}^{PC}\widehat{%
\Phi }^{-1}\mathbf{H}^{-1}-\widehat{\mathbf{\beta }}^{PC\prime }\mathbb{M}%
_{1_{N}}\left[ \frac{1}{NT}\sum_{t=1}^{T}\widetilde{\mathbf{u}}_{t}%
\widetilde{\mathbf{u}}_{t}^{\prime }\right] \widehat{\mathbf{\beta }}^{PC}%
\widehat{\Phi }^{-1}\mathbf{H}^{-1}.
\end{eqnarray*}%
Following exactly the same steps as in the proof of Lemma \ref{l2norm}, it
can be shown that%
\begin{equation}
\frac{\widehat{\mathbf{\beta }}^{PC\prime }\mathbb{M}_{1_{N}}\mathbf{\beta }%
}{N}\mathbf{=}\frac{\widehat{\mathbf{\beta }}^{PC\prime }\mathbb{M}_{1_{N}}%
\widehat{\mathbf{\beta }}^{PC}}{N}\mathbf{H}^{-1}+o_{a.s.}\left( \frac{%
\left( \log N\log T\right) ^{1+\epsilon }}{T^{1/2}}\right) +O\left( \frac{1}{%
N}\right) ,  \label{beta-pc-ennesimo}
\end{equation}%
so that, in (\ref{lam-pc})%
\begin{equation*}
I=\mathbf{H}^{-1}\lambda +o_{a.s.}\left( \frac{\left( \log N\log T\right)
^{1+\epsilon }}{T^{1/2}}\right) +O\left( \frac{1}{N}\right) .
\end{equation*}%
Indeed, by the same token it also holds that 
\begin{equation}
\frac{\widehat{\mathbf{\beta }}^{PC\prime }\mathbb{M}_{1_{N}}\mathbf{\beta }%
}{N}\mathbf{=H}^{\prime }\mathbf{S}_{\beta }+o_{a.s.}\left( \frac{\left(
\log N\log T\right) ^{1+\epsilon }}{T^{1/2}}\right) +O\left( \frac{1}{N}%
\right) =O_{a.s.}\left( 1\right) .  \label{beta-pc-unaltro}
\end{equation}%
Recalling that $\overline{v}=o_{a.s.}\left( T^{-1/2}\left( \log T\right)
^{1+\epsilon }\right) $, using (\ref{inv-den-pc}) and (\ref{beta-pc-unaltro}%
) it follows that%
\begin{equation*}
II=o_{a.s.}\left( T^{-1/2}\left( \log T\right) ^{1+\epsilon }\right) .
\end{equation*}%
Finally, we study%
\begin{equation*}
\frac{\widehat{\mathbf{\beta }}^{PC\prime }\mathbb{M}_{1_{N}}\overline{%
\mathbf{u}}}{N}=\mathbf{H}^{\prime }\frac{\mathbf{\beta }\mathbb{M}_{1_{N}}%
\overline{\mathbf{u}}}{N}+\frac{\left( \widehat{\mathbf{\beta }}^{PC}-%
\mathbf{\beta H}\right) ^{\prime }\mathbb{M}_{1_{N}}\overline{\mathbf{u}}}{N}%
.
\end{equation*}%
We already know from the proof of (\ref{proof-a-caso}) that%
\begin{equation*}
\left\Vert \frac{\mathbf{\beta }^{\prime }\mathbb{M}_{1_{N}}\overline{%
\mathbf{u}}}{N}\right\Vert =o_{a.s.}\left( \frac{\left( \log N\log T\right)
^{1+\epsilon }}{\sqrt{NT}}\right) .
\end{equation*}%
Also, note that%
\begin{equation*}
\left\vert \frac{\left( \widehat{\mathbf{\beta }}^{PC}-\mathbf{\beta H}%
\right) ^{\prime }\mathbb{M}_{1_{N}}\overline{\mathbf{u}}}{N}\right\vert
\leq \left\Vert \frac{\widehat{\mathbf{\beta }}^{PC}-\mathbf{\beta H}}{%
N^{1/2}}\right\Vert \left\Vert \frac{\mathbb{M}_{1_{N}}\overline{\mathbf{u}}%
}{N^{1/2}}\right\Vert ;
\end{equation*}%
we know from Lemma \ref{l2norm} that 
\begin{equation*}
\left\Vert \frac{\widehat{\mathbf{\beta }}^{PC}-\mathbf{\beta H}}{N^{1/2}}%
\right\Vert =o_{a.s.}\left( \frac{\left( \log N\log T\right) ^{1+\epsilon }}{%
T^{1/2}}\right) +O\left( \frac{1}{N}\right) ,
\end{equation*}%
and, by standard passages%
\begin{eqnarray*}
\left\Vert \frac{\mathbb{M}_{1_{N}}\overline{\mathbf{u}}}{N^{1/2}}%
\right\Vert &\leq &\left\Vert \frac{\overline{\mathbf{u}}}{N^{1/2}}%
\right\Vert +\frac{1}{N}N^{1/2}\left\vert \frac{\sum_{i=1}^{N}\overline{u}%
_{i}}{N^{1/2}}\right\vert \\
&=&\frac{1}{N^{1/2}}\left( \sum_{i=1}^{N}\left( \frac{1}{T}%
\sum_{t=1}^{T}u_{i,t}\right) ^{2}\right) ^{1/2}+\left\vert \frac{1}{NT}%
\sum_{i=1}^{N}\sum_{t=1}^{T}u_{i,t}\right\vert \\
&=&o_{a.s.}\left( \frac{\left( \log N\log T\right) ^{1+\epsilon }}{\sqrt{T}}%
\right) +o_{a.s.}\left( \frac{\left( \log N\log T\right) ^{1+\epsilon }}{%
\sqrt{NT}}\right) .
\end{eqnarray*}%
The final result follows by putting all together.
\end{proof}
\end{lemma}

\begin{lemma}
\label{lambda-pc-hA}We assume that Assumptions \ref{error}-\ref{exogeneity},
and \ref{loadings}-\ref{fact-idios} are satisfied. Then it holds that, under 
$\mathbb{H}_{A}$%
\begin{equation*}
\widehat{\lambda }^{PC}-\mathbf{H}^{-1}\lambda =\frac{1}{N}\left( \mathbf{H}%
^{\prime }\mathbf{S}_{\beta }\mathbf{H}\right) ^{-1}\mathbf{H}^{\prime }%
\mathbf{\beta }^{\prime }\mathbb{M}_{1_{N}}\mathbf{\alpha }+o_{a.s.}\left( 
\frac{\left( \log T\log N\right) ^{1+\epsilon }}{T^{1/2}}\right) +O\left( 
\frac{1}{N}\right) ,
\end{equation*}%
for all $\epsilon >0$.

\begin{proof}
The proof follows by combining the arguments in Lemmas \ref{lambda-hA} and %
\ref{l2norm_lambda}.
\end{proof}
\end{lemma}

\begin{lemma}
\label{beta-i-numezzi}We assume that Assumptions \ref{error}-\ref{exogeneity}%
, and \ref{loadings}-\ref{fact-idios} are satisfied. Then it holds that%
\begin{equation*}
\sum_{i=1}^{N}\left\Vert \widehat{\beta }_{i}^{PC}-\mathbf{H}^{\prime }\beta
_{i}\right\Vert ^{\nu /2}=o_{a.s.}\left( NT^{-\nu /4}\left( \log N\log
T\right) ^{\left( 1+\epsilon \right) \nu /2}\right) +O\left( N^{1-\nu
/2}\right) ,
\end{equation*}%
for all $\epsilon >0$.

\begin{proof}
Using (\ref{beta-pc-scalar}), it holds that%
\begin{eqnarray*}
&&\sum_{i=1}^{N}\left\Vert \widehat{\beta }_{i}^{PC}-\mathbf{H}^{\prime
}\beta _{i}\right\Vert ^{\nu /2} \\
&\leq &c_{0}\left\Vert \widehat{\Phi }^{-1}\right\Vert ^{\nu /2}\left\Vert 
\widehat{\mathbf{\beta }}^{PC}\right\Vert ^{\nu /2}\left\Vert \frac{1}{NT}%
\sum_{t=1}^{T}\widetilde{\mathbf{u}}_{t}\widetilde{v}_{t}^{\prime
}\right\Vert ^{\nu /2}\sum_{i=1}^{N}\left\Vert \beta _{i}\right\Vert ^{\nu
/2} \\
&&+c_{0}\left\Vert \widehat{\Phi }^{-1}\right\Vert ^{\nu /2}\left\Vert 
\widehat{\mathbf{\beta }}^{PC}\right\Vert ^{\nu /2}\left\Vert \mathbf{\beta }%
\right\Vert ^{\nu /2}\sum_{i=1}^{N}\left\Vert \frac{1}{NT}\sum_{t=1}^{T}%
\widetilde{v}_{t}\widetilde{u}_{i,t}\right\Vert ^{\nu /2} \\
&&+c_{0}\left\Vert \widehat{\Phi }^{-1}\right\Vert ^{\nu
/2}\sum_{i=1}^{N}\left\Vert \frac{1}{NT}\sum_{t=1}^{T}\left( \widehat{%
\mathbf{\beta }}^{PC}\right) ^{\prime }\widetilde{\mathbf{u}}_{t}\widetilde{u%
}_{i,t}\right\Vert ^{\nu /2} \\
&=&I+II+III.
\end{eqnarray*}%
We begin by noting that, by Lemma \ref{inversion}, $\left\Vert \widehat{\Phi 
}^{-1}\right\Vert ^{\nu /2}=O_{a.s.}\left( 1\right) $; further, $\left\Vert 
\widehat{\mathbf{\beta }}^{PC}\right\Vert ^{\nu /2}=c_{0}N^{\nu /4}$ by
construction, and $\left\Vert \mathbf{\beta }\right\Vert ^{\nu
/2}=c_{0}N^{\nu /4}$ by the identification restriction $\mathbf{\beta }%
^{\prime }\mathbf{\beta }=N\mathbb{I}_{K}$. Finally, Assumption \ref%
{loadings} entails $\sum_{i=1}^{N}\left\Vert \beta _{i}\right\Vert ^{\nu
/2}=O\left( N\right) $. Consider now%
\begin{equation*}
\left\Vert \frac{1}{NT}\sum_{t=1}^{T}\widetilde{\mathbf{u}}_{t}\widetilde{v}%
_{t}^{\prime }\right\Vert ^{\nu /2}\leq \left\Vert \frac{1}{NT}\sum_{t=1}^{T}%
\mathbf{u}_{t}v_{t}^{\prime }\right\Vert ^{\nu /2}+\left\Vert \frac{1}{N}%
\overline{\mathbf{u}}\overline{v}\right\Vert ^{\nu /2}.
\end{equation*}%
We know from the above that%
\begin{equation*}
\left\Vert \frac{1}{NT}\sum_{t=1}^{T}\mathbf{u}_{t}v_{t}^{\prime
}\right\Vert =o_{a.s.}\left( N^{-1/2}T^{-1/2}\left( \log T\log N\right)
^{1+\epsilon }\right) ,
\end{equation*}%
and therefore 
\begin{equation*}
\left\Vert \frac{1}{NT}\sum_{t=1}^{T}\mathbf{u}_{t}v_{t}^{\prime
}\right\Vert ^{\nu /2}=o_{a.s.}\left( T^{-\nu /4}N^{-\nu /4}\left( \log
T\log N\right) ^{\left( 1+\epsilon \right) \nu /2}\right) .
\end{equation*}%
Also, seeing as (as shown above) it holds that $\overline{v}=o_{a.s.}\left(
T^{-1/2}\left( \log T\right) ^{1+\epsilon }\right) $ and $\left\Vert 
\overline{\mathbf{u}}\right\Vert =o_{a.s.}\left( N^{1/2}T^{-1/2}\left( \log
T\log N\right) ^{1+\epsilon }\right) $, we have%
\begin{equation*}
\left\Vert \frac{1}{N}\overline{\mathbf{u}}\overline{v}\right\Vert ^{\nu
/2}=o_{a.s.}\left( T^{-\nu /2}N^{-\nu /4}\left( \log T\log N\right) ^{\left(
1+\epsilon \right) \nu /2}\right) .
\end{equation*}%
Hence, combining all the results above, it follows that%
\begin{equation*}
I=o_{a.s.}\left( NT^{-\nu /4}\left( \log T\log N\right) ^{\left( 1+\epsilon
\right) \nu /2}\right) .
\end{equation*}%
Turning to $II$, note that%
\begin{equation*}
\sum_{i=1}^{N}\left\Vert \frac{1}{NT}\sum_{t=1}^{T}\widetilde{v}_{t}%
\widetilde{u}_{i,t}\right\Vert ^{\nu /2}\leq \sum_{i=1}^{N}\left\Vert \frac{1%
}{NT}\sum_{t=1}^{T}v_{t}u_{i,t}\right\Vert ^{\nu
/2}+\sum_{i=1}^{N}\left\Vert \frac{1}{N}\overline{v}\overline{u}%
_{i}\right\Vert ^{\nu /2}
\end{equation*}%
with%
\begin{equation*}
\sum_{i=1}^{N}\mathbb{E}\left\Vert \frac{1}{NT}\sum_{t=1}^{T}v_{t}u_{i,t}%
\right\Vert ^{\nu /2}=c_{0}N^{1-\nu /2}T^{-\nu /4},
\end{equation*}%
using Lemma \ref{summability}. Hence%
\begin{equation*}
\sum_{i=1}^{N}\left\Vert \frac{1}{NT}\sum_{t=1}^{T}v_{t}u_{i,t}\right\Vert
^{\nu /2}=o_{a.s.}\left( N^{1-\nu /2}T^{-\nu /4}\left( \log N\log T\right)
^{2+\epsilon }\right) ,
\end{equation*}%
and 
\begin{equation*}
II=o_{a.s.}\left( NT^{-\nu /4}\left( \log N\log T\right) ^{2+\epsilon
}\right) .
\end{equation*}%
Also, by exactly the same passages%
\begin{equation*}
\mathbb{E}\sum_{i=1}^{N}\left\Vert \frac{1}{N}\overline{u}_{i}\right\Vert
^{\nu /2}=c_{0}N^{1-\nu /2}T^{-\nu /4},
\end{equation*}%
and therefore%
\begin{equation*}
II=o_{a.s.}\left( NT^{-\nu /4}\left( \log N\log T\right) ^{2+\epsilon
}\right) ,
\end{equation*}%
so that $I$ dominated $II$. Finally, considering $III$ it holds that%
\begin{equation*}
\sum_{i=1}^{N}\left\Vert \frac{1}{NT}\sum_{t=1}^{T}\left( \widehat{\mathbf{%
\beta }}^{PC}\right) ^{\prime }\widetilde{\mathbf{u}}_{t}\widetilde{u}%
_{i,t}\right\Vert ^{\nu /2}\leq \sum_{i=1}^{N}\left\Vert \frac{1}{NT}%
\sum_{t=1}^{T}\left( \widehat{\mathbf{\beta }}^{PC}\right) ^{\prime }\mathbf{%
u}_{t}u_{i,t}\right\Vert ^{\nu /2}+\sum_{i=1}^{N}\left\Vert \frac{1}{NT}%
\sum_{t=1}^{T}\left( \widehat{\mathbf{\beta }}^{PC}\right) ^{\prime }%
\overline{\mathbf{u}}\overline{u}_{i}\right\Vert ^{\nu /2}.
\end{equation*}%
Consider%
\begin{equation*}
\sum_{i=1}^{N}\left\Vert \frac{1}{NT}\sum_{t=1}^{T}\left( \widehat{\mathbf{%
\beta }}^{PC}\right) ^{\prime }\widetilde{\mathbf{u}}_{t}u_{i,t}\right\Vert
^{\nu /2}\leq \sum_{i=1}^{N}\left\Vert \frac{1}{NT}\sum_{t=1}^{T}\left( 
\widehat{\mathbf{\beta }}^{PC}-\mathbf{\beta H}\right) ^{\prime }\mathbf{u}%
_{t}u_{i,t}\right\Vert ^{\nu /2}+\sum_{i=1}^{N}\left\Vert \frac{1}{NT}%
\sum_{t=1}^{T}\mathbf{H}^{\prime }\mathbf{\beta }^{\prime }\mathbf{u}%
_{t}u_{i,t}\right\Vert ^{\nu /2}.
\end{equation*}%
We have 
\begin{equation*}
\sum_{i=1}^{N}\left\Vert \frac{1}{NT}\sum_{t=1}^{T}\left( \widehat{\mathbf{%
\beta }}^{PC}-\mathbf{\beta H}\right) ^{\prime }\mathbf{u}%
_{t}u_{i,t}\right\Vert ^{\nu /2}\leq \left\Vert \widehat{\mathbf{\beta }}%
^{PC}-\mathbf{\beta H}\right\Vert ^{\nu /2}\sum_{i=1}^{N}\left\Vert \frac{1}{%
NT}\sum_{t=1}^{T}\mathbf{u}_{t}u_{i,t}\right\Vert ^{\nu /2}
\end{equation*}%
with%
\begin{equation*}
\sum_{i=1}^{N}\left\Vert \frac{1}{NT}\sum_{t=1}^{T}\mathbf{u}%
_{t}u_{i,t}\right\Vert ^{\nu /2}\leq \sum_{i=1}^{N}\left\Vert \frac{1}{NT}%
\sum_{t=1}^{T}\left( \mathbf{u}_{t}u_{i,t}-\mathbb{E}\left( \mathbf{u}%
_{0}u_{i,0}\right) \right) \right\Vert ^{\nu /2}+\sum_{i=1}^{N}\left\Vert 
\frac{1}{N}\mathbb{E}\left( \mathbf{u}_{0}u_{i,0}\right) \right\Vert ^{\nu
/2}.
\end{equation*}%
It holds that%
\begin{eqnarray*}
&&\left( NT\right) ^{-\nu /2}\sum_{i=1}^{N}\mathbb{E}\left\Vert
\sum_{t=1}^{T}\left( \mathbf{u}_{t}u_{i,t}-\mathbb{E}\left( \mathbf{u}%
_{0}u_{i,0}\right) \right) \right\Vert ^{\nu /2} \\
&=&\left( NT\right) ^{-\nu /2}\sum_{i=1}^{N}\mathbb{E}\left(
\sum_{j=1}^{N}\left( \sum_{t=1}^{T}\left( u_{j,t}u_{i,t}-\mathbb{E}\left(
u_{j,0}u_{i,0}\right) \right) \right) ^{2}\right) ^{\nu /4} \\
&\leq &c_{0}\left( NT\right) ^{-\nu /2}N^{\nu
/4-1}\sum_{i=1}^{N}\sum_{j=1}^{N}\mathbb{E}\left( \sum_{t=1}^{T}\left(
u_{j,t}u_{i,t}-\mathbb{E}\left( u_{j,0}u_{i,0}\right) \right) \right) ^{\nu
/2}\leq c_{0}T^{-\nu /4}N^{1-\nu /4},
\end{eqnarray*}%
and, since $\left\Vert \mathbb{E}\left( \mathbf{u}_{0}u_{i,0}\right)
\right\Vert \leq \left\Vert \mathbb{E}\left( \mathbf{u}_{0}u_{i,0}\right)
\right\Vert _{1}\leq c_{0}$ 
\begin{equation*}
\sum_{i=1}^{N}\left\Vert \frac{1}{N}\mathbb{E}\left( \mathbf{u}%
_{0}u_{i,0}\right) \right\Vert ^{\nu /2}=O\left( N^{1-\nu /2}\right) .
\end{equation*}%
Combining these results with Lemma \ref{l2norm}, it follows that%
\begin{eqnarray*}
&&\sum_{i=1}^{N}\left\Vert \frac{1}{NT}\sum_{t=1}^{T}\left( \widehat{\mathbf{%
\beta }}^{PC}-\mathbf{\beta H}\right) ^{\prime }\mathbf{u}%
_{t}u_{i,t}\right\Vert ^{\nu /2} \\
&=&o_{a.s.}\left[ \left( N^{\nu /4}T^{-\nu /4}\left( \log N\log T\right)
^{(1+\epsilon )\nu /2}+N^{-\nu /4}\right) \left( \left( T^{-\nu /4}N^{1-\nu
/4}\left( \log N\log T\right) ^{1+\epsilon }\right) +N^{1-\nu /2}\right) %
\right] .
\end{eqnarray*}%
Also note that%
\begin{equation*}
\sum_{i=1}^{N}\left\Vert \frac{1}{NT}\sum_{t=1}^{T}\mathbf{H}^{\prime }%
\mathbf{\beta }^{\prime }\mathbf{u}_{t}u_{i,t}\right\Vert ^{\nu
/2}=O_{a.s.}\left( 1\right) \sum_{i=1}^{N}\left\Vert \frac{1}{NT}%
\sum_{t=1}^{T}\mathbf{\beta }^{\prime }\mathbf{u}_{t}u_{i,t}\right\Vert
^{\nu /2},
\end{equation*}%
and%
\begin{equation*}
\sum_{i=1}^{N}\left\Vert \frac{1}{NT}\sum_{t=1}^{T}\mathbf{\beta }^{\prime }%
\mathbf{u}_{t}u_{i,t}\right\Vert ^{\nu /2}\leq \sum_{i=1}^{N}\left\Vert 
\frac{1}{NT}\sum_{t=1}^{T}\mathbf{\beta }^{\prime }\left( \mathbf{u}%
_{t}u_{i,t}-\mathbb{E}\left( \mathbf{u}_{0}u_{i,0}\right) \right)
\right\Vert ^{\nu /2}+\sum_{i=1}^{N}\left\Vert \frac{1}{NT}\sum_{t=1}^{T}%
\mathbf{\beta }^{\prime }\mathbb{E}\left( \mathbf{u}_{0}u_{i,0}\right)
\right\Vert ^{\nu /2}
\end{equation*}%
It is immediate to see that, by Assumption \ref{idiosyncratic}\textit{(vi)} 
\begin{equation*}
\sum_{i=1}^{N}\mathbb{E}\left\Vert \frac{1}{NT}\sum_{t=1}^{T}\mathbf{\beta }%
^{\prime }\left( \mathbf{u}_{t}u_{i,t}-\mathbb{E}\left( \mathbf{u}%
_{0}u_{i,0}\right) \right) \right\Vert ^{\nu /2}\leq c_{0}N\left( NT\right)
^{-\nu /4}.
\end{equation*}%
Also%
\begin{eqnarray*}
&&\sum_{i=1}^{N}\left\Vert \frac{1}{NT}\sum_{t=1}^{T}\mathbf{\beta }^{\prime
}\mathbb{E}\left( \mathbf{u}_{0}u_{i,0}\right) \right\Vert ^{\nu /2} \\
&=&\sum_{i=1}^{N}\left\Vert \frac{1}{N}\mathbf{\beta }^{\prime }\mathbb{E}%
\left( \mathbf{u}_{0}u_{i,0}\right) \right\Vert ^{\nu
/2}=\sum_{i=1}^{N}\left\Vert \frac{1}{N}\sum_{j=1}^{N}\beta _{j}\mathbb{E}%
\left( u_{j,0}u_{i,0}\right) \right\Vert ^{\nu /2} \\
&\leq &\sum_{i=1}^{N}\left\Vert \max_{1\leq j\leq N}\left\vert \beta
_{j}\right\vert \frac{1}{N}\sum_{j=1}^{N}\left\vert \mathbb{E}\left(
u_{j,0}u_{i,0}\right) \right\vert \right\Vert ^{\nu /2}\leq c_{0}N^{1-\nu
/2},
\end{eqnarray*}%
by virtue of Assumption \ref{idiosyncratic}\textit{(iv)}. By the same logic
as above, it can be shown that the term $\sum_{i=1}^{N}\left\Vert \left(
NT\right) ^{-1}\sum_{t=1}^{T}\left( \widehat{\mathbf{\beta }}^{PC}\right)
^{\prime }\overline{\mathbf{u}}\overline{u}_{i}\right\Vert ^{\nu /2}$\ is
dominated. Putting all together, the final result obtains.
\end{proof}
\end{lemma}

\begin{lemma}
\label{normalising_PC} We assume that Assumptions \ref{error}-\ref%
{exogeneity}, and \ref{loadings}-\ref{fact-idios} are satisfied. Then it
holds that%
\begin{eqnarray*}
\liminf_{\min \left\{ N,T\right\} \rightarrow \infty }\frac{1}{NT}%
\sum_{i=1}^{N}\sum_{t=1}^{T}\left( \widehat{u}_{i,t}^{PC}\right) ^{2} &>&0,
\\
\limsup_{\min \left\{ N,T\right\} \rightarrow \infty }\frac{1}{NT}%
\sum_{i=1}^{N}\sum_{t=1}^{T}\left( \widehat{u}_{i,t}^{PC}\right) ^{2}
&<&\infty .
\end{eqnarray*}

\begin{proof}
We let $K=1$ for simplicity and without loss of generality. Let $\widehat{%
\mathbf{u}}_{t}^{PC}=\left( \widehat{u}_{1,t}^{PC},...,\widehat{u}%
_{N,t}^{PC}\right) ^{\prime }$. It holds that%
\begin{eqnarray*}
&&\widehat{\mathbf{u}}_{t}^{PC} \\
&=&\widetilde{\mathbf{y}}_{t}-\widehat{\mathbf{\beta }}^{PC}\widehat{f}%
_{t}^{PC}=\mathbf{\beta }\widetilde{v}_{t}+\widetilde{\mathbf{u}}_{t}-%
\widehat{\mathbf{\beta }}^{PC}\widehat{f}_{t}^{PC} \\
&=&\mathbf{\beta }\widetilde{v}_{t}+\widetilde{\mathbf{u}}_{t}-\frac{1}{N}%
\widehat{\mathbf{\beta }}^{PC}\widehat{\mathbf{\beta }}^{PC\prime }\left( 
\mathbf{\beta }\widetilde{v}_{t}+\widetilde{\mathbf{u}}_{t}\right) \\
&=&\mathbf{\beta }\widetilde{v}_{t}+\widetilde{\mathbf{u}}_{t}-\frac{1}{N}%
\left[ \mathbf{\beta H+}\left( \widehat{\mathbf{\beta }}^{PC}-\mathbf{\beta H%
}\right) \right] \widehat{\mathbf{\beta }}^{PC\prime }\left[ \widehat{%
\mathbf{\beta }}^{PC}\mathbf{H}^{-1}+\mathbf{\beta -}\widehat{\mathbf{\beta }%
}^{PC}\mathbf{H}^{-1}\right] \widetilde{v}_{t} \\
&&-\frac{1}{N}\widehat{\mathbf{\beta }}^{PC}\mathbf{H}^{\prime }\mathbf{%
\beta }^{\prime }\widetilde{\mathbf{u}}_{t}-\frac{1}{N}\widehat{\mathbf{%
\beta }}^{PC}\left[ \widehat{\mathbf{\beta }}^{PC}-\mathbf{\beta H}\right]
^{\prime }\widetilde{\mathbf{u}}_{t} \\
&=&\widetilde{\mathbf{u}}_{t}-\left( \widehat{\mathbf{\beta }}^{PC}-\mathbf{%
\beta H}\right) \mathbf{H}^{-1}\widetilde{v}_{t}-\frac{1}{N}\mathbf{\beta H}%
\widehat{\mathbf{\beta }}^{PC\prime }\left( \mathbf{\beta -}\widehat{\mathbf{%
\beta }}^{PC}\mathbf{H}^{-1}\right) \widetilde{v}_{t} \\
&&-\frac{1}{N}\left( \widehat{\mathbf{\beta }}^{PC}-\mathbf{\beta H}\right) 
\widehat{\mathbf{\beta }}^{PC\prime }\left( \mathbf{\beta -}\widehat{\mathbf{%
\beta }}^{PC}\mathbf{H}^{-1}\right) \widetilde{v}_{t}-\frac{1}{N}\widehat{%
\mathbf{\beta }}^{PC}\mathbf{H}^{\prime }\mathbf{\beta }^{\prime }\widetilde{%
\mathbf{u}}_{t} \\
&&-\frac{1}{N}\widehat{\mathbf{\beta }}^{PC}\left[ \widehat{\mathbf{\beta }}%
^{PC}-\mathbf{\beta H}\right] ^{\prime }\widetilde{\mathbf{u}}_{t}.
\end{eqnarray*}%
We bound the following terms:%
\begin{eqnarray*}
&&\frac{1}{NT}\sum_{t=1}^{T}\left\Vert \left( \widehat{\mathbf{\beta }}^{PC}-%
\mathbf{\beta H}\right) \mathbf{H}^{-1}\widetilde{v}_{t}\right\Vert ^{2} \\
&\leq &c_{0}\left( \frac{1}{T}\sum_{t=1}^{T}\widetilde{v}_{t}^{2}\right) 
\frac{\left\Vert \widehat{\mathbf{\beta }}^{PC}-\mathbf{\beta H}\right\Vert
^{2}}{N}=o_{a.s.}\left( \frac{\left( \log N\log T\right) ^{2+\epsilon }}{T}%
\right) +O\left( \frac{1}{N^{2}}\right) ,
\end{eqnarray*}%
by Lemma \ref{l2norm}; 
\begin{eqnarray*}
&&\frac{1}{NT}\sum_{t=1}^{T}\left\Vert \frac{1}{N}\mathbf{\beta H}\widehat{%
\mathbf{\beta }}^{PC\prime }\left( \mathbf{\beta -}\widehat{\mathbf{\beta }}%
^{PC}\mathbf{H}^{-1}\right) \widetilde{v}_{t}\right\Vert ^{2} \\
&\leq &c_{0}\left( \frac{1}{T}\sum_{t=1}^{T}\widetilde{v}_{t}^{2}\right) 
\frac{\left\Vert \mathbf{\beta }\right\Vert ^{2}}{N}\frac{\left\Vert 
\widehat{\mathbf{\beta }}^{PC\prime }\right\Vert ^{2}}{N}\frac{\left\Vert 
\mathbf{\beta -}\widehat{\mathbf{\beta }}^{PC}\mathbf{H}^{-1}\right\Vert ^{2}%
}{N} \\
&=&o_{a.s.}\left( \frac{\left( \log N\log T\right) ^{2+\epsilon }}{T}\right)
+O\left( \frac{1}{N^{2}}\right) ,
\end{eqnarray*}%
again using Lemma \ref{l2norm};%
\begin{eqnarray*}
&&\frac{1}{NT}\sum_{t=1}^{T}\left\Vert \frac{1}{N}\left( \widehat{\mathbf{%
\beta }}^{PC}-\mathbf{\beta H}\right) \widehat{\mathbf{\beta }}^{PC\prime
}\left( \mathbf{\beta -}\widehat{\mathbf{\beta }}^{PC}\mathbf{H}^{-1}\right) 
\widetilde{v}_{t}\right\Vert ^{2} \\
&\leq &c_{0}\left( \frac{1}{T}\sum_{t=1}^{T}\widetilde{v}_{t}^{2}\right)
\left( \frac{\left\Vert \mathbf{\beta -}\widehat{\mathbf{\beta }}^{PC}%
\mathbf{H}^{-1}\right\Vert ^{2}}{N}\right) ^{2}\frac{\left\Vert \widehat{%
\mathbf{\beta }}^{PC\prime }\right\Vert ^{2}}{N},
\end{eqnarray*}%
and therefore it is dominated by the previous terms;%
\begin{equation*}
\mathbb{E}\frac{1}{NT}\sum_{t=1}^{T}\left\Vert \frac{1}{N}\widehat{\mathbf{%
\beta }}^{PC}\mathbf{H}^{\prime }\mathbf{\beta }^{\prime }\widetilde{\mathbf{%
u}}_{t}\right\Vert ^{2}\leq c_{0}\frac{1}{N^{2}T}\sum_{t=1}^{T}\mathbb{E}%
\left( \sum_{j=1}^{N}\beta _{j}\widetilde{u}_{j,t}\right) ^{2}\leq
c_{1}N^{-1},
\end{equation*}%
so that this term is bounded by $o_{a.s.}\left( N^{-1}\left( \log T\log
^{2}N\right) ^{2+\epsilon }\right) $;%
\begin{eqnarray*}
&&\frac{1}{NT}\sum_{t=1}^{T}\left\Vert \frac{1}{N}\widehat{\mathbf{\beta }}%
^{PC}\left[ \widehat{\mathbf{\beta }}^{PC}-\mathbf{\beta H}\right] ^{\prime }%
\widetilde{\mathbf{u}}_{t}\right\Vert ^{2} \\
&=&\frac{1}{N^{3}T}\sum_{t=1}^{T}\widetilde{\mathbf{u}}_{t}^{\prime }\left[ 
\widehat{\mathbf{\beta }}^{PC}-\mathbf{\beta H}\right] \widehat{\mathbf{%
\beta }}^{PC\prime }\widehat{\mathbf{\beta }}^{PC}\left[ \widehat{\mathbf{%
\beta }}^{PC}-\mathbf{\beta H}\right] ^{\prime }\widetilde{\mathbf{u}}_{t} \\
&=&\frac{1}{N^{2}T}\sum_{t=1}^{T}\widetilde{\mathbf{u}}_{t}^{\prime }\left[ 
\widehat{\mathbf{\beta }}^{PC}-\mathbf{\beta H}\right] \left[ \widehat{%
\mathbf{\beta }}^{PC}-\mathbf{\beta H}\right] ^{\prime }\widetilde{\mathbf{u}%
}_{t} \\
&\leq &\left( \frac{1}{NT}\sum_{t=1}^{T}\left\Vert \widetilde{\mathbf{u}}%
_{t}\right\Vert ^{2}\right) \frac{\left\Vert \mathbf{\beta -}\widehat{%
\mathbf{\beta }}^{PC}\mathbf{H}^{-1}\right\Vert ^{2}}{N};
\end{eqnarray*}%
it is not hard to see that%
\begin{equation*}
\frac{1}{NT}\sum_{t=1}^{T}\left\Vert \widetilde{\mathbf{u}}_{t}\right\Vert
^{2}=\frac{1}{NT}\sum_{t=1}^{T}\mathbb{E}\left( \sum_{i=1}^{N}\widetilde{u}%
_{i,t}^{2}\right) =O_{a.s.}\left( 1\right) ,
\end{equation*}%
and therefore%
\begin{equation*}
\frac{1}{NT}\sum_{t=1}^{T}\left\Vert \frac{1}{N}\widehat{\mathbf{\beta }}%
^{PC}\left[ \widehat{\mathbf{\beta }}^{PC}-\mathbf{\beta H}\right] ^{\prime }%
\widetilde{\mathbf{u}}_{t}\right\Vert ^{2}=o_{a.s.}\left( \frac{\left( \log
N\log T\right) ^{2+\epsilon }}{T}\right) +o_{a.s.}\left( \frac{\left( \log
T\log ^{2}N\right) ^{2+\epsilon }}{N}\right) .
\end{equation*}%
After some algebra and repeated use of the Cauchy-Schwartz inequality, the
above entails that%
\begin{equation*}
\widehat{s}_{NT}^{PC}=\frac{1}{NT}\sum_{t=1}^{T}\left\Vert \widehat{\mathbf{u%
}}_{t}^{PC}\right\Vert ^{2}=\frac{1}{NT}\sum_{t=1}^{T}\left\Vert \widetilde{%
\mathbf{u}}_{t}\right\Vert ^{2}+o_{a.s.}\left( \frac{\left( \log N\log
T\right) ^{2+\epsilon }}{T}\right) +O\left( \frac{1}{N^{2}}\right) .
\end{equation*}%
From hereon, the proof follows by similar arguments as above.
\end{proof}
\end{lemma}

\begin{lemma}
\label{psi-pc}We assume that Assumptions \ref{error}-\ref{exogeneity}, and %
\ref{loadings}-\ref{fact-idios} are satisfied. Then it holds that%
\begin{equation*}
\sum_{i=1}^{N}\psi _{i,NT}^{PC}=o_{a.s.}\left( 1\right) .
\end{equation*}

\begin{proof}
The proof is essentially the same as the proof of Lemmas \ref{psi} and \ref%
{psi-fm}, and we simply discuss the different parts.

Under both $\mathbb{H}_{0}$ and $\mathbb{H}_{A}$, it holds that%
\begin{equation}
\widehat{\alpha }_{i}^{PC}=\alpha _{i}+\beta _{i}^{\prime }\overline{v}+%
\overline{u}_{i}-\left( \widehat{\beta }_{i}-\mathbf{H}^{\prime }\beta
_{i}\right) ^{\prime }\lambda -\beta _{i}^{\prime }\mathbf{H}\left( \widehat{%
\lambda }-\mathbf{H}^{-1}\lambda \right) -\left( \widehat{\beta }_{i}-%
\mathbf{H}^{\prime }\beta _{i}\right) ^{\prime }\left( \widehat{\lambda }-%
\mathbf{H}^{-1}\lambda \right) .  \label{alpha-PC}
\end{equation}%
Under $\mathbb{H}_{0}$ we have 
\begin{eqnarray*}
&&\sum_{i=1}^{N}\psi _{i,NT}^{PC} \\
&=&\frac{C_{NT}^{1/2}}{\left\vert \widehat{s}_{NT}^{PC}\right\vert ^{\nu /2}}%
\sum_{i=1}^{N}\left\vert \widehat{\alpha }_{i}^{PC}\right\vert ^{\nu /2} \\
&\leq &c_{0}\left[ \frac{C _{NT}^{1/2}}{\left\vert \widehat{s}%
_{NT}^{PC}\right\vert ^{\nu /2}}\left( \sum_{i=1}^{N}\left\Vert \beta
_{i}\right\Vert ^{\nu /2}\right) \left\Vert \overline{v}\right\Vert ^{\nu
/2}+\frac{C _{NT}^{1/2}}{\left\vert \widehat{s}_{NT}^{PC}\right\vert ^{\nu
/2}}\sum_{i=1}^{N}\left\vert \overline{u}_{i}\right\vert ^{\nu /2}+\frac{%
C_{NT}^{1/2}}{\left\vert \widehat{s}_{NT}^{PC}\right\vert ^{\nu /2}}\left(
\sum_{i=1}^{N}\left\Vert \widehat{\beta }_{i}-\mathbf{H}^{\prime }\beta
_{i}\right\Vert ^{\nu /2}\right) \left\Vert \lambda \right\Vert ^{\nu
/2}\right. \\
&&\left. +\frac{C _{NT}^{1/2}}{\left\vert \widehat{s}_{NT}^{PC}\right\vert
^{\nu /2}}\left( \sum_{i=1}^{N}\left\Vert \beta _{i}\right\Vert ^{\nu
/2}\right) \left\Vert \widehat{\lambda }-\mathbf{H}^{-1}\lambda \right\Vert
^{\nu /2}+\frac{C _{NT}^{1/2}}{\left\vert \widehat{s}_{NT}^{PC}\right\vert
^{\nu /2}}\left( \sum_{i=1}^{N}\left\Vert \widehat{\beta }_{i}-\mathbf{H}%
^{\prime }\beta _{i}\right\Vert ^{\nu /2}\right) \left\Vert \widehat{\lambda 
}-\mathbf{H}^{-1}\lambda \right\Vert ^{\nu /2}\right] \\
&=&I+II+III+IV+V.
\end{eqnarray*}%
Starting from Lemmas \ref{l2norm_lambda} and \ref{normalising_PC}, we can
show that $I$, $II$ and $IV$ are $o_{a.s.}(1)$ proceeding as in the proof of
Lemma \ref{psi-fm} - indeed, on account of Lemma \ref{lam-pc} $IV$ contains
the extra term%
\begin{equation*}
\frac{C_{NT}^{1/2}}{\left\vert \widehat{s}_{NT}^{PC}\right\vert ^{\nu /2}}%
\left( \sum_{i=1}^{N}\left\Vert \beta _{i}\right\Vert ^{\nu /2}\right)
O\left( N^{-\nu /2}\right) =O\left(C _{NT}N^{1-\nu /2}\right) ,
\end{equation*}%
but this can be shown to be $o(1)$ by routine calculations. As far as $III$
is concerned, Assumption \ref{beta-bar}, and Lemmas \ref{beta-i-numezzi}-\ref%
{normalising_PC} imply, after some algebra 
\begin{equation*}
III=C _{NT}^{1/2}\left\{ o_{a.s.}\left( NT^{-\nu /4}\left( \log N\log
T\right) ^{\left( 1+\epsilon \right) \nu /2}\right) +O\left( N^{1-\nu
/2}\right) \right\} =o_{a.s.}(1),
\end{equation*}%
where recall $\nu \geq 4$. That $V=o_{a.s.}(1)$ readily follows from Lemma \ref%
{l2norm_lambda} and the result on $III$.
\end{proof}
\end{lemma}

\clearpage\newpage

\renewcommand*{\thesection}{\Alph{section}}

\setcounter{equation}{0} \setcounter{lemma}{0} \setcounter{theorem}{0} %
\renewcommand{\theassumption}{C.\arabic{assumption}} 
\renewcommand{\thetheorem}{D.\arabic{theorem}} \renewcommand{\thelemma}{D.%
\arabic{lemma}} \renewcommand{\theproposition}{C.\arabic{proposition}} %
\renewcommand{\thecorollary}{D.\arabic{corollary}} \renewcommand{%
\theequation}{D.\arabic{equation}} \renewcommand{\theremark}{D.%
\arabic{remark}}

\section{Proofs\label{proofs}}

\begin{proof}[Proof of Theorem \protect\ref{asy-max}]
We begin by proving (\ref{th-null}). The proof follows a similar approach to
the proof of Theorem 3 in \citet{he2024online}, which we refine. 
To begin with, note that, for all $-\infty <x<\infty $ 
\begin{equation*}
\mathbb{P}^{\ast }\left( \frac{Z_{N,T}-b_{N}}{a_{N}}\leq x\right) =P^{\ast
}\left( Z_{N,T}\leq a_{N}x+b_{N}\right) ,
\end{equation*}
where recall that $z_{i,NT}=\psi _{i,NT}+\omega _{i}$. Seeing as $\omega
_{i} $ is, by construction, independent across $i$ and independent of the
sample, it follows that%
\begin{equation*}
\mathbb{P}^{\ast }\left( Z_{N,T}\leq a_{N}x+b_{N}\right)
=\dprod\limits_{i=1}^{N}\mathbb{P}^{\ast }\left( z_{i,NT}\leq
a_{N}x+b_{N}\right) =\dprod\limits_{i=1}^{N}\mathbb{P}^{\ast }\left( \omega
_{i}\leq a_{N}x+b_{N}-\psi _{i,NT}\right) .
\end{equation*}%
Let $\Phi \left( \cdot \right) $ denote the standard normal distribution; we
have 
\begin{equation}  \label{eq:prod_Probs}
\dprod\limits_{i=1}^{N}\mathbb{P}^{\ast }\left( \omega _{i}\leq
a_{N}x+b_{N}-\psi _{i,NT}\right) =\exp \left( \sum_{i=1}^{N}\log \Phi \left(
a_{N}x+b_{N}-\psi _{i,NT}\right) \right) .
\end{equation}
Note now that 
\begin{equation}  \label{eq:log_CDF}
\log \Phi \left( a_{N}x+b_{N}-\psi _{i,NT}\right) =\log \Phi \left(
a_{N}x+b_{N}\right) +\log \frac{\Phi \left( a_{N}x+b_{N}-\psi _{i,NT}\right) 
}{\Phi \left( a_{N}x+b_{N}\right) };
\end{equation}
using Lagrange's theorem, there exists an $a_{i}^{\ast }\in \left(
a_{N}x+b_{N}-\psi _{i,NT},a_{N}x+b_{N}\right) $ such that $\Phi \left(
a_{N}x+b_{N}-\psi _{i,NT}\right) $ $=$ $\Phi \left( a_{N}x+b_{N}\right) $ $-$
$\varphi \left( a_{i}^{\ast }\right) \psi_{i,NT}$, where $\varphi \left(
\cdot \right) $ denotes the density function of the standard normal, so that
ultimately%
\begin{equation*}
\log \frac{\Phi \left( a_{N}x+b_{N}-\psi _{i,NT}\right) }{\Phi \left(
a_{N}x+b_{N}\right) }=\log \left( 1-\frac{\varphi \left( a_{i}^{\ast
}\right) }{\Phi \left( a_{N}x+b_{N}\right) }\psi _{i,NT}\right) =\log \left(
1-c_{i}\psi _{i,NT}\right) .
\end{equation*}%
By elementary arguments, it follows that 
\begin{eqnarray*}
&&\exp \left( \sum_{i=1}^{N}\log \frac{\Phi \left( a_{N}x+b_{N}-\psi
_{i,NT}\right) }{\Phi \left( a_{N}x+b_{N}\right) }\right) \\
&=&\exp \left( \sum_{i=1}^{N}\log \left( 1-c_{i}\psi_{i,NT}\right) \right) =
\exp \left(\frac{N}{N}\log\left(\dprod_{i=1}^{N}\left(1 -
c_{i}\psi_{i,NT}\right)\right)\right) \\
&\leq &\exp \left( N\log \left( \frac{1}{N}\sum_{i=1}^{N}\left(1-c_{i}%
\psi_{i,T}\right) \right) \right)=\exp \left( N\log \left( 1-\left( \frac{1}{%
N}\sum_{i=1}^{N}c_{i}\psi_{i,NT}\right) \right) \right) \\
&=&\exp \left( \sum_{h=1}^{\infty }N^{-h+1}\frac{\left(-1\right) ^{h}\left(
\sum_{i=1}^{N}c_{i}\psi _{i,NT}\right) ^{h}}{h}\right),
\end{eqnarray*}
having used the arithmetic/geometric mean inequality to move from the second
to the third line, and a Taylor expansion of $\log(1+x)$ around $x=0$ in the
last line. Since $c_{i}\leq \left( 2\pi \right) ^{-1/2}\left[ \Phi
\left(a_{N}x+b_{N}\right) \right] ^{-1}\leq \overline{c}$, and, by Lemma \ref%
{psi}, 
\begin{equation*}
\mathbb{P}\left( \omega :\lim_{\min \left\{ N,T\right\} \rightarrow \infty
}\sum_{i=1}^{N}\psi _{i,NT}=0\right) =1,
\end{equation*}
we can assume that $\lim_{\min \left\{ N,T\right\} \rightarrow \infty
}\sum_{i=1}^{N}\psi _{i,NT}=0$, it follows from elementary arguments that%
\begin{equation}
\lim_{\min \left\{ N,T\right\} \rightarrow \infty }\exp \left(
\sum_{i=1}^{N}\log \frac{\Phi \left( a_{N}x+b_{N}-\psi_{i,NT}\right) }{\Phi
\left( a_{N}x+b_{N}\right) }\right) =1.  \label{exp-remainder}
\end{equation}%
Thus we have%
\begin{eqnarray}
&&\lim_{\min \left\{ N,T\right\} \rightarrow \infty
}\dprod\limits_{i=1}^{N}P^{\ast }\left( \omega _{i}\leq
a_{N}x+b_{N}-\psi_{i,NT}\right) \\
&=&\left( \lim_{N\rightarrow \infty }\Phi ^{N}\left( a_{N}x+b_{N}\right)
\right) \times \left( \lim_{\min \left\{ N,T\right\} \rightarrow \infty
}\exp \left( \sum_{i=1}^{N}\log \frac{\Phi \left( a_{N}x+b_{N}-\psi
_{i,T}\right) }{\Phi \left( a_{N}x+b_{N}\right) }\right) \right) \\
&=&\exp \left( -\exp \left( -x\right) \right) ,
\end{eqnarray}
using the relations in \eqref{eq:prod_Probs} - \eqref{eq:log_CDF} to move
from the first to the second line, and the Fisher--Tippett--Gnedenko Theorem
(see Theorem 3.2.3 in \citealp{embrechts2013modelling}, among others) along
with the limit in (\ref{exp-remainder}) to obtain the final result.

We now turn to showing (\ref{th-alt}). Under the alternative, there exists a
set of $1\leq m\leq N$ indices $\mathcal{I}=\left\{ i_{1},\dots
,i_{m}\right\} \subseteq \left\{ 1,\dots ,N\right\} $ such that $\left\vert
\alpha _{i}\right\vert >0$ whenever $i\in \mathcal{I}$; in these cases, $%
\psi _{i,NT}$ diverges almost surely at the rate $T^{1/2}$, i.e.%
\begin{equation*}
T^{-1/2}\psi _{i,NT}\overset{a.s.}{\rightarrow }c>0\text{ whenever }i\in 
\mathcal{I}\text{.}
\end{equation*}%
Hence, we can assume that%
\begin{equation*}
\lim_{T\rightarrow \infty }T^{-1/2}\psi _{i,NT}=c>0,
\end{equation*}%
whenever $i\in \mathcal{I}$. Note now that, for any $-\infty <x<\infty $ we
have 
\begin{equation*}
\begin{aligned} P^{\ast }\left( Z_{N,T}\leq a_{N}x+b_{N}\right)&
=\dprod\limits_{i=1}^{N}P^{\ast }\left( z_{i,T}\leq
a_{N}x+b_{N}\right)\leq\dprod\limits_{i\in\mathcal{I}}P^{\ast }\left( \omega
_{i}\leq a_{N}x+b_{N}-\psi_{i,NT}\right)\\ & =
\dprod_{i\in\mathcal{I}}\Phi\left(a_{N}x+b_{N} - \psi_{i,NT} \right).
\end{aligned}
\end{equation*}%
Equation (5) in \cite{borjesson1979simple} entails that 
\begin{equation}
\Phi \left( a_{N}x+b_{N}-\psi _{i,NT}\right) \leq \frac{\mathrm{exp}\left( -%
\frac{1}{2}\left( a_{N}x+b_{N}-\psi _{i,NT}\right) ^{2}\right) }{\sqrt{2\pi }%
\left\vert a_{N}x+b_{N}-\psi _{i,NT}\right\vert }.
\label{eq:upper_bound_gaussian_cdf}
\end{equation}%
Seeing as, by construction, $a_{N}x+b_{N}=O(\sqrt{2\log N})$\ for each $%
-\infty <x<\infty $, by Assumption \ref{asymptotics} it follows that $%
a_{N}x+b_{N}-\psi _{i,NT}\overset{a.s.}{\rightarrow }-\infty $ whenever $%
i\in \mathcal{I}$. Hence, as $\min \left\{ N,T\right\} \rightarrow \infty $
it holds that%
\begin{equation*}
0\leq \Phi \left( a_{N}x+b_{N}-\psi _{i,NT}\right) \leq \frac{\mathrm{exp}%
\left( -\frac{1}{2}\left( a_{N}x+b_{N}-\psi _{i,NT}\right) ^{2}\right) }{%
\sqrt{2\pi }\left\vert a_{N}x+b_{N}-\psi _{i,NT}\right\vert }\overset{a.s.}{%
\rightarrow }0,
\end{equation*}%
which, by dominated convergence, entails that $\Phi \left( a_{N}x+b_{N}-\psi
_{i,T}\right) =o_{a.s.}(1)$. As long as $\mathcal{I}$ is not empty, this
immediately entails that 
\begin{equation*}
\lim_{\min \left\{ N,T\right\} \rightarrow \infty }\mathbb{P}^{\ast }\left(
Z_{N,T}\leq a_{N}x+b_{N}\right) =0,
\end{equation*}%
for almost all realisations of $\left\{ \left( u_{i,t},f_{t}^{\prime
}\right) ^{\prime },1\leq i\leq N,1\leq t\leq T\right\} $. 

In conclusion, we note that the theorem still holds, with the proof
virtually unchanged for the more general statistic in \eqref{psi-delta}. The
only difference consists in replacing $T^{-1/2}$ with $T^{-\delta \nu /2}$
when discussing the behavior under the alternative. No further calculations
or arguments are required with respect to the current proof.
\end{proof}

\begin{proof}[Proof of Theorem \protect\ref{strong-rule}]
Write, for short, $Q_{N,T,B}\left( \tau \right) =Q_{\tau }$. Recall 
\begin{equation*}
Q_{\tau }=\frac{1}{B}\sum_{b=1}^{B}\mathbb{I}\left( Z_{N,T}^{\left( b\right)
}\leq c_{\tau }\right) ,
\end{equation*}%
and let $X_{N,T}^{\left( b\right) }=\mathbb{I}\left( Z_{N,T}^{\left(
b\right) }\leq c_{\tau }\right) $ for short. Note that, by similar passages
as in the proof of Theorem \ref{asy-max}%
\begin{eqnarray}
&&\mathbb{E}^{\ast }\left( X_{N,T}^{\left( b\right) }\right)  \label{exb} \\
&=&\mathbb{P}^{\ast }\left( Z_{N,T}^{\left( b\right) }\leq c_{\tau }\right)
=\dprod\limits_{i=1}^{N}\mathbb{P}^{\ast }\left( \omega _{i}^{\left(
b\right) }\leq c_{\tau }-\psi _{i,NT}\right)  \notag \\
&=&\exp \left( \sum_{i=1}^{N}\log \mathbb{P}^{\ast }\left( \omega
_{i}^{\left( b\right) }\leq c_{\tau }-\psi _{i,NT}\right) \right)  \notag \\
&=&\exp \left( \sum_{i=1}^{N}\log \left[ \mathbb{P}^{\ast }\left( \omega
_{i}^{\left( b\right) }\leq c_{\tau }\right) \left( 1-\frac{\mathbb{P}^{\ast
}\left( c_{\tau }-\psi _{i,NT}\leq \omega _{i}^{\left( b\right) }\leq
c_{\tau }\right) }{\mathbb{P}^{\ast }\left( \omega _{i}^{\left( b\right)
}\leq c_{\tau }\right) }\right) \right] \right)  \notag \\
&=&\exp \left( \sum_{i=1}^{N}\log \mathbb{P}^{\ast }\left( \omega
_{i}^{\left( b\right) }\leq c_{\tau }\right) -c_{0}\sum_{i=1}^{N}\mathbb{P}%
^{\ast }\left( c_{\tau }-\psi _{i,NT}\leq \omega _{i}^{\left( b\right) }\leq
c_{\tau }\right) \right)  \notag \\
&=&\exp \left( \sum_{i=1}^{N}\log \mathbb{P}^{\ast }\left( \omega
_{i}^{\left( b\right) }\leq c_{\tau }\right) \right) \exp \left(
-c_{1}\sum_{i=1}^{N}\psi _{i,NT}\right)  \notag
\end{eqnarray}%
for some positive, finite constants $c_{0}$ and $c_{1}$, and having used the
fact that $\psi _{i,NT}$ implies $\mathbb{P}^{\ast }\left( \omega
_{i}^{\left( b\right) }\leq c_{\tau }-\psi _{i,NT}\right) =\mathbb{P}^{\ast
}\left( \omega _{i}^{\left( b\right) }\leq c_{\tau }\right) -\mathbb{P}%
^{\ast }\left( c_{\tau }-\psi _{i,NT}\leq \omega _{i}^{\left( b\right) }\leq
c_{\tau }\right) $ to move from the fourth to the fifth line.

We now start by showing (\ref{st1}). It holds that%
\begin{eqnarray}
\left\vert \frac{Q_{\tau }-\left( 1-\tau \right) }{\sqrt{\tau \left( 1-\tau
\right) }}\right\vert  &=&\left\vert \frac{Q_{\tau }-\left( 1-\tau \right) }{%
\left( \mathcal{V}^{\ast }\left( Q_{\tau }\right) \right) ^{1/2}}\right\vert 
\frac{\left( \mathcal{V}^{\ast }\left( Q_{\tau }\right) \right) ^{1/2}}{%
\sqrt{\tau \left( 1-\tau \right) }}  \label{q-alpha} \\
&=&\left\vert \frac{Q_{\tau }-\mathbb{E}^{\ast }\left( Q_{\tau }\right) }{%
\left( \mathcal{V}^{\ast }\left( Q_{\tau }\right) \right) ^{1/2}}\right\vert 
\frac{\left( \mathcal{V}^{\ast }\left( Q_{\tau }\right) \right) ^{1/2}}{%
\sqrt{\tau \left( 1-\tau \right) }}+\left\vert \frac{\mathbb{E}^{\ast
}\left( Q_{\tau }\right) -\left( 1-\tau \right) }{\left( \mathcal{V}^{\ast
}\left( Q_{\tau }\right) \right) ^{1/2}}\right\vert \frac{\left( \mathcal{V}%
^{\ast }\left( Q_{\tau }\right) \right) ^{1/2}}{\sqrt{\tau \left( 1-\tau
\right) }}  \notag \\
&=&I+II,  \notag
\end{eqnarray}%
where $\mathcal{V}^{\ast }\left( Q_{\tau }\right) $ denotes the variance of $%
Q_{\tau }$ conditional on the sample, with 
\begin{equation*}
\mathcal{V}^{\ast }\left( Q_{\tau }\right) =\mathbb{E}^{\ast }\left(
X_{N,T}^{\left( b\right) }\right) \left[ 1-\mathbb{E}^{\ast }\left(
X_{N,T}^{\left( b\right) }\right) \right] .
\end{equation*}%
Let%
\begin{equation*}
B\frac{Q_{\tau }-\mathbb{E}^{\ast }\left( Q_{\tau }\right) }{\left( \mathcal{%
V}^{\ast }\left( Q_{\tau }\right) \right) ^{1/2}}=\sum_{b=1}^{B}\frac{%
X_{N,T}^{\left( b\right) }-\mathbb{E}^{\ast }\left( Q_{\tau }\right) }{%
\left( \mathcal{V}^{\ast }\left( Q_{\tau }\right) \right) ^{1/2}}%
=\sum_{b=1}^{B}\chi _{b},
\end{equation*}%
where $\chi _{b}$, conditionally on the sample, is \textit{i.i.d.} with zero
mean and unit variance and - being uniformly distributed - has a finite
moment generating function in a neighborhood of zero. Therefore, by the Koml%
\'{o}s-Major-Tusn\'{a}dy approximation (see e.g. Theorem 2.6.1 in \citealp{csorgo2014strong})
yields that, for each $B$, on a suitably enlarged probability space there
exists a standard Wiener process $\left\{ W_{B}\left( \left\lfloor
Bu\right\rfloor \right) ,0\leq u\leq 1\right\} $ whose distribution does not
depend on $B$ such that (conditional on the data)%
\begin{equation}
\sup_{0\leq u\leq 1}\left\vert \sum_{b=1}^{\left\lfloor Bu\right\rfloor
}\chi _{b}-W_{B}\left( \left\lfloor Bu\right\rfloor \right) \right\vert
=O_{a.s.}\left( \log B\right) .  \label{revesz}
\end{equation}%
Hence we have%
\begin{eqnarray}
\left\vert \frac{Q_{\tau }-\mathbb{E}^{\ast }\left( Q_{\tau }\right) }{%
\left( \mathcal{V}^{\ast }\left( Q_{\tau }\right) \right) ^{1/2}}\right\vert
&\leq &\frac{1}{B}\sup_{0\leq u\leq 1}\left\vert \sum_{b=1}^{\left\lfloor
Bu\right\rfloor }\chi _{b}\right\vert  \\
&\leq &\frac{1}{B}\sup_{0\leq u\leq 1}\left\vert W_{B}\left( \left\lfloor
Bu\right\rfloor \right) \right\vert +\frac{1}{B}\sup_{0\leq u\leq
1}\left\vert \sum_{b=1}^{\left\lfloor Bu\right\rfloor }\chi _{b}-W_{B}\left(
\left\lfloor Bu\right\rfloor \right) \right\vert   \notag \\
&=&O_{a.s.}\left( \sqrt{\frac{2\log \log B}{B}}\right) +O_{a.s.}\left( \frac{%
\log B}{B}\right) ,  \label{lilB}
\end{eqnarray}%
where the first rate comes from the Law of the Iterated Logarithm for Wiener
processes, and the second one from (\ref{revesz}). We now turn to studying $%
II$ in (\ref{q-alpha}), beginning with an estimate for (recall that the
sequence $\left\{ \omega _{i}^{\left( b\right) },1\leq i\leq N\right\} $ is
independent on the data)%
\begin{eqnarray*}
&&\left\vert \mathbb{E}^{\ast }\left( Q_{\tau }\right) -\left( 1-\tau
\right) \right\vert  \\
&\leq &\left\vert \left( \mathbb{P}\left( \omega _{i}^{\left( b\right) }\leq
c_{\tau }\right) \right) ^{N}-\left( 1-\tau \right) \right\vert +\left(
1-\tau \right) \left\vert \exp \left( -c_{1}\sum_{i=1}^{N}\psi
_{i,NT}\right) -1\right\vert .
\end{eqnarray*}%
Using equation (10) in \citet{hall1979rate} - with, in his notation, $%
x=-\log \left( -\log \left( 1-\tau \right) \right) $ - it holds that%
\begin{equation*}
\left\vert \left( \mathbb{P}\left( \omega _{i}^{\left( b\right) }\leq
c_{\tau }\right) \right) ^{N}-\left( 1-\tau \right) \right\vert \leq c_{0}%
\frac{1}{\log N}.
\end{equation*}%
Also, by a standard application of the Mean Value Theorem, equation (\ref%
{psi_rateas}) and Assumption \ref{asymptotics}, we have%
\begin{equation*}
\left\vert \exp \left( -c_{1}\sum_{i=1}^{N}\psi _{i,NT}\right) -1\right\vert
\leq \left\vert \sum_{i=1}^{N}\psi _{i,T}\right\vert =O_{a.s}\left(
N^{-\varepsilon }\right) ,
\end{equation*}%
for some $\varepsilon >0$. Hence we have%
\begin{equation}
\sqrt{\frac{B}{\log \log B}}\left\vert \mathbb{E}^{\ast }\left( Q_{\tau
}\right) -\left( 1-\tau \right) \right\vert =O_{a.s}\left( \sqrt{\frac{B}{%
\log \log B}}\frac{1}{\log N}\right) =o_{a.s}\left( 1\right) .
\label{mediaQ}
\end{equation}%
This also yields%
\begin{equation}
\left\vert \sqrt{\frac{\mathcal{V}^{\ast }\left( Q_{\tau }\right) }{\tau
\left( 1-\tau \right) }}-1\right\vert =O_{a.s}\left( \frac{1}{\log N}\right)
.  \label{varQ}
\end{equation}%
Combining (\ref{mediaQ}) with (\ref{lilB}), it holds that (for $I$ in
equation (\ref{q-alpha}))%
\begin{equation*}
I=O_{a.s}\left( \sqrt{\frac{\log \log B}{B}}\right) +O_{a.s}\left( \frac{%
\log B}{B}\right) .
\end{equation*}%
Further, as far as $II$ in (\ref{q-alpha}) is concerned, (\ref{mediaQ}) and (%
\ref{varQ}) yield%
\begin{equation*}
II=O_{a.s}\left( \frac{1}{\log N}\right) .
\end{equation*}%
Putting all together, it follows that%
\begin{equation}
\left\vert \frac{Q_{\tau }-\left( 1-\tau \right) }{\sqrt{\tau \left( 1-\tau
\right) }}\right\vert =O_{a.s}\left( \sqrt{\frac{\log \log B}{B}}\right)
+O_{a.s}\left( \frac{\log B}{B}\right) +O_{a.s}\left( \frac{1}{\log N}%
\right) ,  \label{revesz2}
\end{equation}%
whence the desired result follows from noting that, by the definition of $B$%
\begin{equation*}
\lim_{\min \left\{ B,N\right\} \rightarrow \infty }\sqrt{\frac{B}{\log \log B%
}}\frac{1}{\log N}=0.
\end{equation*}

Under the alternative, we write%
\begin{eqnarray*}
Q_{\tau } &=&\mathbb{E}^{\ast }\left( Q_{\tau }\right) +Q_{\tau }-\mathbb{E}%
^{\ast }\left( Q_{\tau }\right)  \\
&=&\frac{1}{B}\sum_{b=1}^{B}\mathbb{E}^{\ast }\left( X_{N,T}^{\left(
b\right) }\right) +\frac{1}{B}\sum_{b=1}^{B}\left( X_{N,T}^{\left( b\right)
}-\mathbb{E}^{\ast }\left( X_{N,T}^{\left( b\right) }\right) \right)  \\
&=&\mathbb{E}^{\ast }\left( X_{N,T}^{\left( 1\right) }\right) +\frac{1}{B}%
\sum_{b=1}^{B}\left( X_{N,T}^{\left( b\right) }-\mathbb{E}^{\ast }\left(
X_{N,T}^{\left( b\right) }\right) \right) =I+II.
\end{eqnarray*}%
We know from the proof of Theorem \ref{asy-max} that, under $\mathbb{H}_{A}$%
, $I=o_{a.s.}\left( 1\right) $. Moreover, note that, due to $X_{N,T}^{\left(
b\right) }$ being \textit{i.i.d.} across $1\leq b\leq B$%
\begin{equation*}
\mathcal{V}^{\ast }\left( \frac{1}{B}\sum_{b=1}^{B}\left( X_{N,T}^{\left(
b\right) }-\mathbb{E}^{\ast }\left( X_{N,T}^{\left( b\right) }\right)
\right) \right) =B^{-1}\mathcal{V}^{\ast }\left( X_{N,T}^{\left( 1\right)
}\right) \leq c_{0}B^{-1},
\end{equation*}%
a.s., and therefore, by the Law of the Total Variance, it also holds that 
\begin{equation*}
\mathcal{V}\left( \frac{1}{B}\sum_{b=1}^{B}\left( X_{N,T}^{\left( b\right) }-%
\mathbb{E}^{\ast }\left( X_{N,T}^{\left( b\right) }\right) \right) \right)
\leq c_{0}B^{-1};
\end{equation*}%
Lemma \ref{stout} then entails that $II=o_{a.s.}\left( 1\right) $. The
desired result now follows automatically.
\end{proof}

\begin{proof}[Proof of Theorem \protect\ref{ff-obs}]
The proof is very similar to that of Theorem \ref{asy-max}, and therefore we
only report its main arguments. Under $\mathbb{H}_{0}$, the result follows
immediately from Lemma \ref{psi-fm}. Under $\mathbb{H}_{A}$, using (\ref%
{alpha-ff}) it follows that, for $1\leq i\leq N$%
\begin{equation*}
\widehat{\alpha }_{i}^{FM}=\alpha _{i}+\beta _{i}^{\prime }\overline{v}+%
\overline{u}_{i}-\left( \widehat{\beta }_{i}-\beta _{i}\right) ^{\prime
}\lambda -\beta _{i}^{\prime }\left( \widehat{\lambda }-\lambda \right)
-\left( \widehat{\beta }_{i}-\beta _{i}\right) ^{\prime }\left( \widehat{%
\lambda }-\lambda \right) ,
\end{equation*}%
where recall that, by Lemma \ref{lambda-hA}%
\begin{equation*}
\widehat{\lambda }-\lambda =\frac{1}{N}\mathbf{S}_{\beta }^{-1}\mathbf{\beta 
}^{\prime }\mathbb{M}_{1_{N}}\mathbf{\alpha }+o_{a.s.}\left( T^{-1/2}\left(
\log T\right) ^{1+\epsilon }\right) .
\end{equation*}%
We consider, for simplicity, the case where $\alpha _{1}\neq 0$, and $\alpha
_{i}=0$ for $i\geq 2$. Other, more complicated cases, can be studied by the
same token, and we only add some discussion at the end of this proof. In
such a case, it holds that%
\begin{eqnarray*}
\widehat{\alpha }_{1}^{FM} &=&\alpha _{1}+\frac{1}{N}\mathbf{S}_{\beta }^{-1}%
\mathbf{\beta }^{\prime }\mathbb{M}_{1_{N}}\mathbf{\alpha }+r_{N,T}, \\
\widehat{\alpha }_{i}^{FM} &=&\frac{1}{N}\mathbf{S}_{\beta }^{-1}\mathbf{%
\beta }^{\prime }\mathbb{M}_{1_{N}}\mathbf{\alpha }+r_{N,T},
\end{eqnarray*}%
where the remainder term can be studied along the same lines as under the
null. Note that, under the case considered%
\begin{equation*}
\frac{1}{N}\mathbf{S}_{\beta }^{-1}\mathbf{\beta }^{\prime }\mathbb{M}%
_{1_{N}}\mathbf{\alpha =}\frac{1}{N}\frac{N-1}{N}\mathbf{S}_{\beta
}^{-1}\beta _{1}\alpha _{1}-\frac{1}{N^{2}}\mathbf{S}_{\beta
}^{-1}\sum_{i=1}^{N}\beta _{i}\alpha _{1}=O\left( \frac{1}{N}\right) ,
\end{equation*}%
and therefore $\widehat{\alpha }_{1}^{FM}=\alpha _{1}$ plus a negligible
remainder. The proof now is the same as the proof of Theorem \ref{asy-max}.
In conclusion, note that, under general alternatives, we need to guarantee
that%
\begin{equation}
\max_{1\leq i\leq N}\left\vert \alpha _{i}+\frac{1}{N}\mathbf{S}_{\beta
}^{-1}\mathbf{\beta }^{\prime }\mathbb{M}_{1_{N}}\mathbf{\alpha }\right\vert
>0,  \label{power-ff}
\end{equation}%
in order to be able to apply the arguments in the proof of Theorem \ref%
{asy-max}. However, this condition is always satisfied. Indeed, the only way
it cannot hold is if%
\begin{equation}
\alpha _{i}+\frac{1}{N}\mathbf{S}_{\beta }^{-1}\mathbf{\beta }^{\prime }%
\mathbb{M}_{1_{N}}\mathbf{\alpha }=0,  \label{power-ff-2}
\end{equation}%
for all $1\leq i\leq N$, which, in turn, requires under $\alpha _{i}=\alpha $
for all $1\leq i\leq N$. However, in such a case, $\mathbb{M}_{1_{N}}\mathbf{%
\alpha =0}$, and therefore the only way (\ref{power-ff-2}) can be satisfied
is if $\alpha _{i}=0$ for all $1\leq i\leq N$; but this cannot hold under
the alternative hypothesis.
\end{proof}

\begin{proof}[Proof of Theorem \protect\ref{ff-unobs}]
The proof is essentially the same as that of Theorem \ref{ff-obs}, \textit{%
mutatis mutandis}.
\end{proof}

\begin{proof}[Proof of Theorem \protect\ref{family}]
We henceforth omit, whenever possible, the statement that $\mathbb{H}%
_{0}^{\left( v\right) }$ is true, $1\leq v\leq W$, to save space. By
conditional independence across $v$ 
\begin{eqnarray*}
&&\mathbb{P}^{\ast }\left[ \text{reject any null hypothesis}\right]  \\
&=&1-\mathbb{P}^{\ast }\left[ Q_{1}\left( \tau \right) \geq 1-\tau -f\left(
B\right) ,Q_{2}\left( \tau \right) \geq 1-\tau -f\left( B\right) ,...\right]
.
\end{eqnarray*}%
We note some facts. Conditional on the sample, $Q_{v}\left( \tau \right) $
is \textit{i.i.d.} across the windows $1\leq v\leq W$; also, the Cram\'{e}r
condition holds%
\[
\mathbb{E}^{\ast }\exp \left( tI\left[ Z_{NT}^{\left( b\right) }\leq c_{\tau
}\right] \right) <\infty ,
\]%
for all $\left\vert t\right\vert <H$ and some $H>0$. Therefore, it holds
that 
\begin{eqnarray*}
&&1-\mathbb{P}^{\ast }\left[ Q_{1}\left( \tau \right) \geq 1-\tau -f\left(
B\right) ,Q_{2}\left( \tau \right) \geq 1-\tau -f\left( B\right) ,...\right] 
\\
&=&1-\left\vert 1-\mathbb{P}^{\ast }\left[ Q_{1}\left( \tau \right) <1-\tau
-f\left( B\right) \right] \right\vert ^{W}.
\end{eqnarray*}%
We will now estimate%
\begin{eqnarray*}
&&\mathbb{P}^{\ast }\left[ Q_{1}\left( \tau \right) <1-\tau -f\left(
B\right) \right]  \\
&=&\mathbb{P}^{\ast }\left[ B^{1/2}\frac{Q_{1}\left( \tau \right) -\mathbb{E}%
^{\ast }Q_{1}\left( \tau \right) }{\sqrt{\mathcal{V}^{\ast }Q_{1}\left( \tau
\right) }}<B^{1/2}\frac{1-\tau -\mathbb{E}^{\ast }Q_{1}\left( \tau \right)
-f\left( B\right) }{\sqrt{\mathcal{V}^{\ast }Q_{1}\left( \tau \right) }}%
\right] ;
\end{eqnarray*}%
we will use henceforth the short-hand notation%
\begin{eqnarray*}
S_{B} &=&B^{1/2}\frac{Q_{1}\left( \tau \right) -\mathbb{E}^{\ast
}Q_{1}\left( \tau \right) }{\sqrt{\mathcal{V}^{\ast }Q_{1}\left( \tau
\right) }}, \\
x &=&-B^{1/2}\frac{1-\tau -\mathbb{E}^{\ast }Q_{1}\left( \tau \right)
-f\left( B\right) }{\sqrt{\mathcal{V}^{\ast }Q_{1}\left( \tau \right) }},
\end{eqnarray*}%
and we note that, for large $N,T,B$, it holds that $x>0$ and $x=o\left(
B^{1/2}\right) $. Finally, let the distribution function of the standard
normal be denoted as $\Phi \left( \cdot \right) $; using Theorem 5.23 in
\citet{petrov1995}, it follows that 
\[
\mathbb{P}^{\ast }\left( S_{B}<-x\right) =\Phi \left( -x\right) \left( \exp %
\left[ -\frac{x^{3}}{B^{1/2}}\lambda \left( -\frac{x}{B^{1/2}}\right) \right]
\right) \left[ 1+O\left( \frac{x}{B^{1/2}}\right) \right] ,
\]%
where $\lambda \left( u\right) =\sum_{k=0}^{\infty }c_{k}u^{k}$ is the
so-called Cram\'{e}r function. Note that, by elementary arguments%
\[
\limsup_{B\rightarrow \infty }\lim_{L\rightarrow \infty
}\sum_{k=1}^{L}c_{k}\left( \frac{x}{B^{1/2}}\right) ^{k}<\infty ;
\]%
hence, it follows that  
\[
\lambda \left( -\frac{x}{B^{1/2}}\right) =c_{0},
\]%
where $c_{0}$\ is a finite and positive constant. We use the upper bound %
\[
\Phi \left( -x\right) \leq \exp \left( -\frac{1}{2}x^{2} - c_0\log x\right) ,
\]%
and that%
\[
x=c_{0}\tau ^{-1/2}B^{1/2}f\left( B\right) .
\]%
Using the above and omitting the log terms in the exponential, we have that, for large $B$%
\begin{eqnarray*}
&&\left\vert 1-\mathbb{P}^{\ast }\left[ Q_{1}\left( \tau \right) <1-\tau
-f\left( B\right) \right] \right\vert ^{W} \\
&\geq&\left\vert 1-c_{0}\exp \left( -c_{1}\tau ^{-1}B\left\vert f\left(
B\right) \right\vert ^{2}-c_{2}\tau ^{-3/2}B\left\vert f\left( B\right)
\right\vert ^{3}\right) \right\vert ^{W} \\
&=&\exp \left[ -c_{0}W\exp \left( -c_{1}\tau ^{-1}B\left\vert f\left(
B\right) \right\vert ^{2}-c_{2}\tau ^{-3/2}B\left\vert f\left( B\right)
\right\vert ^{3}\right) \right] .
\end{eqnarray*}%
Then we will have that 
\[
\mathbb{P}^{\ast }\left[ \text{reject any null hypothesis}\right]
\rightarrow 0\text{ \ \ a.s.,}
\]%
a.s. as long as%
\begin{eqnarray*}
&&W\exp \left( -c_{1}\tau ^{-1}B\left\vert f\left( B\right) \right\vert
^{2}-c_{2}\tau ^{-3/2}B\left\vert f\left( B\right) \right\vert ^{3}\right) 
\\
&=&\exp \left( \log W-c_{1}\tau ^{-1}B\left\vert f\left( B\right)
\right\vert ^{2}-c_{2}\tau ^{-3/2}B\left\vert f\left( B\right) \right\vert
^{3}\right) \rightarrow 0,
\end{eqnarray*}%
which is satisfied as long as (\ref{suff-family}) holds.
\end{proof}

\end{document}